\documentclass{article}
\usepackage{hyperref, graphicx}
\usepackage{amsmath}
\usepackage{amsfonts}
\usepackage{tocloft}
\usepackage{amssymb}
\usepackage{slashed}
\usepackage{yfonts}
\usepackage[T1]{fontenc}
\usepackage[utf8]{inputenc}
\usepackage[english]{babel}
\usepackage{amsthm}
\usepackage{tikz-cd}
\usetikzlibrary{nfold}
\usepackage{quiver}
\usepackage{mathabx}
\usepackage{stmaryrd}
\usepackage{braket}
\usepackage[most]{tcolorbox}
\usepackage{physics}
\usepackage{geometry}
\usepackage{csquotes}
\usepackage[
backend=biber,
style=numeric-comp,
sorting=none,
]{biblatex}
\newtheorem{theorem}{Theorem}[section]
\newtheorem{lemma}[theorem]{Lemma}
\newtheorem{corollary}[theorem]{Corollary}
\newtheorem{proposition}[theorem]{Proposition}

\usepackage{authblk}
    \allowdisplaybreaks

\hypersetup{
pdfstartview = {FitH},
}
\hypersetup{
	colorlinks=true,         
	linkcolor=blue,          
	citecolor=red,        
	urlcolor=blue            
}

\DeclareMathAlphabet{\mathpzc}{OT1}{pzc}{m}{it}

\makeatletter
\def\@font@info#1{}
\makeatother

\theoremstyle{remark}
\newtheorem{remark}{Remark}[section]
\newenvironment{rmk}{\begin{remark}}{\hfill$\Diamond$\end{remark}}

\newtheorem{example}[remark]{Example}

\theoremstyle{definition}
\newtheorem{definition}[theorem]{Definition}
\usepackage{mathtools}

\usepackage{mathrsfs}

\numberwithin{equation}{section}

\makeatletter
\newcommand{\ostar}{\mathbin{\mathpalette\make@circled\star}}
\newcommand{\make@circled}[2]{%
  \ooalign{$\m@th#1\smallbigcirc{#1}$\cr\hidewidth$\m@th#1#2$\hidewidth\cr}%
}
\newcommand{\smallbigcirc}[1]{%
  \vcenter{\hbox{\scalebox{0.77778}{$\m@th#1\bigcirc$}}}%
}
\makeatother

\newcommand{\beq}{\begin{eqnarray}}
\newcommand{\eeq}{\end{eqnarray}}

\def\g{\mathfrak g}

\def\h{\mathfrak h}

\def\G{\mathfrak{G}}
\def\C{\mathfrak{C}}

\def\bbZ{\mathbb{Z}}
\def\R{\mathbb{R}}
\def\bbC{\mathbb{C}}

\newcommand{\cA}{{\cal A}}
\newcommand{\cC}{{\cal C}}
\newcommand{\cD}{{\cal D}}

\newcommand{\cG}{{\cal G}}

\newcommand{\cB}{{\cal B}}

\newcommand{\cH}{{\cal H}}
\newcommand{\cM}{{\cal M}}
\newcommand{\cN}{{\cal N}}

\newcommand{\cV}{{\cal V}}

\makeatletter
\newsavebox{\@brx}
\newcommand{\llangle}[1][]{\savebox{\@brx}{\(\m@th{#1\langle}\)}%
  \mathopen{\copy\@brx\kern-0.5\wd\@brx\usebox{\@brx}}}
\newcommand{\rrangle}[1][]{\savebox{\@brx}{\(\m@th{#1\rangle}\)}%
  \mathclose{\copy\@brx\kern-0.5\wd\@brx\usebox{\@brx}}}
\makeatother

\newcommand{\id}{\operatorname{id}}

\begin{document}

\title{Combinatorial quantization of 4d 2-Chern-Simons theory II:\\ Quantum invariants of higher ribbons in $D^4$}
\author[1]{{ \sf Hank Chen}\thanks{hank.chen@uwaterloo.ca}\thanks{chunhaochen@bimsa.cn}}

\affil[1]{\small Beijing Institute of Mathematical Sciences and Applications, Beijing 101408, China}

\maketitle

\begin{abstract}
This is a continuation of the first paper of this series, where the framework for the combinatorial quantization of the 4d 2-Chern-Simons theory with an underlying compact structure Lie 2-group $\mathbb{G}$ was laid out. In this paper, we continue our quest and characterize additive module *-functors $\omega:\C_q(\mathbb{G}^{\Gamma^2})\rightarrow\mathsf{Hilb}$, which serve as a categorification of linear *-functionals (ie. a \textit{state}) on a $C^*$-algebra. These allow us to construct non-Abelian Wilson surface correlators $\widehat{\C}_q(\mathbb{G}^{P})$ on the discrete 2d simple polyhedra $P$ partitioning 3-manifolds. By proving its stable equivalence under 3d handlebody moves, these Wilson surface states extend to decorated 3-dimensional marked bordisms in a 4-disc $D^4$. This provides a definition of an \textit{invariant of framed oriented 2-ribbons} in $D^4$ from the data of a quantum 2-group $\C_q(\mathbb{G}^{\Gamma^2})$. We find that these 2-Chern-Simons-type 2-ribbon invariants are given by bigraded $\mathbb{Z}$-modules, similar to the lasagna skein modules of Manolescu-Walker-Wedrich.
\end{abstract}

\newpage

\tableofcontents

\newpage

\section{Introduction}
This paper is the second part of the series dedicated to the combinatorial quantization of the Hamiltonian 2-Chern-Simons theory. This essentially completes the analysis of \cite{Chen1:2025?}, and constructs the 2-ribbon invariants that one obtains from the underlying Wilson surface observables.

To set the stage, we introduce first the following notions.
We first recall the following well-known definitions (see eg. \cite{baez2004,Yetter:1993dh,Martins:2010ry,Chen:2012gz,Baez:2005sn,Bai_2013,Chen:2023integrable}).
\begin{definition}
    A \textbf{strict Lie 2-group} $\mathbb{G}=\mathsf{H}\xrightarrow{\mathsf{t}}G$ is the data of a pair $\mathsf{H},G$ of Lie groups, a Lie group homomorphism $\mathsf{t}:\mathsf{H}\rightarrow G$ and a smooth action $\rhd: G\rightarrow \operatorname{Aut}\mathsf{H}$ satisfying 
    \begin{equation*}
        t(g\rhd h) = gt(h)g^{-1},\qquad t(h)\rhd h' = hh'h^{-1}
    \end{equation*}
    for all $g\in G,~h,h'\in\mathsf{H}$. 

    A \textbf{Lie 2-algebra/$L_2$-algebra} $\G=\h\xrightarrow{\mu_1}\g$ is a graded vector space $\G=\h\oplus\g$ equipped with $n$-nary skew-symmetric brackets $\mu_n\in \operatorname{Hom}^{n-2}(\G^{\wedge 2},\G)$ with $1\leq n\leq 2$, satisfying the graded Leibniz rules 
    $$\mu_1(\mu_2(x, y)) = \mu_2(x,\mu_1(y)),\qquad \mu_2(\mu_1(y),y')=\mu_2(y,\mu_1(y'))$$ for all $x\in\g,~y,y'\in\h$,
    as well as the Koszul identities. We shall assign $\h$ a degree of $(-1)$, and $\g$ a degree of 0.
    
\end{definition}
The following "2-Lie theorem" is also well-known \cite{Chen:2012gz}.
\begin{theorem}
    There is a one-to-one correspondence between Lie 2-algebras and connected, simply-connected Lie 2-groups. The differential $\mu_1$ is integrated to $\mathsf{t}$.
\end{theorem}
\noindent Equivalently \cite{Porst2008Strict2A,baez2004}, $\mathbb{G}$ is a category internal to the category $\mathsf{LieGrp}$ of Lie groups, with surjective submersive source/target maps \cite{MACKENZIE200046,Chen:2012gz,Schommer_Pries_2011}
\begin{equation*}
    \mathsf{H}\rtimes G\underset{t}{\overset{s}{\rightrightarrows}} G ,\qquad s(h,g) = g,\quad t(h,g) = g\mathsf{t}(h),
\end{equation*}
and a unit section $\id_g = (1,g)$. This is the central perspective that we shall take for the rest of this paper.

\medskip

We say the Lie 2-algebra $\G$ is \textbf{balanced} \cite{Soncini:2014} iff it has equipped a graded-symmetric non-degenerate invariant pairing form $\langle-,-\rangle: \G^{\otimes2}\rightarrow \bbC[1]$ of degree-1; namely it only has support on $\g\otimes\h\oplus\h\otimes \g.$  The classical \textbf{2-Chern-Simons action} \cite{Soncini:2014,Song_2023} then reads
\begin{equation*}
    S_{2CS}[A,B] = \int_{M^4}\langle B,F_A-\frac{1}{2}tB\rangle,\qquad A\in\Omega^1(M^4,\g),\quad B\in\Omega^2(M^4,\h),
\end{equation*}
where $M^4$ is a smooth 4-manifold. This action is part of the \textit{derived} family of homotopy-Chern-Simons theories constructed from $L_\infty$-algebras in \cite{Jurco:2018sby,Ritter:2016}.

2-Chern-Simons theory has been analyzed thoroughly classically in the literature, including its Hamiltonian analysis \cite{Martins:2010ry,Mikovic:2016xmo} and its classical moduli space of 2-flat connections/2-holonomies \cite{Baez:2004in,schreiber2013connectionsnonabeliangerbesholonomy,Chen:2024axr,Sati:2008eg}. As informed by the Fock-Rosly approach \cite{Fock:1998nu}, its quantization should then begin with a graded Poisson structure on the \textit{categorified} moduli space. 

A model for such a quantization framework in the discrete combinatorial context was pinned down in the previous paper \cite{Chen1:2025?}. This led to the definition of the "quantum 2-graph states" $\C_q(\mathbb{G}^{\Gamma^2})$, which can be understood as the categorical/higher-dimensional version of the compact quantum group \cite{Woronowicz1988} on a lattice \cite{Alekseev:1994pa}. It was found that they form a Hopf cocategory (cf. \cite{DAY199799}) \textit{internal} to the measureable fields of Crane-Yetter \cite{Crane:2003gk,Yetter2003MeasurableC}, consistent with the categorical ladder proposal of Baez-Dolan \cite{Baez:1995xq} and Crane-Frenkel \cite{Crane:1993if,Crane:1994ty}.

\begin{rmk}\label{categorificatione}
    Here, by "categorification" we mean the promotion of $\bbC$-valued functions, for instance, to vector space-valued functions. This is why we explicitly work with the categorified version of $L^2$-spaces --- namely the Crane-Yetter measureable fields. This procedure is well-known \cite{Kapranov:1994,Baez1996HigherDimensionalAI}, specifically in the context of topological quantum field theories (TQFTs) and topological orders \cite{Atiyah:1988,lurie2008classification,KitaevKong_2012,Bullivant:2016clk,Kong:2020,Johnson-Freyd:2020usu,Bullivant:2021,Carqueville:2023aak,Delcamp:2023kew,Kong2024-vr,Carqueville:2016kdq,Carqueville:2017aoe,Davydov:2011kb}, but their physical significance to lattice gauge theory has only been noted recently \cite{Jacobson:2023cmr,Chen:2024ddr,Chen1:2025?}. Although higher structures are already known to be required to capture instantons/defects/anomalies in gauge field theory \cite{Carey_1997,Fiorenza2020TwistorialCI,Sati:2009ic,Fiorenza:2020iax} since around the turn of the century, they can be missed by a na{\"i}ve truncation of the degrees-of-freedom on a lattice. The goal of categorification is to \textit{re}capture these anomaly data,\footnote{Indeed, the need for a "derived/higher categorical geometry" in AKSZ/$L_\infty$-algebra models of field theories cannot be understated \cite{Johnson-Freyd:2013oea,costello_gwilliam_2016,Calaque:2021sgp}. See \cite{Kong:2011jf} for a review.} specifically in higher-dimensions, reminiscent of the Villain lattice construction \cite{Chevyrev:2024eng}. In the present context of 2-Chern-Simons theory, its higher homotopy anomalies (ie. the \textit{Postnikov classes} \cite{Kapustin2017,Chen:2022hct,Benini:2018reh}, which we will discuss a bit more in \textit{Remark \ref{weak2ribbons}} later) are known to an play important role for geometric \textit{string structures} \cite{Kim:2019owc,Waldorf2012ACO,Baez:2005sn,Schommer_Pries_2011,Soncini:2014,schreiber2013connectionsnonabeliangerbesholonomy,Sati:2008eg,Fiorenza:2013jz}.
\end{rmk}

The higher representation theory of the quantum categorical symmetries of the {2-Chern-Simons TQFT}, ie. $\operatorname{2Rep}(\mathbb{U}_q\G)$, was studied in \cite{Chen:2025?}. It was found that they exhibit data and properties that categorify the notion of \textit{ribbon tensor categories} \cite{etingof2016tensor,Etingof:2004,book-quasihopf,SHUM199457}, which are well-known to play a central role in the construction of quantum ribbon invariants in 3d \cite{Reshetikhin:1991tc,Turaev:1992,Reshetikhin:1990pr,Turaev+2016}. 

The goal of this paper is therefore to explain and construct the invariants of  higher-dimensional ribbons arising from 2-Chern-Simons TQFT. Towards this, we once again take inspiration from the seminal works of Alekseev-Grosse-Schmerus, now their second paper \cite{Alekseev:1994au}, and develop a higher-dimensional analogue of the Chern-Simons algebra on the standard graph associated to a compact punctured Riemann surface (Def. 12 in \cite{Alekseev:1994au}).

\subsection{Main results}
Starting from the quantum 2-graph states $\C_q(\mathbb{G}^{\Gamma^2})$ of \cite{Chen1:2025?} on a 2-simplex geometry $\Gamma^2$, we characterize additive measureable *-functors in the ambient 2-category $\mathsf{Meas}$ of Crane-Yetter measureable categories. These are categorical models for \textit{states} on a $C^*$-algebra. The main ingredient will be the following Yoneda embedding theorem in the \textit{infinite-dimensional} context.
\begin{theorem}
    (\ref{yoneda}.) There is a fully-faithful embedding $\C_q(\mathbb{G}^{\Gamma^2})\hookrightarrow \operatorname{Fun}_\mathsf{Meas}(\C_q(\mathbb{G}^{\Gamma^2}),\mathsf{Hilb})$. 
\end{theorem}
\noindent Due to the infinite-dimensional nature of measureable categories, this embedding is a priori \textit{not} an equivalence. These invariant *-functors are formalized by the notion of a \textbf{cointegral} for Hopf cocategories (see \S \ref{cointegral}).

These additive *-functors allow us to define the \textbf{non-Abelian Wilson surface states} $\widehat{\C}_q(\mathbb{G}^{\Gamma_P})$, where $\Gamma_P=\Gamma^2$ denotes a combinatorial triangulation of a simple 2d polyhedron $P$. By considering $P$ as a piecewise linear (PL) 2-manifold, we prove the invariance of $\widehat{\C}_q(\mathbb{G}^{\Gamma_P})$ under 2d Pachner moves (\textbf{Theorem \ref{invariance}}), which gives us the \textbf{2-Chern-Simons 2-algebra on the standard simple polyhedron} in \S \ref{3handles}.

This standard 2-algebra is then the central ingredient for the construction of the higher-ribbon invariants arising from 2-Chern-Simons theory. These are defined as \textit{monoidal} functors between certain \textit{double categories} \cite{Kerler2001,shulman2009framedbicategoriesmonoidalfibrations},
\begin{equation}
    \Omega: \LaTeXunderbrace{\operatorname{PLRib}'_{(1+1)+\epsilon}(D^4)}_{\text{geometry}}\rightarrow \LaTeXunderbrace{\widehat{\C}_q(\mathbb{G})}_{\text{algebra}},\label{g2ribboninvar}
\end{equation}
as a higher-categorical analogue of the quantum group ribbon invariants in Reshetikhin-Turaev TQFT \cite{Reshetikhin:1990pr,Reshetikhin:1991tc,Turaev:1992,Turaev+2016}. Here, the left-hand "geometry side" consist of the so-called \textbf{marked PL 2-ribbons}. These are 2-dimensional framed, oriented PL geometries, embedded in a PL 4-disc $D^4$, which are equipped with transverse boundary graphs and diffeomorphisms on top.

\begin{rmk}\label{diffbords}
    The work of \cite{douglas2016internalbicategories} establishes a framework in which one can model bordisms with diffeomorphisms on top of them as \textit{categories internal to $\mathsf{Mfld}$}. They called these the "$(n+1+\epsilon)$-dimensional bordisms"  $\operatorname{Bord}_{\langle n,n-1\rangle+\epsilon}$, where the "$\epsilon$" is supposed to indicate the diffeomorphisms on top of the $n$-bordisms and their $(n-1)$-boundaries. The definition of these PL 2-ribbons are based on a PL version of this construction --- they are categories {internal} to the PL manifolds $\mathsf{PLTop}$. This is the \textit{raison d'{\^e}tre} for working with \textit{internal} structures here --- the categorical types match exactly with the geometry; this is crucial for \S \ref{Gfuncmon} later.
\end{rmk}

These invariants $\Omega$ are therefore not only \textit{functorial} by construction, but also \textit{monoidal} against a certain connected summation operation between the PL 2-ribbons. Through the theory of handlebody decompositions \cite{matveev2007algorithmic}, this monoidality turned out to be central in the following.
\begin{theorem}
    (\ref{handlebodyinvariance}.) The \textbf{2-ribbon invariants of 2-Chern-Simons theory} $\Omega(\,_{B_1}P_{B_2})\in\widehat{\C}_q(\mathbb{G}^P)$ are invariant under handlebody moves (see fig. \ref{fig:handlemoves}) on the 2d simple polyhedron $P$.
\end{theorem}
\noindent By the stable equivalence result of \cite{Sakata2022-il}, this means that $\Omega(\,_{B_1}P_{B_2})$ can be interpreted as certain decorated stratified 3-manifolds \cite{hudson1969piecewise,Liu:2024qth} embedded in $D^4$. 

\medskip

Isomorphism classes of 2-Chern-Simons 2-ribbon invariants \eqref{g2ribboninvar} involve the \textit{smooth equivariant} cohomology. The cohomolgoy on the classifying space (2-stack) of the Lie 2-group $\mathbb{G}$ has been studied in various guises in, for instance, \cite{Angulo2024TheVE,Schommer_Pries_2011,Waldorf2012ACO,Nikolaus2011FOUREV}.
\begin{proposition}
    (\ref{bigradedQ}.) Isomorphism classes of 2-Chern-Simons 2-ribbon invariants $2\mathcal{CS}^\mathbb{G}_{q}(D^4)$ are parameterized by assignments of $\mathbb{G}$-equivariant cohomology classes in $H_\mathbb{G}(B\mathbb{G},\mathbb{Z})[t][q,q^{-1}]$ to marked PL 2-ribbons up to diffeomorphism. 
\end{proposition}
\noindent This result is interesting, as it seems to imply a close relation between $2\mathcal{CS}^\mathbb{G}_{q}(D^4)$ and another type of higher-tangle invariant that exists in the literature: the \textit{higher lasagna modules} of Manolescu-Walker-Wedrich \cite{Manolescu2022SkeinLM}, which are based on the derived, multiply-graded $\mathfrak{gl}_N$ Khovanov-Rozansky homology theory $\operatorname{KhR}^N$ \cite{Khovanov:2000,Khovanov:2006,Khovanov_2010,rouquier:hal-00002981,Rouquier2005CategorificationOS}. 

This may not as surprising as one may first think, since 2-Chern-Simons theory $S_{2CS}$ itself involves \textit{derived} fields and host Wilson surface operators that can end on knots \cite{Witten:2011zz}.\footnote{Furthermore, the gauge-field equations (1.1) in \cite{Gaiotto:2011} can be (mostly) reproduced by the fake-flatness $F_A-\mu_1B=0$ equation of motion in 2-Chern-Simons theory, by restricting to a 2-gauge sector of a certain field multiplet configuration $(A,B=0)\in \Omega^\bullet(M^4)\otimes\G$.} However, $2\mathcal{CS}^\mathbb{G}_{q}(D^4)$ do differ from the lasagna invariants $\mathcal{S}^{\mathfrak{gl}_N}_0(D^4)$ in a crucial manner; more details can be found in \S \ref{conclusion} and \S \ref{higherskein}. 

\medskip

We will also make use of the *-operations and the above Yoneda embedding result to define distinguished \textit{categorical} pairing forms from the geometry. They will play a central role in the notion of \textbf{reflection-positivity} for the corresponding 2-Chern-Simons 2-ribbon invariants $2\mathcal{CS}^\mathbb{G}_{q}(D^4)$.

\subsubsection*{Physical interpretations.}  Higher-gauge theory in general has been known to be deeply relevant to various fields of physics \cite{Chen:2022hct,Ritter:2016}, from quantum gravity \cite{Mikovic:2015hza,Mikovic:2011si,Girelli:2021khh,Baez:1995ph}, high-energy theory \cite{Baez:2002highergauge,Benini:2018reh,Cordova:2018cvg,Song_2023,Song:2021,Gaiotto:2014kfa}, condensed matter \cite{Wen:2019,Kong:2020wmn,Wang:2016rzy,Delcamp:2018kqc,Bullivant:2016clk,Bochniak_2021,Kapustin:2013uxa,Thorngren2015,Dubinkin:2020kxo}, to string theory \cite{Kim:2019owc,Sati:2009ic,Schreiber:2013pra}. As such, it is worthwhile to provide physical interpretations for some of our results. This will be expressed in {\color{purple}purple} in the following.

However, a prevailing slogan the author would like to emphasize here is the following:
\begin{center}
    \large{\emph{Gauge symmetries are internal, global symmetries are enriched.}}
\end{center}
A few comments in \textit{Remarks \ref{CSvsDW}, \ref{gaugedoubles}} will be made which highlight this slogan.
 
\subsection{Overview}
The outline of the paper is as follows. We will begin with a broad overview of the formal mathematical setup in \S \ref{prelim}. We will introduce the measureable categories of Crane-Yetter, definitions of categories/cocategories internal to a bicategory as well as the higher-categorical Hopf structures based on this internal model. This section serves as the foundation for the rest of this paper.

Then, in \S \ref{firstpaper}, we will give a concise but comprehensive review of the key concepts and results of the first paper \cite{Chen1:2025?}. Note that the language of \S \ref{prelim} is slightly different from that used in \cite{Chen1:2025?}, but they are equivalent; this will be explained clearly in \S \ref{classical2states} and \textit{Remark \ref{modelchange}}.

In \S \ref{objA}, we set out to pin down the combinatorial 2-simplex geometry underlying the 2-graph states $\phi\in\C_q(\mathbb{G}^{\Gamma^2})$. We show how the geometry (see figs. \ref{fig:interchanger}, \ref{fig:triple}) of 2d simple polyhedra $P$ can kept track of. {\color{purple}These 2-graph states $\phi$ serve as \textit{extended} operator insertions in discretized 2-Chern-Simons theory, and their operator products are governed abstractly by the braid relations \eqref{braid}.}

We will then prove the following two key results:
\begin{itemize}
    \item \S \ref{invarbdy}: invariance modulo boundary (\textbf{Theorem \ref{bdyinvar}}) --- namely that {\color{purple} the \textit{extended gauge charges} can be probed by ending the Wilson surfaces on boundaries \cite{Freidel:2020xyx,Freidel:2020svx,Ciambelli:2021nmv}}, and
    \item \S \ref{commbdy}: disjoint commutativity/braiding (\textbf{Theorem \ref{brmon}}) --- which is a {\color{purple} realization of the open-closed duality \cite{Kong:2011jf} between the Wilson surface sectors.} 
\end{itemize}

Categorical linear *-functionals on these 2-holonomy states are then studied in \S \ref{1morA}. The so-called "cone" functors are categorifications of the {\color{purple} quantum correlation functions between Wilson surface operators}. We completely characterize them within the ambient 2-category $\mathsf{Meas}$, and prove the Yoneda embedding.

Equipped with these states, we then move on to \S \ref{Gdecorations} where we first define the relevant geometry of \textit{marked} PL 2-ribbons (see figs. \ref{fig:2ribbontwist}, \ref{fig:sum}, \textbf{Proposition \ref{markedPL2skeins}}). The 2-ribbon invariants $\Omega$ \eqref{g2ribboninvar} are then defined in \S \ref{Gfuncmon}.  \S \ref{reflectionpositivity} treats the reflection-positivity/{\color{purple}unitarity} of $\Omega$ (see fig. \ref{fig:4-ball}).

The final section \S \ref{stableequivG2ribbons} is then dedicated to proving the invariance of $\Omega$ under stable equivalence/handlebody moves. The resulting decorated stratified 3-manifold can be interpreted as the {\color{purple} Hilbert space of 2-Chern-Simons Wilson surface states on a Cauchy slice}; see also \S \ref{alterfold} and figs. \ref{fig:handlecollar}, \ref{fig:cornersummation}.

In the conclusion \S \ref{conclusion}, we will frame the results of this paper in the larger context of categorical quantum algebras. In a companion work, we pin down a theory of categorical characters which will allow us to compute the 2-ribbon invariants constructed in this paper.

\medskip

The appendix will provide additional information. Specifically, \S \ref{prevworks} outlines the relation of 2-Chern-Simons 2-ribbon invariants to previous works in the literature. These include
\begin{enumerate}
    \item Chern-Simons standard graph algebra \cite{Alekseev:1994pa,Alekseev:1994au} (\S \ref{boundaryCS}),
    \item 2-tangles in 4-dimensions \cite{baez19982,BAEZ2003705,CARTER19971,Kharlamov:1993} (\S \ref{baezlangford}), and finally
    \item the higher lasagna skein modules \cite{Manolescu2022SkeinLM,Morrison2019InvariantsO4} (\S \ref{higherskein}).
\end{enumerate}
The idea that {\color{purple} higher-gauge theory is able to model codimension-2 defects} has been used in the condensed matter literature as well \cite{Pretko:2017fbf,Pretko:2020,Dubinkin:2020kxo,Lam:2023xng}.

\subsubsection*{Acknowledgments}
The author would like to thank Hao Zheng, Yilong Wang and Zhi-Hao Zhang for insightful discussions throughout the completion of this paper. This work is supported by the RFIS-I program of the National Science Foundation of China (NSFC), Grant Number: W2533012.

\section{Preliminaries}\label{prelim}
Suppose $X$ were a connected smooth Reimannian manifold equipped with a complete metric. Further, we will also assume $X$ is equipped with a Borel measure $\mu$, and let $\mathcal{U}\rightarrow X$ denote a corresponding $\mu$-measureable covering of Borel open sets. The central example is where $X$ is a locally compact topological/Lie group equipped with a Haar measure.

\subsection{Measureable fields and sheaves of Hermitian sections}
Recall the definition of a measureable field $H^X$ \cite{Crane:2003gk,Yetter2003MeasurableC,Baez:2012}.
\begin{definition}
    A \textbf{measureable field} $H^X$ over the measure space $(X,\mu)$ is the data of a family of Hilbert spaces $\{H_x\}_{x\in X}$ and the \textit{measureable sections} $\cM_H\subset \coprod_{x\in X}H_x$ such that
\begin{enumerate}
    \item the norm map $x\mapsto |\xi_x|_{H_x}$ is $\mu$-measureable for all $\xi\in\cM_H$,
    \item if $x\mapsto \langle \eta_x,\xi_x\rangle_{H_x}$ is $\mu$-measureable for all $\xi\in\cM_H$, then $\eta\in\cM_H,$ and
    \item $\cM_H$ is sequentially dense in $\coprod_{x\in X}H_x$.
\end{enumerate}
The collection of all measureable fields $H^X$ and bounded linear measureable operators $\phi:H^X\rightarrow H'^X$ (preserving the measureable sections) form the \textit{measureable category $\cH^X=\mathsf{Meas}_X$} of Crane-Yetter over $X.$ 
\end{definition}
\noindent We shall considerably leverage the theory of sheaves on smooth manifolds \cite{Fausk:2003,Kashiwara1990SheavesOM} in this paper.

\begin{rmk}\label{measureablefieldsassheaves}
    In the language of sheaves, the measureable category $\mathsf{Meas}_X$ over $(X,\mu)$ is equivalent to the category of sheaves of the so-called Hilbert \textit{$W^*$-modules} over $X$, where the $W^*$-algebra is given by the bounded functions $L^\infty(X,\mu)$. We are interested in better-behaved measureable fields in this paper here, however, for which we have access to \textbf{Proposition \ref{vectorbundles}} later. The reason will be clear in \S \ref{PLhomology}.
\end{rmk}

\medskip

One of the central results in \cite{Crane:2003gk,Yetter2003MeasurableC} is the construction of the 2-category $\mathsf{Meas}$ of measureable categories; we will recall its 1- and 2-morphisms in \S \ref{meas12mor}. A few more baisc facts about it is the following.
\begin{proposition}\label{basicfacts}
    Let $X,Y$ be measureable spaces and $\cH^X,\cH^Y$ the measureable categories on them.
    \begin{enumerate}
        \item The direct integral $\int_X^\oplus d\mu_X:\cH^X\rightarrow \mathsf{Hilb}$ is a $\bbC$-linear additive functor, which produces the Hilbert space $H^X\mapsto \int_X^\oplus d\mu_x H_x$ of $\mu$-almost everywhere (a.e.) equivalence classes of sections $\xi\in\cM_H$.
        \item $\mathsf{Meas}$ is symmetric monoidal with $\mathsf{Hilb}\simeq\cH^\emptyset$ as the monoidal unit.
        \item There are equivalences $\cH^{X\times Y}\simeq\cH^X\times \cH^Y$.
    \end{enumerate}
\end{proposition}
\begin{proof}
    These are Thms. 27 and 50 in \cite{Yetter2003MeasurableC}, respectively. The equivalence in the third statement is given by 
    \begin{equation}
        \operatorname{pr}_X^*(-|_X)\otimes \operatorname{pr}^*_Y(-|_Y): \mathsf{Meas}_X\times \mathsf{Meas}_Y\xrightarrow{\sim}\mathsf{Meas}(X\times Y),\label{factorizable}
    \end{equation} 
    where $X\xleftarrow{\operatorname{pr}_X} X\times Y\xrightarrow{\operatorname{pr}_Y} Y$ are the projections of measureable spaces and $\mathsf{Meas}_X\xleftarrow{-|_X}\mathsf{Meas}_X\times\mathsf{Meas}_Y\xrightarrow{-|_Y}\mathsf{Meas}_Y$ are the restriction functors on measureable fields.
\end{proof}
\noindent We will use the third statement freely throughout this paper.

\medskip

Similar to \cite{TRENTINAGLIA2010750}, we shall restrict to better-behaved collection of Hilbert fields.
\begin{definition}\label{hermitian}
    Suppose $X$ admits a $\mu$-measureable cover $\mathcal{U}\rightarrow X$ (ie. we have a Borel measureable algebra on $X$). The \textbf{measureable sheaves of (finite-rank) Hermitian sections} $\cV^X\subset\cH^X$ over $(X,\mu)$ is the full additive measureable subcategory consisting of measureable fields $H^X$ such that its direct integral over $U\in\mathcal{U}$,
\begin{equation*}
    \Gamma_c(H^X): U\mapsto \int_U^\oplus d\mu_x H_x,\qquad U\in\mathcal{U}
\end{equation*} 
defines a coherent sheaf of locally finitely-generated free projective $C(X)$-modules.
\end{definition}

By the classical Serre-Swan theorem \cite{Serre:1955,Swan1962VectorBA}, we can view objects in $\cV^X$ as Hermitian vector bundles (more correctly, \textit{coherent sheaves}) over $(X,\mu).$
\begin{proposition}\label{vectorbundles}
    There is a forgetful functor $\cV^X \rightarrow \operatorname{Bun}_\bbC(X)$ sending a sheaf of sections $\Gamma_c(H^X)$ to its underlying complex vector bundle $H^X$ over $X$. 
\end{proposition}
\noindent Alternatively, $\cV^X\subset\cH^X$ can be understood as the full measureable subcategory which admits a forgetful functor into $\operatorname{Bun}_\bbC(X)$. As $\operatorname{Bun}_\bbC(X)$ is additive and exact, so is $\cV^X$.

Let $\mathsf{Meas}_\text{Herm} \subset\mathsf{Meas}$ denote the full 2-subcategory of measureable sheaves of Hermitian sections (and their completions) $\cV^X$.

\subsection{(Co)Categories internal to 2-categories}
We consider \textit{strict} categories $C$ \textit{internal} to $\mathsf{Meas}_\text{herm}$. This is a "strictified" version of the notion of a \textbf{category object in a 2-category $\cC$} (with pushouts and pullbacks).
\begin{definition}\label{internalcats}
    A \textbf{category $C$ internal to $\cC$} is a strict category object in a bicategory $\cC$ with pushouts and pullbacks (such as $\cC=\mathsf{Meas}$). It consists of the data:
\begin{itemize}
    \item a pair of objects $C_1,C_0\in\cC$,
    \item a pair of \textit{fibrant} 1-morphisms $s,t: C_1\rightarrow C_0$ in $\cV$ called the \textit{source/target}, and their pullback $C_1\,_t\times_sC_1$,
    \item a 1-morphism $\circ: C_1\,_t\times_sC_1\rightarrow C_1$ in $\cV$, called the \textit{composition law}, and
    \item a 1-morphism $\eta:C_0\rightarrow C_1$, called the \textit{unit}, such that
    \begin{enumerate}
        \item the composition law $\circ$ is strictly associative: the 2-morphism
\begin{tikzcd}
	{C_1\times_{C_0}C_1\times_{C_0}C_1} & {C_1\times_{C_0}C_1} \\
	{C_1\times_{C_0}C_1} & {C_1}
	\arrow["{\id\times \circ}", from=1-1, to=1-2]
	\arrow["{\circ\times \id}"', from=1-1, to=2-1]
	\arrow["\cong", shorten <=8pt, shorten >=8pt, Rightarrow, from=1-1, to=2-2]
	\arrow["\circ", from=1-2, to=2-2]
	\arrow["\circ"', from=2-1, to=2-2]
\end{tikzcd}
is invertible, 
        \item $\circ,1$ satisfy strict unity: for each $f\in C_1$ with $s(f) = x$ and $t(f)=y$, we have invertible 2-morphisms $1_y\circ f\cong f \cong f\circ 1_x$,
        \item the invertible compositional unitors and associators satisfy 
        \begin{enumerate}
            \item the exchange equation (which we call the \textit{interchange law}), 
            \item the left- and right-pentagon equations, and
            \item the left-, middle- and right-triangle equations,
        \end{enumerate}
        on the pullbacks $C_1^{[n]}=C_1\times_{C_0}C_1\times_{C_0}\dots\times_{C_0}C_1$. 
\end{enumerate}
\end{itemize}
A \textbf{cocategory $D$ internal to $\cC$} is a strict category object in $\cC^\text{op}$. It is equipped with \textit{cofibrant} 1-morphisms $u,v: D_0\rightarrow D_1$, a strict counit $\epsilon:D_1\rightarrow D_0$ and a strictly coassociative cocomposition law $\Delta_v: D_1\rightarrow D_1~_v\times_u D_1$ along the pushout.
\end{definition}
\noindent More details can be found in \cite{douglas2016internalbicategories}. Keep in mind that internal categories do not have cocompositions, and cocategories do not have compositions.\footnote{Note a category object in $\mathsf{Cat}$, the bicategory of categories, is a \textit{double category}; see Def. 10 of \cite{Ehresmann1963} and \S 12 of \cite{book-2cats}, and also \cite{Brown,Grandis2004}.}

A (strict) functor $F:C\rightarrow D$ of categories internal to $\cC$ is of course a pair of 1-morphisms $F_{i}:C_i\rightarrow D_i$ for $i=0,1$, equipped with invertible 2-morphisms
\[\begin{tikzcd}
	{C_1} && {D_1} \\
	& \cong \\
	{C_0} && {D_0}
	\arrow["{F_1}", from=1-1, to=1-3]
	\arrow["{s_C}"', shift right, curve={height=6pt}, from=1-1, to=3-1]
	\arrow["{t_C}", shift left, curve={height=-6pt}, from=1-1, to=3-1]
	\arrow["{s_D}"', shift right, curve={height=6pt}, from=1-3, to=3-3]
	\arrow["{t_D}", shift left, curve={height=-6pt}, from=1-3, to=3-3]
	\arrow["{F_0}", from=3-1, to=3-3]
\end{tikzcd},\qquad F(\circ)\cong \circ(F\times F),\qquad F_1(\eta)\cong \eta_{F_0}\]
which ensures that $F$ commutes with the fibrant source/target maps and the composition. 


\begin{rmk}
The insistence on working with \textit{internal} categories, as opposed to \textit{enriched} categories, may at first appear strange to some seasoned readers in higher-categorical algebras. However, internal categories have recently seen explicit applications in geometry and algebraic quantum field theory \cite{douglas2016internalbicategories,Bunk_2025}, specifically in the study of bordism categories with extra structure.
\end{rmk}


 \begin{rmk}\label{truncation}
Write $\cV= \mathsf{Meas}_\text{herm}$ the 2-category of measureable coherent sheaves, and let $\operatorname{Cat}_{\cV},\operatorname{Cocat}_{\cV}$ denote the collection of \textit{additive} categories/cocategories {internal} to $\cV$, respectively. A(n additive) co/category object $C$ internal to $\cV$ can be viewed as a double category \cite{Miranda:2025}, whose vertical 1-cells and 2-cells are given by measureable sheaf morphisms; see also \textit{Remark \ref{doublecocat}} later.  If $C$ were a category object in the full (2,1)-subcategory $\pi_{\leq  2}\cV\subset\cV$ consisting of only invertible 2-morphisms, then all of its 2-cells and \textit{vertical} 1-cells are invertible. The 2-truncation $\pi_{<2}C$ --- given by for instance taking the isomorphism classes of sheaves $\pi_{<2}C=[C]$ (see \S \ref{PLhomology}) --- is then an ordinary additive category. 
 \end{rmk}

\subsection{Internal Hopf categories}\label{inthopfcat}
We now define the notion of internal Hopf (co)categories that we shall use, which is heavily inspired by the frameworks of trialgebras \cite{Pfeiffer2007} and Hopf (op-)algebroids \cite{DAY199799}.

\medskip

Suppose $\cV$ is symmetric monoidal, with a monoidal unit object $I\in\cV$. As an abuse of notation, we will also denote by $I$ its discrete category $I\rightrightarrows I$ internal to $\cV$.
\begin{definition}\label{internalhopf}
    Let $(\cV,\times,I)$ be a symmetric monoidal 2-category.
    \begin{itemize}
        \item A \textbf{Hopf monoidal category $\cH$ in $\cV$} is a Hopf algebra object in $\operatorname{Cat}_{\cV}$. Namely, it is equipped with the following internal functors:
    \begin{enumerate}
        \item the \textit{product} $\otimes: \cH\times\cH\rightarrow\cH$ (with a unit $\iota\in \cH$),
        \item the strictly monoidal \textit{coproduct} $\Delta:\cH\rightarrow\cH\times\cH$ (with a counit $\epsilon: \cH\rightarrow I$), and
        \item the strictly op-comonoidal op-monoidal \textit{antipode} $S:\cH\rightarrow \cH^{\text{m-op},\text{c-op}}$,
    \end{enumerate}
    as well as internal natural transformations
    \begin{enumerate}
        \item the \textit{associators} $a^\otimes: \otimes\circ (\otimes\times1_\cH) \Rightarrow \otimes\circ(1_\cH\times\otimes)$ and \textit{unitors} $r^\otimes: (-\otimes \iota)\Rightarrow 1_\cH,~\ell^\otimes: (\iota\otimes -)\rightarrow 1_\cH$ satisfying the strict pentagon and triangle axioms, 
        \item the \textit{coassociators} $a^\Delta: (\Delta\times 1_\cH)\circ\Delta\Rightarrow (1_\cH\times\Delta)\Rightarrow\Delta$ and \textit{counitors} $r^\Delta:( \epsilon\times 1_\cH)\circ\Delta \Rightarrow1_\cH ,~ \ell^\Delta: (1_\cH\times \epsilon)\circ\Delta\Rightarrow 1_\cH$ satisfying the strict copentagon and cotriangle axioms,
        \item the invertible \textit{bimonoidal natural transformations}
        \begin{equation*}
            \Delta \circ \otimes \cong (1_\cH\times \sigma\times 1_\cH) \circ (\otimes\times\otimes)\circ \Delta
        \end{equation*}
        \item the \textit{antipode relations}
        \begin{equation*}
            \otimes\circ (S\times 1_\cH)\circ\Delta \cong \iota\otimes\epsilon\cong \otimes\circ(1_\cH\times S)\circ\Delta,
        \end{equation*}
    \end{enumerate}
    such that these internal natural transformations are mutually coherent.
    \item We say a Hopf monoidal category internal to $\cV$ is \textbf{strict} iff the above internal natural transformations are invertible and only have identity components.
    \item A \textbf{(strict) Hopf comonoidal cocategory in $\cV$} is a (strict) Hopf monoidal category in $\cV^\text{op}$. 
    \item We say $\cH$ is \textit{additive} if both of its objects and morphisms have $\cV$-internal direct sum biproducts, and all of its Hopf internal structures are additive functors/natural transformations.
    \end{itemize}
\end{definition}

As mentioned in \textit{Remark \ref{weakcoherence}}, there are of course {lax} versions of the above, where the coherence 2-cells above are not necessarily invertible. We will not need this much generality, even for the quantization of weak 2-Chern-Simons theory. We will make several brief remarks throughout this paper which explains how the Postnikov associator of $\mathbb{G}$ modifies our results.

\begin{rmk}
    Generally, (co)algebras in $\cV$ have a (co)composition law as well as a (co)monoidal product, which together satisfy the (co)interchange law. These are common structures in bicategories and 2-groups \cite{douglas2016internalbicategories,Baez:2003fs,Chen:2012gz}. It is worth emphasizing that Hopf cocategories do \textit{not} have a composition law for its morphisms.
\end{rmk}

\begin{rmk}\label{weakcoherence}
    We shall call a lax (Hopf monoidal) category object $C$ in $\cV$, whose invertible structural coherence morphisms (including those for the composition) are not necessarily concentrated at the identity component, a \textbf{(Hopf monoidal) $\cV$-pseudocategory}. $\mathsf{LieGrp}$-pseudocategories, internal to the bicategory $\mathsf{LieGrpd}$ of Lie groupoids, was examined in \cite{Ferreira:2015,Bunk_2025}. 
\end{rmk}

In the following, we will recall how combinatorial 2-Chern-Simons theory, based on a strictly associative structure Lie 2-group $\mathbb{G}$, gives rise to the structure of strict\footnote{In the case of non-associative \textit{smooth} 2-groups \cite{Schommer_Pries_2011}, we obtain instead Hopf psuedo-co/categories, but the coherence morphisms remain invertible.} Hopf co/monoidal intenral co/categories on the lattice. The co/monoidal co/associator morphisms of these Hopf co/categories receive contributions directly from the Postnikov anomaly mentioned in \textit{Remark \ref{categorificatione}}.

\section{A comprehensive overview of the first paper}\label{firstpaper}
Let us begin with a brief overview of the first paper, following the more formal perspective of the above section. We shall mainly focus on the central players: the 2-graph states $\C(\mathbb{G}^{\Gamma^2})$ and the 2-gauge transformations $\mathbb{U}\G^{\Gamma^1}$ on a lattice $\Gamma$. We will also state without proof some of their structural results that will be useful later; the interested reader is directed towards \cite{Chen1:2025?} for the proofs. 

The following was obtained by discretizing 2-holonomies of 2-connections \cite{Baez:2004in,schreiber2013connectionsnonabeliangerbesholonomy,Chen:2024axr}.
\begin{definition}\label{2holdef}
    Denote by $\Gamma^2$ a simply-connected 2-truncated topological simplicial complex. Objects of the 2-functor 2-category $F\in\operatorname{2Fun}(\Gamma^2,B\mathbb{G})$ are called \textbf{2-holonomies}, denoted $\mathbb{G}^{\Gamma^2}$, which consist of maps $F:\Gamma^2\rightarrow B\mathbb{G}$ satisfying the \textbf{fake-flatness condition}
\begin{equation*}
    t(b_f) = h_{\partial f},\qquad ~ ~F:(e,f)\mapsto (h_e,b_f)\in \mathsf{H}\rtimes G.
\end{equation*}
\begin{enumerate}
    \item The 1-morphisms/psuedonatural transformations $\eta: F\Rightarrow F'$ are called \textbf{2-gauge transformations}, and they act by horizontal conjugation
\begin{equation*}
    (h'_e,b'_f) = \operatorname{hAd}_{(a_v,\gamma_e)}^{-1}(h_e,b_f),\qquad \eta:(v,e)\mapsto (a_v,\gamma_e)\in\mathsf{H}\rtimes G
\end{equation*}
via the decorated 1-simplices $\mathbb{G}^{\Gamma^1}$.
\item The 2-morphisms/modifications $m:\eta\Rrightarrow \eta'$ are called \textbf{secondary gauge transformations}, and they act by vertical conjugation
\begin{equation*}
    (a_v',\gamma_e') = \operatorname{vAd}_{m_v}^{-1}(a_v,\gamma_e),\qquad m: v\mapsto m_v \in \mathsf{H}.
\end{equation*}
\end{enumerate}
\end{definition}
\noindent In \S \ref{2simplex}, we will set up the geometry such that $\Gamma^2$ can be seen as the combinatorial triangulation of a stratified PL 2-(pseudo)manifold.

\begin{tcolorbox}[breakable]
\subsubsection*{Slight foray into measure theory.}
Let $\mathbb{G} = \mathsf{H}\xrightarrow{\mathsf{t}}G$ be compact; namely it is a locally compact Hausdorff Lie groupoid and $G$ itself is compact.
\begin{definition}\label{2grouphaar}
    A \textbf{Haar measure} $\mu$ on $\mathbb{G}$ is a Radon measure equipped with a \textit{disintegration} (cf. \cite{Pachl_1978,Baez:2012}) $\{\nu^a\}_{a\in G}$ along the source map $s:\mathbb{G}\rightarrow G$ such that
    \begin{enumerate}
        \item the family $\{\nu^a\}_{a\in G}$ is a Haar system (cf. \cite{Williams2015HaarSO}), and
        \item the pushforward measure $\sigma = \mu\circ s^{-1}$ is an Haar-Radon measure on $G$.
    \end{enumerate}
    We say $\mu$ is an \textbf{invariant Haar measure} if the family $\{\nu^a\}_{a\in G}$ is $G$-equivariant and if $\sigma$ is an invariant measure on $G$.
\end{definition}
\noindent Though Haar systems on Lie groupoids are not unique \cite{Williams2015HaarSO}, we have the following analogue of Haar measures on ordinary compact Lie groups.

\begin{proposition}\label{haarunique}
    The Haar measure on compact connected Lie 2-groups $\mathbb{G}$, if it exists, is unique up to equivalence.
\end{proposition}
\begin{proof}
    By \textbf{Definition \ref{2grouphaar}} , the uniqueness of disintegrations \cite{Pachl_1978} (see also Lemma 2.3 of \cite{ACKERMAN_2016}) states that $\nu$ is unique on all points of continuity, which by compactness is the entire Lie 2-group. Additionally, since the pushforward $\sigma=\mu\circ s^{-1}$ is required to be a Lie group Haar measure for $G$, which we know is unique up to equivalence for compact $G$, the result follows.
\end{proof}

Given $\Gamma^2$ is finite, there is an induced invariant Haar measure on $\mathbb{G}^{\Gamma^2}$ denoted by
\begin{equation*}
    d\mu_{\Gamma^2}\big(\{(h_e,b_f)\}_{(e,f)}\big) = \prod_{e\in\Gamma^1}d\sigma(h_e)\prod_{f:e\rightarrow \in\Gamma^2}d\nu^{h_e}(b_f),
\end{equation*}
where $\sigma = \mu\circ s^{-1}$ and $f$ is a face with source edge $e$. Similarly, we can also define an invariant Haar measure on $\mathbb{G}^{\Gamma^1}$,
    \begin{equation*}
    d\mu_{\Gamma^1}\big(\{(a_v,\gamma_e)\}_{(a,e)}\big) = \prod_{v\in\Gamma^0}d\sigma(a_v)\prod_{e:v\rightarrow \in\Gamma^1}d\nu^{a_v}(\gamma_e).
\end{equation*}

We will assume that the Haar measure $\mu$ is Borel: namely all $\mu$-measureable subsets are open in the smooth topology of $\mathbb{G}$.
\end{tcolorbox}

\begin{rmk}\label{tildeChaar}
    We will show that an invariant Haar measure $\mu$ equips the 2-graph states  with a \textit{Hopf cocategorical cointegral} (see \S \ref{cointegral}). In analogy with Hopf algebras \cite{RADFORD19921,Delvaux2006ANO,Radford:1976}, this should have several significant structural implications for Hopf categories, some of which have been mentioned in \cite{Chen:2025?}.
\end{rmk}

\subsection{Geometric 2-graph states}\label{classical2states}
Recall $\cV^X\subset\mathsf{Meas}_X$ is the full monoidal subcategory of measureable sheaves of Hermitian sections over $X$, and $\mathsf{Meas}_\text{herm}\subset \mathsf{Meas}$ is the corresponding full 2-subcategory over the site $\mathsf{Mfld}$ of smooth manifolds (equipped with a measure).

Objects of $\cV$ are measureable sheaves of Hermitian sections $\cV^X$ over $X\in \mathsf{Mfld}$. We shall leverage the measure $\mu$ to redefine the regularity of $\cV^X.$
\begin{definition}\label{geometric2graphs}
    A \textbf{geometric 2-graph state} $\phi$ is an object in the full monoidal subcategory $\C(\mathbb{G}^{\Gamma^2})\subset\cV^X$ over $X=(\mathbb{G}^{\Gamma^2},\mu_{\Gamma^2})$, consisting of those measureable sheaves of smooth Hermitian sections $\Gamma_c(H^X)$. Namely, they are sheaves of countably-generated Hilbert $L^2(X,\mu_{\Gamma^2})$-modules.

    Moreover, if $\Gamma=v$ is a single vertex, then $\C_q(\mathbb{G}^v) \simeq \mathsf{Hilb}$ is trivial. We equip $\C(\mathbb{G}^{\Gamma^2})$ with a unit $\eta:\mathsf{Hilb}\rightarrow \C(\mathbb{G}^{\Gamma^2})$ represented by the trivial line bundle $\underline{\bbC}$ over $X=(\mathbb{G}^{\Gamma^2},\mu_{\Gamma^2})$.
\end{definition}
\noindent The separability condition is natural from the physical point of view, but it was not necessary in \cite{Chen1:2025?,Chen:2025?}. It will also not strictly be necessary in this paper, but it shall be important for computations down the line.

\begin{proposition}
    If $\Gamma,\Gamma'$ are disjoint 2-graphs, then there are equivalences $\C(\mathbb{G}^{\Gamma\coprod\Gamma'})\simeq\C(\mathbb{G}^\Gamma\times\mathbb{G}^{\Gamma
    })\simeq \C(\mathbb{G}^{\Gamma})\times\C(\mathbb{G}^{\Gamma'})$ as measureable categories.
\end{proposition}
\noindent This is immediate from the third statement in \textbf{Proposition \ref{basicfacts}}, which concerns only the \textit{external} structure of $\C(\mathbb{G})$ as a measureable category. 

Internally, $\mathbb{G}$ itself has equipped source/target maps $s,t:\mathsf{H}\rtimes G \rightarrow G$, for which $G$ is equipped with the pushforward Haar measure $\sigma= \mu\circ s^{-1}$. These structure maps then induce pullback/inverse image functors $s^*,t^*:\C(G^{\Gamma^1})\rightarrow \C\big((\mathsf{H}\rtimes G)^{\Gamma^2}\big)$ of measureable sheaves \cite{Kashiwara1990SheavesOM,Baez:2012}. 

Crucially, we require $s,t$ to be surjective submerions \cite{Schommer_Pries_2011,MACKENZIE200046},\footnote{We will also require $s,t$ to induce maps of classifying (2-)stacks $B\mathbb{G}\rightarrow BG$. We will need this in \S \ref{nonabeliansurface} and \S \ref{bigradedQ}.} whence the induced pullbacks are strict cofibrant. Thus they admit a left-section functor $\varepsilon: \C\big((\mathsf{H}\rtimes G)^{\Gamma^2}\big)\rightarrow \C(G^{\Gamma^1})$ satisfying 
\begin{equation*}
    \varepsilon\circ s^* = \id_{\C(G^{\Gamma^1})},\qquad \varepsilon\circ t^* = \id_{\C(G^{\Gamma^1})},
\end{equation*}
which serves as the cocompositional unit on $\C(\mathbb{G}^{\Gamma^2})$. 

\begin{rmk}\label{doublecocat}
     It is useful to organize the 2-graph states by leveraging the notion of a \textit{double cocategory} \cite{Kerler2001}, where the "external/internal" structures are placed vertically/horizontally.\footnote{We can always do this for (co)categories $C$ internal to a bicategory $\cV$ which admits a 2-functor to $\mathsf{Cat}$ that preserves pullbacks and pushouts; see \textit{Remark \ref{modelchange}} later.} More precisely,  for $\phi,\phi'\in\C((\mathsf{H}\rtimes G)^{\Gamma^2})$ we write
\begin{equation}\begin{tikzcd}
	{\phi_1} & {\phi_2} \\
	{\phi_1'} & {\phi_2'}
	\arrow[""{name=0, anchor=center, inner sep=0}, "\psi", "\shortmid"{marking}, from=1-1, to=1-2]
	\arrow["{U_1}"', from=1-1, to=2-1]
	\arrow["{U_2}", from=1-2, to=2-2]
	\arrow[""{name=1, anchor=center, inner sep=0}, "{\psi'}"', "\shortmid"{marking}, from=2-1, to=2-2]
	\arrow["u", shorten <=4pt, shorten >=4pt, Rightarrow, from=0, to=1]
\end{tikzcd}\label{dudesquare},\end{equation}
where the vertical arrows $U_1,U_2,u$ are measureable morphisms and the horizontal \textit{co}arrows $\psi,\psi'\in \C(G^{\Gamma^1})$ are 1-holonomy states satisfying
\begin{equation*}
    s^*\psi = \phi_1,\qquad t^*\psi = \phi_2,\qquad s^*\psi' = \phi'_1,\qquad t^*\psi' = \phi_2'.
\end{equation*}
\end{rmk}

\subsubsection{Measureable functors and measureable natural transformations}\label{meas12mor}
To proceed, we first recall the notion of measureable functors and measureable natural transformations \cite{Yetter2003MeasurableC,Crane:2003gk,Baez:2012}.
\begin{definition}
    A {\bf measureable functor} $F:\cH^X\rightarrow \cH^Y$ between measureable categories $\cH^X,\cH^Y$ is a family $\{f_y\}_{y\in Y}$ of measures on $X$, together with a field $F$ of Hilbert spaces on $Y\times X$, such that 
    \begin{enumerate}
        \item the map $y\mapsto f_y(A)$ is measureable for all measureable subsets $A\subset X$, and
        \item $f_y(X\setminus\operatorname{cl}(\operatorname{supp}_yF))=0$ where $\operatorname{supp}_yF = \{x\in X\mid F_{y,x}\neq 0\}$.
    \end{enumerate}
    For $H^X\in\cH^X$, the target measureable field $F(H^X)\in\cH^Y$ is given by a direct integral
        \begin{equation*}
            (FH)_y = \int_X^\oplus df_y(x) F_{y,x}\otimes H_x.
        \end{equation*}

        The composition $F\circ G:\cH^X\rightarrow\cH^Z$ of measureable functors is given by the $Z$-family $\{(fg)_z\}_z$ of measures,
\begin{equation*}
    (fg)_y = \int_X df_z(y) g_y,
\end{equation*}
and the field of Hilbert spaces 
\begin{equation*}
    (F\circ G)_{z,x} = \int_Y^\oplus dk_{z,x}(y)F_{z,y}\otimes G_{y,z}
\end{equation*}
where $k$ is the {\it $f,g$-disintegration} measure \cite{Pachl_1978} satisfying
\begin{equation}
    \int_X d(fg)_z(x)\int_Y dk_{z,x}(y)F(y,x) = \int_Ydf_z(y) \int_X dg_y(x) F(y,x),\qquad \forall~ F\in L^0(Y\times X).\label{disintegration}
\end{equation}

    The identity functor $1_{\cH^X}$ is the dirac measure $\{\delta_x\}_{x\in X}$ and the rank-1 field $(1_{\cH^X})_{x,x'} = \bbC$.
\end{definition}
\noindent Note that not all tensor products of sections in $F_{y,-},H$ will define a section of $FH^X$. Only those which, for every $y\in Y$, that give rise to $L^2$-sections over $X$ will.

We also have the following notion, from Def. 48 of \cite{Yetter2003MeasurableC}.
\begin{definition}\label{measnat}
    A measureable natural transformation $\beta:(F,f)\Rightarrow (G,g): \cH^X\rightarrow \cH^Y$ is the data of a field of $g$-essentially bounded linear operators $\beta: F\rightarrow G$ such that on each component $H^X\in\cH^X$ we have a map
    \begin{equation*}
        F_y = \int^\oplus_X df_y(x) F_{y,x} \mapsto \int_X^\oplus dg_y(x)\sqrt{\frac{d\tilde f_y(x)}{dg_y(x)}} \id_{H_x}\otimes \beta_{y,x}(F_{y,x}),\qquad \forall y\in Y,
    \end{equation*}
    where $\tilde f_y$ is the dominated component of $f_y=\tilde f_y+\hat f_y$ which is absolutely continuous with respect to $g_y$.
\end{definition}
\noindent The 2-category $\mathsf{Meas}$ of measureable categories was constructed by Yetter, and it is in fact \textit{symmetric monoidal} with the identity $\cH^\emptyset \simeq \mathsf{Hilb}$; see Thm. 50 in \cite{Yetter2003MeasurableC}.

\begin{proposition}
    Two measureable functors $(F,f),(G,g): \cH^X\rightarrow \cH^Y$ are isomorphic iff (i) the underlying measures $f,g$ are equivalent $f\ll g,~g\ll f$ and (ii) the field of operators $\beta$ is invertible. 
\end{proposition}
\begin{proof}
    This is immediate from \textbf{Definition \ref{measnat}}.
\end{proof}
\noindent We say $F,G$ are \textbf{unitarily} isomorphic iff they are isomorphic and $\beta$ is in addition a field of unitary operators.

Note \textbf{Definition \ref{measnat}} says that the 2-category $\mathsf{Meas}$ is 2-enriched in measureable fields, similar to how, in the finite-dimensional case, $\mathsf{2Hilb}$ is 2-enriched in $\mathsf{Hilb}$ \cite{Baez1996HigherDimensionalAI,Ganter:2006}.

\subsubsection{2-gauge transformations}\label{2gauge} We now turn to the 2-gauge transformations acting on $\C(\mathbb{G}^{\Gamma^1})$. These are parameterized by the so-called \textit{decorated 1-graphs}, which are maps $\Gamma^1\rightarrow\mathbb{G}$ that assign Lie 2-group elements to edges in $\Gamma$, $$\zeta=\left\{(v\xrightarrow{e}v')\mapsto (a_v\xrightarrow{\gamma_e}a_{v'})\right\}_{(v,e)},\qquad \mathsf{t}(\gamma_e) = a_{v}^{-1}a_{v'}.$$ 
\begin{definition}\label{additive2gaugesymms}
    Denote by $\mathbb{U}\G^{\Gamma^1}$ the additive monoidal category internal to $\mathsf{Meas}$ additively generated\footnote{Here, by "additively generated" we mean that every objects in $\mathbb{U}\G^{\Gamma^1}$ is a direct sum  of homogeneous elements (cf. \cite{SOZER2023109155}), which is given by $\zeta\in \big(\mathbb{U}_q\G^{\Gamma^1}\big)^\text{hom} = \mathbb{G}^{\Gamma^1}$. This will be made more precise in a follow up work.} by 2-gauge parameters/decorated 1-graphs $\mathbb{G}^{\Gamma^1}$ equipped with fibrant source/target maps
    \begin{equation*}
        \tilde s,\tilde t:\mathbb{U}\G^{\Gamma^1}\rightrightarrows\mathbb{U}\g^{\Gamma^0},\qquad \zeta= a_v\xrightarrow{\gamma_e}a_{z'}\iff \begin{cases}
            \tilde s(\zeta) = a_v \\ 
            \tilde t(\zeta) = a_{v'},
        \end{cases}
    \end{equation*}
    and a unit section $\tilde\eta: a_v\mapsto \id_{a_v}$ given by the groupoid unit in $(\mathsf{H}\rtimes G)^{\Gamma^1}$. 
\end{definition}

The way these decorated 1-graphs act on the decorated 2-graphs $\mathrm{z}=(h_e,b_f)\in\mathbb{G}^{\Gamma^2}$ is through the \textit{inverse horizontal conjugation} action,
\begin{equation*}
    \operatorname{hAd}_\zeta^{-1}:  (h_e,b_f)\mapsto \zeta^{-1} \cdot (h_e,b_f)\cdot \zeta,\qquad \zeta=(a_v\xrightarrow{\gamma_e}a_{v'}).
\end{equation*}
Since the 2-graph states can be viewed as sections of Hermitian vector bundles $H^X\rightarrow X$ over $X=(\mathbb{G}^{\Gamma^2},\mu_{\Gamma^2})$, we can construct the pull-back bundle $(\Lambda_\zeta H)^X=(\operatorname{hAd}_\zeta^{-1})^*H^X$ along $\operatorname{hAd}^{-1}_\zeta$.

In \cite{Chen1:2025?}, this pullback $(\Lambda_\zeta H)^X$ was used in order to realize the 2-gauge transformations $\Lambda_\zeta$ concretely as bounded linear operators $U_\zeta$. For the purposes of this paper, however, we shall instead describe 2-gauge transformations directly as a measureable functor form the get-go. 

\medskip

Recall the notion of a \textit{direct image functor} of sheaves \cite{Kashiwara1990SheavesOM}.
\begin{definition}\label{2gaugetransfo}
    Let $X=(\mathbb{G}^{\Gamma^2},\mu_{\Gamma^2})$. A \textbf{2-gauge transformation} on $\C(\mathbb{G}^{\Gamma^2})$ is,  for each $\zeta\in\mathbb{U}\G^{\Gamma^1}$, an additive measureable invertible endofunctor $\Lambda_\zeta: \C(\mathbb{G}^{\Gamma^2})\rightarrow \C(\mathbb{G}^{\Gamma^2})$ given by the direct image functor $(\operatorname{hAd}^{-1}_\zeta)_*$ of sheaves along the horizontal conjugation automorphism $\operatorname{hAd}^{-1}_\zeta:X\rightarrow X$, such that there are identifications
    \begin{equation}
        s^*(\Lambda_\zeta\phi) = \Lambda_{\tilde s\zeta}(s^*\phi),\qquad t^*(\Lambda_\zeta\phi) =\Lambda_{\tilde t\zeta}  (t^*\phi),\qquad \forall~ \zeta\in\mathbb{U}\G^{\Gamma^1},~\phi\in\C(\mathbb{G}^{\Gamma^2}) \label{stequiv}
    \end{equation}
    against the cofibrant cosource/cotarget maps $s^*,t^*$ on the 2-graph states. Moreover, the counit is $\mathbb{U}\G^{\Gamma^1}$-invariant, $\epsilon(\Lambda_\zeta\phi)=\Lambda_{\tilde\eta}\epsilon(\phi)\cong\epsilon(\phi)$.
\end{definition}

In other words, $\Lambda$ determines $\C(\mathbb{G}^{\Gamma^2})$ as a measureable $\mathbb{U}\G^{\Gamma^1}$-module category,
\begin{equation*}
    \Lambda: \mathbb{U}\G^{\Gamma^1}\times \C(\mathbb{G}^{\Gamma^2})\rightarrow \C(\mathbb{G}^{\Gamma^2}),
\end{equation*}
which by \eqref{stequiv} is internal to $\mathsf{Meas}_\text{herm}$. It is crucial to emphasize here that the "morphisms layer" in $\Lambda$, as written here, are not populated by the 1-cells in $\G^{\Gamma^1}$ (ie. the decorated edges $(\mathsf{H}\rtimes G)^\text{edges}$), but instead by the monoidality witness/module associators $\alpha^\Lambda_{\zeta,\zeta'}: \Lambda_\zeta\circ\Lambda_{\zeta'}\Rightarrow\Lambda_{\zeta\cdot\zeta'}$ of 2-gauge transformations. 

\begin{rmk}\label{modelchange}
To treat the decorated 1-edges as 1-cells in $\mathbb{U}\G^{\Gamma^1}$, we recall the 2-truncation \textit{Remark \ref{truncation}} for co/categories internal to $\cV$. By treating $\pi_{<2}\C(\mathbb{G}^{\Gamma^2}),~\pi_{<2}\mathbb{U}\G^{\Gamma^1}$ as additive, \textit{$\mathsf{Meas}$-enriched} categories in this way, we see that the the corresponding 2-gauge transformations understood as an action functor $$\Lambda: \pi_{<2}\mathbb{U}\G^{\Gamma^1}\rightarrow \operatorname{Aut}_\mathsf{Cat}\big(\pi_{<2}\C(\mathbb{G}^{\Gamma^2})\big)$$ for which the 2-gauge transformations are equipped with the following structure 
\begin{equation}\begin{tikzcd}[ampersand replacement=\&]
	{\pi_{<2}\C(\mathbb{G}^{\Gamma^2})} \&\& {\pi_{<2}\C(\mathbb{G}^{\Gamma^2})}
	\arrow[""{name=0, anchor=center, inner sep=0}, "{\Lambda_{a_v}}", curve={height=-18pt}, from=1-1, to=1-3]
	\arrow[""{name=1, anchor=center, inner sep=0}, "{\Lambda_{a_{v'}}}"', curve={height=18pt}, from=1-1, to=1-3]
	\arrow["{\Lambda_{\gamma_e}}", shorten <=5pt, shorten >=5pt, Rightarrow, from=0, to=1]
\end{tikzcd},\qquad \zeta = [a_v\xrightarrow{\gamma_e}a_{v'}]\in\pi_{<2}\mathbb{U}\mathbb{G}^{\Gamma^1}.\label{2gtmodule}\end{equation}
Now if we replace $\pi_{<2}\C(\mathbb{G}^{\Gamma^2})$ with some other category, such as a finite linear semisimple one $\cD\in\mathsf{2Vect}$, then we obtain finite 2-representations of $\pi_{<2}\mathbb{U}\G^{\Gamma^1}$ as studied in \cite{Chen:2025?}.
\end{rmk}

\begin{tcolorbox}[breakable]
    \subsubsection*{Measureable functors and sheaves of bounded linear operators.}The way that this definition is related to the sheaves of bounded operators $U_\zeta$ used in \cite{Chen1:2025?,Chen:2025?} is through Prop. 46 of \cite{Baez:2012}.
\begin{proposition}\label{pullbackmeas}
    All measureable automorphisms on a measureable category $\cH^X$ over $(X,\mu)$ are measureably naturally isomorphic to one induced by pulling back a measureable map $f:X\rightarrow X$.
\end{proposition}
Each automorphism $\Lambda_\zeta,~\zeta\in\mathbb{U}\G^{\Gamma^1}$ is thus measureably naturally isomorphic to one induced by pulling back the smooth measureable automorphism $\operatorname{hAd}_\zeta:X\rightarrow X$ on $X=(\mathbb{G}^{\Gamma^2},\mu_{\Gamma^2})$. The inverses of the operators $(U_\zeta^{-1})^\phi=(\operatorname{hAd}_\zeta^*)^\phi: \phi\rightarrow \phi|_{\operatorname{hAd}_\zeta-}$ are precisely those used in \cite{Chen1:2025?}.
\begin{definition}\label{2gauregularity}
    We say the 2-gauge transformations $\Lambda$ are \textbf{regular} iff the operators $\zeta\mapsto U_\zeta$ define measureable sheaves of (essentially) bounded linear operators over $(\mathbb{G}^{\Gamma^1},\mu_{\Gamma^1})$.
\end{definition}
\end{tcolorbox}


\begin{rmk}\label{secondarygaugetransfos}
    The module associators $\alpha^\Lambda$ are induced from invertible modifications $m:\operatorname{hAd}^{-1}\circ \operatorname{hAd}^{-1}\Rrightarrow \operatorname{hAd}^{-1}$ in the 2-functor 2-groupoid $\mathbb{G}^{\Gamma^2} = \operatorname{2Fun}(\Gamma^2,B\mathbb{G})$ describing the 2-holonomies. In the context of higher-gauge theory, these modifications are known as \textit{secondary gauge transformations} \cite{Kim:2019owc,Schenkel:2024dcd,Chen:2024axr,Chen1:2025?}. These can be ignored when $\mathbb{G}$ is strict, but they have non-identity components when $\mathbb{G}$ has a weak associators $\tau$. In forming the 2-truncation \textit{Remark \ref{modelchange}} they descend to crucial structures for the 2-gauge transformations $\Lambda$.
\end{rmk}

\subsubsection{Locality of states and gauge transformations}\label{locality}
Now a crucial feature of any lattice gauge theory is \textit{locality.} This is the notion that the data attached to the lattice, be it states or gauge transformations, should commute if they have disjoint support. In order to express this notion, we first define the so-called localized states and 2-gauge transformations.
\begin{definition}\label{local2graphstates}
    Let $(e,f)=e\xrightarrow{f}e_f\in\Gamma^2$ denote a 2-graph with source edge $e$. The \textbf{2-graph state localized at $(e,f)$} corresponding to $\phi\in\C(\mathbb{G}^{\Gamma^2})$ is defined by the measureable field $\phi_{(e,f)}$ whose stalk Hilbert spaces are given by
    \begin{equation*}
        (\phi_{(e,f)})_{\{(h_{e'},b_{f'})\}_{(e',f')}} = \chi^{[2]}_{(e,f)}\phi_{\{(e',f')\}_{(e',f')}},
    \end{equation*}
    where $\chi^{[2]}_{(e,f)}$ is the characteristic measure on ${\Gamma^2}$ supported at the face ${(e,f)}$. As a sheaf of smooth sections, $\phi_{(e,f)}$ is the restriction sheaf of $\phi$ along the inclusion $(e,f)\hookrightarrow \Gamma^2$. 
\end{definition}
\noindent More precisely, the restriction sheaf is the direct image of the induced pullback $\mathbb{G}^{\Gamma^2}\rightarrow \mathbb{G}^{(e,f)}$.

With these localized 2-graphs states, the geometry of the 2-graphs become apparent. If we let $\Delta$ denote the pullback measureable field of (group/groupoid) multiplication $\cdot_{h,v}$ in $\mathbb{G}$, such that we have, in Sweedler notation, an isomorphism of stalks 
\begin{equation*}
    (-\otimes-)\Delta(\phi)_{\mathrm{z},\mathrm{z}'} = \bigoplus (\phi_{(1)}^{h,v})_{\mathrm{z}}\otimes (\phi_{(2)}^{h,v})_{\mathrm{z}'}\cong \phi_{\mathrm{z}\cdot_{h,v}\mathrm{z}'},\qquad \mathrm{z},\mathrm{z'}\in\mathbb{G}
\end{equation*}
for all $\phi\in\C(\mathbb{G})$, then we can promote this coproduct to $\mathbb{G}^{\Gamma^2}$ in accordance with the geometry:
\begin{equation*}
    \Delta_{h,v}(\phi_{(e,f)}) = \begin{cases}
        \displaystyle \bigoplus (\phi_{(1)}^{h,v})_{(e_1,f_1)}\times (\phi_{(2)}^{h,v})_{(e_2,f_2)},&; (e,f) = (e_1,f_1)\cup_{h,v}(e_2,f_2)\\
        \phi_{(e_1,f_1)}\times \phi_{(e_2,f_2)} &; (e_1,f_1)\cap (e_2,f_2)=\emptyset
    \end{cases}
\end{equation*}
where $\cup_{h,v}$ are horizontal/vertical 2-graph gluing laws displayed in fig. \ref{fig:2-graph}. In the case where the 2-graphs $(e_1,f_1),(e_2,f_2)$ are disjoint, $(e,f)$ is interpreted as their disjoint union and the coproduct is grouplike/cocommutative.

\begin{figure}[h]
    \centering
    \includegraphics[width=0.8\linewidth]{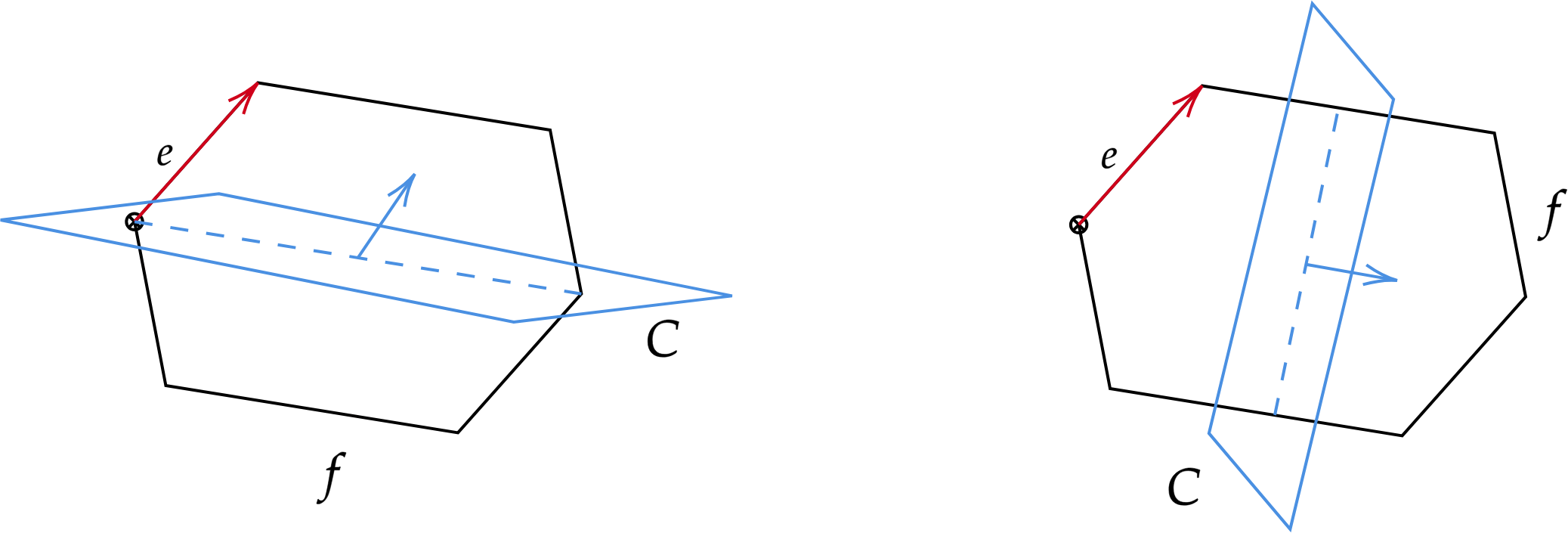}
    \caption{The two ways in which a local 2-graph $(e,f)$ can be decomposed into two 2-graphs, depending on how $(e,f)$ intersects an embedded 2-cell $C$ in the 3d manifold $\Sigma$. The left denotes $\cup_h,$ in which the normal vectors of $C$ are locally tangent to the source edge $e$ of $f$ around $v=s(e)$, while the right denotes $\cup_v$, where the normal vectors are perpendicular to $e$.}
    \label{fig:2-graph}
\end{figure}

We emphasize here that there are \textit{two} coproduct operations hidden in the symbol "$\Delta$", which correspond to the horizontal or the vertical labels $h,v$. These coproducts are required to satisfy the cointerchange law
\begin{equation*}
    (\Delta_h \times \Delta_h)\circ \Delta_v \cong (1\times \sigma\times 1)\circ (\Delta_v\times\Delta_v)\Delta_h
\end{equation*}
on $\C(\mathbb{G}^{\Gamma^2})$, which can be seen to arise from the geometry of \textit{triple intersections} of 2-cells in $\Sigma$. We shall in the following abbreviate $\Delta_{h,v}$ as $\Delta$ when no confusion is possible; explicit details can be found in \cite{Chen1:2025?}.

Similarly to for the 2-gauge transformations, it also inherits its notion of locality from the underlying geometry, this time of the \textit{1-graphs}. Like the 2-graphs states, this is captured by the coproducts $\tilde\Delta$ on $\mathbb{U}_q\G^{\Gamma^1}$. 
\begin{definition}
    Let $(v,e)=v\xrightarrow{e}v_f\in\Gamma^1$ denote a 1-graph with source vertex $v$. The \textbf{2-gauge transformation localized at $(v,e)$} corresponding to $\Lambda$ is a norm-smooth assignment
    \begin{equation*}
        \zeta\mapsto \chi_{(v,e)}^{[1]}\Lambda_\zeta,\qquad \zeta\in\mathbb{G}^{\Gamma^1}
    \end{equation*}
    of measureable direct image endofunctors on $\C(\mathbb{G}^{\Gamma^2})$, where $\chi^{[1]}_{(v,e)}$ is the characteristic measure on ${\Gamma^1}$ supported at the edge ${(v,e)}$.    
\end{definition}
\noindent In contrast to the 2-graph states, the way local 2-gauge transformations stack geometrically are dictated by its \textit{products} --- for homogeneous elements $\zeta,\zeta'\in\mathbb{U}\G^{\Gamma^1}$, we have
\begin{equation*}
      \zeta^{(v,e)}\zeta'^{(v',e')} = \begin{cases}
          (\zeta\cdot \zeta')^{(v,e)} &; (v',e')=(v,e) \\ 
          \zeta^{(v,e)}\circ\zeta'^{(v',e')} &; v' = t(e) \\ 
          0&; \text{otherwise}
      \end{cases},
\end{equation*}
where $\cdot,\circ$ denotes the group/horizontal and gorupoid/vertical composition of 2-gauge parameters $\mathbb{G}^{\Gamma^1}$.

\medskip

Recall from \textbf{Definition \ref{2gaugetransfo}} that the 2-gauge transformation operation $\Lambda$ makes $\C(\mathbb{G}^{\Gamma^2})$ into a $\mathbb{U}\G^{\Gamma^1}$-module. It is n fact a \textit{bimodule} over $\mathbb{U}\G^{\Gamma^1}$, equipped with a natural measureable natural isomorphism called the \textit{bimodule associator}
\begin{equation*}
    (\alpha^\bullet)_{\zeta,\zeta'}^{\phi}: \phi\bullet(\zeta\cdot\zeta') \xrightarrow{\sim}(\phi\bullet\zeta)\bullet\zeta',
\end{equation*}
which has only components on the identity due to the strict associativity of $\mathbb{G}$. This bimodule structure is related to the 2-gauge transformations though
\begin{equation}
    \zeta^{-1}\bullet \phi \bullet\zeta = \Lambda_\zeta\phi,\qquad\forall~ \phi\in\C(\mathbb{G}^{\Gamma^2}),\quad \zeta\in\mathbb{U}\G^{\Gamma^1}.\label{regrep}
\end{equation}
The coproduct $\tilde\Delta$ on $\mathbb{U}\G^{\Gamma^1}$, on the other hand, is instead induced from the consistency of this bimodule structure against the tensor product $\otimes$ of 2-graph states,
\begin{equation}
    \otimes \big((\phi\times\phi')\bullet  {\tilde\Delta_\zeta}\big) \xrightarrow{\sim} (\phi\otimes\phi')\bullet \zeta,\label{derivation}
\end{equation}
where $\phi,\phi'\in\C(\mathbb{G}^{\Gamma^2})$ and $\zeta\in\mathbb{U}\G^{\Gamma^1}$. The geometric interpretation of this so-called "derivation property" \eqref{derivation} is shown in fig. \ref{fig:derivation}.

\begin{figure}
    \centering
    \includegraphics[width=0.75\linewidth]{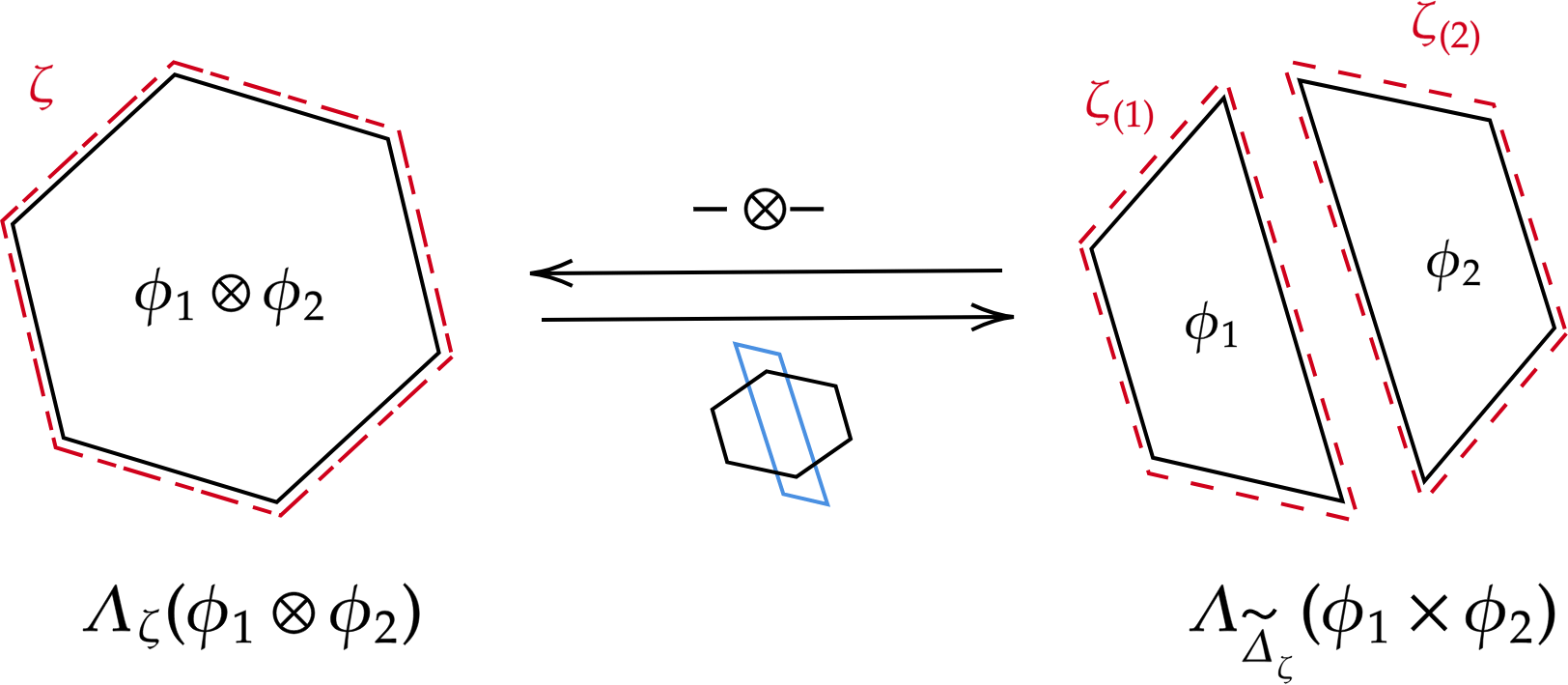}
    \caption{The graphical representation of the derivation property \eqref{derivation}, which implements the geometric consitency between the product $\otimes$ and the 2-gauge transformation action $\Lambda$.}
    \label{fig:derivation}
\end{figure}

Algebraically, \eqref{derivation} implies that $\C_q(\mathbb{G}^{\Gamma^2})$ is a monoidal $\mathbb{U}_q\G^{\Gamma^1}$-module category, since it gives the $\bullet$-module \textit{tensorator}; see \S \ref{observ} later, as well. 
\begin{rmk}
    The reason why \eqref{derivation} is called the "derivation property" is the following. One categorical level down, the same algebraic condition 
\begin{equation*}
    (\psi\psi')\bullet\zeta = \mu\big((\psi\otimes\psi')\tilde\Delta_\zeta\big) = (\psi\zeta)\psi' + \psi(\psi'\zeta)
\end{equation*}
is precisely the Leibniz rule for the derivation action of $\zeta\in U\g$ on functions $\psi,\psi'\in C(G)$ of a compact Lie group $G$.
\end{rmk}

We can now state the central characterization theorem proven in \cite{Chen1:2025?}.
\begin{theorem}
    Let ${\bf C}=\mathsf{LieGrp}\subset\mathsf{Mfld}$ denote the site of Lie groups. 
    \begin{itemize}
        \item The 2-graph states $\C(\mathbb{G}^{\Gamma^2})\in\mathsf{Hopf}(\operatorname{Cocat}_{\cV^{\bf C}})$ is a strict symmetric Hopf comonoidal cocategory internal to $\cV^{\bf C}$. 
        \item Given \textbf{Definition \ref{2gauregularity}}, the \textit{regular} 2-gauge transformations $\mathbb{U}\G^{\Gamma^1}\in \mathsf{Hopf}(\operatorname{Cat}_{\cV^{\bf C}})$ define a strict cosymmetric Hopf monoidal category internal to $\cV^{\bf C}$. 
    \end{itemize}
\end{theorem}
\noindent Here, "(co)symmetric" refers to the (co)monoidal (co)product. This theorem led to the following definitions.
\begin{definition}\label{catcordring}
Suppose $\Gamma$ is a single PL 2-disc, consisting of a single face bounded by an edge loop based at a vertex. 
    \begin{enumerate}
        \item We call $\C(\mathbb{G}^{\Gamma^2})= \C(\mathbb{G})$ the \textbf{categorical coordinate ring} of $\mathbb{G}$.
        \item We call $\mathbb{U}\G^{\Gamma^1} = \mathbb{U}\G$ the \textbf{categorical universal enveloping algebra} of $\G$.
    \end{enumerate}
\end{definition}
\noindent We emphasize here that this name and notation for $\mathbb{U}\G$ is just suggestive: while $\C(\mathbb{G})$ was concretely constructed, $\mathbb{U}\G$ was specified indirectly through the 2-gauge transformation on it.

In the following, we will recall the quantum deformation of these structures introduced by the 2-Chern-Simons theory.

\subsection{Deformation quantization and the combinatorial 2-Fock-Rosly bracket}\label{2fockrosly}
Let us now briefly recall the procedure for deformation quantizing $\C(\mathbb{G}^{\Gamma^2})$. From the classical 2-Chern-Simons action $S_{2CS}$, one can extract the presymplectic form $\omega$ as well as the Lie 2-algebra cobracket $\delta$. The coefficients of these data, as in the usual Chern-Simons theory \cite{Alekseev:1994pa}, combine to give a classical 2-graded $r$-matrix \cite{Bai_2013,Chen:2023integrable} of degree-1
\begin{equation*}
    (1\otimes \mu_1)r=(\mu_1\otimes 1)r,\qquad r\sim \omega + \delta\in(\G^{\otimes 2})_1.
\end{equation*}
It is known \cite{chen:2022} that the semiclassical symmetries of 2-Chern-Simons theory is captured by the Lie 2-bialgebra $(\G;\delta)$ determined by this classical 2-$r$-matrix.

We now leverage the main result in \cite{Chen:2012gz}. 
\begin{theorem}
    There is a one-to-one correspondence between Lie 2-bialgebras and \textbf{Poisson-Lie 2-groups} $(\mathbb{G};\Pi)$, which are Lie 2-groups $\mathbb{G}$ equipped with a \textit{multiplicative} bivector field $\Pi\in\mathfrak{X}^2(\mathbb{G})$. 
\end{theorem}
\noindent Elements of the universal envelope of $\G$, such as the classical 2-$r$-matrix $r$, act on functions of $\mathbb{G}$ by graded derivations \cite{Chen:2012gz}. 

This induces a 2-graded Poisson bracket $\{-,-\}$ \cite{Chen:2012gz,Chen:2023integrable} which gives rise to the following \textbf{combinatorial 2-Fock-Rosly Poisson brackets} (here $\hbar = \frac{2\pi}{k}$)
\begin{align*}
    \{f_{(e_1,f_1)},f_{(e_2,f_2)}\} &= \hbar\big({\delta_{t(e),s(e')}} {r (f_{(e_1,f_1)}\cdot f_{(e_2,f_2)})} - {\delta_{s(e),t(e')}}{(f_{(e_1,f_1)}\cdot f_{(e_2,f_2)}) r^T}\big) \\
    &\equiv \hbar \big((-\cdot-) [r,\Delta_\mathrm{h}(\phi_{(e,f)})]_c\big)
\end{align*}
on localized functions $f_{(e,f)}\in C(X)$ of the decorated 2-graphs $X=\mathbb{G}^{\Gamma^2}$. Here, $(e,f)=(e_1,f_1)\cup_\mathrm{h}(e_2,f_2)$ denotes the 2-graph obtained from \textit{gluing} $(e_1,f_1)$ with $(e_2,f_2)$ such that $e = e_1\ast e_2$ or $e=e_2\ast e_1$ (ie. the \textit{source edges} are composed).

\subsubsection{Quantum 2-graph states}\label{quantum2states}
We now invoke the central result in \cite{Bursztyn2000DeformationQO}: for each smooth Riemannian manifold $X$ and a fixed $\star$-product on the $C^*$-algebra $C(X)$, there is a unique (up to isometry) $\star$-product on the smooth sections $\Gamma(E)$ of a Hermitian vector bundle $E\rightarrow X$, treated as sheaves of $C(X)$-modules over the ring of power series in $\hbar=\frac{2\pi}{k}$. We denote such sheaves by $\Gamma(E)[[\hbar]]$.

As such, the $\star$-product on $C(X)$, obtained from the deformation quantization along the Fock-Rosly 2-group Poisson bracket $\{-,-\}$ above, extends to sections $\Gamma_c(H^X)$ of \textit{any} measureable Hermitian vector bundle $H^X\rightarrow X$ over $X$. This extension, in particular, satisfies the following \textit{semiclassical limit}
\begin{equation*}
    \lim_{\hbar\rightarrow 0}\frac{1}{\hbar}\big(\xi\star\xi' - \xi\star \xi\big) = \{\xi,\xi'\},
\end{equation*}
where $\xi,\xi'$ are sections in the \textit{same} sheaf $\Gamma_c(H^X)$. 

\medskip

Moreover, this deformation quantization also determines a $\star$-product on sections of the tensor product sheaf $\big(\Gamma_c(H^X)\otimes \Gamma_c(H'^X)\big)[[\hbar]]\cong \Gamma_c((H\otimes H')^X)[[\hbar]]$. This allows us to define a \textbf{tensor $\ostar$-product}, as a deformation the usual symmetric tensor product $\otimes$, equipped with sheaf automorphisms
\begin{equation}
    \Gamma_c(H^X)[[\hbar]]\ostar \Gamma_c(H'^X)[[\hbar]]\cong \Gamma_c((H\otimes H')^X)[[\hbar]]\label{quantumautomorphism}
\end{equation}
over the ring of formal power series in $\hbar$. This deformed tensor product then by construction satisfies the following \textbf{Dirac quantization condition}: formally, for each $\phi=\Gamma_c(H^X)[[\hbar]],\phi'=\Gamma_c(H'^X)[[\hbar]]\in\C(\mathbb{G}^{\Gamma^2})$, we have a sheaf automorphism on $\Gamma_c((H\otimes H')^X)[[\hbar]]$ on which
\begin{equation}
       \lim_{\hbar\rightarrow 0}\frac{1}{i\hbar}\big(\xi\ostar\xi' - \xi'\ostar\xi\big) \mapsto \{\xi,\xi'\},\label{semiclassical}
\end{equation}
with respect to the combinatorial 2-group Fock-Rosly Poisson bracket, for sections $\xi\in\phi,~\xi'\in\phi'$ on \textit{different} sheaves.

More formally, if we write "evaluating at $\hbar=0$" as a functor, then the above disgussion renders the following diagram
\begin{equation}
    \begin{tikzcd}[cramped]
	{\C_q(\mathbb{G}^{\Gamma^2})\times \C_q(\mathbb{G}^{\Gamma^2})} & {\C_q(\mathbb{G}^{\Gamma^2})} \\
	{\C(\mathbb{G}^{\Gamma^2})\times \C(\mathbb{G}^{\Gamma^2})} & {\C(\mathbb{G}^{\Gamma^2})}
	\arrow["\ostar", from=1-1, to=1-2]
	\arrow[from=1-1, to=2-1]
	\arrow["\cong", Rightarrow, from=1-1, to=2-2]
	\arrow[from=1-2, to=2-2]
	\arrow["\otimes"', from=2-1, to=2-2]
\end{tikzcd}\label{ostarproduct}
\end{equation}
commutative, up to the homotopy given by the sheaf automorphism \eqref{quantumautomorphism}. 

\begin{definition}\label{quantumhermitian}
    Let $q=e^{i\hbar}=e^{i\frac{2\pi}{k}}$ and $X=(\mathbb{G}^{\Gamma^2},\mu_{\Gamma^2})$. Denote by $\cV^X_q\subset\cH^X$ the category of meaureable sheaves of \textbf{Hermitian $\hbar$-power series sections} $\phi=\Gamma_c(H^X)[[\hbar]]$ --- namely, $\phi$ is a local finitely-generated projective Hilbert $C(X)\otimes_{\bbC}\bbC[[\hbar]]$-module. The morphisms are measureable essentially bounded $\bbC[[\hbar]]$-linear operators.
    
    The associated tensor $\star$-product $\ostar$ \eqref{quantumautomorphism} defines a monoidal functor $\ostar: \cV^X_q\times\cV^X_q\rightarrow \cV^X_q$  satisfying \eqref{semiclassical} and fits into \eqref{ostarproduct}. This makes $\cV^X_q\in\mathsf{Meas}$ into a {non-commutative algebra object} in $\mathsf{Meas}$.

The \textbf{quantum 2-graph states} $\C_q(\mathbb{G}^{\Gamma^2})\subset\cV_q^X$ on $X=(\mathbb{G}^{\Gamma^2},\mu_{\Gamma^2})$ is the full monoidal 2-subcategory whose norm-completions $\Gamma(H^X)[[\hbar]]$ are \textit{separable}: namely they define sheaves of countably-generated Hilbert $L^2(X,\mu_{\Gamma^2})\otimes_\bbC \bbC[[\hbar]]$-modules over $X$. 
\end{definition}
\noindent These quantum 2-graph states $\C_q(\mathbb{G}^{\Gamma^2})$ are precisely those which underlie the discrete degrees-of-freedom in quantum 2-Chern-Simons theory.


Recall that elements of the universal envelope of $\G$ acts on $C(\mathbb{G})$, and hence sections on $\mathbb{G}$, by derivations.
\begin{theorem}\label{quantumhopfcocat}
    The quantum 2-graph states $\C_q(\mathbb{G}^{\Gamma^2})$ is a strict (non-symmetric) Hopf comonoidal cocategory internal to $\cV^{\bf C}$, equipped with an invertible \textbf{cobraiding} $\mathfrak{R}=(R,R')$ of the following form:
    \begin{enumerate}
        \item suppressing the strict associators, $\mathfrak{R}$ is an invertible \textit{bimodule internal endofunctor}
    $$\mathfrak{R}=R\ostar-\ostar R^{-1}: \C_q(\mathbb{G}^{\Gamma^2})\times \C_q(\mathbb{G}^{\Gamma^2})\to \C_q(\mathbb{G}^{\Gamma^2})\times \C_q(\mathbb{G}^{\Gamma^2})$$
        induced by the "conjugation" by a so-called \textbf{2-$R$-matrix} $R\in \C_q((\mathsf{H}\rtimes G)^{\Gamma^2})\times \C_q((\mathsf{H}\rtimes G)^{\Gamma^2})$, satisfying the quasitriangularity condition/\textit{2-Yang-Baxter relations} (cf. \cite{Chen:2023tjf,Chen1:2025?})
\begin{equation}
    (\Delta\times 1)R \cong R^{13}\ostar R^{12},\qquad (1\times\Delta) R \cong R^{13}\ostar R^{23},\label{2yb}
\end{equation}
\item there is an internal natural transformation $R': \mathfrak{R}\circ\Delta\Rightarrow \Delta^\text{op}$  whose components witness the intertwining relations,
\begin{equation*}
    \Delta^\text{op}(\phi)\ostar R \cong  R\ostar  \Delta(\phi),\qquad\phi\in \C_q(\mathbb{G}^{\Gamma^2})
\end{equation*}
for each $\phi\in \C_q(\mathbb{G}^{\Gamma^2})$.
    \end{enumerate}
\end{theorem} 

\begin{rmk}\label{undeformedverticalcomposition}
    Since the \textit{vertical}/groupoid direction direction remains undeformed, we will often denote the monoidal product $\phi\ostar\phi'=\phi\otimes\phi'$ as the undeformed tensor product when $\phi,\phi'$ are localized on 2-simplces which meet \textit{vertically}.
\end{rmk}

We emphasize once again from \textit{Remark \ref{weakcoherence}}, in the current case of the \textit{strict} 2-Chern-Simons theory (ie. in the absence of the weak associator $\tau$), the Hopf structures are strict with invertible coherence morphisms.

\subsubsection{Quantum 2-gauge transformations}\label{deformed2gau}
Upon quantization, the 2-gauge transformations must also deformed accordingly. We shall do this indirectly by preserving certain consistency conditions under the new $\bbC[[\hbar]]$-module structure aforded by deformation quantization.

More precisely, it was proven in \cite{Chen1:2025?} that, if $\C_q(\mathbb{G}^{\Gamma^2})$ is to remain a $\mathbb{U}\G^{\Gamma^1}$-module category satisfying the property \eqref{derivation}, then $\mathbb{U}\G^{\Gamma^1}$ must itself inherit a non-trivial cobraiding $\tilde R$, and a quantum deformed \textit{coproduct} $\tilde\Delta$. This makes  $\mathbb{U}\G^{\Gamma^1}\rightsquigarrow \mathbb{U}_q\G^{\Gamma^1}$ into a Hopf category which is non-cosymmetric. 

Moreover, these newly deformed coproduct and invertible cobraiding structures come equipped with invertible natural transformations 
\begin{align}
      &\varpi: (-\ostar-)\circ \Lambda_{\tilde\Delta} \cong \Lambda\circ  (-\ostar-): \tilde{\mathbb{U}}_q\G^{\Gamma^1}\times\mathfrak{C}_q(\mathbb{G}^{\Gamma^2})^{\times 2}\rightarrow \mathfrak{C}_q(\mathbb{G}^{\Gamma^2})\label{tensorator},\\
      &\Lambda_{\tilde\Delta } \big(\mathfrak{R} \circ \Delta \big)\cong \Lambda_{\tilde{\mathfrak R}\circ  \tilde \Delta }\Delta:  \tilde{\mathbb{U}}_q\G ^{\Gamma^1}\times\mathfrak{C}_q(\mathbb{G}^{\Gamma^2}) \rightarrow \mathfrak{C}_q(\mathbb{G}^{\Gamma^2})^{\times 2}, \label{inducrmat}
\end{align}
which are crucial in preserving the derivation property \eqref{derivation} and \textbf{ensuring that $\C_q(\mathbb{G}^{\Gamma^2})$ remains a \textit{monoidal} measureable $\mathbb{U}_q$-module} under the 2-gauge transformation operation $\Lambda$.

\begin{rmk}\label{induced2Rs}
    We shall refer to the first 2-morphism $\varpi$ \eqref{tensorator} in the above as the \textbf{tensorator}. To clarify what the second 2-morphism \eqref{inducrmat} is really doing, recall from \textbf{Theorem \ref{quantumhopfcat}} that $\tilde{\mathfrak{R}}$ is also induced from a 2-$R$-matrix $\tilde R\in\mathbb{U}_q\G^{\Gamma^1}\times\mathbb{U}_q\G^{\Gamma^1}$. In this case, components of the 2-morphism \eqref{inducrmat} can then be written as invertible measureable morphisms for which
    \begin{gather*}
        (\Delta(\phi)\ostar R^{-1})\bullet \tilde\Delta(\zeta)\cong \Delta(\phi)\bullet(\tilde R \cdot \tilde \Delta(\zeta)) ,\\ 
        (R \ostar \Delta(\phi))\bullet \tilde\Delta(\zeta)\cong \Delta(\phi)\bullet\big({\tilde \Delta(\zeta)\cdot \tilde R^{-1}}\big)
    \end{gather*}
    in terms of the right-module structure $\bullet$ \eqref{regrep}, where $\cdot$ denotes the monoidal structure on $\mathbb{U}_q\G^{\Gamma^1}$.
\end{rmk}

The following characterization can then be obtained.
\begin{theorem}\label{quantumhopfcat}
    $\mathbb{U}_q\G^{\Gamma^1}$ is a strict (non-cosymmetric) Hopf monoidal category internal to $\mathsf{Meas}$, also equipped with an invertible cobraiding $\tilde{\mathfrak{R}}=(\tilde R,\tilde R')$ of the form similar in \textbf{Theorem \ref{quantumhopfcocat}}.
\end{theorem}

\medskip

In accordance with the above, we can now introduce the categorical versions of compact quantum groups, in analogy to the quantum coordinate rings of Woronowicz \cite{Woronowicz1988} or the quantum enveloping algebras of Drinfel'd-Jimbo \cite{Drinfeld:1986in,Jimbo:1985zk}.
\begin{definition}\label{categoricalquantumgroup}
    Suppose $\Gamma$ is a PL 2-disc.
    \begin{enumerate}
        \item $\C_q(\mathbb{G}^{\Gamma^2}) = \C_q(\mathbb{G})$ is called the \textbf{quantum categorical coordinate ring}.
        \item $\mathbb{U}_q\G^{\Gamma^1} = \mathbb{U}_q\G$ is called the \textbf{quantum categorical enveloping algebra}.
    \end{enumerate}
\end{definition}


It is reasonable to expect a parallel, categorical analogue of the Drinfel'd-Jimbo construction for $\mathbb{U}_q\G$, as well as a categorical analogue of the quantum Fourier theory \cite{Semenov1992} which ties them together. We will not pursue this in this paper, however.



\subsection{The Lattice 2-algebra}\label{lattice2alg}
Equipped with the above structures, \cite{Chen1:2025?} defined the \textit{lattice 2-algebra} of 2-Chern-Simons theory. It is endowed with certain conditions which are categorical analogues of those in the lattice algebra for Chern-Simons theory \cite{Alekseev:1994pa}.
\begin{definition}
    The \textbf{lattice 2-algebra} $\mathcal{B}^\Gamma$ for 2-Chern-Simons theory on the lattice $\Gamma$ is the monoidal semidirect product  (cf. \cite{Fuller:2015}) $\C_q(\mathbb{G}^{\Gamma^2})\rtimes \mathbb{U}_q\G^{\Gamma^1}$ through the right action $\bullet$, such that each $\phi\in\C_q(\mathbb{G}^{\Gamma^2})$ satisfy:
    \begin{enumerate}
        \item the \textbf{left-covariance} condition\footnote{This can be understood as a version of \eqref{regrep} in the general Hopf categorical context.}
        \begin{equation}
        \phi\bullet (a_v,\gamma_e) \cong (1 \otimes \Lambda)_{\tilde\Delta(a_v,\gamma_e)}\bullet\phi,\qquad \forall ~(a_v,\gamma_e)\in\mathbb{U}_q\G^{\Gamma^1},\label{leftreg}
    \end{equation}
        \item on local 2-graph states, there exist sheaf isomorphism witnessing the \textbf{braid relations}
                \begin{equation}
            \phi_{(e,f)}\times\phi_{(e',f')} \cong \begin{cases}\phi_{(e',f')}\times\phi_{(e,f)} &; (e,f) \cap (e',f')=\emptyset\\
                (\Lambda\times \Lambda)_{\tilde R_e}(\phi_{(e',f')}\times\phi_{(e,f)}) &; e\cup \partial f' \neq \emptyset
            \end{cases},\label{braid}
        \end{equation}
        where $\tilde R_e\in \mathbb{U}_q\G^{\Gamma^1}\times \mathbb{U}_q\G^{\Gamma^1}$ is the cobraiding 2-$R$-matrix (cf. \textit{Remark \ref{induced2Rs}}) localized on the common edge $e$.
    \end{enumerate}
\end{definition}
\noindent The braid relations ensure that both sides of \eqref{braid} furnish the same $\mathbb{U}_q\G^{\Gamma^1}$-representation, up to intertwining homotopy; they shall play an important role in \S \ref{commbdy}.

Now as mentioned previously in \S \ref{deformed2gau}, the derivation property \eqref{derivation} and its underlying coherent monoidal module tensorator $\varpi$ \eqref{tensorator} ensures that $\C_q(\mathbb{G}^{\Gamma^2})$ remains a monoidal measureable $\mathbb{U}_q\G^{\Gamma^1}$-module category under quantum deformation.

It is also worth mentioning here that the witness for the the braid relations \eqref{braid} can be \textit{explicitly} obtained from the invertible cobraiding, as well as the coherence 2-morphism \eqref{inducrmat}.

\subsubsection{2-Chern-Simons lattice observables}\label{observ}
In a field theory, from the purely algebraic perspective, observables should be defined as the "gauge invariants" --- in an appropriate sense --- of all possible configurations. This philosophy takes different guises in different physical contexts: such as in the invertible TQFT context \cite{huang2023tannaka} and in the perturbative QFT context \cite{costello_gwilliam_2016}.

In our case in the context of the 2-category $\mathsf{Meas}$, this idea takes the form of the following explicit definition.
\begin{definition}
The \textbf{observables of 2-Chern-Simons theory} $\mathcal{O}^\Gamma$ consist of those 2-graph states $\phi\in\C_q(\mathbb{G}^{\Gamma^2})$ equipped with natural measureable sheaf isomorphisms
    \begin{equation}
        \phi\bullet \zeta \cong \zeta\bullet\phi,\qquad \forall~ \zeta\in A\label{invarstatement},
    \end{equation}
    witnessing the \textit{invariance condition}, where $A\subset\mathbb{U}_q\G^{\Gamma^1}$ runs over all Borel measureable subsets. By construction, there is a fully-faithful internal functor $\mathcal{O}^\Gamma\rightarrow \mathcal{B}^\Gamma$ into the lattice 2-algebra.
\end{definition} 
By \eqref{leftreg}, the observables $\mathcal{O}^\Gamma$ are equivalently those 2-graph states $\phi$ which are equipped with measureable natural isomorphisms \eqref{invarstatement} $\Lambda_\zeta\phi\cong \phi$ for all $\zeta\in\mathbb{U}_q\G^{\Gamma^1}$. 



In other words, $\mathcal{O}^\Gamma = \big(\C_q(\mathbb{G}^{\Gamma^2})\big)^{\mathbb{U}_q\G^{\Gamma^1}}$ are the \textit{homotopy fixed-points}, or equivalently the \textbf{equivariantization} of $\C_q(\mathbb{G}^{\Gamma^2})$ under the 2-gauge transformations $\mathbb{U}_q\G^{\Gamma^1}$, with respect to the module structure \eqref{2gtmodule}. This is a categorical analogue of the Chern-Simons observables defined in \cite{Alekseev:1994pa} --- as invariants of the algebra of observables.




\begin{rmk}
    Suppose the PL 2-manifold $S$, embedded in a 3d manifold $\Sigma$, has two triangulations $\Gamma,\Gamma'$ that are refinements of each other --- that is, there is an embedding $\boldsymbol{\Delta}\supset\boldsymbol{\Delta}'$ of their corresponding simplicial complexes  --- then there is a monoidal restriction functor of sheaves $f_{\Gamma\supset\Gamma'}:\mathcal{B}^{\Gamma}\rightarrow\mathcal{B}^{\Gamma'}$ on the lattice 2-algebras. The family $\big(\mathcal{B}^\Gamma,f_{\Gamma\supset\Gamma'}\big)_\Gamma$ thus forms a direct system in the double bicategory of cobraided Hopf cocategories $\cA=\operatorname{cobHopf}_{\cV^{\bf C}}$ in $\cV^{\bf C}$, where ${\bf C}=\mathsf{LieGrp}$. If 2-colimits exist in $\cA$, then we can take the direct limit to obtain the "universal" 2-Chern-Simons algebra $\mathcal{B}=\lim_{\Gamma\rightarrow}\mathcal{B}^\Gamma$.
\end{rmk}

\begin{rmk}\label{2colimit}
    Since each $\mathcal{B}^\Gamma$ is a monoidal semidirect product and each functor $f_{\Gamma\supset\Gamma'}$ is monoidal, $\mathcal{B}$ can also be written as a monoidal semidirect product $\overline{\C}_q(\mathbb{G})\rtimes\overline{\mathbb{U}}_q\G$ (these may not coincide on-the-nose with \textbf{Definition \ref{categoricalquantumgroup}}). The homotopy fixed points $\mathcal{O}=(\overline{\C}_q(\mathbb{G}))^{\overline{\mathbb{U}}_q\G}$ would then, analogous to the lattice algebra in Chern-Simons theory \cite{Alekseev:1994pa}, be able to be interpreted as a model for the \textit{quantum categorified moduli space of flat 2-connections.} 
\end{rmk}



\subsubsection{2-$\dagger$ unitarity of the 2-holonomies}\label{2dagger}
Recall in the above theorems that $\C_q(\mathbb{G}^{\Gamma^2})$ is a \textit{Hopf} cocategory, which has equipped a antipode functor. Similar to the coproducts $\Delta,$ these antipode functors $S$ are intimately tied to the geometry of the underlying 2-graphs. Specifically, $S$ is induced from \textit{orientation reversal}.

Following {\bf Example 5.5} of \cite{ferrer2024daggerncategories}, we take the embedded graph $\Gamma\subset\Sigma$ as a framed piecewise-linear (PL) 2-manifold, then the PL-group $\operatorname{PL}(2) = O(2) = SO(2)\rtimes \bbZ_2$ tells us directly what the 2-dagger structure on $\Gamma$ is --- $\dagger_2$ is given by the orientation reversal $\bbZ_2$ subgroup and $\dagger_1$ is a $2\pi$-rotation in framing $SO(2)$-factor.

Crucially, these daggers are involutive $\dagger_2^2 = \id,~\dagger^2_1 \cong \id$ and they {\it strongly commute} 
\begin{equation}
    \dagger_2\circ \dagger_1 = \dagger_1^\text{op}\circ\dagger_2\label{commutedagger}.
\end{equation}
For edges in $\Gamma^1$, on the other hand, $\dagger_2$ implements an orientatino reversal $e^{\dagger_2} = \bar e$ while $\dagger_1$ rotates its framing: if $\nu$ is a trivialization of the normal bundle along the embedding $e\hookrightarrow \Sigma$, then $(e,\nu)^{\dagger_1} =(e,-\nu)$. Let us denote this frame rotation by the shorthand $e^T = (e,-\nu)$.

We denote the induced maps on the measureable Lie 2-groups by
$X=\mathbb{G}^{\Gamma^2}\xrightarrow{\sim} \overline{X}^{\mathrm{h,v}} = \mathbb{G}^{(\Gamma^2)^{\dagger_2,\dagger_1}}$. 
    \begin{definition}\label{unitary2hol}
Define the \textbf{antipode fucntors}
\begin{equation}
  S_\mathrm{v}:\C_q(\mathbb{G}^{\Gamma^2})\rightarrow  \C_q(\mathbb{G}^{\Gamma^2})^{{\text{op}}},\qquad S_\mathrm{h}: \C_q(\mathbb{G}^{\Gamma^2})\rightarrow \C_q(\mathbb{G}^{\Gamma^2})^{\text{m-op,c-op}},\label{antilinear}
\end{equation} 
      where "$-^{{\text{op}}}$" denotes taking the opposite internal category, and "$-^{\text{m-op,c-op}}$" denotes taking the reverse internal monodal/comonoidal structure. The \textbf{2-$\dagger$ unitarity of the 2-holonomies} is the property that:
        \begin{itemize}
            \item For each 2-graph state in $\mathfrak{C}_q(\mathbb{G}^{\Gamma^2})$, we have stalk-wise for each $\mathrm{z}=\{(h_e,b_f)\}_{(e,f)}\in\mathbb{G}^{\Gamma^2}$,
            \begin{align*}
        (S_h\phi)_{\mathrm{z}} &= \bar\phi_{\mathrm{z}^{\dagger_1}},\qquad \mathrm{z}^{\dagger_1}=\{(h_{e^{\dagger_1}},b_{f^{\dagger_1}})\}_{(e,f)}\\ (S_v\phi)_{\mathrm{z}} &= \phi^T_{\mathrm{z}^{\dagger_2}},\qquad \mathrm{z}^{\dagger_2} = \{(h_{e^{\dagger_2}},b_{f^{\dagger_2}})\}_{(e,f)}
    \end{align*}
        where $\bar\phi$ is the measureable field $(H^*)^X$ complex linear dual to $\phi$, and $\phi^T$ is the same sheaf underlying $\phi\in\C_q(\mathbb{G}^{\Gamma^2})$ but equipped with the adjoint sheaf morphisms. 
    \item For the 2-gauge transformation $\Lambda: \mathbb{U}\G^{\Gamma^1}\times \C(\mathbb{G}^{\Gamma^2})\rightarrow \C(\mathbb{G}^{\Gamma^2})$, we have pointwise for each $\zeta=\{(a_v,\gamma_{e}\}_{(e,v)}\in\mathbb{U}_q\G^{\Gamma^1}$ (recall $e^T = (e,-\nu)$ denotes a frame rotation of an edge),
    \begin{align*}
            \Lambda_{\tilde S_h\zeta} &= \bar \Lambda_{\zeta^{\dagger_1}},\qquad \zeta^{\dagger_1}=\{ (a_{v'}\xrightarrow{\gamma_{\bar e}}a_v)\}_{(a,v)},\\ \Lambda_{\tilde S_v\zeta}&= \Lambda^\dagger_{\zeta^{\dagger_2}},\qquad \zeta^{\dagger_1}=\{(a_{v}\xrightarrow{\gamma_{e^T}}a_{v'})\}_{(a,v)}
    \end{align*}
    where $\bar \Lambda_\zeta$ is the complex conjugate measureable functor and $\Lambda_\zeta^\dagger$ is the adjoint.
         \end{itemize}
    \end{definition}
Note for $C=\C_q(\mathbb{G}^{\Gamma^2})$, the vertical antipode $S_v:C\rightarrow C^{\bar{\text{op}},\text{c-op}_v}$ reverses both the \textit{external} (ie. in $\mathsf{Meas}_X$) composition and the internal (ie. in $C_1$) cocomoposition $\Delta_v$. On the other hand, the horizontal antipode $S_h:C\rightarrow C^{\text{m-op,c-op}_h}$ is internally op-$\ostar$-monoidal and op-comonoidal.

\medskip

The $\dagger$-unitarity property intertwines the external $\dagger$-adjoint structures and the internal geometry of the underlying 2-graph $\Gamma$.

\subsubsection{*-operations}\label{*-op}
Denote by $\eta_\mathrm{h,v}: \Gamma_c(H^X)[[\hbar]]\rightarrow \Gamma_c(H^{\overline{X}^\mathrm{h,v}})[[\hbar]]$ the $\bbC[[\hbar]]$-linear measureable sheaf morphisms induced on the 2-graph states by the 2-$\dagger$ structure of $\Gamma^2$. 
\begin{definition}\label{daggerpair}
    We say the pair $(\eta_\mathrm{h},\eta_\mathrm{v})$ is a \textbf{2-$\dagger$-intertwining pair} iff for each $\zeta\in\mathbb{U}_q\G^{\Gamma^1}$, we have
    \begin{equation*}
    \eta_\mathrm{h}\big(\Lambda_\zeta\phi_{(e,f)}\big) = \Lambda_{\bar \zeta}(\eta_\mathrm{h}\phi)_{(\bar e',\bar f)},\qquad \eta_\mathrm{v}\big(\Lambda_\zeta \phi_{(e,f)}\big)= \Lambda_\zeta(\eta_\mathrm{v}\phi)_{(e',\bar f)}
\end{equation*}
as operators on each quantum 2-graph state $\phi=\Gamma_c(H^X)[[\hbar]]\in\C_q(\mathbb{G}^{\Gamma^2})$, where $U_\zeta$ denotes the field of bounded invertible operators realizing the 2-gauge transformations $\Lambda_\zeta$.
\end{definition}

We are finally ready to state the *-operations on the 2-graph states and the 2-gauge transformations. Suppose the $R$-matrix $\tilde R$ on $\mathbb{U}_q\G^{\Gamma_1}$ is invertible, in the sense that the induced cobraiding natural transformations $\tilde\Delta\Rightarrow \tilde\Delta^\text{op}$ are invertible.

Due to the locality properties \S \ref{locality}, it suffice to define the *-operations on local pieces.
\begin{definition}\label{starop}
Let $(v,e) = v\xrightarrow{e}v'\in\Gamma^1$ denote a 1-graph, and  let $(e,f)\in\Gamma^2$ denote a 2-graph, with source and target edges $e,e':v\rightarrow v'$.
    \begin{enumerate}
        \item The \textbf{*-operations} on localized homogeneous elements in $\tilde{\cC}$ are given by
    \begin{equation}
        \zeta^{*_1}_{(v,e)} = \bar \zeta ,\qquad \bar\zeta^{*_2}_{(v,e)} = \zeta^T \label{dagger2-gau}
    \end{equation}
    where $v'\xrightarrow{\bar e}v$ is the orientation-reversal and $v\xrightarrow{e^T}v'$ is the framing rotation.
    \item Given the 2-$\dagger$-intertwining pairs in \textbf{Definition \ref{daggerpair}}, the \textbf{*-operations} on localized 2-graph states $\phi_{(e,f)}\in\mathscr{A}^0=\mathfrak{C}_q(\mathbb{G}^{\Gamma^2})$ are given by 
    \begin{align*}
        & \phi_{(e,f)}^{*_1} = (\Lambda\otimes 1)_{\tilde R^{-1}}(\phi_{(\bar e',\bar f)}) \eta_\mathrm{h},\\
        & \phi_{(e,f)}^{*_2} = (\phi_{(e',\bar f)})\eta_\mathrm{v},
    \end{align*}
    where $(\bar e',\bar f) = (e,f)^{\dagger_1}$ and $(e',\bar f) = (e,f)^{\dagger_2}$. Here,  the $\tilde R$-matrix is localized on $\partial f$.
    \item The regular $\bullet$-module structure on $\mathscr{A}^0$ over $\tilde{\cC}$ is *-compatible: there exist natural measureable isomorphisms
    \begin{equation*}
        (\phi\bullet\zeta)^{*_{1,2}} \cong \zeta^{*_{1,2}} \bullet\phi^{*_{1,2}},\qquad \forall~ \phi\in\mathscr{A}^0,\quad \zeta\in\tilde{\cC},
    \end{equation*}
    satisfying the obvious coherence conditions against the $\bullet$-module associator and the tensorator \eqref{tensorator}.
    \end{enumerate}       
\end{definition}
\noindent Note crucially that these *-operations are in general \textit{not} involutive. 

A routine check yields the following \cite{Chen1:2025?}.
\begin{proposition}
    The *-operations strongly commute, $(-^{*_1})^\text{op}\circ -^{*_2}\cong (-^{*_2})^\text{m-op,c-op}\circ -^{*_1}$. 
\end{proposition}
Throughout the following, we will assume that both $-^{*_1},-^{*_2}$ are equivalences of measureable categories, with $-^{*_2}$ is idempotent/involutive but $-^{*_1}$ not necessarily (unless $q=1$; see \S 7, \cite{Chen:2025?}).

\begin{rmk}\label{gluingedges}
    We pause here to note that the definition \eqref{dagger2-gau} essentially states that a frame reversal $(e,\nu)\mapsto (e,-\nu)$ on a 1-graph is implemented by the antipode on the decorations. This is an important fact for gluing localized 2-graphs: the interfacing edge has opposite framing depending on which local 2-graph it is embedded into.
\end{rmk}

To extend the above definition globally to the entire lattice configuration on $\Gamma$, the following was proven in \cite{Chen1:2025?}.
\begin{theorem}
    Given 2-$\dagger$-unitarity holds on each quantum 2-graphs state, the *-operations preserves (i) the $\bullet$-bimodule structure $\C_q(\mathbb{G}^{\Gamma^2})\circlearrowleft\mathbb{U}_q\G^{\Gamma^1}$, (ii) the covariance condition \eqref{leftreg}, and (iii) the braiding relations \eqref{braid}. Thus they extend to the lattice 2-algebra $\mathcal{B}^\Gamma$.
\end{theorem}

In fact, under the unitarity property defined above, the compatibility of the *-operations with \eqref{leftreg}, \eqref{braid} is \textit{equivalent} to the various axioms satisfied by the antipode/cobraiding $\tilde S,\tilde R$ on $\mathbb{U}_q\G^{\Gamma^1}$.

\section{Higher-algebra of dense 2-holonomies/2-monodromies}\label{objA}
We now formally begin the main contents of this paper. Given the underlying 2d lattice $\Gamma$, we model its triangulation as a simplicial complex. Its 2-truncation $\Gamma^2$ is a \textit{2-graph}, whose $2$-groupoid structure describes how the closed 2-simplices are glued together in $\Gamma^2$. Using this idea, we seek to build 2-graph states $\C_q(\mathbb{G}^{\Gamma^2})$ from the local quantum categorified coordinate ring $\C_q(\mathbb{G})\simeq\C_q(\mathbb{G}^{\Delta^2})$ living on each fundamental 2-simplex $\Delta^2=\Delta$. 

\subsection{Setting up the 2-simplex geometry}\label{2simplex}
We shall label a fundamental 2-simplex by specifying its edges and face $({\bf e}=(e_1,e_2,e_3),f)$, such that the 2-holonomy decorations satisfy fake-flatness $t(b_f) = h_{\partial f}$ with $\partial f = e_1-e_2+e_3$. We will in the following identify the first edge $e_1$ as the {\it source edge} of the face $f$. Once this choice is made, the cyclic ordering of the vertices and the rest of the edges are induced by the orientation of the face $f$ in $\Delta$.

Consider an embedded triangulated 2-manifold $\Gamma\subset\Sigma^3$. Its vertex, edge and face ordering is inherited from the orientation of $\Sigma^3$. 
\begin{definition}
    Denote by $\displaystyle \boldsymbol{\Delta}=\coprod_{l\leq k}\Delta_l^{\epsilon_l}$ a collection of ordered 2-simplices with orientation labelled by $\epsilon_l=\pm1$. A \textbf{simplicial decomposition of $\Gamma^2$ by $\boldsymbol{\Delta}$ of legnth $k\geq 1$} is the structure of a simplicial set on $\boldsymbol{\Delta}$ --- namely the data of face and degeneracy maps on the 2-simplices $\Delta_l$ such that  $e_j^{l_j}=\delta^l_j(\Delta_j)$ is the $l$-th face of the $j$-th 2-simplex $\Delta_j\in\boldsymbol{\Delta}^2$, with $1\leq j\leq k$ and $1\leq l\leq 3,$ --- such that $\Gamma^2$ is PL homeomorphic to the 2-truncated simplicial nerve
    \begin{equation*}
        \Gamma^2\cong \big(\boldsymbol{\Delta}^2\mathrel{\substack{\textstyle\rightarrow\\[-0.6ex]
                      \textstyle\rightarrow \\[-0.6ex]
                      \textstyle\rightarrow}} \boldsymbol{\Delta}^1 \rightrightarrows \boldsymbol{\Delta}^0\big).
    \end{equation*}
    Moreover, we say $\boldsymbol{\Delta}$ is \textbf{regular} iff each edge is shared by at most by two distinct 2-simplices.
    \end{definition}
    \noindent If $\boldsymbol{\Delta}$ is regular, then we can write the PL identification as
    \begin{equation*}
        \Gamma^2 \cong \Delta_1^{\epsilon_1}  \,_{e_1^{t_1}}\cup_{e_2^{s_2}}\Delta_2^{\epsilon_2 } \,_{e_2^{t_2}}\cup_{e_3^{s_3}}\dots \,_{e_{k-1}^{t_{k-1}}}\cup_{e_k^{s_{k}}}\Delta_k^{\epsilon_k}.
    \end{equation*}
    The length $k$ is simply the number of distinctly-labelled 2-simplices. 
    
    Here, the "incoming" $e^t\subset\partial\Delta$ and "outgoing" $e'^s\subset\partial\Delta'$ edges of two oriented simplices $\Delta^\epsilon,\Delta'^{\epsilon'}\in\boldsymbol{\Delta}^2$ are glued along a given PL homeomorhism $e^t\cong e'^{s}$, which can be either orientation preserving ($\epsilon=\epsilon'$) or reversing ($\epsilon=-\epsilon'$). Since only relative orientation matters in the gluing, we can always assume the orientation of $\Gamma^2$ agrees with the first simplex $\Delta_1$, ie. $\epsilon_1=1$.

    \begin{definition}
        We call a vertex $v_j$ in $\Delta_j$ the \textbf{root vertex} if $v_j$ is the source vertex of the distinguished source edge $e_1^j$ of $\Delta_j$. We take as base point of $\Gamma$ to be the root vertex $v=v_1$ of $\Delta_1$.
    \end{definition}
\noindent Recall in the definition of the fundamental 2-simplex $\Delta$ that the data of (i) a distinguished source edge, and (ii) its orientation determine the orientation of $\Delta$ itself.
\begin{definition}
    We say that the simplicial decomposition $\boldsymbol{\Delta}$ of $\Gamma^2$ with length $k$ is \textbf{unbroken} if the distinguished source edges of $\Delta_j$, $1\leq j\leq k$, glue into a continuous PL path $p=p_k$ in $\Gamma^2.$ See fig. \ref{fig:unbroken}.
\end{definition}
\noindent Note we can always change the designated source edge label such that $\boldsymbol{\Delta}$ is unbroken. The fact that the path $p$ intersects the root vertices of every 2-simplex $\Delta\in\boldsymbol{\Delta}^2$ is a key property which will be used later on. 

\begin{figure}[h]
    \centering
    \includegraphics[width=0.6\linewidth]{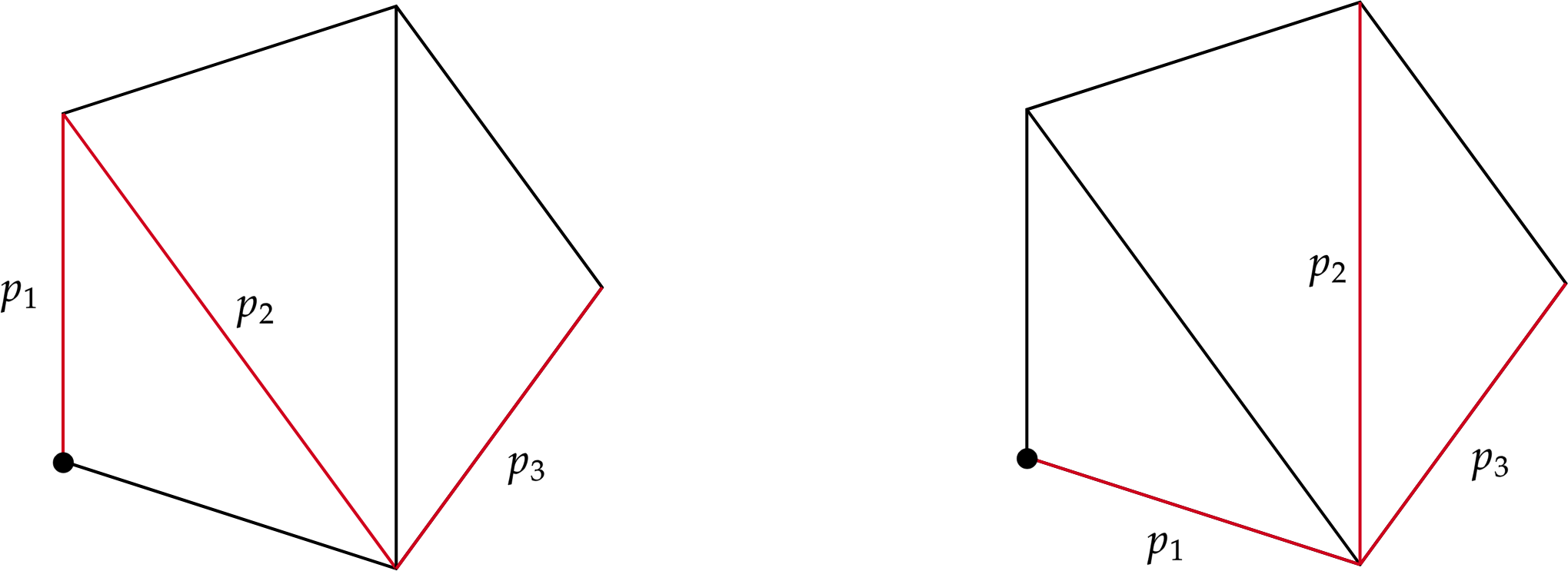}
    \caption{A typical complex of 2-simplices with different choice of source paths, coloured in red. The left is a unbroken configuration, and the right is broken.}
    \label{fig:unbroken}
\end{figure}

It is clear that, if $p$ is an oriented PL path, then its orientation determines uniquely a set of orientation data $\{\epsilon_j\}_{j=1}^k$ for $\boldsymbol{\Delta}$.
\begin{proposition}\label{path}
    Suppose the simplicial decomposition $\boldsymbol{\Delta}$ with length $k$ is regular, then there exists an assignment of source edges to $\{\Delta_j\}_j$ such that it is unbroken, with the length of $p$ bounded by $|p|\leq k-1$.
\end{proposition}
\begin{proof}
    Recall how the source edges are defined: it is the "first" edge in a fundamental 2-simplex $\Delta$, and the rest of the edges are labelled in cyclic order according to the orientation of $\Delta$.
    
    Prior to assuming regularity, we are going to record the indices $(t_j;s_l)$ which label the edges appearing in the gluing data of the simplicial decomposition $\boldsymbol{\Delta}$. 
    \begin{definition}
        Define the set $\cG=\{(t_j;s_l)\}_{j,l}$ of indices, where $j,l$ runs over the indices for which we have a prescribed PL identification $e_j^{t_j}\cong e_l^{s_l}$ of the corresponding gluing edges. 
    \end{definition}
    \noindent The condition of regularity then means that each edge in $\boldsymbol{\Delta}^1$ cannot have more than one gluing data: if $t_j=t_{j'}$ then $(t_j;s_l)=(t_{j'};s_{l'})$ must coincide in $\cG$. This then allows us to take $\cG$ as a subset of $(\bbZ_3)^{k-1}$.
    
    $\boldsymbol{\Delta}$ can in turn be made unbroken provided $t_j\neq s_{l}$ if one of $t_j,s_l$ is not 1 --- namely, we have to remove from $(\bbZ_3)^{k-1}$ the diagonal of the subset $\bbZ_2\subset\bbZ_3$. This guarantees the existence of a PL continuous path $p$ in $\Gamma$. We now partition $\cG$ into two subsets: one $\cG_2$ consisting of members of the form $(1;1)$ and one $\cG_{0.1}=\cG\setminus \cG_2$ that does not; it is from $\cG_{0,1}$ that we have to remove the diagonal. 
    
    These subsets have the following geometric meaning,
    \begin{enumerate}
        \item $\cG_{0,1}$ contains indices for the gluing edges $e_j^{t_j},e_l^{s_l}$ for which at most only one of them is a source edge, and
        \item $\cG_2$ contains those for which \textit{both} of them are source edges.
    \end{enumerate}
    It is then easy to see that gluing two 2-simplices along edges labelled in $\cG_{0,1}$ will increase the length of $p$ by 1, while gluing along those in $\cG_2$ will increase $|p|$ by 0. The length $p$ is therefore bounded by the size of $\cG_{0,1}$, which is $k-1$.
\end{proof}
\noindent Note a length $|p|=0$ of zero is only possible in a regular simplicial decomposition $\boldsymbol{\Delta}$ of length at most 2. The above proposition can be strengthened to ensure that the path $p$ of length $k-1$ is oriented, by including the data $\epsilon_j/\epsilon_l=\pm 1$ of the relative orientations into the set $\cG$.

In the following, we will always assume that $\boldsymbol{\Delta}$ is equipped with a specification of source edges such that it defines a regular and unbroken simplicial decomposition $\boldsymbol{\Delta}$ of $\Gamma$. Further, we shall also assume that the orientation data for the fundamental 2-simplices in $\boldsymbol{\Delta}$ are determined uniquely (up to global orientation reversal) by the PL orientation of the path $p$.

\subsubsection*{Whiskering.}
Fix a base point vertex $v\in \Gamma^2$. We denote by $p_j\subset \coprod_l\partial\Delta^2_l$ some simplex path which connects $v$ to the root vertex of $\Delta_j^2$, for all $1\leq j\leq k$. For a decorated 2-simplex $\mathbb{G}^{\Delta_j}$, let $\phi_j\in\C_q(\mathbb{G}^{\Delta_j})$.
\begin{definition}
   Define the \textit{whiskering} of $\phi_j$ to the base vertex $v\in\Gamma^2$ as the meassureable field $W_{p_j}\phi$ with stalk Hilbert spaces
\begin{equation*}
    (W_{p_j}\phi_j)_\mathrm{z} = (\phi_j)_{h_{p_j}\rhd \mathrm{z}},\qquad p_j=1\implies W_{p_j}\rhd-=\id.
\end{equation*}
From the perspective of sheaves, $W_{p_j}:\C(\mathbb{G}^{\Delta_j}) \rightarrow\C(\mathbb{G}^{p_j\ast\Delta_j})$ is the invertible direct image functor along the whiskering automorphism $h_{p_j}\rhd-: \mathbb{G}\rightarrow \mathbb{G}$, where $p_j\ast \Delta_j$ is the attachment of the path $p_j$ to the root vertex of $\Delta_j$.
\end{definition}
\noindent Note a whiskering by the edge holonomy $h_e$ cannot in general be removed through a 2-gauge transformation! Unless, of course, $h_e=a_v^{-1}a_{v'}$ is a pure gauge.

\medskip

Note we can whisker along any path, not just the ones overlaying the distinguished source path on $\Gamma^2$ obtained from \textbf{Proposition \ref{path}}.
\begin{tcolorbox}[breakable]
    \subsubsection*{Homotopies between whiskerings.} Consider two generic paths $p,p'$ which are homotopic in $\Gamma^2$. Let $D: p\Rightarrow p'$ denote the contractible closed face $D$ spanned by them, which encloses several glued simplices. Due to fake-flatness $h_{p'} =  h_p\mathsf{t}(b_D)$, the whiskering along $p$ vs. that along $p'$ differ by a vertical multiplication of the face holonomy $b_D\in \mathsf{H}$. 
    
    This induces the translation operator $T_{D}:\xi_\mathrm{z}\mapsto \xi_{b_D\circ \mathrm{z}}$ on sections of 2-graph states $\phi\in\C_q(\mathbb{G}^{\Gamma^2})$. More precisely, we achieve the invertible bounded linear operators
\begin{equation*}
    T_{D}^{\phi}:W_{p}(\phi_j) \rightarrow W_{p'}(\phi_j),\qquad \forall ~\phi_j\in\C_q(\mathbb{G}^{\Delta_j})
\end{equation*}
witnessing the difference between the whiskerings along $p,p'$, where $\Delta_j$ is the 2-simplex whose root vertex $v_j=p(1)=p'(1)$ is the endpoint of $p,p'$.  Imposing naturality against measureable morphisms, ie. the commutativity $$T_{D}^{\phi'} \circ W_{p}(f) = W_{p}(f)\circ T_{D},\qquad \forall ~f: \phi_j\to \phi_j',$$ we can lift the above to the following.
\begin{proposition}\label{whiskeringhomotopy}
    Each PL homotopy $D: p\Rightarrow p'$ between oriented paths $p,p'$ on $\Gamma^2$ are witnessed by monoidal  invertible measureable natural transformations $T_{D}:W_{p}\Rightarrow W_{p'}$ between the associated whiskering measureable functors.
\end{proposition}
\noindent The monoidality follows from the fact that the whiskering operation is monoidal,
\begin{equation*}
    W_p(\phi\otimes\phi') \cong W_{p}\phi\otimes W_{p}\phi',\qquad \phi,\phi'\in\C_q(\mathbb{G}^{\Delta})
\end{equation*}
where $\Delta$ is the 2-simplex whose root vertex is the endpoint of the path $p$.
\end{tcolorbox}

As such, provided $\Gamma^2$ is unbroken and simply-connected, and that $p$ starts at the root of $\Gamma^2$, then there is an invertible measureable natural transformation $W_p\Rightarrow W_{p_j}$ which brings the whiskering by $p$ to the whiskering by the source path $p_j$.

\subsection{Dense states of 2-holonomies and 2-monodromies}\label{holdense}
We are finally ready to describe the construction of 2-simplex holonomies. We shall do this iteratively, starting from the case where the regular simplicial decomposition $\boldsymbol{\Delta}$ has $k=2$. Let $\Delta_1,\Delta_2\in\boldsymbol{\Delta}^2$ be the 2-simplices in a regular simplicial decomposition $\boldsymbol{\Delta}$ of $\Gamma$ with the prescribed PL identification $f_\epsilon:e_1^t\xrightarrow{\sim} e_2^s$. Recall $\epsilon=\pm1$ keeps track of the orientation.

We now make use of the degeneracy maps $d_j^l$ in the simplicial set $\boldsymbol{\Delta}$; denote by $d_j(e^{l}_j)$ the degenerate 2-simplex which collapse down to the $l$-th edge $e^l_j$ of the $j$-th 2-simplex. We call $u_{12} = d_1(f_\epsilon(e_1^{t}))\cap d_2(e_2^s)$ the \textbf{$(1,2)$-degeneracy intersection}. This subgraph has the property that its decorations have non-zero measure
\begin{equation*}
    \mu_{\Delta_1\coprod\Delta_2}(\mathbb{G}^{u_{12}})\neq0
\end{equation*}
with respect to the Haar measure $\mu_{\Delta_1\coprod\Delta_2}$ on the disjoint union decorated 2-simplices $\mathbb{G}^{\Delta_1}\times\mathbb{G}^{\Delta_2}=\mathbb{G}^{\Delta_1\coprod \Delta_2}$.

By the classic Tietze extension theorem \cite{brylinski2007loop,book-mathphys1}, we can pick a smoothly interpolating function on $\mathbb{G}^{u_{12}}$ to extend sections of $\phi_1\in \C_q(\mathbb{G}^{\Delta_1})$, say, into the degeneracy intersection $u_{12}$. Recall the notion of \textit{localized 2-graph states} in \textbf{Definition \ref{local2graphstates}}.
\begin{definition}\label{glueamenable}
    Suppose $\Gamma$ is a 2-graph lattice containing two 2-simplices $\Delta_1,\Delta_2$ which meet at an edge $e$, and suppose $\phi_{1,2}\in\C_q(\mathbb{G}^{\Gamma})$ are 2-graph states localized on $\Delta_{1,2}$, respectively. The tuple $(\phi_1,\phi_2)$ is called \textbf{gluing-amenable at $e$} iff there exist an isomorphism of restriction sheaves
    \begin{equation*}
        \alpha_{12}:  \phi_1\mid_{\mathbb{G}^{u_{12}}}~\cong~ \phi_2\mid_{\mathbb{G}^{u_{12}}},\qquad \alpha_{12}=\alpha_{21}^{-1}.
    \end{equation*}
    We denote the gluing-amenable 2-graph states by $\C_q(\mathbb{G}^{\Delta_1})\times_e\C_q(\mathbb{G}^{\Delta_2})$, where $e$ is the gluing edge.
\end{definition}
\noindent In essence, this condition allows us to "concatenate" $\phi_1,\phi_2$ along the glued edges $f_\epsilon:e_1^t\xrightarrow{\sim}e_2^s$. 

What this definition means more explicitly is the following. Let $\Gamma_c(H_j^{X_{12}})[[\hbar]]$ denote the measureable sheaf of Hermitian sections corresponding to the restricted 2-graph states $\phi_j\mid_{X_{12}}$, where $j=1,2$ and $X_{12}=\mathbb{G}^{u_{12}}$. The gluing-amenability condition is then the existence of a *-isomorphism $\Gamma_c(H_1^{X_{12}})[[\hbar]]_{/V}\cong \Gamma_c(H_2^{X_{12}})[[\hbar]_{/V}$ of free $C(U)[[\hbar]]$-modules for each such Borel open $V\subset X_{12}$.

Let $p=p_2$ denote the PL path from $v$ to the root of $\Delta_2$, we then use the quantum deformed monoidal structure \S \ref{quantum2states} to define the \textbf{2-holonomy state}
\begin{equation*}
    {\Phi} = 
        \phi_1\ostar(h_{p_2}\rhd \phi_2),\qquad \phi_1,\phi_2 ~\text{ gluing-amenable}
\end{equation*}
associated to $\phi_1,\phi_2$. The resulting 2-graph state $\Phi$ is clearly localized on $\Delta_1\cup_e \Delta_2$.

\medskip

We now wish to extend the notion of gluing-amenability to a regular simplicial decomposition $\boldsymbol{\Delta}$ of $\Gamma$ containing $k>2$ number of fundamental 2-simplices. In order to do so, we first have to spell out the necessary coherence structure.

\subsubsection{Interchangers; vertices of trisecitons}\label{interchanger}
In \S \ref{holdense}, we have described how we can build 2-graphs $\Gamma$ and 2-graph states on them from local data on each 2-simplex within it. We pause here to introduce a special geometric configuration of particular importance. 

Let $\Delta_1,\dots,\Delta_4$ denote four oriented fundamental 2-simplices, which glues into the graph $\Gamma_+$ specified by the following gluing configurations:
\begin{equation*}
    \Delta_{2i-1}^+\,_{e_{2i-1}^2}\cup_{e_{2i}^3}\Delta_{2i}^+,\qquad \Delta_{i}^+\,_{e_i^{1}}\cup_{e_{i+2}^1}\Delta_{i+2}^-,
\end{equation*}
where $i=1,2$. In other words, the resulting graph $\Gamma_+$ is obtained by gluing a pair of the 2-simplices horizontally, and then gluing them vertically. Here, we have chosen the source edges to be $e_i=e_{i}^1\cong -e_{i+2}^1$ for $i=1,2$, which is completely internal in $\Gamma_+$. We denote by the other glued edges by $e_i'=e_{2i-1}^2\cong e_{2i}^3$, and the corresponding degeneracy intersection by $u_{1234} = u_{12}\cap u_{34}\cap u_{13}\cap u_{24}$ around the central vertex.

The fact that $\Gamma_+$ is well-defined means that the simplicial decomposition $\boldsymbol{\Delta}=\{\Delta^{\epsilon_i}_i\}_{i=1}^4$ is unambiguous. This manifests as a certain \textit{interchanger} isomorphism.
\begin{definition}\label{intch}
    Let $\Delta_1,\dots,\Delta_4$ denote 2-simplices for which $\Gamma_+=\coprod_{i=1}^4\Delta_i\subset\Gamma$ is a 2-subgraph, and let $\phi_i$ be 2-graph states localized on $\Delta_i$ for $1\leq i\leq 4$. We say the tuple $(\phi_1,\phi_2,\phi_3,\phi_4)$ is \textbf{gluing-amenable at $\Gamma_+$} iff (i) they are pairwise gluing-amenable, and (ii) they have equipped a measureable sheaf isomorphism
    \begin{equation*}
        \beta_{12}^{34}:(\phi_1\ostar\phi_2)\ostar(\phi_3\ostar\phi_4)\xrightarrow{\sim} (\phi_1\ostar\phi_3)\ostar(\phi_2\ostar\phi_4),
    \end{equation*}
    called the \textbf{interchanger}. By \textit{Remark \ref{undeformedverticalcomposition}}, we will denote this measureable natural transformation by 
    \begin{equation*}
        \beta: (-\otimes-)\circ(-\ostar- \times -\ostar-) \xRightarrow{\sim} (-\ostar-)\circ(-\otimes-\times-\otimes -)\circ(1\times \sigma\times 1),
    \end{equation*}
    where $\sigma: \C_q(\mathbb{G}^{\Gamma})\times \C_q(\mathbb{G}^{\Gamma})\rightarrow \C_q(\mathbb{G}^{\Gamma})\times \C_q(\mathbb{G}^{\Gamma})$ is a swap of products.
\end{definition}


\begin{rmk}\label{trisection}
    Geometrically, $\beta$ witnesses the equivalence between the two ways  in which the decorated 2-simplices on $\coprod_{i=1}^4{\Delta_i}$ can be glued onto $\Gamma_+$; see fig. \ref{fig:interchanger}. Hence, \textbf{Definition \ref{intch}} is saying that each such trisection in a 2-graph is assigned a natural interchange isomorphism $\beta$.
\end{rmk}

\begin{figure}[h]
    \centering
    \includegraphics[width=0.8\linewidth]{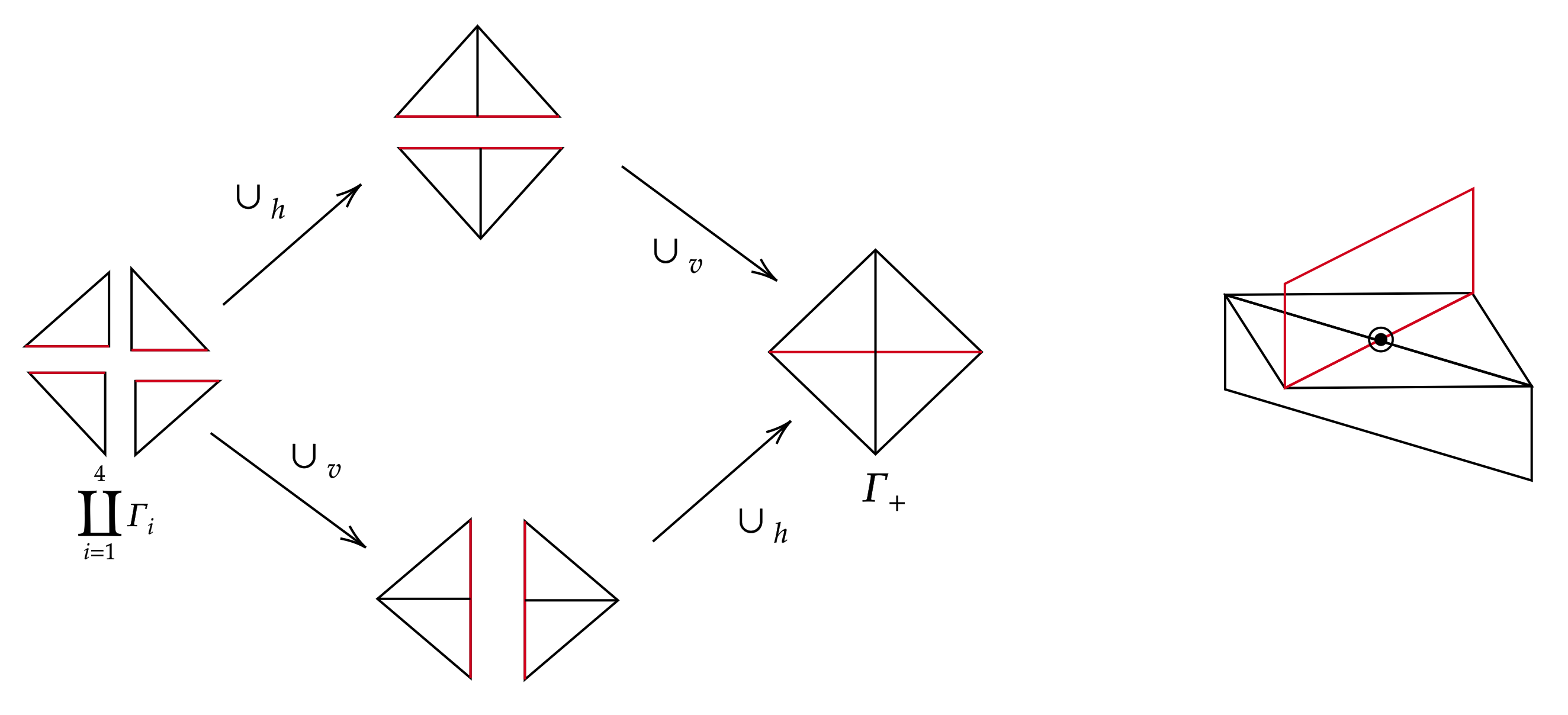}
    \caption{The left illustrates the geometric configuration of 2-simplices upon which the interchanger $\beta$ is defined.  This geometry is precisely the \textit{vertex} in a trisected singular graph \cite{matveev2007algorithmic} as displayed on the right; see also fig. 2 (c) of \cite{Sakata2022-il}.}
    \label{fig:interchanger}
\end{figure}

This isomorphism $\beta$ witnesses the equivalence between the two valid ways of constructing holonomy-dense 2-graph states in $\C_q(\mathbb{G}^{\Gamma_+})$; since the deformed products of 2-graph states are used in the construction, the data of the interchanger $\beta$ will also depend on $q$. We say $\C_q(\mathbb{G}^{\Gamma_+})$ \textbf{holonomy-dense} if the above functors $\C_q(\mathbb{G}^{\coprod_{i=1}^4\Delta_i})^+\rightarrow \C_q(\mathbb{G}^{\Gamma_+})$ are equivalences.



\begin{rmk}
    Another geometric interpretation of the subgraph $\Gamma_+$ is the following. Consider a graph $\Gamma\subset\Sigma$ embedded in a 3d manifold $\Sigma$, and two other dsjoint 2-cells $C,C'\subset\Sigma$ in general position, such that their transversal intersection $\Gamma\cap C\cap C'$ forms a "cross". This cross is precisely what the internal tree $E'$ of the glued edges in $\Gamma_+$ looks like. As such, the data $\beta$ can also be interpreted as a witnesses for \textit{triple intersections} of surfaces in $\Sigma$. The fact that higher-gauge theories in 4d can detect triple intersections was also noted in  \cite{PUTROV2017254}.
\end{rmk}

Recall from the proof of \textbf{Proposition \ref{path}} that the set $\cG$ keeps tracks of the edge gluing data in $\boldsymbol{\Delta}$. 
\begin{definition}
    Suppose $\boldsymbol{\Delta}$ has length $k>2$. The tuple $(\phi_1,\dots,\phi_k)$ of 2-graph states $\phi_i$ localized on a regular simplicial decomposition $\boldsymbol{\Delta}$ of $\Gamma$ is \textbf{gluing-amenable} iff 
    \begin{enumerate}
        \item each adjacent localized pair $\phi_j,\phi_l$ of 2-graph states is gluing-amenable over the $(j,l)$-degeneracy intersection $u_{jl},$ where $j,l$ run over the indices of the set $\cG$,
        \item for each 2-subgraph of the form $\Gamma_+\subset\Gamma$, every localized 4-tuple $(\phi_1,\dots,\phi_4)$ on it has equipped a natural interchanger isomorphism $\beta_{12}^{34}$.
    \end{enumerate}
\end{definition}

\subsubsection{Graphical 2-holonomies and holonomy-density}
The data of the interchanger $\beta,$ as well as the strong associativity\footnote{In the undeformed case, this simply follows from the strong associativity of graph gluing. In the quantum case, we also require the strict Jacobi identity of the combinatorial 2-Fock-Rosly Poisson brackets. This is explained in more detail in \cite{Chen1:2025?}.} of $\ostar$, then allow us to construct 2-holonomy states on a generic regular simplicial decomposition $\boldsymbol{\Delta}$ of length $k> 2$ in a non-ambiguous manner.
\begin{definition}\label{2hol2mon}
    Let $(\phi_1,\dots,\phi_k)$ denote a tuple of 2-simplex states in $\C_q(\mathbb{G}^{\Gamma})$ which are gluing-amenable, then the associated \textbf{2-holonomy state on $\Gamma$} is the product
    \begin{equation}
        \Phi = \phi_1\ostar (W_{p_1}\phi_2)\ostar\dots\ostar (W_{p_k}\phi_k)\in\C_q(\mathbb{G}^{\Gamma}).\label{2holstate}
    \end{equation}
    When $\Gamma^2$ has no boundary, we call the associated sheaf $\Phi$ the \textbf{2-monodromy state}.
\end{definition}


From here on, we consider $\Gamma$ as a fixed lattice graph embedded in a PL 3d manifold $\Sigma$. The orientation of the PL path $p$ described in \textbf{Proposition \ref{path}} determines an orientation of the 2-simplices underlying the associated regular unbroken simplicial decomposition $\boldsymbol{\Delta}$ of $\Gamma$.

\begin{definition}
    Let $\boldsymbol{\Delta}$ denote a regular unbroken oriented simplicial decomposition of $\Gamma$. We say $\C_q(\mathbb{G}^{\Gamma})$ is \textbf{holonomy-dense} iff for every $\phi\in\C_q(\mathbb{G}^{\Gamma})$ there exist a gluing-amenable tuple $(\phi_1,\dots,\phi_k)$, localized on 2-simplices $\Delta_1,\dots,\Delta_k$ appearing in $\boldsymbol{\Delta}$, such that $\phi$ is measureably naturally isomorphic to 2-holonomy states $\Phi$ of the form \eqref{2holstate}. 
\end{definition}
We are then able to iteratively construct 2-graphs states $\phi\in\C_q(\mathbb{G}^{\Gamma^2})$ from the products of (gluing-amenable) states living on the fundamental 2-simplices $\Delta_j\in\boldsymbol{\Delta}.$ This is another expression of locality in our theory.

\subsection{Invariance modulo boundary}\label{invarbdy}
Fix a regular unbroken oriented simplicial decomposition $\boldsymbol{\Delta}$ of $\Gamma$. The above formulation of $\Phi$ is a direct generalization formulas given for the Chern-Simons holonomies in \cite{Alekseev:1994pa}, and they have the following analogous property.
\begin{theorem}
    Let $E^1=\{e^{t_j}_j\cong e^{s_l}_l\}_{(t_j;s_l)\in\cG}\subset\Gamma^1$ denote the \textit{rooted tree} of internal 1-graphs of $\Gamma$, consisting of edges across which the 2-simplices $\Delta\in\boldsymbol{\Delta}^2$ are glued upon. If $\phi\in\C_q(\mathbb{G}^{\Gamma^2})$ were holonomy-dense, then there is a measureable isomorphism $\Lambda_\zeta\phi\xrightarrow{\sim}\phi$ for all $\zeta\in\mathbb{U}_q\G^{E^1}$.   
\end{theorem}
\begin{proof}
    
    Recall from \S \ref{2gauge}, \S \ref{locality} that the the geometry/locality of the 2-gauge parameters $\mathbb{U}_q\G^{\Gamma^1}$ are dictated by the coproducts $\tilde\Delta$. We shall use to this to describe how 2-gauge transformations act on gluing-amenable 2-graph states.
    
    By definition, a 2-gauge transformation localized to the edges $e_j^{t_j}\in\partial\Delta_j,e_l^{s_l}\in\partial\Delta_l$ act as the measureable endofunctors $$\Lambda_{\zeta_j} = \Lambda_{(a_{v_j},\gamma_{e_j^{t_j}})},~\Lambda_{\zeta_l} = \Lambda_{(a_{v_l},\gamma_{e_l^{s_l}})}: \C_q(\mathbb{G}^{\Gamma})\rightarrow \C_q(\mathbb{G}^{\Gamma})$$ near the simplices $\Delta_j,\Delta_l$. Suppose now we specify the gluing data, namely a PL identification $f_\epsilon: e_j^{t_j}\cong e_l^{s_l}$ across which the localized 2-graph states $\phi_j,\phi_l$ are gluing-amenable. The derivation property \eqref{derivation} then supplies a module tensorator $\varpi$ \eqref{tensorator} such that 
        \begin{equation*}
        \big(\varpi_\zeta^{W_{p_j}\phi_j,W_{p_l}\phi_l}\big)^{-1}: \Lambda_{\zeta_j}(W_{p_j}\phi_j)\ostar \Lambda_{\zeta_l}(W_{p_l}\phi_l) \xrightarrow{\sim}(-\ostar-)\big ((\Lambda\times\Lambda)_{\tilde\Delta(\zeta)}(W_{p_j}\phi_j\times W_{p_l}\phi_l)\big),
    \end{equation*}
    as an invertible measureable natural transformation in $\C_q(\mathbb{G}^{p_j\ast\Delta_j\coprod p_l\ast\Delta_l})$. 
    
    By definition, the 2-gauge parameter $\zeta = \zeta_{f_\epsilon(v_j,e_j^{t_j})}\cdot \zeta_{(v_l,e_l^{s_l})}$ is obtained by horizontally stacking the 2-gauge transformations. However, given the path $p$ is endowed with a framing which agrees with $\Delta_j$, then the 2-simplex $\Delta_l$ interfacing with it must have the opposite framing. This framing reversal thus comes, according to \eqref{dagger2-gau}, with an antipode $\tilde S$ on $\mathbb{U}_q\G$,
    \begin{equation*}
        \zeta = (\tilde S\zeta)_{(v_j,e_j^{t_j})}\times \zeta_{(v_l,e_l^{s_l})};
    \end{equation*}
    see \textit{Remark \ref{gluingedges}}. 

    Given the counit $\tilde\epsilon$ and the unit $\tilde\eta=(1_v,({\bf 1}_1)_e)$ in $\mathbb{U}_q\G$ such that
    \begin{equation*}
        \Lambda_{\tilde\epsilon(\zeta)} = \id_\zeta,\qquad \Lambda_{\tilde\eta}=1_{\C_q(\mathbb{G}^{\Gamma^2})},
    \end{equation*}
    the Hopf axioms 
    $$(\tilde S\otimes 1)\tilde \Delta \cong (1\otimes \tilde S)\tilde \Delta\cong\tilde\epsilon\otimes \tilde\eta$$
    then provide an invertible natural transformation  
    \begin{equation*}
        (-\ostar-)\big ((\Lambda\times\Lambda)_{\tilde\Delta(\zeta)}(W_{p_j}\phi_j\times W_{p_l}\phi_l)\big) \cong \Lambda_{\tilde\epsilon(\zeta)\cdot\tilde\eta} (W_{p_j}\phi_j\ostar W_{p_l}\phi_l) \cong W_{p_j}\phi_j\ostar W_{p_l}\phi_l.
    \end{equation*}
    

    Due to the locality of the edges in $E^1$, we can repeat the above argument for each edge in $E^1$ such that we achieve an invertible measureable natural transformation $\varphi_\zeta$ on the 2-holonomy states,
    \begin{equation*}
        \varphi_\zeta^\Phi: \Lambda_\zeta\Phi\xrightarrow{\sim} \Phi,\qquad \forall~ \Phi. 
    \end{equation*}
    By holonomy-density, we can then extend this to all $\phi\in\C_q(\mathbb{G}^{\Gamma^2})$.
\end{proof}
\noindent Note this isomorphism is natural against measureable morphisms between holonomy-dense measureable fields: $f\circ \varphi^\phi_\zeta= \varphi_\zeta^{\phi'}\circ f$ for all $f:\phi\rightarrow\phi'$.

An immediate corollary is therefore the following.
\begin{corollary}\label{bdyinvar}
    If $\C_q(\mathbb{G}^{\Gamma^2})$ were holonomy-dense, then it is a homotopy fixed point under $\mathbb{U}_q\G^{E^1}$. Therefore, holonomy-dense \textbf{2-monodromy states are observable}: $$\partial\Gamma=\emptyset\implies \C_q(\mathbb{G}^{\Gamma})\subset \mathcal{O}^\Gamma.$$
\end{corollary}
\begin{proof}
    It is quick to verify that the natural transformations $\varphi$ satisfy the triangle axioms
    \begin{equation*}
        \varphi_{\zeta\cdot\zeta'}=\varphi_\zeta\ast(\Lambda_\zeta\circ \varphi_{\zeta'})\ast \alpha^\Lambda_{\zeta,\zeta'}
    \end{equation*}
    against the $\Lambda$-module associator $\alpha^\Lambda_{\zeta,\zeta'}:\Lambda_{\zeta\cdot\zeta'}(-) \Rightarrow \Lambda_\zeta\circ (\Lambda_{\zeta'}(-))$.
    The result then follows.
    
    

    The second statement follows directly from the invariance condition \eqref{invarstatement}.
\end{proof}

Note if $\zeta,\zeta'$ are localized on disjoint edges in $E^1$, then $\Lambda_\zeta,\Lambda_{\zeta'}$ commute up to a natural measureable isomorphism by locality (see \S \ref{locality}). 




    
\subsection{Disjoint commutativity modulo boundary}\label{commbdy}
We now turn to general simplicial decompositions of a 2-graph, in which each edge is not shared by necessarily at most two faces in $\boldsymbol{\Delta}$. To build such a structure up from the regular one, we first set up the local geometry, where a 2-simplex intersects a graph $\Gamma$ at one of its \textit{internal} edges. 

Provided $\Gamma$ itself has equipped a regular (unbroken oriented) simplicial decomposition $\boldsymbol{\Delta}$, there is then a 2-subgraph $\Gamma_e$ local to an internal edge $e\in E^1$, satisfying the property that its induced regular simplicial decomposition $\boldsymbol{\Delta}_e\subset\boldsymbol{\Delta}$ has size $k=2$. 

We fix the labels $\Delta_1,\Delta_2\in\boldsymbol{\Delta}_e$ and the associated gluing data on $e$ as a PL identification $e=e_1^{t_1}\xrightarrow{\sim}e_2^{s_2}$.  For simplicity, we shall pick the base point of $\Gamma_e$ to be contained within the glued edge. This is such that no whiskering needs to be performed when forming local holonomy-dense 2-graph states on $\Gamma_e$. 

Now suppose a third fundamental simplex $\Delta'$ intersects $\Gamma_e$ at its internal gluing edge $e$, whence this edge is shared by \textit{three} simplices. We denote the resulting graph by $\Gamma'_e$, which is equipped with a non-regular simplicial decomposition.

\subsubsection{Non-regular 2-graphs; triple points}\label{sec:triplepoint}
Prior to studying properties of the holonomy-dense 2-graph states on $\C_q(\mathbb{G}^{\Gamma'_e})$, we first promote our notion of "gluing-amenability" to non-regular simplicial decompositions. 

Suppose three fundamental 2-simplices $\Delta_1,\Delta_2,\Delta_3$ are incident upon the same edge $e$. Denote by $u_{123} = u_{12}\cap u_{13}\cap u_{13}$ the triple intersection of the pairwise degeneracy intersections $u_{12},u_{13},u_{23}$, and we label the pairwise sheaf automorphisms (here the indices $i,j,k$ are defined modulo 3) $$\alpha_{ij}: \phi_i\mid_{\mathbb{G}^{u_{ij}}} \cong \phi_j\mid_{\mathbb{G}^{u_{jk}}},\qquad 1\leq i<j<k\leq 3$$ as provided in \textbf{Definition \ref{glueamenable}}.

Under this configuration, we now introduce a $U(1)$-phase (resp. natural isomorphism of sheaves) $\sigma_{123}$ localized at the gluing edge $e$, which directly receives contribution from the \textit{Postnikov class} (resp. associator) $\tau$ of $\mathbb{G}$.
\begin{definition}\label{nonreg-glue}
    We say the triple $(\phi_1,\phi_2,\phi_3)_\sigma\in\C_q(\mathbb{G}^{\Gamma})$ is \textbf{gluing-amenable on the non-regular 2-subgraph $\Gamma_e'$} iff there is a $U(1)$-phase $\sigma_{123}\in U(1)$, localized on 2-holonomy decorations on the triple intersection $u_{123}$, such that the associated sheaf isomorphisms $\alpha_{ij}$ satisfies
\begin{equation*}
        \alpha_{23}\circ\alpha_{12} = \sigma_{123} \cdot \alpha_{13}.
    \end{equation*}
    If $\Delta_4$ is another 2-simplex incident upon this same edge $e$, then on the quadruple intersection $u_{1234}$ this phase satisfies the pentagon condition
    \begin{equation*}
        (\delta\sigma)_{1234} = \sigma_{234}\sigma_{134}^{-1}\sigma_{124}\sigma_{123}^{-1}=1,
    \end{equation*}
    which ensures that the assignment of $U(1)$-phases $\sigma$ is well-defined 
\end{definition}
\noindent The above pentagon condition bears a striking resemblance to {\v C}ech 2-cocycle conditions, hence we shall denote by $\check{H}(\mathbb{G}^{u_{123}},U(1))$ the space in which such $U(1)$-phases $\sigma$ live. As a slight abuse of language, we shall refer to $\sigma$ as the \textit{$U(1)$-gerbe} attached to a triple degeneracy intersection $u_{123}$.

\begin{rmk}\label{triplepoint}
    Geometrically, $\sigma$ witnesses the equivalence between the two ways in which the decorated 2-graphs on $\coprod_{i=1}^3\Delta_i$ into decorations on $\Gamma_e'$; see fig. \ref{fig:triple}. As such, \textbf{Definition \ref{intch}} is saying that each such triple point in a 2-graph is assigned a natural isomorphism $\sigma$. In the strict cast, these $\sigma$'s only have components proportional to the identity, and hence reduces to a $U(1)$-valued phase.
\end{rmk}

\begin{figure}[h]
    \centering
    \includegraphics[width=0.8\linewidth]{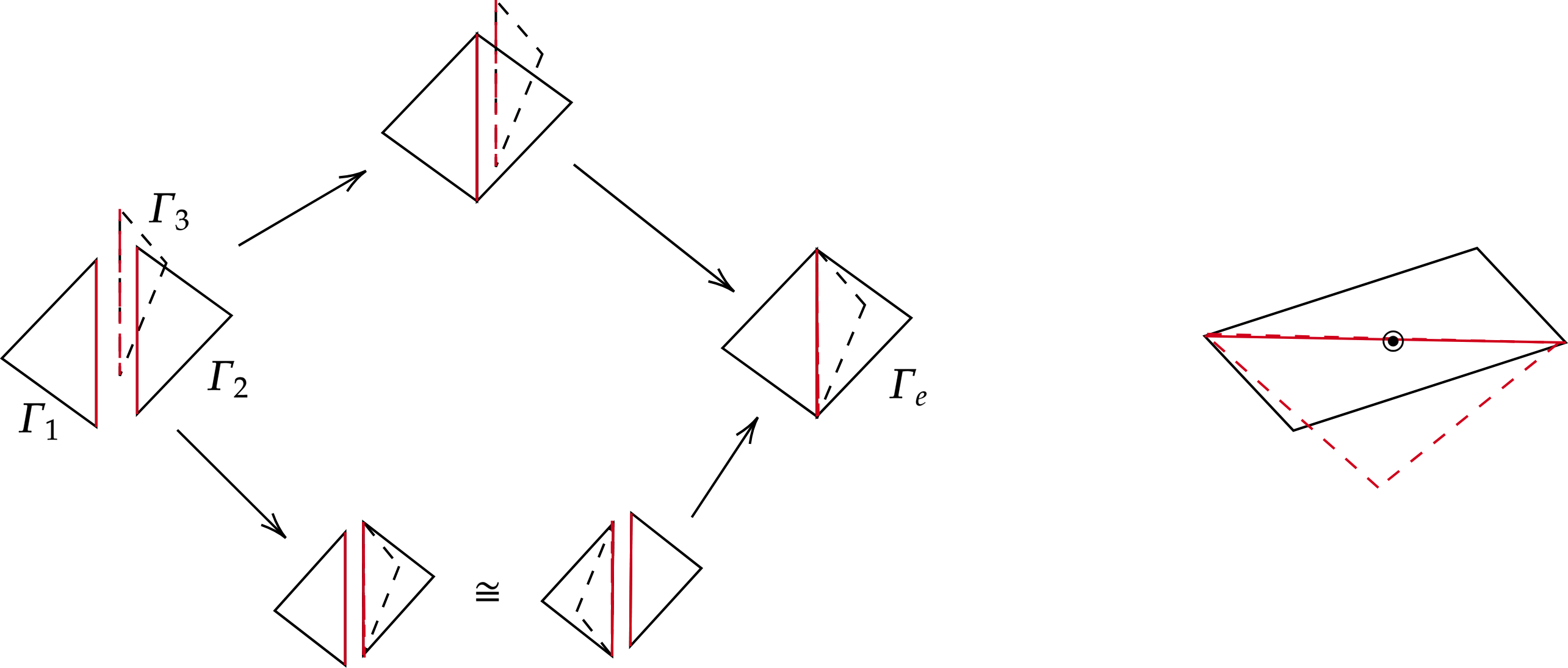}
    \caption{The left illustrates the geometric configuration of 2-simplices upon which the $U(1)$-gerbe $\sigma$ is defined. This geometry is precisely the \textit{triple point} in a singular graph \cite{matveev2007algorithmic} as displayed on the right; see also fig. 2 (b) of \cite{Sakata2022-il}.  }
    \label{fig:triple}
\end{figure}

The data $\sigma$ will be implicit in the following. 
\begin{theorem}\label{onecellcommute}
    Let $\C_q(\mathbb{G}^{\Gamma})$ be holonomy-dense, and let $\Delta'$ be another fundamental 2-simplex which intersects $\Gamma$ at one of its internal edges $e\in E^1$. Then provided $(\Phi_e,\phi')\in\C_q(\mathbb{G}^{\Gamma\coprod \Delta'})$ is gluing-amenable at the non-regular 2-subgraph $\Gamma_e'$, there exists measureable isomorphisms of sheaves
    \begin{equation*}
        \phi'\times \Phi_e \xrightarrow{\sim}\Phi_e\times \phi',\qquad \Phi_e\times \phi' \xrightarrow{\sim}\phi'\times \Phi_e
    \end{equation*}
    in $\C_q(\mathbb{G}^{\Gamma_e\coprod \Delta'})\cong \C_q(\mathbb{G}^{\Gamma_e})\times\C_q(\mathbb{G}^{\Delta'})$.
\end{theorem}
\begin{proof}
    Denote by the involved non-regular 2-subgraph $\Gamma_e'= \Gamma_e\cup\Delta'$, where $\Gamma_e\subset\Gamma$ is the 2-subgraph of the \textit{regular} 2-graph $\Gamma$ which meets the 2-simplex $\Delta'$ non-regularly. We then use holonomy-density to write $\phi\cong\Phi_e = \phi_1\ostar\phi_2\in\C_q(\mathbb{G}^{\Gamma})$ for any 2-graph state localized at $\Gamma_e'$, where $(\phi_1,\phi_2)\in\C_q(\mathbb{G}^{\Gamma})$ denote a tuple of 2-graph states, localized on $\Delta_1,\Delta_2$, which are gluing-amenable at the common edge $e\in E^1$. Note no whiskering needs to be done on $\Gamma_e$ as we have assumed that the base point $v$ of $\Gamma_e$ is contained in $e$.

    Now take some $\phi'\in\C_q(\mathbb{G}^{\Delta'})$. Given this setup, we then have a dense inclusion of sheaves of sections 
    \begin{equation*}
        (1\times -\ostar-) (\phi'\times\phi_1\times\phi_2)= \phi'\times (\phi_1\ostar\phi_2) \subset \phi'\times \Phi_e. 
    \end{equation*}
    By hypothesis, $\partial\Delta'\cap e\neq \emptyset$. If we pick the local framing of the interface $e$ to coincide with the framings of $\Delta_1$, then we have a measureable isomorphism of sheaves
    \begin{equation*}
        \phi'\times\phi_1 \cong (\Lambda\times \Lambda)_{\tilde R_e}(\phi_1 \times\phi')
    \end{equation*}
    by the \textit{braid relations} \eqref{braid}, where $\tilde R_e$ is the 2-$R$-matrix on $\mathbb{U}_q\G^e$. On the other hand, once we have fixed the framing of $e$ as above, it must be opposite to that of $\Delta_2$. Hence \eqref{dagger2-gau}
    \begin{equation*}
        \phi'\times\phi_2 \cong (\Lambda\times \Lambda)_{(1\times \tilde S)\tilde R_e}(\phi_2\times\phi');
    \end{equation*}
    see \textit{Remark \ref{gluingedges}}. 

We now combine these two computations through the gluing-amenability condition \textbf{Definition \ref{nonreg-glue}}. Using the module associator $$(\alpha^{\Lambda\times\Lambda}_{\tilde R,(1\times\tilde S)\tilde R}): (\Lambda\times\Lambda)_{\tilde R}\circ (\Lambda\times \Lambda)_{(1\times \tilde S)\tilde R}\Rightarrow (\Lambda\times\Lambda)_{\tilde R ~\hat\cdot~ (1\times \tilde S) \tilde R},$$
   together with one of the quasitriangularity axioms satisfied by the cobraiding $\tilde R$, 
   \begin{equation*}
        \tilde R \cdot (1\times \tilde S) \tilde R = \tilde\eta\times\tilde \eta,
    \end{equation*}
     we finally achieve a measureable isomorphism of sheaves
     \begin{equation*}
         \phi'\times\Phi_e\cong \phi'\times (\phi_1\ostar\phi_2) \xrightarrow{\sim} (\phi_1\ostar\phi_2)\times\phi'\cong \Phi_e\cong \phi',
     \end{equation*}
     as desired.

     Similar argument applies to produce a sheaf isomorphism 
     $\Phi_e\times \phi'  \xrightarrow{\sim} \phi'\times \Phi_e$
     from the other quasitriangularity axiom $$(\tilde S\times1) \tilde R \cdot \tilde R  = \tilde\eta\times\tilde \eta.$$
\end{proof}
\noindent Keep in mind that, in general, the above sheaf isomorphisms need \textit{not} be inverses of each other.

\medskip

In the following, we will often abuse notation to denote "$\phi\in\C_q(\mathbb{G}^{\Gamma'})$" by a 2-graph state $\phi\in\C_q(\mathbb{G}^{\Gamma})$ which is localized, in the sense of \textbf{Definition \ref{local2graphstates}}, on a 2-subgraph $\Gamma'\subset \Gamma$. The fact that $\C_q(\mathbb{G}^{\Gamma'})\subset \C_q(\mathbb{G}^{\Gamma})$ is a full measureable subcategory will be implicitly understood.

\subsubsection{Consistency with the interchanger}
We now wish to extend the above argument to \textit{any} regular graph $\Gamma'$ which meets the given $\Gamma$ at a collection of internal edges of $\Gamma$ in $E^1$. To do this, however, we need to understand how the $U(1)$-gerbes $\sigma$ "stack" against each other. This involves the planar interchanger $\beta$. 

The geometric setup is the following. Let $\Gamma_e,\Gamma_{e'}$ denote graphs of the form above: each consisting of three fundamental 2-simplices glued at the same edges $e,e'$, respectively. Given then edges $e,e'$ are composable
\begin{equation*}
    \exists v_\circ,\qquad e\cup_{v_\circ} e' = v\xrightarrow{e}v_\circ\xrightarrow{e'}v',
\end{equation*}
we can introduce additional gluing data which stacks these graphs together along (all) their source edges: $e_i^1\xrightarrow{\sim}e_i'^1$. We denote the resulting graph by $\Gamma=\Gamma_{e\cup_{v_\circ} e'}$. 


\begin{figure}[h]
    \centering
    \includegraphics[width=1\linewidth]{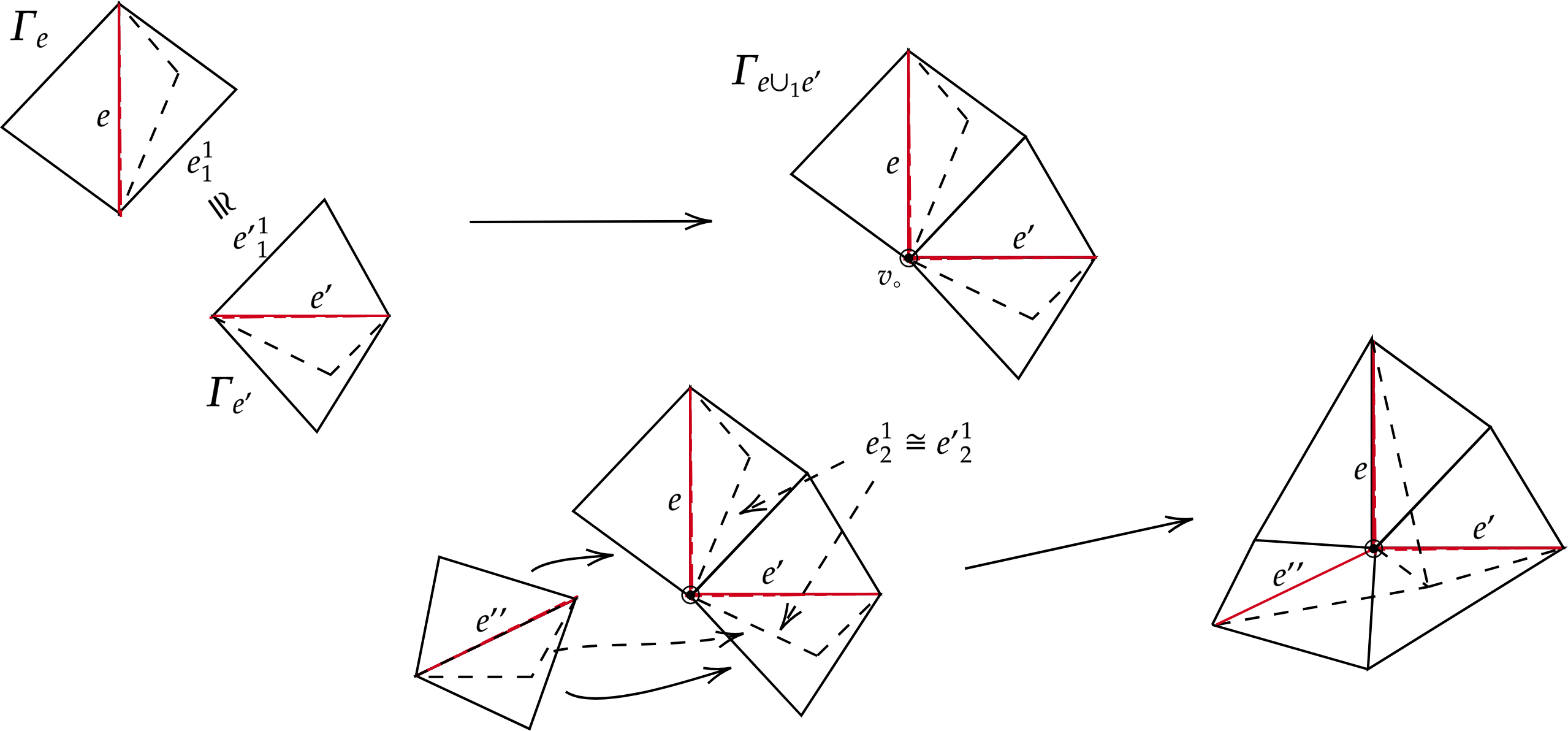}
    \caption{The geometric configurations involving the different gluing operations between non-regular triple point 2-graphs $\Gamma_e,\Gamma_{e'}$. The upper row displays their gluing along a single source edge $e_1^1\cong e_1'^1$, while the lower row displays a trivalent vertex formed by triple point 2-graphs.}
    \label{fig:triplecompose}
\end{figure}

The degeneracy neighbourhood around the central vertex $v_\circ$ then carries the data of \textit{both} of the $U(1)$ phases $\sigma_{123},\sigma_{1'2'3'}$ provided by \textbf{Definition \ref{nonreg-glue}}. This stacking of the graphs induces a "fusion operation" (cf. \cite{Waldorf2015TransgressiveLG}) on the $U(1)$-gerbes,
$$\cup_3=\cup: \check{H}^2(\mathbb{G}^u,U(1)) \otimes \check H^2(\mathbb{G}^u,U(1))\rightarrow \check H^2(\mathbb{G}^u,U(1)),\qquad u=u_{123}\cap u_{123}',$$
along the vertical composition operation $\otimes$ on the 2-graph sheaves. As such, we can denote the $U(1)$-gerbe attached to $\Gamma$ by $\sigma\cup\sigma'$.

\medskip

On the other hand, for $i=1,2,3$, let $\Gamma_i$ denote the graph consisting of two fundamental 2-simplices $\Delta_i,\Delta_i'$ glued along their source edges $e_i^1\cong e_i'^1$. If we introduce the following additional gluing data $(e^{(')})_i^2\cong (e^{(')})_{i+1}^3$ for $i-1 \in\bbZ_3$, then we also obtain the graph $\Gamma$ as defined above; see fig. \ref{fig:compat}. However, the $U(1)$ phase which is obtained in this manner is given instead by the following composite sheaf isomorphisms
\begin{equation*}
    (\alpha_{23}\ostar \alpha_{2'3'})\circ (\alpha_{12}\ostar \alpha_{1'2'}) = \sigma_{(11')(22')(33')}\cdot (\alpha_{13}\ostar \alpha_{2'3'})
\end{equation*}
near the central vertex $v_\circ.$ This also defines a $U(1)$-gerbe, which we denote by $\sigma\cdot\sigma'\in \check H^2(\mathbb{G}^{u},U(1))$. 

The notion of "gluing-amenability" for generic non-regular simplicial decompositions therefore must involve consistency relations between the $U(1$)-gerbes $\sigma\cup\sigma',\sigma\cdot\sigma'$ living on subgraphs of the form $\Gamma$. This is stated as follows. 

\medskip

Let $u=u_{123}\cap u_{123}'$ denote the degeneracy intersection around the central vertex $v_\circ$ of $\Gamma=\Gamma_{e\cup_{v_\circ}e'}=\bigcup_{i=1}^3\Gamma_i$. We now introduce the $U(1)$-phases $\gamma_{12},\gamma_{23},\gamma_{13}$ (see \textit{Remark \ref{commutephase}}) which witness the commutativity of $\alpha$ with the interchangers $\beta$,\footnote{Here we have abbreviated $\alpha_i=\alpha_{i,i+1}: \phi_i\mid_{u_{i,i+1}}\cong \phi_{i+1}\mid_{u_{i,i+1}}$ for $i=1,2,3$, where $\alpha_3=\alpha_{3,1}$.} 
    \begin{align*}
        \beta_{23}^{2'3'}\circ\big((\alpha_1\ostar\alpha_2)\otimes(\alpha_{1'}\ostar \alpha_{2'})\big)&= \gamma_{12}\cdot\big((\alpha_{1}\otimes\alpha_{1'})\ostar(\alpha_{2}\otimes\alpha_{2'})\big) \circ \beta_{12}^{1'2'},\\
        \beta_{31}^{3'1'}\circ\big((\alpha_2\ostar\alpha_3)\otimes(\alpha_{2'}\ostar \alpha_{3'})\big) &= \gamma_{23}\cdot \big((\alpha_{2}\otimes\alpha_{2'})\ostar(\alpha_{3}\otimes\alpha_{3'})\big) \circ \beta_{23}^{2'3'},\\
        \beta_{21}^{2'1'}\circ\big((\alpha_1\ostar\alpha_3)\otimes(\alpha_{1'}\ostar \alpha_{3'})\big) &= \gamma_{13} \cdot\big((\alpha_{1}\otimes\alpha_{1'})\ostar(\alpha_{3}\otimes\alpha_{3'})\big) \circ \beta_{13}^{1'3'}.
    \end{align*}
Geometrically, these $U(1)$-phases $\gamma$ witness the compatibility of the configuration of simplices indicated in fig. \ref{fig:compat}.

\begin{figure}[h]
    \centering
    \includegraphics[width=0.65\linewidth]{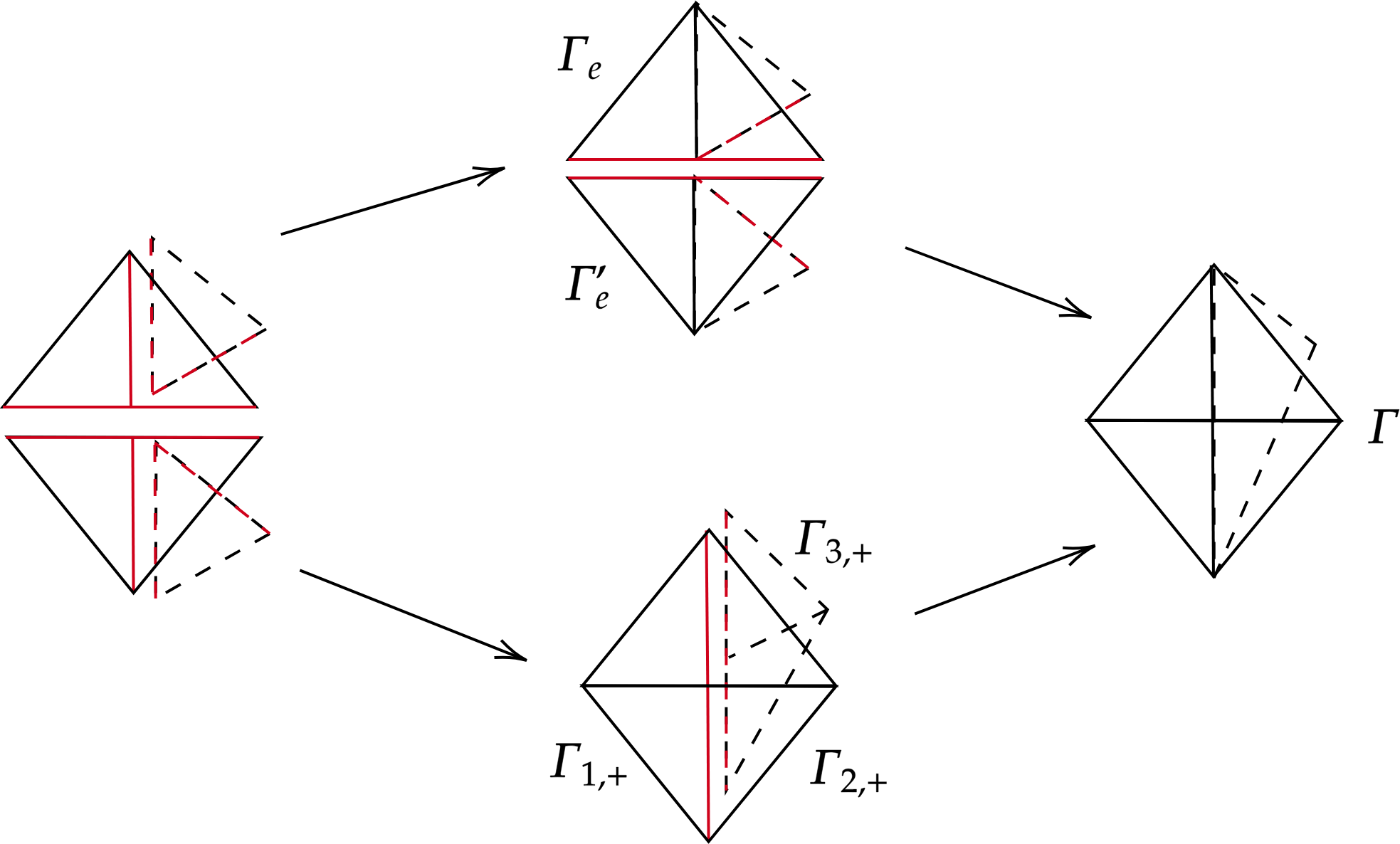}
    \caption{The figure illustrates the geometric configuration of 2-simplices upon which $\gamma$ witnesses the compatibility of the interchanger $\beta$ with the sheaf isomorphisms $\alpha$.}
    \label{fig:compat}
\end{figure}

The condition is then that these phases implements the consistency of the products $\cdot,\cup$,
\begin{equation}
    \gamma_{12} \gamma_{13}^{-1}\gamma_{23}=(\sigma_{123}\otimes\sigma_{123}') (\sigma_{(11')(22')(33')})^{-1},\nonumber
\end{equation}
By translating this into the language of the {\v C}ech cocycle $\delta \eta$, we have the following.

\begin{definition}\label{consistentgerbe}
    We say that the tuple $(\phi_1,\phi_2,\phi_3;\phi_1',\phi_2',\phi_3')$ of 2-graph states in $\C_q(\mathbb{G}^{\Gamma})$ localized, respectively, on the 2-simplices $\Delta_1,\Delta_1',\dots,\Delta_3,\Delta_3'$, is \textbf{gluing-amenable} on $\Gamma=\coprod_i(\Delta_i\cup\Delta_i')$ iff there exists a {\v C}ech 1-cocycle $\gamma \in Z^1(\mathbb{G}^{u},U(1))$ such that 
    \begin{equation}
        (\sigma\cup\sigma') = \delta\gamma  (\sigma\cdot \sigma')\label{gerbeinterchange}.
    \end{equation}
   In other words, the two operations $\cdot,\cup$ coincide in {\v C}ech cohomology on $\mathbb{G}^u$, where $u=u_{123}\cap u_{123}'$.
\end{definition}
\noindent This condition ensures that the $U(1)$-gerbe attached to states on graphs of the form $\Gamma=\Gamma_{e\cup_{v_\circ}e'}=\bigcup_{i=1}^3\Gamma_i$ is unambiguously $\sigma\cup\sigma'$. 

\begin{rmk}\label{commutephase}
    The quantity $\eta$ in general defines a sheaf isomorphism on quadruple tensor products of 2-graph states, hereby abbreviated as "$\phi^4$". However, $\eta$ is natural and only have components proportional to the identity in this case, which gives a $U(1)$-phase similar to $\sigma$ mentioned in \textit{Remark \ref{triplepoint}}.
\end{rmk}


\begin{example}\label{S3decomposition}
    Let $P\subset\R^3$ denote the union of the three coordinate planes in $\R^3$, and consider a 2-graph $\Gamma^2$ which triangulates $P\cap D^3$, where $D^3$ is the unit 3-disc. This geometric configuration consists of the stacking of two subgraphs of the form $\Gamma_e' \cup_e \Delta_4$, where $\Gamma_e'$ is the graph around a triple point as described in \textit{Remark \ref{triplepoint}}. In accordance with \textbf{Definition \ref{consistentgerbe}}, gluing-amenable 2-graph states on each wedge in $\Gamma=\Gamma^2$ has attached a $U(1)$-valued {\v C}ech 2-cocycle of the form $\sigma\cup\sigma'$. The difference between these gerbes across the wedges are described by precisely \textit{the Leibniz rule},
\begin{equation*}
    \delta (\sigma\cup\sigma') = \delta\sigma\cup\sigma' + \sigma\cup \delta\sigma' ,
\end{equation*}
whence the 2-cocycle condition in \textbf{Definition \ref{nonreg-glue}} says that the $U(1)$-gerbe attached to $\Gamma$ is unambiguously given by the {\v C}ech cohomology class of $\sigma\cup\sigma'$.    
\end{example}



 \subsubsection{Braiding properties of the 2-graph operator products}
We can now finally examine how each 2-graph states behave depending on the locality of the 2-graphs.


\begin{theorem}\label{brmon}
    For each 2-graph $\Gamma,\Gamma'$, define the functor
    \begin{equation}
        c:\C_q(\mathbb{G}^{\Gamma})\times \C_q(\mathbb{G}^{\Gamma'}) \rightarrow \C_q(\mathbb{G}^{\Gamma'})\times \C_q(\mathbb{G}^{\Gamma}),\qquad \Phi\times \Phi' \mapsto \text{flip}\circ\big((\Lambda\times \Lambda)_{\tilde R} \Phi\times\Phi'\big),\label{braidfunctor}
    \end{equation}
    where $\text{flip}$ is the swap of the Cartesian product factors. If $\Gamma^1\cap \partial\Gamma'$ contains at most 0-simplices, then there exists a trivialization $c\cong \text{flip}$.
\end{theorem}
\begin{proof}
    Note the functor $c$, as defined, depends on where $\mathbb{U}_q\G^{\Gamma^1}$ is localized --- namely how the 1-graph $\Gamma^1\hookrightarrow\Sigma$ is embedded into the 3d PL Cauchy surface in relation to the graphs $\Gamma,\Gamma'$. By \textbf{Corollary \ref{bdyinvar}}, 2-gauge transformations $\Lambda$ act non-trivially only on the boundary, hence we can without loss of generality assume $\Gamma^1 \subset \partial\Gamma$ is localized to the boundary of, say, the graph $\Gamma$.
    \begin{enumerate}
        \item \textbf{$\Gamma^1\cap \partial\Gamma'$ contains only 0-simplices}: $\Gamma'$ ends on a set $E=\Gamma^1$ of internal edges of $\Gamma$. We can then decompose $\Gamma=\Gamma_1\cup_E\Gamma_2$ further, whence by holonomy-density, we can apply the argument in \textbf{Theorem \ref{onecellcommute}} to each local graph intersection along $e\in E$. The condition \textbf{Definition \ref{consistentgerbe}} then allows us to extend this argument along composite internal edges $e\cup_{v_\circ}e'$, and hence to the entire collection $E$. This gives a natural isomorphism $c\Rightarrow {flip}$ which trivializes the braiding on  the gluing-amenable states $\C_q(\mathbb{G}^{\Gamma})\times_{E} \C_q(\mathbb{G}^{\Gamma'})$.
        \item \textbf{$\Gamma^1\cup\partial\Gamma'=\emptyset$ is empty}: in this case, $\Gamma,\Gamma'$ are disjoint, whence $\tilde R$ acts trivially by \eqref{braid}. The braiding functor $c$ {\it is} just the flip functor. 
    \end{enumerate}

    The final statement follows immediately from \textbf{Definition \ref{2hol2mon}}.
\end{proof}
\noindent  In other words, the extended operator insertions commute on 2-graphs with "decloaized boundaries".  This is the categorical analogue of Thm. 1 in \cite{Alekseev:1994au}: the closed plaquette elements $c^I(P)$ are central in $\cA_{CS}$. 


\begin{rmk}\label{2tangle}
    Recall \textbf{Definition \ref{categoricalquantumgroup}}. From \S \ref{deformed2gau} and \S \ref{lattice2alg}, the categorical quantum coordinate ring $\C_q(\mathbb{G})\in\operatorname{Mod}^*_\mathsf{Meas}(\mathbb{U}_q\G)$ is a measureable *-module category over $\mathbb{U}_q\G$. Due to the comonoidality of the cobraiding $\tilde R$ (or the \textit{higher-Yang-Baxter equations} satisfied by the 2-$R$-matrix, cf. \cite{Chen:2023tjf,Chen1:2025?}), the functor $c$ \eqref{braidfunctor} induces a braided monoidal structure on $\operatorname{Mod}^*_\mathsf{Meas}(\mathbb{U}_q\G)$ \cite{neuchl1997representation,Chen:2025?}.  If we further replace $\mathsf{Meas}$ with its finite-dimensional version $\mathsf{2Hilb}$, then we would recover the \textit{ribbon tensor} 2-category $\operatorname{2Rep}({\mathbb{U}}_q\G)$ of 2-representations studied in \cite{Chen:2025?}; see also \textit{Remark \ref{modelchange}}. 
\end{rmk}

An immediate consequence of \textbf{Theorem \ref{brmon}} is that 2-monodromy states --- namely the \textit{closed} Wilson surface states --- commute with all other 2-graph states. In the context of \textit{Remark \ref{2tangle}}, it means that the "closed-surface sector" of 2-Chern-Simons theory is contained within the \textit{$E_2$-centre} $Z_2(\operatorname{Mod}^*_\mathsf{Meas}(\mathbb{U}_q\G))$. This fact is a concrete manifestation of the general idea that the closed-brane sector of higher-dimensional QFT lies, in an appropriate sense, in the centre of the open-brane sector \cite{Kong:2011jf}. 

\begin{rmk}\label{openclosed}
By definition of the 2-holonomies $\mathbb{G}^{\Gamma^2}$ (\textbf{Definition \ref{2holdef}}), open Wilson surface states can only be described by the theory of \textit{non-Abelian gerbes} afforded by principal 2-bundles \cite{schreiber2013connectionsnonabeliangerbesholonomy,Waldorf:2012,Nikolaus2011FOUREV,Baez:2012}. Indeed, 2-gauge theories with a trivial structure map $\mu_1=0$ can only describe Abelian Wilson operators on closed surfaces \cite{Kapustin:2013uxa,Chen:2024axr,Alvarez:1997ma}, and not open-brane sectors.
\end{rmk}

\subsection{Orientation reversals and frame rotations}\label{*-opiso}
To close this section off, let us investigate the what the *-operations defined in \S \ref{*-op} imply through holonomy-density.
\begin{proposition}  
    Let $\C_q(\mathbb{G}^{\Gamma})$ be holonomy-dense, then there are meausreable natural isomorphisms
    \begin{equation*}
        -^{*_1}\xRightarrow{\sim} (\Lambda\otimes 1)_{{\tilde R^{-1}}}\circ (-^{\dagger_1}),\qquad -^{*_2}\xRightarrow{\sim}  -^{\dagger_2}
    \end{equation*}
    whose underlying measureable morphism at each component $\phi$ is given by the 2-$\dagger$-intertwining pair $\eta$. Here, each relevant $\tilde R$-matrices are localized on $\partial\Gamma$. If $\partial\Gamma$, then $\tilde R$ is localized on the base point $v\in \Gamma$.
\end{proposition}
\begin{proof}
    By holonomy-density, this follows directly from \textbf{Definition \ref{daggerpair}} and \textbf{Definition \ref{starop}}.
\end{proof}

What this means more explicitly is that there exist natural measureable isomorphisms which identify the following 2-graph states
\begin{equation*}
        \phi^{*_1} \cong  (\Lambda\otimes 1)_{\tilde R^{-1}}(\phi \circ -^{\dagger_1}),\qquad 
          \phi^{*_2} \cong \phi\circ(-^{\dagger_2})
    \end{equation*}
coming from the conditions in  \textbf{Definition \ref{2dagger}} as well as the module associator $\alpha^{\Lambda\otimes 1}_{\tilde R,\tilde R^{-1}}: (\Lambda\otimes 1)_{\tilde R}\circ(\Lambda\otimes 1)_{\tilde R^{-1}} \cong (\Lambda\otimes 1)_{{\tilde R~\hat\cdot \tilde R^{-1}}}\cong 1_{\C_q(\mathbb{G}^{\Gamma})}$.

Further, these natural isomorphisms commutes with those coming from the strong-commutativity $(-^{*_1})^\text{op}\circ -^{*_2}\cong (-^{*_2})^\text{m-op,c-op}\circ -^{*_1}$ of the *-operations. 

\medskip

\begin{definition}\label{2flat}
    The \textbf{flatness of the 2-holonomies} is the notion that, if $V$ is a contractible 3-cell, then $\prod_{f\in\partial V}b_f=1$ for all $\mathrm{z}=\{(h_e,b_f)\}_{(e,f)}\in \mathbb{G}^{\partial V}$. As such, if $V$ is represented by a PL homotopy $\Gamma\Rightarrow\Gamma'$ then the 2-holonomies on $\Gamma,\Gamma'$ are 2-gauge equivalent.

    This is well-known fact in \textit{strict} higher-gauge theory \cite{Baez:2004in,schreiber2013connectionsnonabeliangerbesholonomy,Martins:2010ry,Mikovic:2016xmo,Bullivant:2016clk,Bochniak_2021}.
\end{definition}
\noindent By "full-stacking", we mean a PL identification of two 2-simplices \textit{everywhere} (ie. not just at one of their edges).

 \begin{rmk}\label{weak2ribbons}
        In weak 2-Chern-Simons theory, the Postnikov class of $\mathbb{G}$ \cite{Schommer_Pries_2011,Chen:2022hct,Kapustin:2013uxa,Baez2023HoangXS} gives the anomaly/defect that breaks precisely the 2-flatness condition \cite{Kim:2019owc,schreiber2013connectionsnonabeliangerbesholonomy,Sati:2008eg}: $\prod_{f\in\partial V}h_v = \tau_{h_{e_1},h_{e_2},h_{e_3}}$. This leads to non-trivial modifications between whiskering pseudonaturals as described in \textit{Remark \ref{pachnertau}}, and also induce a \textit{first descendant} modification between 2-gauge transformations (this was described in \cite{Chen1:2025?}). The presence of $\tau$ necessitates the categorification step \textit{Remark \ref{categorificatione}}, and one in general should not truncate the 2-gauge transformations to an internal 0-category.
    \end{rmk}

 We now leverage 2-flatness to prove a categorical, "basis-independent" analogue of Prop. 7 in \cite{Alekseev:1994au}. 
\begin{proposition}\label{orientationreverse}
    Suppose $\Gamma=\Delta\cup_\Delta\bar\Delta$ consist of the full-stacking of a fundamental 2-simplex $\Delta$ with its orientation reversal $\bar\Delta=\Delta^{\dagger_1}$, then holonomy-dense 2-graph states on $\Gamma$ is trivial: $\C_q(\mathbb{G}^{\Gamma})\simeq\mathsf{Hilb}$.
\end{proposition}
\begin{proof}
   The full-stacking of $\Delta$ and its orientation reversal $\bar\Delta$ gives rise to a \textit{closed} 2-graph $\Gamma^2$ which comes equipped with a null-homotopy $\Gamma^2 \simeq v$. Thus by 2-flatness \textbf{Definition \ref{2flat}}, the 2-holonomies on $\bar\Delta,\Delta$ are 2-gauge equivalent: for each fixed $\mathrm{z}\in\mathbb{G}^{\bar\Delta}$ and $\mathrm{z}'\in  \mathbb{G}^\Delta$, we can find a 2-gauge transformation $\zeta\in\mathbb{G}^{\Gamma^1}$ for which $\operatorname{hAd}^{-1}_\zeta\mathrm{z} = \mathrm{z}'$ --- or, in other words, $\mathrm{z}^{-1_h}\cdot_h\mathrm{z}'$ is a \textit{pure 2-gauge}.

    Therefore, through holonomy-density and 2-$\dagger$ unitarity \S \ref{2dagger}, each 2-graph state $\Phi={\phi}\otimes\phi'\in\C_q(\mathbb{G}^{\Gamma})$ by 2-flatness is a pure 2-gauge state (namely one with support only on pure 2-gauge 2-holonomies). By construction, pure 2-gauge holonomy configurations can be removed by a 2-gauge transformation \S \ref{2gauge}. But since $\Gamma^2$ has no boundary, $\C_q(\mathbb{G}^{\Gamma^2})$ only has 2-monodromy states, which are 2-gauge invariant up to homotopy by \textbf{Proposition \ref{bdyinvar}}. 
    
    This means that there is a measureable isomorphism $\Phi \cong \eta$ to the unit $\eta \in\C_q(\mathbb{G}^\Gamma)$, which removes \textit{all} of the 2-holonomy decorations on any 2-graph state $\Phi\in\C_q(\mathbb{G}^\Gamma)$. The unit, by definition, can be viewed as a full measureable functor $\C_q(\mathbb{G}^v)\simeq\mathsf{Hilb}\rightarrow \C_q(\mathbb{G}^\Gamma)$ from states on the trivial 2-graph $v$. The above argument then means that every 2-graph state in $\C_q(\mathbb{G}^\Gamma)$ lives in the essential image of this functor, giving us the desired equivalence
    \begin{equation*}
        \C_q(\mathbb{G}^{\Gamma^2})\simeq\mathsf{Hilb}.
    \end{equation*}
\end{proof}

Gluing-amenability then allows us to extend \textbf{Proposition \ref{orientationreverse}} to entire 2-graphs.
\begin{proposition}\label{reversestack}
    Let $\bar\Gamma=\Gamma^{\dagger_1}$ denote the orientation reversed simplicial complex of $\Gamma$, then there is an equivalence
    \begin{equation*}
        \C_q(\mathbb{G}^{\Gamma\cup_\Gamma\bar\Gamma}) \simeq\mathsf{Hilb}
    \end{equation*}
on holonomy-dense 2-graph states on the full-stacking $\Gamma\cup_\Gamma\bar\Gamma$.
\end{proposition}
\begin{proof}
    By gluing-amenability, we can use the interchanger isomorphisms $\beta$ \textbf{Definition \ref{intch}} to break 2-graph states on $\Gamma\cup_\Gamma\bar\Gamma$ to a product of 2-graph states on the stacking $\Delta_j\cup_{\Delta_j}\bar\Delta_j$ of each fundamental 2-simplex $\Delta_j$ contained in $\Gamma$. The result then follows by applying \textbf{Proposition \ref{orientationreverse}} repeatedly.
\end{proof}

\medskip

The results of these sections, \S \ref{invarbdy}, \S \ref{commbdy} and \S \ref{*-opiso}, are direct higher-dimensional generalizations of part (1), (2) and (3) of Proposition 2, 3 in \cite{Alekseev:1994au}.\footnote{That is, except the first formula in part (3) of these propositions. This formula expands the tensor products of the quantum algebra $C_q(G^{\Gamma^1})$ in a basis, resulting in the Clebsch-Goran coefficients. We had not done this here, as to do so for $\C_q(\mathbb{G}^{\Gamma^2})$ we require a categorical Peter-Weyl theorem. We leave this to a companion work.} Though many subtleties arise in the weak case (cf. \textit{Remark \ref{weakcoherence}}), we expect lax versions of the results of these sections to continue to hold.


\section{Categorified states: additive measureable *-functors}\label{1morA}
Recall that the usual notion of a \textit{normalized state} on a unital $C^*$-algebra $A$ is a linear funciotnal $\psi:A\rightarrow\bbC$ for which $\psi(1)=1$ \cite{book-operators,Woronowicz1988}. The space of such linear functionals serves as the physical Hilbert space of states in the quantum theory. 

The goal in this section is to introduce a categorified version of these states. The guiding principle is once again $\mathsf{Meas}$, the 2-category of measureable categories \cite{Yetter2003MeasurableC}. Indeed, there is a natural equivalence $\mathsf{Hilb}\simeq\cH^\emptyset$ with the measureable category over the empty set. Moreover, considering $\mathsf{Meas}$ as a monoidal bicategory (see Thm. 50, \cite{Yetter2003MeasurableC}), $\mathsf{Hilb}$ is the monoidal identity.

\begin{tcolorbox}[breakable]
    \subsubsection*{Global measureable change of basis.}Let $\{H_x\}_{x\in X}$ be a family of Hilbert spaces over the measure space $(X,\mu)$ and let $R$ be a local ring over $\bbC$ (such as when $R=C(Y),L^2(Y,\mu')$ for some other manifold/measure space $(Y,\mu')$). The following proposition will be useful.
\begin{proposition}\label{analysisandalgebra}
    If each $H_x$ is a (finitely-generated projective) $R$-module, then the direct integral $\displaystyle\int_X^\oplus d\mu(x) H_x$ is a (finitely-generated projective) $R$-module. Conversely, if $H$ is a $R$-module and admits a direct integral decomposition $\displaystyle\int_X^\oplus d\mu(x) H_x$, then each $H_x$ is also a $R$-module.
\end{proposition}
\begin{proof}
    If $v\sim_\mu v$ are $\mu$-a.e. equivalent sections in $\coprod_{x\in X}H_x$, then $v-u\sim_\mu 0$, hence $r\cdot u - r\cdot v = r\cdot (u-v)\sim_\mu 0$ and hence $r\cdot u\sim_\mu r\cdot u$ are also $\mu$-a.e. equivalent sections for any $r\in R$. The converse is a special case of a theorem in the work of Segal \cite{Segal1951DecompositionsOO} (see also Thm. 1.2 (iii) in \cite{Frank:1993}), where we simply replace the $W^*$-algebra $A\cong L^\infty(X,\mu)$ with $A\otimes_\bbC R$.
\end{proof}
\noindent In other words, if $R$ is "constant across $X$", then the direct integral will also inherit the $R$-module structure and vice versa.\footnote{
    The author believes that there should be a much more general version of the above statement where $R$ is allowed to be local along $X$, provided the local $R_x$-module structure is allowed to vary in a $\mu$-essentially bounded manner across $x\in X$. We will not need such a powerful statement here, however.}
\end{tcolorbox}


\subsection{Categorical linear *-functionals on 2-graph states}\label{3handles}
In accordance with the above setup, we will model such "categorical linear functionals" as an additive measureable functor of sheaves 
\begin{equation*}
    \omega: \C_q(\mathbb{G}^{\Gamma^2})\rightarrow \mathsf{Hilb},
\end{equation*}
where we are considering $\mathsf{Hilb}$ as the category of sections of Hermitian vector bundles over the singleton $\ast$. Here, additive means that $\omega$ respects the direct sum of sheaves, but it need \textit{not} respect any monoidal structure!

In this section, we will prove a Yoneda embedding \textbf{Proposition \ref{yoneda}} for $\C_q(\mathbb{G}^{\Gamma})\subset \mathsf{Meas}_X$ by just treating it as a full subcategory of  measureable fields of over $X=(\mathbb{G}^{\Gamma},\mu_{\Gamma^3})$, as in \textbf{Definition \ref{quantumhermitian}}. We will come back to deal with the \textit{internal}/double cocategory structure in \S \ref{nonabeliansurface}.


\subsubsection{Evaluation states; cone functors on $\Lambda\Gamma^2$}
We begin with a connected PL 2-manifold $S$ equipped with an oriented simplicial decomposition $\boldsymbol{\Delta}$. The resulting graph $\Gamma$ of $S$, obtained from the gluing data attached to $\boldsymbol{\Delta}$ is a convex simplicial space.

To set up the geometry, we first recall from \cite{hudson1969piecewise}.
\begin{definition}
    The \textbf{convex sum} of two convex sets $A,B\subset \mathbb{R}^N$ is 
    \begin{equation*}
        A+_cB=\{\lambda a+(1-\lambda)b\mid a\in A,~b\in B,~\lambda\in[0,1]\}.
    \end{equation*}
    The \textbf{one-point suspension} $\Lambda A$ of $A$ is the convex set $A+_c\{\ast\}$ where $\ast\in\mathbb{R}^N$ is some point which is non-colinear with any $a\in A$.
\end{definition}
\noindent The non-colinearity assumption is required such that, if $A=\Delta^{n}$ is a $n$-simplex, then its one-point suspension $\Delta^{n+1}= \Lambda\Delta^n$ is the $(n+1)$-simplex. 

Suppose $\Sigma \cong CS$ is the \textit{PL cone} over $S$, then if $S$ has equipped a simplicial decomposition by the graph $\Gamma$, then $\Sigma$ has equipped a simplicial decomposition given by the on-point suspension $\Lambda\Gamma$. For instance, if $S=S^2$ were the PL 2-sphere, then $\Sigma$ is homeomorphic to the PL 3-disc $D^3$. 

We shall focus on this case first. Let $\Gamma$ be a \textit{connected} 2-graph. 
\begin{definition}\label{linearfunctors}
    Denote by $\eta\in\C_q(\mathbb{G}^{\Gamma^2})$ the unit, and $\emptyset = \mathbb{G}^{\emptyset}$ the trivial decorated 2-graph. A {\bf categorical state} associated to the one-point suspension $\Lambda\Gamma$, also referred to as a \textbf{cone functor}, is an additive measureable functor $$\omega=\omega_{\Lambda\Gamma}:\C_q(\mathbb{G}^{\Gamma})\rightarrow\cH^{\emptyset}\simeq\mathsf{Hilb},$$ for which $\omega(\eta)\in\mathsf{Hilb}^\text{f.d.}$ is of finite-dimension. 
\end{definition}

By definition, $\omega$ comes with an underlying field $\underline{\omega}$ of Hilbert spaces on $\ast\times\mathbb{G}^{\Gamma}=\mathbb{G}^{\Gamma}$, such that
\begin{equation*}
    \omega(\phi) = \int_{\mathbb{G}^{\Gamma}}^\oplus d\nu_{\Gamma^2}(\mathrm{z}) \underline{\omega}_\mathrm{z}\otimes \phi_\mathrm{z},\qquad \phi\in\C_q(\mathbb{G}^{\Gamma})
\end{equation*}
where $\nu_{\Gamma^2}$ is another Haar measure $\mu'_{\Gamma^2}$ on $\mathbb{G}^{\Gamma}$.


\subsubsection*{An infinite-dimensional Yoneda embedding.}
One crucial fact to keep in mind is that the data $\underline{\omega}$ does \textit{not} itself determine a measureable field in general. Indeed, the space $\cM_\omega\subset\coprod_x \underline{\omega}_x$ of measureable sections is not specified. 

However, we do have access to a \textbf{Yoneda embedding}, which in the context of \textit{Remark \ref{doublecocat}} is a instance of the double Yoneda lemma (Thm. 4.1.2 in \cite{froehlich2024yonedalemmarepresentationtheorem}).
\begin{proposition}\label{yoneda}
    There is a fully-faithful embedding $\C_q(\mathbb{G}^{\Gamma^2})^{\text{m-op}}\rightarrow \operatorname{Fun}(\C_q(\mathbb{G}^{\Gamma^2}),\mathsf{Hilb})$, where $\C_q(\mathbb{G}^{\Gamma^2})^{\text{m-op}}$ denotes the opposite algebra object in $\mathsf{Meas}$.
\end{proposition}
\begin{proof}
    The embedding takes a 2-graph state $\bar\phi'\in \C_q(\mathbb{G}^{\Gamma^2})^{\text{m-op}}$, linear dual to one $\phi'\C_q(\mathbb{G}^{\Gamma^2})$, to a measureable functor $(\omega_{\phi'},\mu_{\Gamma^2})$ of the form
    \begin{equation}
    \omega_{\phi'}(\phi) = \int_{\mathbb{G}^{\Gamma}}^\oplus d\mu_{\Gamma^2}(\mathrm{z}) \bar\phi'_\mathrm{z}\otimes \phi_\mathrm{z},\qquad \phi\in\C_q(\mathbb{G}^{\Gamma})^{\text{m-op}};\label{pairing}
\end{equation}
see \textit{Remark \ref{opposites}}.


The full-faithfulness is obvious by \textbf{Definition \ref{measnat}}: each natural transformation $\omega_\phi\Rightarrow \omega_{\phi'}$ correspond to a bounded linear operator $\beta: \phi=\Gamma_c(H^X)\rightarrow \Gamma_c(H'^X)=\phi'$ of measureable sheaves.
\end{proof}

\begin{rmk}\label{opposites}
    We emphasize that, by $C^{\text{m-op}}$ for a category $C=(C_0,C_1,\operatorname{id},\circ)$ internal to $\cV=\mathsf{Meas}$, it means the monoidal structure $\otimes$ and the compositions on the measureable categories $C_0,C_1$ are reversed. On the other hand, for the 2-graph states, the direct image functors induced by the 2-$\dagger$ structures on $\Gamma$ are a priori \textit{covariant} on $C_{1}\rightarrow C_{1}$ in $\cV$, but reverses the "internal" composition $\circ$. The unitarity property of \textbf{Definition \ref{unitary2hol}} mixes both, and makes the *-operations into an $\text{m-op}$ contravariant functor.
\end{rmk}


\begin{rmk}\label{pair}
    This embedding, and the formula \eqref{pairing}, determines a \textit{categorical pairing form}
\begin{equation}
    \C_q(\mathbb{G}^{\Gamma^2})^{\text{m-op}} \times \C_q(\mathbb{G}^{\Gamma^2})\rightarrow \operatorname{Fun}(\C_q(\mathbb{G}^{\Gamma^2}),\mathsf{Hilb})\times\C_q(\mathbb{G}^{\Gamma^2})\xrightarrow{\operatorname{eval}}\mathsf{Hilb},\label{pairfunctor}
\end{equation}
which was used in \cite{Chen:2025?} as a "duality evaluation" for $\C_q(\mathbb{G}^{\Gamma^2})$.\footnote{Such pairing functors, if Frobenius, was also used by \cite{Crane:1994ty} as part of the definition of a Hopf category. However, we will not be using that notion here.} This categorfies the pairing functional $\langle\Psi_2\mid\Psi_1\rangle = \omega(\bar\psi_2(U)\psi_1(U))$ defined on the 3d Chern-Simons holonomies $\psi(U)$  as constructed in \S 6.2 of \cite{Alekseev:1994pa}. 
\end{rmk}

A perhaps unfortunate fact is the following.
\begin{proposition}\label{notequiavlence}
    The embedding $\phi'\mapsto\omega_{\phi'}$ \eqref{pairing} is $\mathrm{not}$ essentially surjective.
\end{proposition}
\begin{proof}
    By \textbf{Definition \ref{measnat}}, a measureable natural isomorphism $(\omega,\nu_{\Gamma^2})\Rightarrow \omega_{\phi}$ to one coming from a 2-graph state $\phi$ consist of (i) a Haar measure equivalent  to $\mu_{\Gamma^2}$, and (ii) a field of $\mu_{\Gamma^2}$-essentially bounded sheaf of invertible operators $K: \omega\rightarrow\phi$. 
    
    We know from \textbf{Proposition \ref{haarunique}} that (i) is not problematic. On the other hand, if a sheaf of invertible operator $K$ in (ii) exists, then $\omega\in\cV^X$ itself must be a measureable sheaf of Hermitian sections. The existence of $K$ for all $\omega$ means that $\cV^X\simeq\cH^X$ are equivalent, which is of course not the case.
    
    Indeed, in the language of sheaves \textit{Remark \ref{measureablefieldsassheaves}}, (ii) says that we can find a field of bounded isomorphisms from \textit{any} Hilbert $W^*$-module to a Hilbert $C^*$-module, which is not possible in general.
\end{proof}
\noindent This issue is a consequence of the infinite-dimensional nature of the structures involved. 

Indeed, this result is a categorical analogue of the fact that there is no isomorphism between test functions and tempered distributions \cite{book-mathphys1}. This will show up again later in \textbf{Proposition \ref{cylinderembeddings}}.

\subsubsection{Transition states; cylinder functors on $\Gamma^2\times [0,1]$}
Consider the following geometry. Let $\Sigma \cong S\times[0,1]$ be a manifold diffeomorphic to the cylinder on $S$. Equip $\Sigma$ with a PL structure $C: \boldsymbol{\Delta}\rightarrow \Sigma$ which defines a homotopy between the given PL structures $\Gamma_0,\Gamma_1:\boldsymbol{\Delta}\rightarrow S\times\{0,1\}$ on the two copies of $S$. 

We now wish to define the categorical functional $\omega_C$ associated to the cylinder graph $C$. 
\begin{definition}
    The {categorical functional associated to the homotopy $C$}, or simply a \textbf{cylinder functor}, is a unit-preserving additive measureable functor $$\omega=\omega_{C}:\C_q(\mathbb{G}^{\Gamma_0})\rightarrow\C_q(\mathbb{G}^{\Gamma_1}),$$ such that the target is once again a 2-graph state.
\end{definition}
\noindent Let us spell out what this means. Keep in mind that $\Gamma_0,\Gamma_1$ are \textit{disjoint}.

A priori, the data of this additive measureable functor $\omega_C$ involves an underlying field $\underline{\omega}$ of Hilbert spaces over $\mathbb{G}^{\Gamma_1}\times\mathbb{G}^{\Gamma_0}$, together with a $\mathbb{G}^{\Gamma_1}$-family of measures $\{\nu_{{\mathrm{z}}}\}_{{\mathrm{z}}\in\mathbb{G}^{\Gamma_1}}$ on $\mathbb{G}^{\Gamma_0}$, such that
\begin{equation*}
    \omega_C(\phi)_{{\mathrm{z}}} = \int_{\mathbb{G}^{\Gamma_0}}^\oplus d\nu_{{\mathrm{z}}}(\mathrm{z}') \underline{\omega}_{{\mathrm{z}},\mathrm{z}'}\otimes \phi_\mathrm{z'},\qquad \phi\in\C_q(\mathbb{G}^{\Gamma_0}),~ \bar{\mathrm{z}}\in\mathbb{G}^{\bar\Gamma_1}.
\end{equation*}
This is not enough, however, as general measureable functors $\omega_C$ may not produce a 2-graph state. An additional requisite condition is the following: that for each Borel subset $ U\subset \mathbb{G}^{\Gamma_1}$, the assignment
\begin{equation*}
    U\mapsto \int_{U}^\oplus d\mu_{\Gamma_1}({\mathrm{z}}) \omega_C(\phi)_{{\mathrm{z}}}
\end{equation*}
defines a sheaf of $L^2$-sections $\Gamma_c(H^{X_1})$ of a Hermitian vector bundle $H^{X_1}\rightarrow X_1$ over $ X_1=(\mathbb{G}^{\Gamma_1},\mu_{\Gamma_1})$. This puts constraints on $\underline{\omega}$. 

\medskip

Prior to proceeding, we first introduce the following notion.
\begin{definition}
    We say the Radon measures $(\mu,\mu')$ are a \textbf{disintegration pair} on $Y\times X$ iff for each $Y$-family $\{\nu_y\}_{y\in Y}$ of disintegration measures, there is a $X$-family $\{\nu'_x\}_{x\in X}$ of disintegration measures such that 
    \begin{equation*}
        \int_{Y}d\mu(y)\int_{X}d\nu_y(x)f(y,x) = \int_{Y\times X}d\lambda(y,x) f(y,x) = \int_Xd\mu'(x)\int_Yd\nu'_x(y)f(y,x)
    \end{equation*}
    for all measureable function $f$ on $X\times Y$. Here, $\lambda$ is a measure on $Y\times X$ which is obtained by "integrating" $\nu_y$ against $\mu$, or "integrating" $\nu'_x$ against $\mu'$.
\end{definition}
The existence and uniqueness of disintegration pairs \cite{Pachl_1978} (see also Thm. 23 in \cite{Baez:2012} and Lemma 2.3 in \cite{ACKERMAN_2016}) gives the following.
\begin{proposition}\label{uniqueness}
    We have a disintegration pair $(\mu,\mu')$ whenever
\begin{equation*}
    \mu(U)=0\implies \lambda(U\times X) = 0,\qquad \mu'(V) \implies \lambda(Y\times V)=0
\end{equation*}
for each measureable $U\subset Y,~V\subset X$. In which case, they are unique.
\end{proposition}


\subsubsection*{Characterizing $\omega_C$ and pairings along the cylinder.} 
Let us now try to characterize $\omega_C$ on the cylinder under the assumption that the given Haar measures $(\mu_{\Gamma_0},\mu_{\Gamma_1})$ form a disintegration pair on $\mathbb{G}^{\Gamma_1}\times\mathbb{G}^{\Gamma_0}$. 

For each Borel $U\subset \mathbb{G}^{\Gamma_1}$, we rewrite the direct integral of $\omega_C(\phi)$ in the following way,
    \begin{align*}
        \int_{U}^\oplus d\mu_{\Gamma_1}({\mathrm{z}}) \omega_C(\phi)_{{\mathrm{z}}}&= \int_{U}^\oplus d\mu_{\Gamma_1}({\mathrm{z}})\int_{\mathbb{G}^{\Gamma_0}}^\oplus d\nu_{\mathrm{z}}(\mathrm{z}')\underline{\omega}_{{\mathrm{z}},\mathrm{z}'}\otimes \phi_\mathrm{z'}\\
        &= \int_{\mathbb{G}^{\Gamma_0}}^\oplus d\mu_{\Gamma_0}({\mathrm{z}}')\int_{U}^\oplus d\nu'_{\mathrm{z}'}(\mathrm{z})\underline{\omega}_{{\mathrm{z}},\mathrm{z}'}\otimes \phi_\mathrm{z'}\equiv  \int_{\mathbb{G}^{\Gamma_0}}^\oplus d\mu_{\Gamma_0}({\mathrm{z}'}) (\Omega_{\mathrm{z}'})_{/U}\otimes \phi_\mathrm{z'},
    \end{align*}
    which gives us a $\mathsf{Hilb}$-valued presheaf on $\mathbb{G}^{\Gamma_1}$,
   \begin{equation*}
        \Omega_\mathrm{z'}: U\mapsto (\Omega_\mathrm{z'})_{/U} = \int_{U}^\oplus d\nu'_{\mathrm{z}'}(\mathrm{z})\underline{\omega}_{{\mathrm{z}},\mathrm{z}'},\qquad \mathrm{z'}\in\mathbb{G}^{\Gamma_0}
    \end{equation*}
    for each $\mathrm{z'}\in\mathbb{G}^{\Gamma_0}$.
        
Recall from Lemma 4.3 of \cite{Bridgeland:1999} that a $S$-family of sheaves on $X$ is a sheaf on $X\times S$ which is flat over $S$. We then have the following characterization.
\begin{proposition}\label{cylinder}
    Suppose $(\mu_{\Gamma_0},\mu_{\Gamma_1})$ forms a disintegration pair. Then $\omega_C(\phi)\in\C_q(\mathbb{G}^{\Gamma_1})$ is a 2-graph state for all $\phi\in\C_q(\mathbb{G}^{\Gamma_0})$ iff $\Omega$ defines a $\mathbb{G}^{\Gamma_0}$-family of sheaves of finitely-generated projective $C(\mathbb{G}^{\Gamma_1})$-modules of $L^2$-sections on $\mathbb{G}^{\Gamma_1}$.
\end{proposition}
\begin{proof}
    The hypotheses guarantee that the sheaf $U\mapsto \displaystyle\int_{\mathbb{G}^{\Gamma_0}}^\oplus d\mu_{\Gamma_0}({\mathrm{z}'}) (\Omega_{\mathrm{z}'})_{/U}\otimes \phi_\mathrm{z'} = \int_{U}^\oplus d\mu_{\Gamma_1}({\mathrm{z}}) \omega_C(\phi)_{{\mathrm{z}}} $ is well-defined, and that it is equivalent to a sheaf of sections of a Hermitian vector bundle over $\mathbb{G}^{\Gamma_1}$ by the Serre-Swan theorem \cite{Serre:1955,Swan1962VectorBA}.

    Conversely, suppose the above sheaf defines a 2-graph state for all $\phi$. Evaluating $\omega_C$ on the unit, $$\omega_C(\eta_0) = \int_{\mathbb{G}^{\Gamma_0}}^\oplus d\mu_{\Gamma_0}(\mathrm{z}) (\Omega_\mathrm{z})_{/U}\otimes \eta_0\cong \int_{\mathbb{G}^{\Gamma_0}}^\oplus d\mu_{\Gamma_0}(\mathrm{z}) (\Omega_\mathrm{z})_{/U},$$ implies that $U\mapsto (\Omega_\mathrm{z})_{/U}$ defines a sheaf. Since each stalk $(\Omega_\mathrm{z})_\mathrm{z'}$ is finitely-generated and projective as a $C(\mathbb{G}^{\Gamma_1})$-module, so is the sheaf $U\mapsto (\Omega_\mathrm{z})_{/U}$ by \textbf{Proposition \ref{analysisandalgebra}}.
\end{proof}
\noindent By definition, measureable natural transformations between cylinder functors $\omega_C,\omega'_C$ correspond to ($\mu_{\Gamma_1}$-essentially) bounded linear operators of sheaves on $\mathbb{G}^{\Gamma_0}\times \mathbb{G}^{\Gamma_1}$.

\medskip

By leveraging this characterization, there are embeddings that can be written down.
\begin{proposition}\label{cylinderembeddings}
    Let $C: \Gamma_0\Rightarrow \Gamma_1$ denote a homotopy between 2-graphs.
    \begin{itemize}
        \item There are fully-faithful embeddings
\begin{enumerate}
    \item $\C_q(\mathbb{G}^{\Gamma_0})^{\text{m-op}}\times\C_q(\mathbb{G}^{\Gamma_1})\rightarrow \operatorname{Fun}_\mathsf{Meas}(\C_q(\mathbb{G}^{\Gamma_0}),\C_q(\mathbb{G}^{\Gamma_1}))$,
    \item $\operatorname{Fun}_\mathsf{Meas}(\C_q(\mathbb{G}^{\Gamma_0}),\C_q(\mathbb{G}^{\Gamma_1})) \rightarrow \operatorname{Fun}_{\mathsf{Meas}}(\C_q(\mathbb{G}^{\Gamma_0})\times\C_q(\mathbb{G}^{\Gamma_1})^{\text{m-op}},\mathsf{Hilb})$.
\end{enumerate}
\item Neither of which are equivalences in general.
    \end{itemize}
\end{proposition}
\begin{proof}
\begin{itemize}
    \item We will explicitly construct the embeddings in the following.
     \begin{enumerate}
        \item      The goal is to construct a $\mathbb{G}^{\Gamma_0}$-family of sheaves of Hermitian $L^2$-sections on $\mathbb{G}^{\Gamma_1}$ from a pair of 2-graph states $\bar\phi_0\in  \C_q(\mathbb{G}^{\Gamma_0})^{\text{m-op}},~\phi_1\in \C_q(\mathbb{G}^{\Gamma_1})$. Here we emphasize that $\bar\phi$ is the \textit{linear dual}, not the *-operations.
        
        To do so, we use the monoidal product on $\mathsf{Meas}$ in Thm. 50 of \cite{Yetter2003MeasurableC}. Consider a 2-graph state $\Phi=\phi_0\times\phi1$ on $\mathbb{G}^{\Gamma_0}\times \mathbb{G}^{\Gamma_1}$ subject to the following conditions.
        \begin{itemize}
            \item $\Phi$ is \textbf{factorizable}: we have $\operatorname{pr}_1^*\Phi=\phi_1$ and $\operatorname{pr}_0^*\Phi=\bar\phi_0$ as sheaves along the projection functors \eqref{factorizable}, and
            \item $\Phi$ is equipped with a bounded Radon measure $\lambda$ on $\mathbb{G}^{\Gamma_0}\times \mathbb{G}^{\Gamma_1}$, for which the given Haar measures $\mu_{\Gamma_{0,1}} = \lambda\circ\operatorname{pr}_{0,1}^{-1}$ are the corresponding pushfowards.
        \end{itemize}
    These surjective submersive projections make $(\mu_{\Gamma_0},\mu_{\Gamma_1})$ into a disintegration pair. 
    
    Since projective modules are flat, the presheaf $\Phi_{\mathrm{z}}: U\mapsto (\Phi_{\mathrm{z}})_{/U},~ U\subset \mathbb{G}^{\Gamma_1}$ is a $\mathbb{G}^{\Gamma_0}$-family of finitely-generated projective sheaves on $\mathbb{G}^{\Gamma_1}$, which defines a cylinder functor $\omega_{\Phi}$ as desired. 

    The full-faithfulness is clear from definition: measureable natural transformations between cylinder functors of the form $\omega_\Phi,\omega_{\Phi'}$ are precisely sheaves of ($\mu_{\Gamma_1}$-essentially) bounded linear operators $\Phi\rightarrow\Phi'$.

   \item Now consider a cylinder functor $\omega_C$. Given its associated family of sheaves $\Omega$, the linear dual gives rise to a $\mathbb{G}^{\Gamma_0}$-family $\bar\Omega$ of finitely-generated projective $L^2$-sheaves on $\mathbb{G}^{\bar\Gamma_1}$. 

    Now let $\phi_0\in  \C_q(\mathbb{G}^{\Gamma_0}),~\bar\phi_1\in \C_q(\mathbb{G}^{\Gamma_1})^{\text{m-op}}$, and denote by $\tilde\Phi=\phi_0\times\bar\phi_1$ the associated factoriazable 2-graph state defined from along the canonical projections.

    Given the Radon measure $\lambda$ as above, we can then define a cone functor $\Omega_C$ by
   \begin{equation*}
       \Omega_C(\phi_0\times \bar\phi_1) = \int_{\mathbb{G}^{\Gamma_0}\times \mathbb{G}^{\bar\Gamma_1}}^\oplus d\lambda(\mathrm{z},\mathrm{z}') \bar\Omega_{\mathrm{z},\mathrm{z}'}\otimes \tilde\Phi_{\mathrm{z},\mathrm{z}'}\in\mathsf{Hilb}.
   \end{equation*}
Once again, the full-faithfulness is clear: measureable natural transformations $\Omega_C\Rightarrow\Omega'_C$ of the form above are precisely bounded linear operators between families of sheaves $\Omega\rightarrow \Omega'$.
    \end{enumerate}

    \item Given \textbf{Proposition \ref{uniqueness}}, the reasons for the non-essential surjectiveness is the following.
    \begin{enumerate}
        \item First, cylinder functors of the form $\omega_\Phi$ comes from factorizable sheaves $\Phi$ which are projective in both coordinates $\mathbb{G}^{\Gamma_0}\times \mathbb{G}^{\Gamma_1}$, whereas the characterization \textbf{Proposition \ref{cylinder}} only requires flatness along $\mathbb{G}^{\Gamma_0}$.\footnote{Though any flat module over a Noetherian ring is projective, it is well-known that continuous functions $C(X)$ over any manifold $X$, with $\operatorname{dim}X>0$, is not Noetherian.} 
        \item Second, cone functors of the form $\Omega_C$ come from families of very well-behaved sheaves, while generically their underlying field of Hilbert spaces $\underline{\omega}$ have no constraint. Thus the issue is the same as in \textbf{Proposition \ref{notequiavlence}}.
    \end{enumerate}
\end{itemize}
\end{proof}

\begin{rmk}
    The composition of the embeddings in the above theorem gives a full-faithful functor
        \begin{equation}
            \C_q(\mathbb{G}^{\Gamma_0})^{\text{m-op}}\times\C_q(\mathbb{G}^{\Gamma_1})\rightarrow \operatorname{Fun}_{\mathsf{Meas}}(\C_q(\mathbb{G}^{\Gamma_0})\times\C_q(\mathbb{G}^{\Gamma_1})^{\text{m-op}},\mathsf{Hilb}),\label{disjointpair}
        \end{equation}
        which extends the {categorical pairing form} (see \textit{Remark \ref{pairing}}) to disjoint homotopic graphs $\Gamma_0,\Gamma_1$. In fact, it is clear that, if $\Gamma_1 = v$ is trivial, then under the equivalence $\C_q(\mathbb{G}^\ast)\simeq\mathsf{Hilb}$ this functor \eqref{disjointpair} reproduces precisely the Yoneda embedding in \textbf{Proposition \ref{yoneda}}.
\end{rmk}

 Consider the one-point suspension of the disjoint union $\Gamma_0\coprod\Gamma_1$. It is PL homeomorphic to two  tetrahedra on $\Gamma_0,\Gamma_1$ identified at the cone point (a PL cylinder "pinched" at the centre), which is an \textit{irregular point} in the stratification (see fig. 2 in \cite{Liu:2024qth}). This leads to the fact that the right-hand side of \eqref{disjointpair}, ie. the cone functors $\operatorname{Fun}_{\mathsf{Meas}}(\C_q(\mathbb{G}^{\Gamma_0})\times\C_q(\mathbb{G}^{\bar\Gamma_1})^\text{op},\mathsf{Hilb})$, being "too large": it contains geometries which are not cylinders. Irregular points are also undesireable from the lattice theoretic perspective \cite{Liu:2024qth}, as they lead to ambiguities.

\subsection{Gauge *-invariance of categorical states}
Recall from \S \ref{lattice2alg} that $\C_q(\mathbb{G}^{\Gamma})\subset \cV_q^X$, for each 2-graph $\Gamma$, is a right *-module over $\mathbb{U}_q\G^{\Gamma^1}$. The categorical linear functionals, which are supposed to define states on the \textit{physical} degrees-of-freedom, should therefore be \textit{invariant} under $\mathbb{U}_q\G^{\Gamma^1}$. Such notions are captured by \textit{module functors}. 

These are by now very well-known, specifically in the theory of tensor categories \cite{etingof2016tensor,maclane:71,Bartsch:2022mpm,Delcamp:2023kew}. 
\begin{definition}
    Let $A,B$ denote two $\bbC$-linear monoidal categories. We say $\cM$ is an $A$-module if it comes equipped with a functor $\rhd: A\times \cM\rightarrow \cM$ and the module associator natural transformation $(-\otimes -)\rhd-\Rightarrow -\rhd (-\rhd-)$.
    \begin{enumerate}
        \item An \textbf{$A$-module functor} $F: \cM\rightarrow \cN$ is a functor equipped with natural transformations $F_a: F\circ (a\rhd_\cM-)\Rightarrow (a\rhd_{\cN} -)\circ F$, satisfying monoidal coherence conditions in $A$.
        \item Let $\cN$ be a $B$-module. A monoidal functor $f: A\rightarrow B$ induces the \textbf{restriction of scalars} functor $f^*\times 1_\cN:-\rhd -\mapsto f(-)\rhd  -$, which turns $(\cN,\rhd_B)$ into an $A$-module: $a\rhd_fn = f(a)\rhd n$.
    \end{enumerate}
\end{definition}
We will also recall the notion of a \textbf{rigid dagger} category \cite{Jones:2017}. 
\begin{definition}\label{rigiddagger}
     Let $\cM,\cN$ be rigid dagger categories. A \textbf{rigid dagger functor} $F:\cM\rightarrow \cN$ is a functor equipped with natural isomorphisms
        \begin{equation}
    F^\text{m-op}\circ (-^*)_\cM \cong (-^*)_\cN\circ F,\qquad F^\text{op}\circ (-^\dagger)_\cM\cong(-^\dagger)_\cN\circ F\label{daggercommute}
\end{equation}
preserving the rigid duality data, and satisfying the obvious coherence conditions against the rigid monoidal structures of $\cM,\cN$.
\end{definition}
In fact, when the rigid duality is involutive, a rigid duality structure can be thought of as a $\bbZ_2\times B\bbZ_2$-module structure on $\cM$. This gives the delooping $B\cM$ the structure of a coherent 2-$\dagger$ structure \cite{ferrer2024daggerncategories}.

\subsubsection{Invariant categorical linear functionals}
Consider a PL continuous map $\Gamma'^2\rightarrow \Gamma^2$ between two 2-graphs, and denote by $h:\Gamma'^1\rightarrow\Gamma^1$ the induced PL continuous map on their 1-skeleta, which by definition is a functor of PL 1-simplex groupoids.

We construct a functor $h^*:\mathbb{U}_q\G^{\Gamma^1}\rightarrow \mathbb{U}_q\G^{\Gamma'^1}$ on the 2-gauge parameters by pulling back $h$, which is easily seen to be \textit{strictly} monoidal
\begin{align*}
    h^*(\zeta\cdot_\mathrm{h}\zeta') &= h^*\big((aa')_v\xrightarrow{\gamma_e(a_v\rhd\gamma'_e}(aa')_{v'}\big)=(aa')_{h(v)}\xrightarrow{\gamma_{h(e)}(a_{h(v)}\rhd \gamma'_{h(e)})} (aa')_{h(v)} \\
    &=  (a_{h(v)}\xrightarrow{\gamma_{h(e)}}a_{h(v')})\cdot_\mathrm{h}(a'_{h(v)}\xrightarrow{\gamma_{h(e)}'}a'_{h(v')}) = h^*(\zeta)\cdot_\mathrm{h} h^*(\zeta),\\
    h^*(\zeta\cdot_\mathrm{v}\zeta') &= h^*\big(a_v\xrightarrow{\gamma_e}a_{v_0}\xrightarrow{\gamma_{e'}}a_{v'}\big) = h^*\big(a_v\xrightarrow{\gamma_{e\ast e'}}a_{v'}\big) =a_{h(v)}\xrightarrow{\gamma_{h(e\ast e')}}a_{h(v')} \\
    &= a_{h(v)}\xrightarrow{\gamma_{h(e)}\gamma_{h(e')}}a_{h(v')}= a_{h(v)}\xrightarrow{\gamma_{h(e)}}a_{h(v_0)}\xrightarrow{\gamma_{h(e')}}a_{h(v')} = 
    h^*(\zeta)\cdot_\mathrm{v}h^*(\zeta')
\end{align*}
for each (horizontally/vertically) composable $\zeta,\zeta'\in\mathbb{U}_q\G^{\Gamma^1}$.

\medskip

This monoidal functor $h^*$ then induces a restriction of scalars, sending $\mathbb{U}\G^{\Gamma^1}$-modules to $\mathbb{U}\G^{\Gamma'^1}$-modules. We can therefore introduce the following notion.
\begin{definition}
    Suppose there is a PL continuous map $\Gamma'^2\rightarrow \Gamma^2$, then a \textbf{measureable $\bullet$-module functor} $F:\C_q(\mathbb{G}^{\Gamma^2})\rightarrow\C_q(\mathbb{G}^{\Gamma'^2})$ is a measureable functor $\omega^h=(\omega,\nu)$ --- with $\nu$ a ($\operatorname{pr}_2$-filtred) measure on $\mathbb{G}^{\Gamma^2}\times\mathbb{G}^{\Gamma'^2}$ --- equipped with a measureable natural transformation
    \begin{equation*}
        \omega_\zeta: \omega \circ  (-\bullet \zeta) \Rightarrow (-\bullet h^*\zeta)\circ\omega
    \end{equation*}
    for all $\zeta\in\mathbb{U}_q\G^{\Gamma^1}$, such that the diagram against the module associator $\alpha^\Lambda$,
\[\begin{tikzcd}
 & {(-\bullet h^*\zeta)\circ\omega\circ(-\bullet  \zeta')} & \\
	{\omega\circ(-\bullet\zeta)\circ(-\bullet\zeta')}  & & {(-\bullet h^*\zeta)\circ(-\bullet  h^*\zeta')\circ\omega} \\
	{\omega\circ(-\bullet \zeta\cdot\zeta')} && {\big(-\bullet h^*(\zeta\cdot\zeta')\big)\circ\omega}
	\arrow["{\omega_\zeta\circ(-\bullet  \zeta')}", Rightarrow, from=2-1, to=1-2]
	\arrow["{\omega\circ\alpha^\bullet_{\zeta,\zeta'}}"', Rightarrow, from=2-1, to=3-1]
	\arrow["{(-\bullet \zeta)\circ\omega_{\zeta'}}", Rightarrow, from=1-2, to=2-3]
	\arrow["{\alpha^\bullet_{\zeta,\zeta'}\circ\omega}", Rightarrow, from=2-3, to=3-3]
	\arrow["{\omega_{\zeta\cdot\zeta'}}"', Rightarrow, from=3-1, to=3-3]
\end{tikzcd},\]
commutes. Here $\cdot$ denotes either the horizontal or vertical composition, depending on the composability of $\zeta,\zeta'$.
\end{definition}

Explicitly, the natural transformation $\omega_\zeta$ is the data of a field of bounded linear operators
\begin{equation*}
    (\omega_\zeta)_{\mathrm{z}',\mathrm{z}}: (\omega \Lambda_\zeta)_{\mathrm{z}',\mathrm{z}}\rightarrow (\Lambda_{h^*\zeta} \omega)_{\mathrm{z}',\mathrm{z}},\qquad \mathrm{z},\mathrm{z}'\in X=\mathbb{G}^{\Gamma^2},
\end{equation*}
with measureability class $\sqrt{(\nu\lambda_\zeta)(\lambda_\zeta \nu)}$ \cite{Baez:2012}, where $\lambda_\zeta$ is the measure on $X\times X$ underlying $\Lambda_\zeta$. We will assume $\omega_\zeta$ is invertible in the following.

By inducing $\Lambda_\zeta$ from a pullback (see \S \ref{2gauge}), $\lambda_\zeta=\delta$ is the delta measure and $f\delta=f=\delta f \implies \sqrt{ff} = f$ by Radon-Nikodym.  Taking the PL continuous map $h$ to be the identity, we recover the notion of "measureable module endofunctors" introduced in the appendix of \cite{Chen:2025?}, through the model change \textit{Remark \ref{modelchange}}.

\medskip

This gives us the following \textit{invariance} property.
\begin{proposition}\label{invarfunctor}
    Cone $\bullet$-module functors $\omega$ are $\mathbb{U}_q\G^{\Gamma^1}$-invariant, hence they descend to categorical states on the 2-Chern-Simons observables $\omega\in \operatorname{Fun}_\mathsf{Meas}(\mathcal{O}^\Gamma,\mathsf{Hilb})$.
\end{proposition}
\begin{proof}
    Recall $\cH^\emptyset\simeq\mathsf{Hilb}$. Consider the constant PL continuous map $\ast\rightarrow \Gamma^2$ sending a point to the root $v\in \Gamma^2$ of a 2-graph, which gives rise the same trivial map on the 1-skeleta $h: \ast\rightarrow \Gamma^1$. Since the point $\ast$ is undecorated, the induced map on the decorated 1-graphs is the monoidal counit $h^*(\zeta) = \tilde\epsilon(\zeta)$ in $\mathbb{U}_q\G^{\Gamma^1}$ (ie. the trivial transformation for all $\zeta$).
    
    By \textbf{Definition \ref{linearfunctors}}, the $\bullet$-module structure on cone functors $\omega\in\operatorname{Fun}_\mathsf{Meas}(\C_q(\mathbb{G}^{\Gamma^2}),\mathsf{Hilb})$ then reads
    \begin{equation*}
        \omega_\zeta: \omega \circ(-\bullet \zeta) \Rightarrow (-\bullet  h^*\zeta)\circ\omega = (-\bullet \tilde \epsilon(\zeta))\circ\omega \cong \omega.
    \end{equation*}
    where we have by definition $-\bullet {\tilde\epsilon(\zeta)} \cong -\otimes\mathsf{Hilb} \cong 1_{\C_q(\mathbb{G}^{\Gamma^2})}$ for all $\zeta\in\mathbb{U}_q\G^{\Gamma^1}$.

    Now given 2-gauge transformations can be written in terms of the $\bullet$-bimodule structure \eqref{leftreg}, the last statement follows immediately.
\end{proof}
\noindent This is a categorified version of the invariance condition, eq. (6.7) of \cite{Alekseev:1994pa}, for linear functionals in discrete Chern-Simons theory.\footnote{Note we do not require the monoidality of categorical linear funcitonals under the monoidal structure given by $\ostar$, since such functors decategorifies into an algebra map, which does not correspond to a state on a $C^*$-algebra.}

\subsubsection{*-functors and cointegrals for Hopf categories}
Recall from \S \ref{*-op} that the *-operations give the cocategory $\C_q(\mathbb{G}^\Gamma)$ with a dagger *-structure (in which the duality is \textit{not} necessarily involutive). The unitarity property stated in \textbf{Definition \ref{unitary2hol}} then allows us to construct the duality data on $\C_q(\mathbb{G}^{\Gamma^2})$ (specifically the evaluation measureable functors; see the appendix of \cite{Chen:2025?}).

\begin{rmk}
    In the following, we will only focus on the property \eqref{daggercommute}. This is because infinite-dimensional Hilbert spaces do \textit{not} have coevaluation maps that satisfy the snake equation against the canonical evaluation map, and hence any infinite-dimensional analogue of $\mathsf{Hilb}$ will not be rigid. Indeed, evaluation module functors on $\C_q(\mathbb{G}^\Gamma)$ have been written down in the appendix of \cite{Chen:2025?} using the *-operations, but it does not have coevaluations.
\end{rmk}

Focusing on the cone functors $\omega=\omega_{\Lambda\Gamma^2}$ for clarity, we define the following.
\begin{definition}
    A \textbf{measaureable (cone) $\bullet$-module *-functor} is a cone $\bullet$-module functor $\omega: \C_q(\mathbb{G}^{\Gamma})\rightarrow \mathsf{Hilb}$ equipped with invertible measureable $\bullet$-module natural transformations
    \begin{equation*}
        \omega^\dagger: -^\dagger\circ \omega \Rightarrow \omega^\text{op}\circ -^\dagger,\qquad \omega^*: \bar\cdot \circ \omega \Rightarrow \omega^{\text{m-op}}\circ \bar\cdot
    \end{equation*}
    such that the obvious coherence conditions against the *-module natural transformations $\overline{\phi\bullet\zeta} \cong \bar{\zeta}\bullet \bar\phi$ are satisfied.

    Denote by $\operatorname{Fun}_{\mathsf{Meas}}^{\bullet,\ast}(\C_q(\mathbb{G}^{\Gamma^2}),\mathsf{Hilb})$ the hom-category of such $\bullet$-module cone *-functors on $\C_q(\mathbb{G}^{\Gamma^2})$.
\end{definition}
We shall assume these measureable natural transformations are invertible. 

Let us now prove the categorification of eq. (6.8) in \cite{Alekseev:1994pa}.
\begin{proposition}
    Let $\omega$ be a measaureable (cone) $\bullet$-module *-functor, then there are natural measureable isoomrphisms
    \begin{equation*}
    \overline{\omega(\phi)} \cong \omega(\phi^{*_1}),\qquad \omega(\phi)^\dagger \cong \omega(\phi^{*_2})
\end{equation*}
    for each $\phi\in\C_q(\mathbb{G}^{\Gamma^2})$, intertwining the *-operations \textbf{Definition \ref{starop}}.
\end{proposition}
\begin{proof}
    To begin, by definition, for each $\phi\in\C_q(\mathbb{G}^{\Gamma^2})$ we have {linear} isomorphisms
\begin{equation*}
    \overline{\omega(\phi)} \cong \omega(\bar\phi) \cong \omega(S_h\phi^{\dagger_1}),\qquad {\omega(\phi)}^\dagger \cong \omega(S_v\phi^{\dagger_2}),
\end{equation*}
where we have used the unitarity property \textbf{Definition \ref{opposites}} to rewrite $\bar \phi$ in terms of the horizontal/vertical antipodes $S_h,S_v$ and the 2-dagger structures on the 2-graphs, $$(\phi^{\dagger_{1,2}})_\mathrm{z} = \phi_{\mathrm{z}^{\dagger_{1,2}}},\qquad \mathrm{z}\in \mathbb{G}^{\Gamma^2}.$$

However, by definition of the *-operations in \textbf{Definition \ref{starop}}, these 2-dagger structures are related to $-^{*_{1,2}}$ up to an action of the $R$-matrices (as well as the invertible 2-$\dagger$ intertwiner pair $\eta=(\eta_h,\eta_v)$). Due to the invariance property \textbf{Proposition \ref{invarfunctor}} of cone $\bullet$-module functors $\omega$, these are trivialized whence we achieve the natural measureable isomorphisms as desired.
\end{proof}

\subsubsection{Cointegrals for Hopf (co)categories}\label{cointegral}
Equipped with the notion of $\bullet$-module functors, we can then concretely interpret the Haar measure $\mu$ of a Lie 2-group $\mathbb{G}$. Recall that a \textit{left-/right-cointegral} of a Hopf algebra $H$ is a linear functional $\lambda_l,\lambda_r: H\rightarrow \bbC$ for which
\begin{equation*}
    (\lambda_l\otimes 1) \circ\Delta = \eta \circ \lambda_l,\qquad (1\otimes\lambda_r)\circ\Delta =\eta\circ\lambda_r,
\end{equation*}
respectively, where $\Delta: H\rightarrow H\otimes H$ is the coprodut and $\eta:\bbC\rightarrow H$ is the unit. $\lambda:H\rightarrow \bbC$ is simply called a \textbf{cointegral} if it is both a left- and a right-cointegral.

A classic example of a Hopf algebra, which is \textit{not} in general finite-dimensional (but finitely-generated as  $C^*$-algebra), equipped with a cointegral is the (undeformed) compact quantum group $C(G)$ of Woronowicz \cite{Woronowicz1988} for a compact semisimple Lie group $G$. It is given precisely by the Haar measure on $G$.

Let us now introduce the (co)categorical version.
\begin{definition}
    Let $H$ denote a Hopf cocategory internal to a symmetric monoidal bicategory $\cV$. A \textbf{left-/right-cointegral} for $H$ is an internal functor $\Lambda_l,\Lambda_r: H\rightarrow I$ into the discrete internal cocategory $I\leftleftarrows I$ on the monoidal unit $I\in \cV$, such that there exist natural transformations
    \begin{equation}
        (\Lambda_l\times 1) \circ\Delta \Rightarrow \eta \circ \Lambda_l,\qquad (1\times\Lambda_r)\circ\Delta \Rightarrow\eta\circ\Lambda_r,\label{integralnat}
    \end{equation}
    satisfying the obvious coherence conditions against the natural transformations $\Delta\circ m \Rightarrow (m\times m)\circ\Delta$ witnessing the bimonoidal axioms.

    We call $\Lambda_l,\Lambda_r$ \textbf{strong} iff these natural transformations are invertible. We say $\Lambda: H\rightarrow I$ is an \textbf{integral} iff it is both a left- and right-cointegral such that the following diagram 
\[\begin{tikzcd}
	{(1\times \Lambda\times 1) \circ(\Delta \times 1)\circ \Delta } && {(1\times \Lambda\times 1) \circ(1\times \Delta)\circ \Delta } \\
	{(\eta\circ\Lambda \times 1)\circ \Delta} & {\eta\times\eta} & {(1\times\eta\circ\Lambda )\circ \Delta}
	\arrow[Rightarrow, from=1-1, to=1-3]
	\arrow[Rightarrow, from=1-1, to=2-1]
	\arrow[Rightarrow, from=1-3, to=2-3]
	\arrow[Rightarrow, from=2-1, to=2-2]
	\arrow[Rightarrow, from=2-3, to=2-2]
\end{tikzcd}\]
against the coassociator $(\Delta \times 1)\circ \Delta \Rightarrow (1\times\Delta )\circ \Delta$ commutes.
\end{definition}

We can now prove the following.
\begin{proposition}\label{quantum2cointegral}
    Let $\mu$ denote an invariant Haar measure for the compact Lie 2-group $\mathbb{G}$, then the direct integral $\int^\oplus_\mathbb{G}d\mu (-): \C(\mathbb{G})\rightarrow \mathsf{Hilb}$ is a strong cointegral for the geometric 2-graph states $\C(\mathbb{G})$. 
\end{proposition}
\begin{proof}
    By \textbf{Definition \ref{2grouphaar}}, $\mu$ has a disintegration along the source map for which the pushforward $\sigma=\mu\circ s^{-1}$ is itself an invariant Haar measure on $G$. This allows us to define the measureable functor $\int_G^\oplus d\sigma(-): \C(G)\rightarrow\mathsf{Hilb}$ which fits into the strict commutative diagram
\[\begin{tikzcd}
	{\C(\mathsf{H}\rtimes G)} && {\mathsf{Hilb}} \\
	{\C(G)} && {\mathsf{Hilb}}
	\arrow["{\int_\mathbb{G}^\oplus d\mu(-)}", from=1-1, to=1-3]
	\arrow["=", shorten <=10pt, shorten >=10pt, Rightarrow, from=1-1, to=2-3]
	\arrow[shift left, from=2-1, to=1-1]
	\arrow[shift right, from=2-1, to=1-1]
	\arrow["{\int_G^\oplus d\sigma(-)}"', from=2-1, to=2-3]
	\arrow[shift left, from=2-3, to=1-3]
	\arrow[shift right, from=2-3, to=1-3]
\end{tikzcd}.\]
    This casts $\int_\mathbb{G}^\oplus d\mu(-):\C(\mathbb{G})\rightarrow \mathsf{Hilb}$ as a functor of \textit{internal} cocategories.

    To show invariance, we invoke Thm. 28 of \cite{Yetter2003MeasurableC}: 
    \begin{theorem}\label{isointegral}
        Direct integral functors $\displaystyle\int_X^\oplus d\mu,~\int_X^\oplus d\nu$ on a measureable category $\cH^X$ over some measureable space $X$ are measureably naturally isomorphic iff the two measures $\mu,\nu$ are equivalent (namely they are absolutely continuous with respect to each other $\mu\ll\nu,~\nu\ll\mu$).     
    \end{theorem}
    \noindent Therefore any given measure $\mu$ on $\mathbb{G}$ invariant under both left and right 2-group (ie. group and groupoid) multiplications, the induced direct integrals $\displaystyle\int_\mathbb{G}^\oplus d\mu(\mathrm{z}\cdot-) \cong \int_\mathbb{G}^\oplus d\mu \cong\int_\mathbb{G}^\oplus d\mu(-\cdot \mathrm{z})$ are measureably naturally isomorphic. These provide the desired natural isomorphisms required for a cointegral.
    
    The fact that invariance (in the sense of \textbf{Definition \ref{2grouphaar}}) implies both left- and right-invariance of $\mu$ under the 2-group multiplication operations was proven in \S 3.2.2 of \cite{Chen1:2025?}.
\end{proof}

This endows the cone $\bullet$-module *-functors $\omega\in\operatorname{Fun}_{\mathsf{Meas}}^{\bullet,*}(\C_q(\mathbb{G}^\Gamma),\mathsf{Hilb})$ the interpretation of a "quantum" version of a Hopf category cointegral, and the categorical version of the "quantum Haar measure" described in \cite{Alekseev:1994pa}.

\begin{rmk}
    We know from \textbf{Proposition \ref{haarunique}} that Haar measures are unique on compact Lie 2-groups $\mathbb{G}$. Hence, to show $\C(\mathbb{G})$ is unimodular, we just need to show that all cointegrals on $\C(\mathbb{G})$ come from invariant Haar measures via the direct integral. This is not known, however.
\end{rmk}

\subsection{Orientation and framing pairings}\label{higherdaggerpairings}
It is \textit{crucial} that the unitarity property \textbf{Definition \ref{unitary2hol}} relates the "internal" dagger *-structure on $\C_q(\mathbb{G}^{\Gamma^2})$ to the "external" dagger duality on $\mathsf{Meas}$ (see \textit{Remark \ref{opposites}}), since this then allows us to turn the pairing functor of \textit{Remark \ref{pairing}} into a geometric one.
\begin{definition}\label{orientationpairing}
    Let $\bar\Gamma^2 = (\Gamma^2)^{\dagger_1}$ denote the orientation reversed 2-graph. The \textbf{orientation pairing} on 2-graph states is the composite measureable functor
    \begin{equation}
        \C_q(\mathbb{G}^{\bar\Gamma^2})^{\text{c-op}_h}\times\C_q(\mathbb{G}^{\Gamma^2}) \xrightarrow{S_h\times 1} \C_q(\mathbb{G}^{\Gamma^2})^\text{m-op}\times\C_q(\mathbb{G}^{\Gamma^2})\xrightarrow{\eqref{pairfunctor}} \mathsf{Hilb},\label{geometrypair1}
    \end{equation}
    given in terms of the horizontal antipode $S_h:\C_q(\mathbb{G}^{\Gamma^2})\rightarrow \C_q(\mathbb{G}^{\Gamma^2})^{\text{m-op,c-op}_h}$ by \eqref{pairing},
    \begin{equation*}
        (\phi',\phi) \mapsto \omega_{S_h\phi'}(\phi) = \int_{\mathbb{G}^{\Gamma^2}}^\oplus d\mu_{\Gamma^2}(\mathrm{z})(S_h\phi')^{*_1}_{\mathrm{z}}\otimes \phi_{\mathrm{z}}.
    \end{equation*}
\end{definition}

We also have the following notion.
\begin{definition}\label{framingpairing}
    Let $\tilde\Gamma^2 = (\Gamma^2)^{\dagger_2}$ denote the frame-rotated 2-graph. The \textbf{framing pairing} on 2-graph states is the composite measureable functor
    \begin{equation}
        \C_q(\mathbb{G}^{\tilde\Gamma^2})^{\text{c-op}_v}\times\C_q(\mathbb{G}^{\Gamma^2}) \xrightarrow{S_v\times 1} \C_q(\mathbb{G}^{\Gamma^2})^\text{m-op}\times\C_q(\mathbb{G}^{\Gamma^2})\rightarrow \mathsf{Hilb},\label{geometrypair2}
    \end{equation}
    given in terms of the vertical antipode $S_v:\C_q(\mathbb{G}^{\Gamma^2})\rightarrow \C_q(\mathbb{G}^{\Gamma^2})^{\text{m-op,c-op}_v}$,
    \begin{equation*}
        (\phi',\phi) \mapsto \omega_{S_v\phi'}(\phi) = \int_{\mathbb{G}^{\Gamma^2}}^\oplus d\mu_{\Gamma^2}(\mathrm{z})(S_v\phi')^{*_2}_{\mathrm{z}}\otimes \phi_{\mathrm{z}},
    \end{equation*}
\end{definition}
They will play an important role later in \S \ref{2-skeins}.

\section{$\mathbb{G}$-decorated 2-ribbons: $\operatorname{PLRib}'^{\mathbb{G};q}_{(1+1)+\epsilon}(D^4)$}\label{Gdecorations}

\subsection{Handlebody decompositions and the standard 2-algebra}\label{3dhandles}
The above \S \ref{1morA} lays down the foundation for the {\it gluing} of 3d handlebodies onto the 2-graph states, which allows us to reconstruct 3-manifold ribbon invariants through the \textit{handlebody decomposition}.  For details of the following notions, see eg. \cite{Sakata2022-il,matveev2007algorithmic}.

\begin{definition}
    A \textbf{2d polyhedron} $P$ is the underlying space of a non-collapsible locally finite 2-dimensional complex, such that the link of each vertex contains no isolated vertices. We say $P$ is \textbf{simple} if each point has a neighbourhood homeomorphic to either a non-singular point, a triple point or a trisection vertex (see fig. 2, \cite{Sakata2022-il}).
\end{definition}

The idea is that by pasting 3-dimensional handles onto $P$ in a certain way, we can obtain a 3-manifold.
\begin{definition}
    Let $M$ be a closed, connected, oriented 3-manifold. A \textbf{handlebody decomposition of type-$(g_1,\dots,g_n;P)$} for $M$ is a 2d simple polyhedron $P$ such that $M\setminus P=\coprod_{i=1}^nH_i$, where each $H_i$ is the interior of a 3-dimensional handlebody with genus $g_i$. The polyhedron $P$ is called the \textbf{partition} of $M$.
\end{definition}

The central theorem in \cite{Castler:1965} is that \textit{every} 3-manifold can be obtained in this way.
\begin{theorem}\label{handlebody}
    Any closed connected 3-manifold admits a simple handlebody decomposition of type-(0).
\end{theorem}
\noindent Now the point is that a 2-graph $\Gamma^2$ serves {precisely} as the combinatorial triangulation of a \textit{simple} polyhedron $P$, and its 1-skeleton $\Gamma^1$ forms its \textit{singular graph} $B$. 

It is thus possible to determine a handlebody decomposition of a 3-manifold $\Sigma$ by embedding a 2-graph $\Gamma^2$ into it. 

\begin{rmk}\label{lengthhandle}
    Given a handlebody decomposition of type-$(g_1,\dots,g_n;P)$ for a 3-manifold $M$, let us call $n$ its \textbf{length}. Length $n=2$ decompositions are precisely Heegaard splittings, and length $n=3$ are trisections. Generally, handlebody decompositions of larger length and lesser genera "knows" more about the underlying 3-manifold; indeed, 3-manifolds $M$ admitting a length-3 handlebody decomposition with genera $\leq 1$ has been classified completely up to homeomorphism in \cite{Gomez_Larranaga1987-dz}, Thm. 1. Moreover, by Proposition 4.2 of \cite{Sakata2022-il}, any 3-manifold $M$ whose spheres are all separating admits a length-3 decomposition of the type $(0,0,g)$, where $g$ is the Heegaard genus of $M$.
\end{rmk}

The heavy-lifting of \S \ref{objA} --- specifically the specification of the interchangers $\beta$ and the $U(1)$-gerbes $\sigma$ in \textit{Remarks \ref{trisection}, \ref{triplepoint}} --- then defines holonomy-dense 2-graph states on combinatorial triangulations of such simple partitions $P$. We can then give the categorical analogue of Def. 12 in \cite{Alekseev:1994au}.

\begin{definition}\label{standard2alg}
    The \textbf{standard 2-algebra} $\mathcal{B}^P$ associated to a 2d simple polyhedron $P$ is the monoidal semidirect product $\C_q(\mathbb{G}^{(\Gamma_P)^2})\rtimes \mathbb{U}_q\G^{(\Gamma_B)^1}$, where $(\Gamma_P)^2 = \Gamma_P$ is a combinatorial quantization of $P$ and $(\Gamma_B)^1=\Gamma_B$ is the induced triangulation of its underlying singular graph $B$. 
\end{definition}
\noindent In the following, all 2-graph states are holonomy-dense.

\subsubsection{Independence of the 2-graph}\label{pachner}
In this section, we will examine the dependence of the standard 2-algebra under the choice of combinatorial triangulation $\Gamma_P$ of $P$. Treating $P$ as a (framed) PL 2-manifold, will do this through the Pachner moves \cite{Pachner1991Pachner}.

\begin{theorem}\label{invariance}
    The standard 2-algebra $\mathcal{B}^{\Gamma_P}$ associated to a 2d simple polyhedron $P$ in \textbf{Definition \ref{standard2alg}} is independent of the choice of the combinatorial triangulation.
\end{theorem}
\begin{proof}
Let us begin by setting up the geometry of the Pachner moves. In 2-dimensions, there are two of them: a "flip" and a "bistellar subdivision"; see also fig. 3 in \cite{Beck2024-fs}. The way that we are going to perform them is given in fig. \ref{fig:pachner}.

\begin{figure}[h]
    \centering
    \includegraphics[width=1\linewidth]{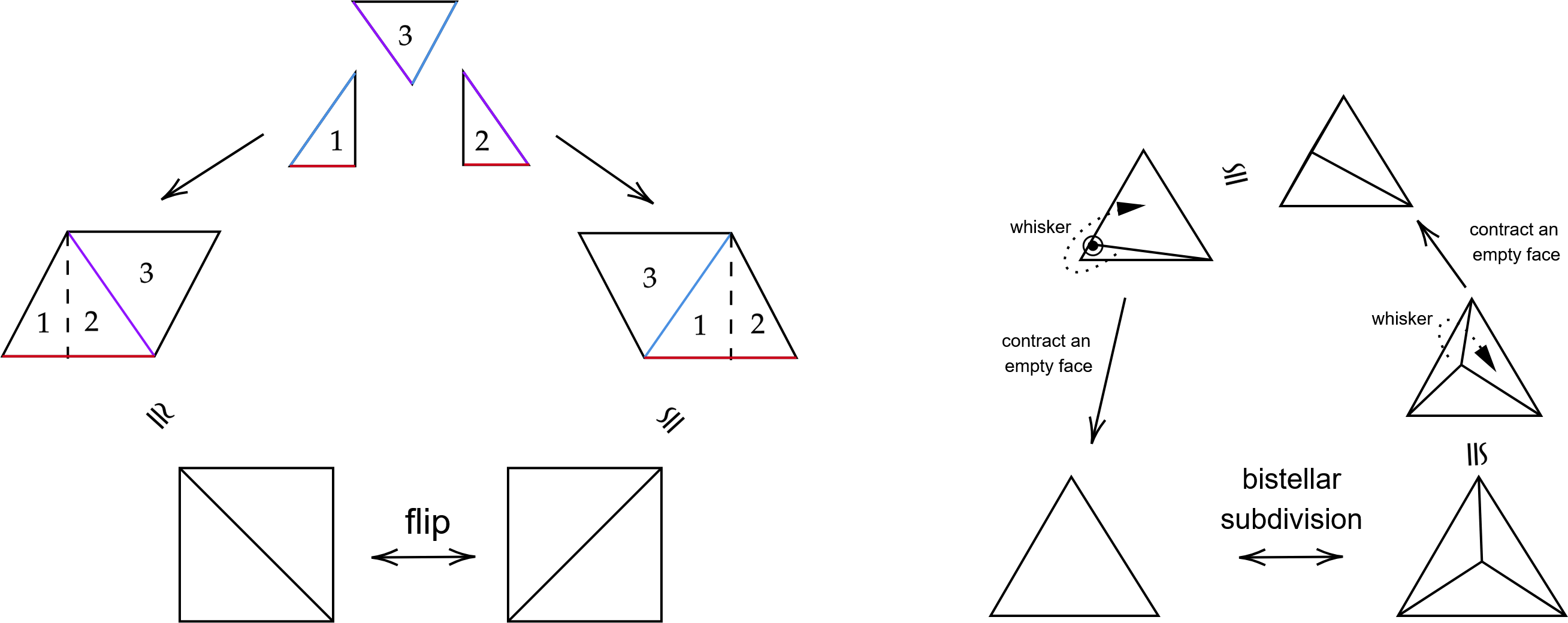}
    \caption{The 2-simplex configurations which witness the 2d Pachner moves.}
    \label{fig:pachner}
\end{figure}

\begin{lemma}
    $\C_q(\mathbb{G}^\Gamma)$ is invariant under flip moves.
\end{lemma}
\begin{proof}
    Let $\Gamma$ and $\Gamma'$ denote two combinatorial triangulations of the unit square which differ by a single flip move. Take three gluing-amenable 2-simplex states $(\phi_1,\phi_2,\phi_3)\in\C_q(\mathbb{G}^{\coprod_{i\leq 3}\Delta_i})$ in accordance with the configuration of 2-simplices $\Delta_1,\Delta_2,\Delta_3$ as arranged on the left-hand side of fig. \ref{fig:pachner}. 
    
    By following the geometric procedure as indicated on the left of the figure, we construct a 2-graph state on $\Gamma$ by first gluing $\Delta_2,\Delta_3$, then with $\Delta_1$:
    \begin{equation*}
        (\phi_2\ostar \phi_3)\ostar \phi_1\in\C_q(\mathbb{G}^\Gamma).
    \end{equation*}
    Similarly, the procedure along the right side produces a 2-graph state on $\Gamma'$ by first gluing $\Delta_1,\Delta_3$, then with $\Delta_2$,
    \begin{equation*}
        \phi_2\ostar (\phi_3\ostar\phi_1)\in\C_q(\mathbb{G}^{\Gamma'}).
    \end{equation*}
    This is precisely the associativity of $\ostar$.

    More generally by holonomy-density, the flip move is equivalent to an invertible natural transformation,
\[\begin{tikzcd}
	{\C_q(\mathbb{G}^{\Delta_2})\times_{23} \C_q(\mathbb{G}^{\Delta_3})\times_{31} \C_q(\mathbb{G}^{\Delta_1})} && {\C_q(\mathbb{G}^{\Delta_2})\times_e\C_q(\mathbb{G}^{\Delta_3\cup_{31}\Delta_1})} \\
	{\C_q(\mathbb{G}^{\Delta_2\cup_{23}\Delta_3})\times_e \C_q(\mathbb{G}^{\Delta_1})} && {\C_q(\mathbb{G}^{\Gamma})=\C_q(\mathbb{G}^{\Gamma'})}
	\arrow["{1\times\ostar}", from=1-1, to=1-3]
	\arrow["{\ostar\times1}"', from=1-1, to=2-1]
	\arrow["\cong"', shorten <=17pt, shorten >=17pt, Rightarrow, from=1-1, to=2-3]
	\arrow["\ostar", from=1-3, to=2-3]
	\arrow["\ostar"', from=2-1, to=2-3]
\end{tikzcd},\]
which witnesses the associativity of $\ostar$, where $e$ denotes the edge at the bottom of the left-hand side of fig. \ref{fig:pachner}, coloured in red. 

By construction (cf. \textbf{Theorem \ref{quantumhopfcocat}} and \eqref{quantumautomorphism}), $\ostar$ is not only associative by also \textit{strictly} so (namely the above associator natural transformation is not only invertible but also only have components at the identity). The statement follows.
\end{proof}
\noindent     The fact that the flip move is related to a certain notion of associativity was noticed also in the construction of 2d TQFTs from $A_\infty$-algebras in \cite{Beck2024-fs}. This is a manifestation of a certain theorem of Gauss.

We now turn to the bistellar subdivision.
\begin{lemma}
    If $\Delta\simeq \Delta_\ast$ are 2-simplices related by a bistellar subdivision, then $\C_q(\mathbb{G}^\Delta)\simeq\C_q(\mathbb{G}^{\Delta_\ast})$.
\end{lemma}
\begin{proof}
    As illustrated on the right-hand side of fig. \ref{fig:pachner}, we can move from the bistellar subdivision $\Delta_\ast$ to $\Delta$ by contracting "empty" faces. However, since each 2-simplex are decorated with 2-holonomies $(h_e,b_f)\in\mathbb{G}$, we need to leverage the composition of 2-holonomies in $\mathbb{G}^\Delta$ to remove decorations on the face that we wish to contract.
    
    This can be done through the fake-flatness condition: if a face $D$ bounds $e$, then its 2-holonomy satisfies $tb_D=h_{\partial D}$. We can thus remove a 2-holonomy by a \textit{whiskering} \cite{Baez:2004} along the inverse of the decoration $h_e$ on the boundary $e=\partial D$, making the underlying 2-simplex undecorated.
   
    Recall the direct image functor on the sheaves $\C_q(\mathbb{G}^\Delta)$ induced by this whiskering operation is denoted by $W_e$. If the edge $e$ is a contractible loop, then we can use \textbf{Proposition \ref{whiskeringhomotopy}} to construct an invertible measureable natural transformation to trivialize it. 

    Now as can be seen in fig. \ref{fig:pachner}, we have to do this whiskering twice. Therefore we have a measureable natural isomorphism 
    \begin{equation}
        T_{D'\ast D}^{-1}: W_{e'}^{-1}\circ W_e^{-1}\Rightarrow 1_{\C_q(\mathbb{G}^\Delta)},\qquad \partial (D'\ast D)= e'\ast e
    \end{equation}
    witnessing the equivalence $\C_q(\mathbb{G}^{\Delta_\ast})\simeq\C_q(\mathbb{G}^\Delta)$ under bistellar subdivision.
\end{proof}

Invariance of the 2-gauge transformations under the 1d Pachner move can be routinely checked. 
\end{proof}

\begin{rmk}\label{pachnertau}
    Here we make the crucial observation that both of the above lemmas hold \textit{up to equivalence} when the associativity in $\mathbb{G}$ is weakened. The weak associator $\tau$ on $\mathbb{G}$ contributes directly not only to the invertible associativity of $\ostar$, but also to the invertible modification $T_{D'_1\ast D_1}^{-1}\Rrightarrow T_{D'_2\ast D_2}^{-1}$ which witnesses the bistellar move. These witnesses of course must be mutually coherent; in terms of higher-gauge theory, these equations take the guise of the \textit{descent equations} for $\tau$ \cite{Kapustin:2013uxa,Baez:2004in}.
\end{rmk}

Thanks to this result, we will denote by $\C_q(\mathbb{G}^P)$ the 2-graph states associated to a 2d simple polyhedron $P$ evaluated on any choice of a combinatorial triangulation $\Gamma_P$ of $P.$



\subsubsection{Example: cone functors on $S^3$}\label{cone3d}
    Let us consider the example of the (unit) 3-sphere $M=S^3$, and consider a 2d polyhedron $P$ partitioning it. We pick the 2-graph underlying $P$ is exactly the one $\Gamma^2=\Gamma_{S^3}$ described in \textit{Example \ref{S3decomposition}}. Note that in $S^3$, this polyhedron $P$ is convex and has no boundary as a 2-graph.
    
    This 2-graph admits a splitting into eight fundamental 2-simplices $\Delta_1,\dots,\Delta_4,\Delta_1',\dots,\Delta_4$, for which $\Gamma_{+,i}=\Delta_i\cup\Delta_{i+1}\cup\Delta_i'\cup\Delta_{i+1}$ defines the geometry described in \S \ref{interchanger} for each $i=1,\dots,4$ (here the indices are modulo 4, $\Delta_{4+1} = \Delta_1$). These are the boundaries of the standard octants in $\R^3$.
    
    Let us first describe how the 2-monodromy states $\Phi\in\C_q(\mathbb{G}^P)$ on $P$ are constructed. To do this, fix a set of eight 2-simpelx states $\phi_i \in\C_q(\mathbb{G}^{\Delta_i}),~\phi'_i\in\C_q(\mathbb{G}^{\Delta_i'}$, $i=1,\dots,4$. There are certain cnoditions that these 2-simplex states must satisfy.
    \begin{enumerate}
        \item First, by \textbf{Definition \ref{intch}}, each 4-tuple $(\phi_i,\phi_{i+1},\phi_i',\phi_{i+1}')\in\C_q(\mathbb{G}^{\Delta_i\coprod \Delta_{i+1}\coprod \Delta_i'\coprod\Delta_{i+1}'})$ must be gluing-amenable for each $i=1,\dots,4$, which provides us with interchanger natural isomorphisms $\beta_i$. We define
        \begin{equation*}
            \Phi_i = \phi_i\ostar\phi_{i+1}\ostar\phi_i'\ostar\phi_{i+1}' \in\C_q(\mathbb{G}^{\Gamma_{+,i}})
        \end{equation*}
        as their product.
        \item Next, by \textbf{Definition \ref{nonreg-glue}}, each triple $(\Phi_i,\Phi_{i+1},\Phi_{i+2})_{(\sigma\cup\sigma')_i}\in \C_q(\mathbb{G}^{\coprod_{j=i}^{j+2}\Gamma_{+,j}})$ must be gluing-amenable for each $i=1,\dots,4$ (recall the indices are mod-4, $\Gamma_{+,5}=\Gamma_{+,1},\Gamma_{+,6}=\Gamma_{+,2}$, etc.). This involves the data of {\v C}ech 2-cocycles $(\sigma\cup\sigma')_i$ attached to each edge $\coprod_{j=i}^{j+2}\Gamma_{+,j}$ in $P$.
    \end{enumerate}

Now notice that a PL 3-disc around the origin of $P$ is L homeomorphic to the configuration seen in the lower-right of fig. \ref{fig:triplecompose}. Therefore by \textit{Example \ref{S3decomposition}}, we have a well-defined {\v C}ech cohomology class/$U(1)$-gerbe $\sigma\cup\sigma'\cup\sigma'' \in {\check{H}}^2(\mathbb{G}^u,U(1))$ attached to $P$ where $u$ is the degeneracy intersection surrounding the central vertex in $P$. 

Thus elements of $\C_q(\mathbb{G}^P)$ are characterized by the data $(\Phi;\sigma\cup_2\sigma'\cup_2\sigma'')$, where
\begin{equation}
    \Phi = \Phi_1\ostar \Phi_2\ostar \Phi_{3}\ostar\Phi_4\label{S^3monodromy}
\end{equation}
is the associated 2-monodromy state.

\medskip

Now consider the one-point suspension $\Lambda P$ of $P$, which by construction bounds a 3-disc. This 3-disc is precisely the genus-0 handlebody $H_0$ arising from a type-0 handlebody decomposition of the 3-sphere $S^3$, for which $P$ is the partition. 
\begin{definition}\label{S3catstate}
    A \textbf{categorical state on $S^3$} is characterized by 
    \begin{enumerate}
        \item a cone functor $\omega \in\operatorname{Fun}(\C_q(\mathbb{G}^P),\mathsf{Hilb})$ on 2-monodromy states of the form \eqref{S^3monodromy}, and
        \item a $U(1)$-gerbe of the form $\sigma\cup_2\sigma'\cup_2\sigma''\in \check{H}^2(\mathbb{G}^P,U(1))$.
    \end{enumerate}
    If $\omega$ lies in the image of the Yodena embedded \eqref{pairing}, then we call it a \textbf{closed Wilson surface state of $S^3$}.
\end{definition}
\noindent See \S \ref{nonabeliansurface} and \textbf{Proposition \ref{closedwilsonsurfaces}} later.

\begin{rmk}
    Note in this definition, categorical states on $S^3$, or \textit{any} 3-manifold without boundary for that matter, are automatically 2-gauge invariant. This is because the underlying 2-graph states are 2-monodomy states, which we know from \S \ref{invarbdy} is $\mathbb{U}_q\G^B$-invariant. 
\end{rmk}

Due to \textbf{Theorem \ref{handlebody}}, the above procedure can be applied to \textit{any} closed connected oriented 3-manifold $M$. If $M$ has boundary, then the underlying 2-graph states are 2-holonomy states, and hence not necessarily $\mathbb{U}_q\G^B$-invariant. In any case, this gives a procedure in which categorical states as in \textbf{Definition \ref{linearfunctors}} can be assigned to a type-0 partition $P$ of a 3-manifold.


\medskip



Throughout the following, we shall arrange the 2d polyhedron $P$ with boundary $\partial P = B_0\coprod \bar B_1$, such that $B_0$ consist precisely of the source edges living on the boundary $\partial\Gamma_P$ of the underlying 2-graph $\Gamma_P$ of $P$. 

\subsubsection{Non-Abelian Wilson surface states of 2-Chern-Simons theory}\label{nonabeliansurface}
By the full-faithful Yoneda embedding $\C_q(\mathbb{G}^P)\hookrightarrow \operatorname{Fun}_\mathsf{Meas}^{\ast}(\C_q(\mathbb{G}^P),\mathsf{Hilb})$ in \textbf{Proposition \ref{yoneda}}, there is a measureable subcategory equivalent to $\C_q(\mathbb{G}^P)$. 

Upon imposing $\bullet$-module structure, there is then a subcategory, denoted by $$\widehat{\C}_q(\mathbb{G}^P)\subset\operatorname{Fun}_\mathsf{Meas}^{\bullet,\ast}(\C_q(\mathbb{G}^P),\mathsf{Hilb}),$$ which is equivalent to the equivariantization/the \textit{lattice observables} $\C_q(\mathbb{G}^P)^{\mathbb{U}_q\G^B}$ (see \S \ref{observ}). We call $\widehat{\C}_q(\mathbb{G}^P)$ the \textbf{non-Abelian Wilson surface states} of the 2-Chern-Simons theory. 

As advertised in the beginning of \S \ref{cone3d}, we now investigate its \textit{internal} properties.
\begin{proposition}\label{wilsonsurfaces}
    $\widehat{\C}_q(\mathbb{G}^P)$ is a category internal to the bicategory $\mathsf{Meas}$.
\end{proposition}
\begin{proof}

     We treat Wilson surface states $\widehat{\C}_q(\mathbb{G}^P)$ as presheafs $\phi\mapsto\omega_\phi$ of the 2-graph states $\C_q(\mathbb{G}^P)$, valued in the category $\mathsf{Hilb}$ which possesses small co/limits. There are then canonically induced restrictions of scalars functors 
     \begin{equation*}
    \hat s: \omega_\phi \mapsto \omega_\phi\circ s^*,\qquad \hat t:\omega_\phi\mapsto \omega_\phi\circ t^*,\qquad \forall~\phi\in\C_q(\mathbb{G}^{\Gamma}),
\end{equation*}
induced by the cofibrant cosource/cotaget maps $s^*,t^*$ on $\C_q(\mathbb{G}^P)$. 

Since the Yoneda embedding preserves limits, the cocomposition $\Delta_v:\C_q(\mathbb{G}^P)\rightarrow \C_q(\mathbb{G}^P)\,_{t^*}\times_{s^*}\C_q(\mathbb{G}^P)$ to the pushout canonically induces a composition operation $\circ :\widehat{\C}_q(\mathbb{G}^P)\,_{\hat t}\times_{\hat s}\widehat{\C}_q(\mathbb{G}^P)\rightarrow \C_q(\mathbb{G}^P)$ on the pullback, making $\widehat{\C}_q(G^{B_0})\xleftarrow{\hat s}\widehat{\C}_q\big((\mathsf{H}\rtimes G)^P\big)\xrightarrow{\hat t}\widehat{\C}_q(G^{B_1})$ into a category internal to $\mathsf{Meas}$.


It is then not hard to see that the associativity of $\circ$ come from the coassociativity of $\Delta_v$.
\end{proof}

\begin{rmk}\label{CSvsDW}
    We emphasize here that Wilson surface states are \textit{not} defined as the 2-holonomies $\mathbb{G}^P$ themselves. They differ by \textit{two} dualities
    $$ \mathbb{G}^P \rightsquigarrow \C_q(\mathbb{G}^P)\rightsquigarrow \operatorname{Fun}_\mathsf{Meas}^{\ast}(\C_q(\mathbb{G}^P),\mathsf{Hilb}),$$
    which can possibly be an equivalence (of monoidal categories) if (i) no non-trivial quantum deformations occur and (ii) all of the Yoneda-type embeddings (\textbf{Propositions \ref{yoneda}, \ref{cylinderembeddings}}) are equivalences. As one expects, the only known case where this happens is when $\mathbb{G}$ is finite in the Morita context of $\mathsf{2Vect}$, not $\mathsf{Meas}$. In which case, we obtain the \textbf{4d 2-group Dijkgraaf-Witten theory} \cite{Yetter:1993dh,Martins:2006hx,Bullivant:2016clk,Delcamp:2017pcw,Bochniak:2020vil}, instead of 2-Chern-Simons theory. Such Djkgraaf-Witten TQFTs appear in the study of topological phases of matter, which explains why many condensed matter literature \cite{Wen:2019,Kapustin:2013uxa,Chen2z:2023,Else:2017yqj,Wang:2016rzy,Wan:2014woa,PUTROV2017254,Kong:2020wmn,walker2012} can get away with reading off the fusion and braiding properties of the underlying anomaly-free non-degenerate gapped state directly from the action. 
\end{rmk}

We will actually need $\widehat{\C}_q(\mathbb{G}^P)$ to be monoidal later, in order to keep track of more geometric data. Such a monoidal structure can be induced from the internal coproduct functor $\Delta_h$ on $\C_q(\mathbb{G}^P)$, but we shall introduce a modified version explicitly in \S \ref{monoidality}. 

\begin{rmk}\label{gaugedoubles}
    There is a very widely-accepted statement in the categorical symmetries literature \cite{Bartsch:2022mpm,KongTianZhou:2020,Baez:2012,Ganter:2014,Delcamp:2023kew,Sean:private}, which is:
    \begin{center}
        \emph{Finite 2-group $\mathbb{G}$ Dijkgraaf-Witten theories are described by the Drinfel'd centre $Z_1\big(\operatorname{2Rep}(\mathbb{G})\big)$.}
    \end{center}
    Given the above remark, this statement is not immediate and requires verification. This was done for the 3+1d $\bbZ_p$-toric code  (and its spin counterpart) in \cite{Chen2z:2023}, where $p$ is prime. The 2-category capturing the Wilson surface states were explicitly matched to well-known 2-categories studied in \cite{Johnson-Freyd:2020,Johnson_Freyd_2023,KongTianZhou:2020,Wen:2019} for $p=2$.\footnote{The 4d gravitational-anomalous boundary of the 5d $\bbZ_2$-protected state $w_2w_3$ \cite{Thorngren2015,Johnson_Freyd_2023,Chen:2021xks}, on the other hand, is known to \textit{not} be a centre.}
\end{rmk}

\subsection{PL 2-ribbons $\operatorname{PLRib}'_{(1+1)+\epsilon}(D^4)$ in a 4-disc}\label{PL2ribbons}
The geometry we will consider is the following. For the time being, imagine a PL 4-disc $D^4=[0,1]^4 \subset \R^4$ whose top/bottom boundaries $D^3\times\{0,1\}$ are equipped with embedded directed graphs $B_{0,1}$, respectively. Let $P$ denote a 2d polyhedron, embedded in $D^4 = [0,1]^4$, such that $P$ intersects the top layer at $B_0$ and the bottom layer at $B_1$, both transversally. We call such a configuration $\,_{B_0}P_{B_1}$.

\begin{definition}\label{2-skeins}
    The monoidal category $\operatorname{PLRib}_{2+\epsilon}'(D^4)$ consist of
    \begin{itemize}
        \item the objects are the slab layers $D^3\times \{0,1\}$ with a framed oriented immersed PL 1-submanifolds $B_0,B_1$ (read: directed graphs), as well as PL homeomorphisms on them, and 
        \item the morphisms are the 4-slabs $D^4$ with a framed oriented immersed PL 2-submanifold $P\subset D^4$ (read: a 2d simple polyhedron) such that $P\cap (D^3\times\{0\})=B_0$ and $P\cap (D^3\times \{1\})=B_1$ transversally, as well as level-preserving PL homeomorphisms\footnote{What this means is that these are diffeomorphisms of the fibre bundles $D^4\rightarrow D^3$ and $D^4\rightarrow [0,1]$.} relative boundary.
    \end{itemize}
    The (horizontal) composition law is given by stacking these slabs long the $[0,1]$ direction: $\,_{B_0}P_{B_1} \circ \,_{B_1}P'_{B_{2}} = \,_{B_0}(P\cup_{B_1} P')_{B_2}$. The monoidal structure is given by disjoint union.
\end{definition}

Now consider PL 2-ribbon configuration of the form $B_0\coprod B'_0 \xRightarrow{P\coprod P'} B_1\coprod B'_1$. By applying a $\pi$-rotation of the \textit{entire} half-slab $D^3\times[1/2,1]$, while holding the top half $D^3\times[0,1/2]$ fixed, we obtain another PL 2-ribbon
$$B_0\coprod B'_0 \xRightarrow{(P\coprod P')^\pi} B'_1\coprod B_1.$$ 
Applying this operation twice, we obtain a PL 2-ribbon $(P\coprod P')^{2\pi}$ (see fig. \ref{fig:2ribbontwist}) which is {not} naturally isomorphic (ie. ambient isotopic relative boundary) to the original 2-ribbon $P\coprod P'$. This is because to undo such a $2\pi$-twist on the half-slab while keeping the boundary graphs fixed, we \textit{must} cross the polyhedra past each other, which is in general not an level-preserving diffeomorphism in $D^3\times[0,1]$.

\begin{figure}[h]
    \centering
    \includegraphics[width=1\linewidth]{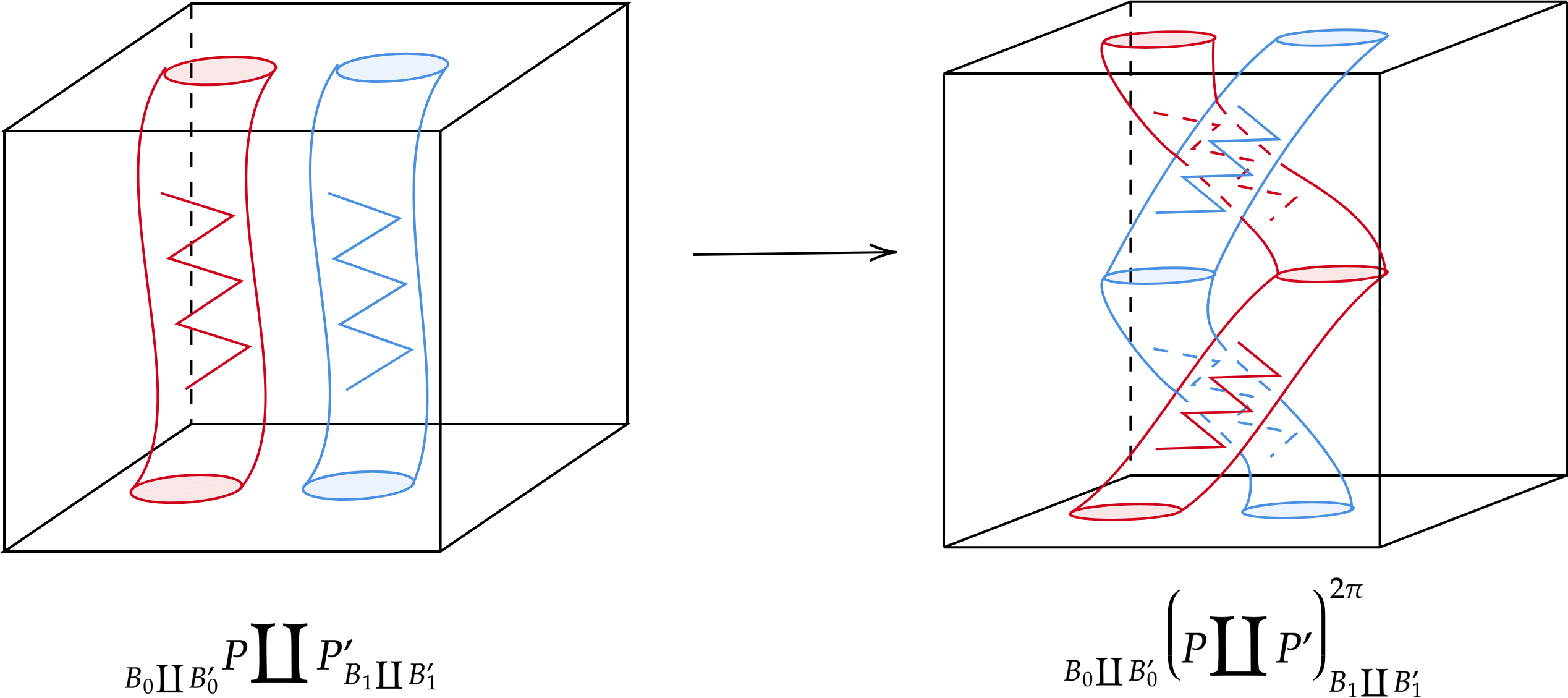}
    \caption{The $2\pi$-twisted PL 2-ribbon.}
    \label{fig:2ribbontwist}
\end{figure}

By a construction analogous to \S 2.1 of \cite{douglas2016internalbicategories}, each PL 2-ribbon in $\operatorname{PLRib}'_{(1+1)+\epsilon}(D^4)$ is a category internal to $\operatorname{PLTop}.$ Indeed, the so-called "$(n+\epsilon)$-dimensional bordisms" constructed there are categories \textit{internal} to $\mathsf{Mfld}$; see \textit{Remark \ref{diffbords}}. 

\begin{proposition}\label{internalribbon}
    $\operatorname{PLRib}'_{(1+1)+\epsilon}(D^4)$ is a double category.
\end{proposition}
\begin{proof}
    Each object $\,_{B_0}P_{B_1}\in\operatorname{PLRib}'_{(1+1)+\epsilon}(D^4)$ and their PL homeomorphisms can be represented as
\[\begin{tikzcd}
	{B_0} & {B_1} \\
	{B_0’} & {B_1’}
	\arrow[""{name=0, anchor=center, inner sep=0}, "P", "\shortmid"{marking}, from=1-1, to=1-2]
	\arrow["{f_0}"', from=1-1, to=2-1]
	\arrow["{f_1}", from=1-2, to=2-2]
	\arrow[""{name=1, anchor=center, inner sep=0}, "{P’}"', "\shortmid"{marking}, from=2-1, to=2-2]
	\arrow["\mathscr{a}"', shorten <=4pt, shorten >=4pt, Rightarrow, from=0, to=1]
\end{tikzcd},\]
    where $f_{0,1}$ are PL homeomorphisms in $D^3\times\{0,1\}$ of the boundary graphs $B_0,B_1$, and $\alpha$ is a PL homeomorphism  of $P$ rel. boundary in $D^4$. 
    
    The vertical and horizontal compositions and their associativity are obvious; the vertical composition unit is the identity PL homeomorphism, while the horizontal composition unit is given by the trivial PL 2-ribbon $B\times[0,1]: B\rightarrow B$. The level-preserving condition ensures that the $\alpha$'s satisfy the interchange law.
\end{proof}

\subsubsection{Horizontal functoriality: stacking on 4-discs}\label{positiveskeins}
For simplicity, we will for now assume that the graphs $B_0,B_1$ embedded in the slab layers are closed. Then, the 2d polyhedron $P$ within the slab has only $B_0,B_1$ as boundary. 

We shall identify $B_0=p_P\cap \partial\Gamma_P$ as precisely the subcomplex of the distinguished source edges $p_P$ (see \textbf{Proposition \ref{path}}) that lives on the boundary of $\Gamma_P$. All other source edges are internal. We will also assume the root vertex $v$ of the 2-graph $\Gamma_P$ to lie on the source boundary $v\in B_0\subset D^3\times\{0\}$. 

\begin{definition}\label{stackable}
    Take two PL 2-ribbon configurations $\,_{B_0}P_{B_1}$ and $\,_{B_0'}P'_{B_1'}$ embedded within $D^3\times[0,1]$ and $D^3\times[1,2]$, respectively. We say these two PL 2-ribbons are \textbf{stackable} iff there exists an orientation \textit{reversing} PL homeomorphism $f:B_1\cong B_0'$. 
    
    Denote by $P\cup_{B_1}P'$ the 2d simple polyhedron (with boundary $B_0,B_1'$) obtained by gluing of $P,P'$ at $B_1\cong B_0'$. Given level-preserving PL homeomorphisms $\mathscr{a},\mathscr{a}'$ on $P,P'$, we also have the concatenation $\mathscr{a}\cup_{B_1}\mathscr{a}'$ along $B_1$. The \textbf{stacking of $P$ and $P'$ along $f$} $\,_{B_0}(P\cup_{B_1}P')_{B_1'}$ is the {horizontal composition} in the double category $\operatorname{PLRib}'_{(1+1)+\epsilon}(D^4)$ obtained by rescaling the glued polyhedron $P\cup_{B_1}P'$ along the vertical axis $[0,2]\xrightarrow{\sim}[0,1]$ by one-half.
\end{definition}
\noindent We call $f$ the \textbf{stacking homeomorphism}. The associativity is obvious.

Now provided the PL 2-ribbon $\,_{B_0}P_{B_1}$ intersects the middle slice $D^3\times\{1/2\}$ transversally at a graph $B_{1/2}$, such that $P_1=P\cap(D^3\times[0,1/2])$ and $P_2= P\cap (D^3\times[1/2,1])$ remain 2d simple polyhedra, then we have
\begin{equation*}
    \,_{B_0}P_{B_1}\cong \,_{B_0}(P_1)_{B_{1/2}}\cup_{B_{1/2}}\,_{B_{1/2}}(P_2)_{B_1}.
\end{equation*}
This can be done for any PL 2-ribbon, since we can apply a PL homeomorphism which slides a neighbourhood of the trisection vertex away from the middle slice,\footnote{The reason we have to do this is because the graphs above and below the central trisection neighbourhood are not PL homeomorphic; see the right-hand side of fig. \ref{fig:interchanger}.} and apply a PL homeomorphism if necessary to ensure that it intersects $P$ transversally.

\begin{proposition}
    For each $\,_{B_0}P_{B_1}\coprod \,_{B_0'}P'_{B_1'} \in \operatorname{PLRib}_{2+\epsilon}'(D^4)$, we have 
    \begin{equation*}
        \,_{B_0\coprod B'_0}(P\coprod P')^{2\pi}_{B'_1\coprod B_1} \cong \Big(\,_{B_0\coprod B'_0}(P\coprod P')^{\pi}_{B'_{1/2}\coprod B_{1/2}}\Big)\cup_{B'_{1/2}\coprod B_{1/2}} \Big(\,_{B_{1/2}\coprod B'_{1/2}}(P\coprod P')^{\pi}_{B'_{1/2}\coprod B_{1/2}}\Big).
    \end{equation*}
\end{proposition}

However, as opposed to the usual 3d embedded ribbon category, the boundary slabs come with embedded 1-simplicial complexes $B$, instead of points. These complexes have more structure --- namely they can be pasted together along certain junctions. The composition along the boundary graphs will give rise to a {monoidal structure} which is \textit{not} just given by the disjoint union in general. Let us describe this in the following.

\subsubsection{Anchored connected summation of PL 2-ribbons}\label{connectdesum}
Let us now relax the assumption that the boundary graphs $B$ are closed, though they still remain connected. Let us describe the data necessary in order to facilitate the conjunction of PL 2-ribbons.

\begin{definition}
    A \textbf{marking} on a PL 2-ribbon $\,_{B_0}P_{B_1}$ is a distinguished framed oriented PL path $\ell: [0,1]\rightarrow D^3\times[0,1]$ embedded in $P$ (ie. its image is contained $\ell([0,1])\subset P$) such that $\ell$ intersects the slab layers $D^3\times\{0,1\}$ transversally at the graphs $B_0,B_1$.
    
    We call the endpoints $\ell(0)\in B_0,~\ell(1)\in B_1$ of a marking $\ell$ the \textbf{anchors}. The PL 2-ribbon $P$ is \textbf{marked} if it has equipped a set $L$ of such markings $\ell\in L$.
\end{definition}

Markings $L$ on a generic PL 2-ribbon $\,_{B_0}P_{B_1}$ is characterized by a bipartition $L = L^+\coprod L^-$, indicating the markings with positive or negative framings; namely, $\ell^\pm\in L^\pm$ iff its anchors $\ell^\pm(0),\ell^\pm(1)$ have positive/negative framing in $D^3\times\{0,1\}$. The set $L$ is therefore characterized by a tuple $(n,m)\in\bbZ_{\geq 0}^2$ for which $n=|L^+|$ and $m=|L^-|$.
\begin{definition}
    We call the anchors with positive framing \textbf{incoming}, while the others \textbf{outgoing}.
\end{definition}
\noindent We are going to assume without much loss of generality that the root vertex $v\in \Gamma_P$ of $P$ is an \textit{incoming} anchor.

    

Let $\,_{B_0}P_{B_1},~\,_{B_0'}P'_{B_1'}\in\operatorname{PLRib}_{2+\epsilon}'(D^4)$ denote two marked PL 2-ribbons. In the following, we will embed each of them into \textit{quarter-slab spaces} instead: $$P\subset D^2\times [0,1]\times [0,1],\qquad P'\subset D^2\times[1,2]\times [0,1],$$ and we will require the PL homeomorphisms on the boundary graphs $B$ to be level-preserving with respect to the fibrations $D^3\rightarrow D^2$ and $D^3\rightarrow[0,1]$.
\begin{definition}\label{summable}
    We say two disjoint marked PL 2-ribbons $P,P'$ with marking sets $L,L'$ are \textbf{connected summable} iff there exists markings $\ell^-\in L^-$ and $\ell'^+\in L'^+$ such that, upon embedding $P\coprod P' \subset D^3\times[0,2]\times[0,1]$, there exists PL \textit{framing-reversing} homotopy $H:\ell^-\Rightarrow \ell'^+$ in $D^2\times[0,2]\times[0,1]$ relative boundary.

    With this homotopy, consider the following PL 2-ribbon (see fig. \ref{fig:sum})
    \begin{equation*}
        \,_{B_0\vee_{\ell(0)}B_0'}(P \#_{H} P')_{B_1\vee_{\ell(1)}B_1'}\subset D^2\times[0,2]\times[0,1],
    \end{equation*}
    where $\vee$ denotes the wedge sum and $P\#_{H} P'\subset D^3\times[0,1]$ is the connected simple 2d polyhedron obtained by pasting the given homotopy $H: [0,1]\times [0,1]\rightarrow D^3\times[0,1]$ with $P\coprod P'$. The \textbf{PL connected sum $(\,_{B_0}P_{B_1})\#_H(\,_{B_0'}P'_{B_1'})$} along $H$ is the rescaling of this PL 2-ribbon along the third coordinate by $1/2$.
\end{definition}
\noindent We call $H$ the \textbf{summation collar} of $P\#_H P'$. Since we have split up the incoming and outgoing anchors along which the PL connection summation can be performed, the strict associativity\footnote{Geometrically here, having only identity components means that the associator is PL (2-)homotopic to the identity, which is true.} of $\#$ is obvious.

\begin{figure}[h]
    \centering
    \includegraphics[width=0.8\linewidth]{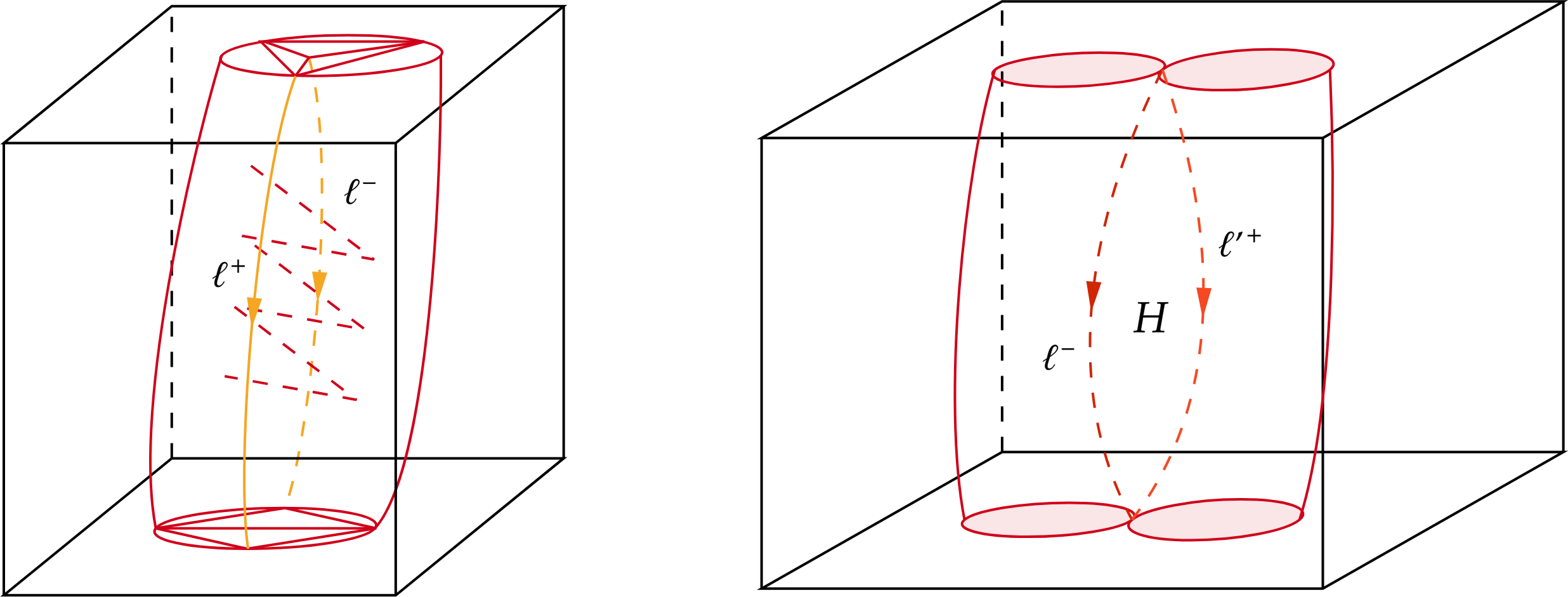}
    \caption{The markings on, and connected summation of, PL 2-ribbons.}
    \label{fig:sum}
\end{figure}

\begin{rmk}\label{skeinint}
    For the stacking of marked PL 2-ribbons, we must make sure that the incoming and outgoing anchors on the boundary graphs agree upon applying the stacking homeomorphism $f:B_1\cong B_0'$. This adds the following additional constraint to \textbf{Definition \ref{stackable}}:
    \begin{equation*}
        f(L^+) = L'^+,\qquad f(L^-) =L'^-.
    \end{equation*}
    This of course implies that the numbers $n=n',~m=m'$ of positively/negatively framed anchors on $B_1$ agrees with those on $B_0'$, otherwise the PL 2-ribbon cannot be stacked. If we consider PL 2-ribbons $P_1,\dots P_4$ for which (i) $P_1,P_3$ and $P_2,P_4$ are stackable and (ii) $P_1,P_2$ and $P_3,P_4$ can be PL connected summed, then we have a level-preserving PL diffeomorpism
    \begin{equation*}
        \mathfrak{b}: (P_1\cup_{B_1} P_3)\#_{H\ast H'} (P_2\cup_{B_2}P_4) \xrightarrow{\sim} (P_1\#_H P_2)\cup_{B_1\vee B_2}(P_3\#_{H'} B_4)
    \end{equation*}
    given by continuously deforming the framing of the underlying summation collars $H\cup H'$ on either side.
\end{rmk}

\subsubsection{Marked PL 2-ribbons as a double bicategory}
Given the 2d polyhedra $P,P'$ under consideration are path-connected, they are PL connected summable whenever their boundary graphs have the same number of framed anchors. The above structures immediately implies the following.
\begin{proposition}\label{markedPL2skeins}
    Marked PL 2-ribbons in the 4-disc $D^4$ are bicategories internal to $\operatorname{PLTop}$ (cf. \S 3.1 in \cite{douglas2016internalbicategories}). Together, they form a $\mathrm{double~bicategory}$ $\mathcal{T}'^{PL}_\text{mrk}$.
\end{proposition}
\begin{proof}
    The objects $n\in\bbZ_{\geq 0}$ are given by $n$ framed points, (horizontal) 1-morphisms $B:n\rightarrow m$ are graphs embedded in $D^3$ with $n,m$ incoming/outgoing (ie. positively-/negatively-framed) external marked points, and (horizontal) 2-morphisms given by $P: B_0\Rightarrow B_1: n\rightarrow m$ given by a marked PL 2-ribbon $\,_{B_0}P_{B_1}$ embedded in $D^4$.

    Composition of 1-morphisms $n\xrightarrow{B_1}m\xrightarrow{B_2}k$ is the wedge sum $B_1\vee_m B_2$. Vertical composition $B_0\xRightarrow{P}B_1\xRightarrow{P'}B_2$ of 2-morphisms is the stacking $P\cup_BP'$, and the horizontal composition $\big(B_0\xRightarrow{P}B_1\big)\#_H \big(B_0'\xRightarrow{P'}B_1'\big)$ is the PL connected summation $P\#_H P'$ over \textit{all} possible summation collars $H:L^-\Rightarrow L'^+$.  
    
    Up to rescaling, the identity 1-morphisms $1_n:n\rightarrow n$ are straight lines $\{1,\dots,n\}\times [0,1]\times\{0\}$, and the identity 2-morphism $\id_B: B\rightarrow B$ is the cylinder $B\times[0,1]\subset D^2\times[0,1]\times[0,1]$. 
    
    By performing diffeomorphisms on the framed points in $D^2$, the PL 2-ribbons $P$ thus form bicategories internal to $\mathsf{PLTop}$. Its collection $\mathcal{T}^{PL}_\text{mrk}$ is a tricategory, which is equivalent to a bicategory internal to $\mathsf{Cat}$ (\S 4.2, \cite{douglas2016internalbicategories}) --- aka. a \textit{double bicategory}. The shape of the 3-cells in $\mathcal{T}^{PL}_\text{mrk}$ takes the form
\[\alpha=\begin{tikzcd}
	n \\
	&&& m \\
	{n’} \\
	&&& {m’}
	\arrow[""{name=0, anchor=center, inner sep=0}, "{B_0}", "\shortmid"{marking}, curve={height=-12pt}, from=1-1, to=2-4]
	\arrow[""{name=1, anchor=center, inner sep=0}, "{B_1}"', "\shortmid"{marking}, curve={height=12pt}, from=1-1, to=2-4]
	\arrow[from=1-1, to=3-1]
	\arrow[from=2-4, to=4-4]
	\arrow[""{name=2, anchor=center, inner sep=0}, "{B_1’}"', "\shortmid"{marking}, curve={height=18pt}, from=3-1, to=4-4]
	\arrow[""{name=3, anchor=center, inner sep=0}, "{B_0’}"{pos=0.7}, "\shortmid"{marking}, curve={height=-18pt}, from=3-1, to=4-4]
	\arrow[""{name=4, anchor=center, inner sep=0}, "P", "\shortmid"{marking}, shorten <=3pt, shorten >=3pt, Rightarrow, from=0, to=1]
	\arrow["{f_0}"', shift right=5, shorten <=9pt, shorten >=9pt, Rightarrow, from=1, to=2]
	\arrow["{f_1}"{pos=0.6}, shift left=5, shorten <=8pt, shorten >=8pt, Rightarrow, from=0, to=3]
	\arrow[""{name=5, anchor=center, inner sep=0}, "{P’}", "\shortmid"{marking}, shorten <=5pt, shorten >=5pt, Rightarrow, from=3, to=2]
	\arrow["\mathscr{a}"', shorten <=9pt, shorten >=9pt, Rightarrow, scaling nfold=3, from=4, to=5]
\end{tikzcd};\]   
the 2-cells $f_0,f_1$ represent the PL homeomorphisms on the graphs $B_0,B_1$, while the 3-cell $\alpha$ is a diffeomorphism rel. boundary on the PL 2-ribbons $P$.

     To ensure the relevant interchange laws --- as described in \textit{Remark \ref{skeinint}} --- and the interchange associativity (see F2-8, F3-8, and F3-14 in \S 3.1 of \cite{douglas2016internalbicategories}, respectively) are satisfied, we require the relevant PL homeomorphisms to be suitably level-preserving. This means that the isotopies $f_0,f_1$ are level-preserving in $D^2\times[0,1]$, while $\alpha$ are "doubly" level-preserving in $D^2\times[0,1]\times[0,1]$ (cf. \cite{BAEZ2003705} and \S \ref{baezlangford}).
\end{proof}

For generic $n,m\in\bbZ_{\geq 0}$, the hom-category $\operatorname{Hom}_{\mathcal{T}'^{PL}_\text{mrk}}(n,m)$ is
\begin{itemize}
    \item  a left $\operatorname{End}_{\mathcal{T}'^{PL}_\text{mrk}}(n)$-module and 
    \item a right $\operatorname{End}_{\mathcal{T}'^{PL}_\text{mrk}}(m)$-module 
\end{itemize}
under PL connected summation $\#$. Notice that if there are no markings $n=0$, then $\#_\emptyset =\coprod$ reduces to the disjoint union. Thus $\operatorname{End}_{\mathcal{T}'^{PL}_\text{mrk}}(0)$ recovers \textbf{Definition \ref{higherGskein}}.

\begin{rmk}\label{emptyboundary}
    A subtlety that should be emphasized here is that all PL 2-ribbons we are considering are based spaces. Hence, by "$0\in\bbZ_{\geq 0}$" we mean an unframed base point $v$. We shall always consider such a point to be \textit{external}, hence 1-morphisms of the form $0\xrightarrow{B}n$ can be thought of as directed graphs with a single incoming vertex, and analogously for $n\xrightarrow{B'}0$. The "trivial 1-endomorphism $\emptyset:0\rightarrow 0$" is thus understood as the trivial graph $v$, not {\it literally} the empty set. This allows us to define 2-morphisms of the form "$B\xRightarrow{P}\emptyset$" as marked PL 2-ribbons such that $\ell(1)=v$ for all paths $\ell\in L$ in the marking set.
\end{rmk}

\subsection{$\mathbb{G}$-decorated ribbons from 2-Chern-Simons theory}\label{Gfuncmon}
Recall the measureable category $\cV^X_q$ over a smooth measureable space $X$ in \textbf{Definition \ref{quantumhermitian}}. The quantum categorical coordinate ring $\C_q(\mathbb{G})\subset \cV^X_q$ is a 2-subcategory for $X=(\mathbb{G},\mu)$, and we let $\widehat{\C}_q(\mathbb{G})$ denote its image under the Yoneda embedding as in \textbf{Proposition \ref{yoneda}}.

In accordance with \textbf{Proposition \ref{wilsonsurfaces}}, we can view $\widehat{\C}_q(\mathbb{G})$ as a double category of measureable fields in $\mathsf{Meas}$.
The {\it raison d'{\^e}tre} \textit{Remark \ref{diffbords}} then allows us to finally define the following.
\begin{definition}\label{higherGskein}
    The \textbf{category of $\mathbb{G}$-decorated ribbons} $$\operatorname{PLRib}'^{\mathbb{G};q}_{(1+1)+\epsilon}(D^4)\equiv \operatorname{Fun}\big(\operatorname{PLRib}'_{(1+1)+\epsilon}(D^4),\widehat{\C}_q(\mathbb{G})\big)$$ is the double category of \textit{double functors} \cite{Kerler2001}
    $$\Omega: \left(\begin{tikzcd}
	{B_0} & {B_1} \\
	{B_0’} & {B_1’}
	\arrow[""{name=0, anchor=center, inner sep=0}, "P", "\shortmid"{marking}, from=1-1, to=1-2]
	\arrow["{f_0}"', from=1-1, to=2-1]
	\arrow["{f_1}", from=1-2, to=2-2]
	\arrow[""{name=1, anchor=center, inner sep=0}, "{P’}"', "\shortmid"{marking}, from=2-1, to=2-2]
	\arrow["\alpha"', shorten <=4pt, shorten >=4pt, Rightarrow, from=0, to=1]
\end{tikzcd}\right) \mapsto \left(\begin{tikzcd}
	{\sigma_0} & {\sigma_1} \\
	{\sigma_0’} & {\sigma_1’}
	\arrow[""{name=0, anchor=center, inner sep=0}, "\omega", "\shortmid"{marking}, from=1-1, to=1-2]
	\arrow["{\Omega f_0}"', from=1-1, to=2-1]
	\arrow["{\Omega f_1}", from=1-2, to=2-2]
	\arrow[""{name=1, anchor=center, inner sep=0}, "{\omega’}"', "\shortmid"{marking}, from=2-1, to=2-2]
	\arrow["\Omega\alpha"', shorten <=4pt, shorten >=4pt, Rightarrow, from=0, to=1]
\end{tikzcd}\right)$$ parameterized by the non-Abelian Wilson surface states $\,_{\sigma_0}\omega_{\sigma_1}\in\widehat{\C}_q(\mathbb{G}^{\,_{B_0}P_{B_1}})$ for which \begin{equation*}
        \begin{cases}
            \hat s(\omega)=\sigma_0 \\ \hat t(\omega)=\sigma_1
        \end{cases},\qquad \widehat{\C}_q\big(G^{B_0})\xleftarrow{\hat s}\widehat{\C}_q\big((\mathsf{H}\rtimes G)^P\big)\xrightarrow{\hat t}\widehat{\C}_q\big(G^{B_1}).
    \end{equation*}
    The ambient PL isotopies $f_{0,1},\alpha$ on the PL 2-ribbons are sent to measureable isomorphisms $\Omega f_{0,1},\Omega\alpha$ on the Wilson surface states. 
\end{definition}

Note $\Omega$ contains not just the Wilson surface states, but also the following data:
\begin{enumerate}
        \item an interchanger sheaf automorphism for each trisection vertex; see \S \ref{interchanger}, and
        \item a $U(1)$-gerbe $\check{H}^2(\mathbb{G},U(1))$ for each triple point; see \S \ref{sec:triplepoint}.
\end{enumerate}
These allow the $\mathbb{G}$-decorated ribbons $\operatorname{PLRib}^{\mathbb{G};q}_{(1+1)+\epsilon}(D^4)$ to capture the geometry of 2d simple polyhedra up to diffeomorphism. This is important for the topology of embedded 3-manifolds, as we have seen in \S \ref{3handles}.

\begin{proposition}\label{closedwilsonsurfaces}
    $\Omega(\emptyset) \simeq \mathsf{Hilb}$ on the empty PL 2-ribbon. For $\partial P = \emptyset$ without boundary (which of course implies $B_0,B_1=\emptyset$), we call $\Omega(P)$ the \textbf{closed Wilson surface states}.
\end{proposition}
\begin{proof}
    These follow immediately from the fact that $\C_q(\mathbb{G}^\emptyset) \simeq\mathsf{Hilb}$. 
\end{proof}
\noindent The $S^3$-state constructed in \S \ref{cone3d}, for instance, define the closed Wilson surface states $\Omega(P_{S^3})$ on $S^3$; recall \textbf{Definition \ref{S^3monodromy}}.

\medskip

The above definition is not fully complete, however, and we shall give the "correct" one later in \textbf{Definition \ref{markedhigherGskein}}. However, it does highlight the following central idea.

\begin{rmk}
     \textbf{Definition \ref{higherGskein}} is the reason for our insistence on working with \textit{internal} categories throughout the quantization scheme we have developed/are developing. Such structures are not only natural from the perspective of higher-gauge principal bundles \cite{Nikolaus2011FOUREV,Baez:2012,Schommer_Pries_2011}, but also from that of extended $(n+1)+\epsilon$-dimensinoal bordisms \cite{douglas2016internalbicategories,Bunk_2025,haioun2025nonsemisimplewrtboundarycraneyetter}.
\end{rmk}

\subsubsection{Functoriality against the stacking of 4-discs}
The functoriality is immediate from \textbf{Proposition \ref{wilsonsurfaces}}, but let us describe it explicitly. To mediate the gluing construction, we require an equivalence $\C_q(G^{B_1}) \simeq \C_q(G^{B_0'})$ of the categories of measureables sheaves, and they must fit into the following cospan diagram
\begin{equation}\begin{tikzcd}
	& {\C_q\big(\mathbb{G}^{\,_{B_0}P_{B_1}}\big)} && {\C_q\big(\mathbb{G}^{\,_{B_0'}P'_{B_1'}}\big)} \\
	{\C_q(G^{B_0})} && {\C_q(G^{B_1})\simeq \C_q(G^{B_0'})} && {\C_q(G^{B_1'})}
	\arrow["{s^*}", from=2-1, to=1-2]
	\arrow["{t^*}"', from=2-3, to=1-2]
	\arrow["{s’^*}", from=2-3, to=1-4]
	\arrow["{t’^*}"', from=2-5, to=1-4]
\end{tikzcd}\label{cospan}
\end{equation}
formed from the cofibrant cosource/cotarget functors on the sheaves/2-graph states within the slab.

Denote by the pushout $\C_q\big(\mathbb{G}^{\,_{B_0}P_{B_1}}\big)\times_{B_1}\C_q\big(\mathbb{G}^{\,_{B_0'}P'_{B_1'}}\big)$ along \eqref{cospan}, we dualize it via the full-faithful imit-preserving Yoneda embedding to obtain the associated pullback $\widehat{\C}_q\big(\mathbb{G}^{\,_{B_0}P_{B_1}}\big)\times_{B_1}\widehat{\C}_q\big(\mathbb{G}^{\,_{B_0'}P'_{B_1'}}\big)$ upon which we can define a canonical additive measureable functor 
    \begin{align*}
        \circ_{B_1}: \widehat{\C}_q\big(\mathbb{G}^{\,_{B_0}P_{B_1}}\big)\times_{B_1}\widehat{\C}_q\big(\mathbb{G}^{\,_{B_0'}P'_{B_1'}}\big)\rightarrow \widehat{\C}_q\big(\mathbb{G}^{\,_{B_0}(P\cup_{B_1}P')_{B_1'}}\big).
    \end{align*}
This functor $\circ_{B_1}$ can be understood as a form of \textit{profunctor composition} 
$$\,_{\sigma_0}(\omega\circ_{B_1}\omega')_{\sigma_1'} = \int^{\sigma_0\cong \sigma_1\in \C_q(G^{B_1})}\,_{\sigma_0}\omega_{\sigma_1}\times \,_{\sigma_0'}\omega'_{\sigma_1'}$$
of Wilson surface states $\,_{\sigma_0}\omega_{\sigma_1}\in \widehat{\C}_q\big(\mathbb{G}^{\,_{B_0}P_{B_1}}\big)$ and $\,_{\sigma_0'}\omega'_{\sigma_1'}\in \widehat{\C}_q\big(\mathbb{G}^{\,_{B_0'}P'_{B_1'}}\big)$.

\medskip

To describe the pushout $\C_q\big(\mathbb{G}^{\,_{B_0}P_{B_1}}\big)\times_{B_1}\C_q\big(\mathbb{G}^{\,_{B_0'}P'_{B_1'}}\big)$ more explicitly, we will leverage \textbf{Theorem \ref{invariance}} and use the degeneracy maps $\delta,\delta'$ in a combinatorial triangulations of the 2d polyhedra $P,P'$. Let $\tilde f$ denote the extension of the gluing PL homeomorphism $f:B_1\cong B_0'$ to the contractible 2-simplex $\delta(B_1)$. We define the degeneracy intersection $$\mathfrak{u}_{12} = \tilde f(\delta(B_1))\cap \delta'(B_0'),\qquad \tilde f(\delta(B_1)) = \delta'(f(B_1))$$ near the middle slab layer $f:B_1\cong B_0'$  (ie. a small\footnote{Since degeneracy 2-simplices are contractible, we can perform PL homeomorphisms that shrink $\mathfrak{u}_{12}$ to be as small as we wish.} collar around $B_1\cong B_0'$).

Define the full measureable subcategory $$\C_q\big(\mathbb{G}^{\,_{B_0}P_{B_1}}\big)\times_{B_1}\C_q\big(\mathbb{G}^{\,_{B_0'}P'_{B_1'}}\big)\subset\C_q\big(\mathbb{G}^{\,_{B_0}P_{B_1}}\big)\times\C_q\big(\mathbb{G}^{\,_{B_0'}P'_{B_1'}}\big)$$ consisting of pairs $(\phi,\phi')$ of 2-graph states for whom there exist a natural measureable sheaf isomorphism
    \begin{equation}
        \phi\mid_{\mathbb{G}^{\mathfrak{u}_{12}}} \cong \phi'\mid_{\mathbb{G}^{\mathfrak{u}_{12}}}.\label{bdyamenable}
    \end{equation}
This additive measureable subcategory defines $\C_q\big(\mathbb{G}^{\,_{B_0}P_{B_1}}\big)\times_{B_1}\C_q\big(\mathbb{G}^{\,_{B_0'}P'_{B_1'}}\big)$.

\medskip

\begin{definition}\label{compositionskein}
    The (horizontal) \textbf{functoriality of $\mathbb{G}$-decorated ribbons} is the data of a measureable natural isomorphism
    \begin{equation*}
\Omega(P\cup_{B_1}P') \cong \omega\circ_{B_1}\omega',\qquad \forall~ P,P'\in \operatorname{PLRib}'_{(1+1)+\epsilon}(D^4)
    \end{equation*}
    where $\Omega(\,_{B_0}P_{B_1}) = \,_{\sigma_0}\omega_{\sigma_1}$ and $\Omega(\,_{B_0'}P_{B_1'}) = \,_{\sigma_0'}\omega'_{\sigma_1'}$, satisfying the obvious coherence conditions against the compositional associators/unitors.
\end{definition}

\begin{rmk}\label{laxdoublefunctor}
There is a more general notion of \textit{double lax functors/pseudofunctors} \cite{Shulman2011,Pare:2011,froehlich2024yonedalemmarepresentationtheorem}, in which functoriality is witnessed by a (not necessarily invertible) double natural transformation $\Omega_\circ:\Omega\circ(-\cup_{B_1} -) \Rightarrow \Omega(-)\circ_{B_1}\Omega(-)$, whose components are given by vertical measureable morphisms
        \begin{equation*}
\Omega_\circ: \Omega(P\cup_{B_1}P') \xrightarrow{\sim} \omega\circ_{B_1}\omega',\qquad \forall~ P,P'\in \operatorname{PLRib}'_{(1+1)+\epsilon}(D^4)
    \end{equation*}
    in $\widehat{\C}_q(\mathbb{G}).$ These morphisms must also satisfy natural commutative conditions against the 2-morphisms $\Omega(\alpha)$ in the data of the double functor $\Omega$. We will assume such data to be trivial $\Omega_\circ=\id$ in the following.
\end{rmk}

Note that if $(\omega,\omega')\in \widehat{\C}_q\big(\mathbb{G}^{\,_{B_0}P_{B_1}}\big)\times_{B_1}\widehat{\C}_q\big(\mathbb{G}^{\,_{B_0'}P'_{B_1'}}\big)$ live in the pullback measureable subcategory, then they \textit{by construction} must satisfy $t^*(\omega) \cong s^*(\omega')$ up to measureable isomorphism, since the degeneracy intersection $\mathfrak{u}_{12}\supset \mathbb{G}^{B_1}$ contains decorations on $B_1\cong B_0'$.

\begin{rmk}
    Recall the local sheaf identifications $\alpha$ introduced in \textbf{Definition \ref{glueamenable}}. By holonomy-density, the sheaf isomorphism \eqref{bdyamenable} can be constructed from the local $\alpha$'s --- more precisely, if $B_1=\bigcup_e e \cong B_0'$ is given by a collection of 1-simplices, then \eqref{bdyamenable} can be written as $\bigotimes_e \alpha_e$ where $\alpha_e$ are the natural sheaf identifications across the edge $e$.
\end{rmk}

Since the composition $\circ$ is canonically induced from the (vertical) cocomposition of the 2-graph states, which is strictly coassociative, it is strictly associative.

\subsubsection{Monoidality under PL connected summation}\label{monoidality}
Recall from \textbf{Proposition \ref{wilsonsurfaces}} that the Wilson surface states $\widehat{\C}_q(\mathbb{G}^P)$ is a monoidal internal category, induced by the horizontal gluing of decorated 2-graphs. We will leverage this fact to define an internal monoidal structure $$\hat\otimes_H: \widehat{\C}_q(\mathbb{G}^P)\times_H\widehat{\C}_q(\mathbb{G}^{P'})\rightarrow \widehat{\C}_q(\mathbb{G}^{P\#_H P'})$$ on the Wilson surface states along the summation collar $H$. 

To begin, we first note that $\partial H = (\ell'^+)^{-1}\ast\ell^-$, and hence the boundary holonomies on $H$ are completely determined by the given edge decorations on the incoming $\ell^-$ and outgoing $\ell'^+$ markings of $P,P'$. Fixing these, we can then parameterize 2-graph states on $H$ as those sheaves $\phi_H\in\C_q((\mathsf{H}\rtimes G)^H)$ whose cosource/cotargets satisfy
\begin{equation*}
    s^*\phi_H \cong \Phi\mid_{\ell^-},\qquad t^*\phi_H\cong \Phi'\mid_{\ell^+}
\end{equation*}
for some given 2-graph states $\Phi\in\C_q(\mathbb{G}^P),~\Phi'\in\C_q(\mathbb{G}^{P'})$. This allows us to paste $\Phi,\Phi'$ across $\phi_H$. By holonomy-density, ${\C}_q(\mathbb{G}^{P\#_H P'})$ consists of 2-graph states of the form $$\Phi\ostar_{\ell_1^-} \phi_H\ostar_{\ell_2^+}\Phi',\qquad \Phi\in\C_q(\mathbb{G}^P)~\Phi'\in\C_q(\mathbb{G}^{P'}),$$ where the subscripts $\ell_{1,2}^+$ indicates the gluing-amenabiity conditions across the markings; see \S \ref{objA}.

Then, for each Wilson surface state (which are cone $\bullet$-module *-functors) $\omega \in \widehat{\C}_q(\mathbb{G}^P),~\omega'\in\widehat{\C}_q(\mathbb{G}^{P'})$, their monoidal product is defined through the coend (cf. \cite{Day:1970,Loregian_2021})
\begin{equation}
    (\omega\hat \otimes_H\omega')(\Phi\ostar_{\ell}\Phi') = \int^{\phi_H\in\C_q((\mathsf{H}\rtimes G)^H)}\omega(\Phi) \otimes\left(\int_{\mathbb{G}^H}^\oplus d\mu_H(\mathrm{z})(\phi_H)_\mathrm{z}\right)\otimes\omega'(\Phi'),\label{summationcollar}
\end{equation}
where $\int_{\mathbb{G}^H}^\oplus d\mu_H(-): \C_q(\mathbb{G}^H)\rightarrow \mathsf{Hilb}$ is the $\bullet$-invariant direct integral; see also \textit{Remark \ref{Dayconvolve}}.

\begin{rmk}\label{Dayconvolve}
    The appearance of the coend in \eqref{summationcollar} is inspired by both the formula in \cite{Alekseev:1994au}, as well as the \textit{Day convolution product} \cite{Day:1970} on the presheaves $\operatorname{Fun}(C^\text{op},\mathsf{Vect})$ of a $\bbC$-linear monoidal category $(C,\otimes,I)$,
    \begin{equation*}
        (F\otimes_\text{Day}G)(c) = \int^{(c_1,c_2)\in C\times C} \operatorname{Hom}(c\otimes c_1,c_2)\otimes F(c_1)\otimes G(c_2),
    \end{equation*}
    for which the Yoneda embedding $C\rightarrow \operatorname{Fun}(C^\text{op},\mathsf{Vect})$ is monoidal. In terms of the Day convolution, \eqref{summationcollar} essentially says that the summation collars $H$ are decorated by trivial face states living in $\mathsf{Hilb}$.
\end{rmk}

This property of being monoidal is shared by \textit{all} end-categories of the marked PL 2-ribbons, as detailed in \textbf{Proposition \ref{markedPL2skeins}}. To put them all together, we consider free formal linear combinations of marked PL 2-ribbons over $\bbC$, and take the formal direct sum 
    \begin{equation*}
        \operatorname{PLRib}'_{(1+1)+\epsilon}(D^4)\equiv \bigoplus_n\operatorname{End}_{\mathcal{T}'^{PL}_\text{mrk}}(n)
    \end{equation*}
as $\bbC$-modules. This allows us to enhance \textbf{Definition \ref{higherGskein}}.

\begin{definition}\label{markedhigherGskein}
The \textbf{marked $\mathbb{G}$-decorated ribbons} is the category $$\operatorname{PLRib}'^{\mathbb{G};q}_{(1+1)+\epsilon}(D^4) =\operatorname{Fun}(\operatorname{PLRib}'_{(1+1)+\epsilon}(D^4),\widehat{\C}_q(\mathbb{G}))$$ of additive \textit{monoidal} internal functors.

The \textbf{monoidality of marked $\mathbb{G}$-decorated ribbons} is the data of a double monoidal natural isomorphism $\Omega_{\hat\otimes}: \Omega(-\#_H -) \Rightarrow\Omega(-)\hat\otimes_H\Omega(-)$, satisfying the following coherence property against the functoriality witness $\Omega_\circ$ of \textit{Remark \ref{laxdoublefunctor}},
\begin{equation*}
    \beta\ast\big(\Omega_{\hat\otimes}\circ(\Omega_\circ\times\Omega_\circ)\big) =\big(\Omega_\circ \circ(\Omega_{\hat\otimes}\times\Omega_{\hat\otimes})\big) \ast \mathfrak{b},
\end{equation*}
where $\mathfrak{b}$ is the interchanger on $\operatorname{PLRib}'_{(1+1)+\epsilon}(D^4)$ (see \textit{Remark \ref{skeinint}}) and $\beta$ is the interchanger on $\cV^X_q$ (see \S \ref{interchanger}).
\end{definition}
\noindent The monoidal condition on $\Omega_{\hat\otimes}$ simply means that it satisfies the obvious coherence diagrams against the associators of $\operatorname{PLRib}'_{(1+1)+\epsilon}(D^4)$ and $\widehat{\C}_q(\mathbb{G})$. We will not write them out here.


\subsubsection{Isomorphism classes of Wilson surfaces}\label{PLhomology}
Prior to moving on, let us examine the measureable isomorphism classes of objects in $\widehat{\C}_q(\mathbb{G})$. As mentioned in \S \ref{nonabeliansurface}, the Yoneda embedding allows us to start with the equivariantization $\C_q(\mathbb{G})^{\mathbb{U}_q\G}\subset \cV^X_q$, where $X=(\mathbb{G},\mu)$.

Recall that $\cV^X,\cV^X_q$ are additive and exact as categories of certain sheaves of sections over $X$. Henceforth, let us denote the resulting ring of isomorphism classes by $[\cV^X],[\cV^X_q]$.
\begin{proposition}\label{bigradedQ}
    There is an injective ring map $[\C_q(\mathbb{G})]\rightarrow H(B\mathbb{G},\mathbb{Z})[t][q,q^{-1}]$ into a bigraded polynomial algebra over the cohomology classes of $\mathbb{G}$.
\end{proposition}
\begin{proof}   
    For this proposition, we shall consider $X=B\mathbb{G}$ as the classifying space (classifying 2-stack) of $\mathbb{G}$, which one can realize geometrically in terms of its {\v C}ech covers \cite{Schommer_Pries_2011,Baez:2004in,schreiber2013connectionsnonabeliangerbesholonomy,Nikolaus2011FOUREV,Henriques2006Integrating,Schreiber:2013pra}.

    
    Consider the classical, undeformed case first. By \textbf{Proposition \ref{vectorbundles}}, there is a forgetful functor $\cV^X\rightarrow \operatorname{Bun}_\bbC(X)$ which simply treats a geometric 2-graph state $\phi$ as a complex vector bundle. This induces an injective ring map $[\C(X)] \rightarrow [\operatorname{Bun}_\bbC(X)]$. 
    
    It is well-known that complex vector bundles are classified by its Chern classes $c_i\in H^{2i}(X,\mathbb{Z})$ \cite{atiyah2018k,book-charclass} up to isomorphism. The total Chern class $c(\phi) \in H^\bullet(X,\mathbb{Z})$ of a complex vector bundle $\phi\rightarrow X$ can be captured by the \textit{Chern polynomial} $$c(\phi;t) = 1+ \sum_{i\leq \operatorname{rk}\phi} c_i(\phi)t^i\in H^\bullet(X,\mathbb{Z})[t]$$ over the cohomology ring. Thus, we can write $[\operatorname{Bun}_\bbC(X)]\cong H(X,\mathbb{Z})[t]$, mapping isomorphism classes of 2-graph states $\phi\mapsto c(\phi;t)$ to its Chern polynomial. 
    
    Now in the quantum case, the sheaves of sections $\Gamma(X)\rightsquigarrow \Gamma(X)[[\hbar]]$ of complex vector bundles become $\star$-deformed over the power series ring $\bbC[[\hbar]]$ {\`a} la \cite{Bursztyn2000DeformationQO}. We let $\operatorname{Bun}_{\bbC,q}(X)$ denote the category of such $\star$-deformed complex vector bundles on $X$, as defined in \cite{Bursztyn2000DeformationQO}, equipped with $\bbC[[\hbar]]$-linear sheaf morphisms.

    This $\star$-deformation endows the Chern polynomials another grading coming from the powers of $q=e^{i\hbar}$. If we denote by the isomorphism classes $[\operatorname{Bun}_{\bbC,q}(X)] \cong H^\bullet(X;\bbZ)[t][q,q^{-1}]$, then the forgetful functor $\cV^X_q\rightarrow \operatorname{Bun}_{\bbC;q}(X)$ induces the desired injective map $[{\C}_q(\mathbb{G})] \rightarrow H(B\mathbb{G},\mathbb{Z})[t][q,q^{-1}]$.  
\end{proof}
\noindent In analogy with the theory of principal $G$-bundles \cite{book-charclass}, principal $\mathbb{G}$-bundles \cite{Baez:2012,schreiber2013connectionsnonabeliangerbesholonomy,Schommer_Pries_2011,Nikolaus2011FOUREV,Wockel2008Principal2A} are determined by pull-backs of cohomology classes in $H^\bullet(B\mathbb{G},\bbZ)$.

\medskip

Now upon equivariantizing by the $\mathbb{U}_q\G$-module structure, we have the 2-Chern-Simons observables $\mathcal{O}^\Gamma = \C_q(\mathbb{G})^{\mathbb{U}_q\G}$ defined in \S \ref{observ}. By \textbf{Proposition \ref{bigradedQ}} and the Yoneda embedding, we then have a map from $[\widehat{\C}_q(\mathbb{G})]\cong [\mathcal{O}^\Gamma]$ into the $[\mathbb{U}_q\G]$-invariant part of the bigraded cohomology ring $H(B\mathbb{G},\mathbb{Z})[t][q,q^{-1}]$. We denote this ring by
\begin{equation*}
    H_{\mathbb{G}}(B\mathbb{G},\bbZ)[q,q^{-1}][t]= \big(H(B\mathbb{G},\bbZ)[q,q^{-1}][t]\big)^{\mathbb{G}},
\end{equation*}
where the subscript "$\mathbb{G}$" denotes, morally, the "\textit{Lie 2-group $\mathbb{G}$-equivariant cohomology}" (cf. \cite{Schreiber:2013pra}) obtained upon equivariantizing by the $\mathbb{U}_q\G$-action. This notation is suggested by the fact (see \textbf{Definition \ref{additive2gaugesymms}}) that $\mathbb{U}_q\G$ has a $\mathbb{G}$-grading as a monoidal category.

By extending the above to Wilson surface states on non-trivial 2-graph lattices $\Gamma$, we then have the following.
\begin{proposition}\label{2ribboninvariant}
    Denote by $\operatorname{PLRib}'_{1+1}(D^4)=[\operatorname{PLRib}'_{(1+1)+\epsilon}(D^4)]$ the additive monoid of (formal linear combinations of) PL homeomorphism classes of PL 2-ribbons. Isomorphism classes of marked $\mathbb{G}$-decorated ribbons, $[\Omega]\in\operatorname{Fun}\big(\operatorname{PLRib}'_{1+1}(D^4),[\widehat{\C}_q(\mathbb{G})]\big)$, are then parameterized by the set
    \begin{equation*}
        \left\{H_{\mathbb{G}^B}(B\mathbb{G}^P,\mathbb{Z})[t][q,q^{-1}]\mid P\in \operatorname{PLRib}'_{1+1}(D^4)\right\},
    \end{equation*}
    where $B$ is the intersection of the singular graph of $P$ with its (oriented) boundary $\partial P$.
\end{proposition}
\begin{proof}
Clearly, the isomorphism class $[\Omega]$ is well-defined under (level-preserving) PL homeomorphism $P\simeq P'$.

Pick a combinatorial triangulation $\Gamma_P$ of $P$, whose underlying 1-skeleton $\Gamma_P^1$ gives a triangulation of its singular graph. Then by \textbf{Proposition \ref{bigradedQ}} (or its straightforward generalization to 2-graph states), isomorphism classes of marked $\mathbb{G}$-decorated ribbons are parameterized by $H_{\mathbb{G}^{\Gamma_{\partial P}^1}}(B\mathbb{G}^{\Gamma_P},\mathbb{Z})[t][q,q^{-1}]$, where $\Gamma_{\partial P}^1 = \Gamma_P^1\cap \partial P.$

Now thanks to \textbf{Theorem \ref{invariance}}, the 2-graph states ${\C}_q(\mathbb{G}^{\Gamma_P})$ do not depend on the choice of the combinatorial triangulation $\Gamma_P$ of a 2d simple polyhedron $P$.  Similarly, $\mathbb{U}_q\G^{\Gamma_{P}^1}$ do not depend on the choice of the induced triangulation on its singular graph $B$. Therefore the bigraded ring $H_{\mathbb{G}^B}(B\mathbb{G}^P,\mathbb{Z})[t][q,q^{-1}]$ does not depend on the triangulation.  This proves the statement.

\end{proof}

\begin{rmk}\label{knothomology}
    It is very interesting that the structure of bigraded cohomology rings appeared here, since the knot categorification program pioneered by Khovanov \cite{Khovanov:2006,Khovanov:2000,Khovanov_2010,rouquier:hal-00002981,Elias2010ADT,webster2013knot} produces bigraded chain complexes. The attentive reader may have also noticed that the definition of the 2-Chern-Simons $\mathbb{G}$-decorated ribbons \textbf{Definition \ref{markedhigherGskein}} bears a striking resemblance to the lasagna higher skein modules arising from Khovanov homology \cite{Manolescu2022SkeinLM,Morrison2019InvariantsO4}. Even further, the $(\infty,2)$-categories arising from categorical quantum groups \cite{Chen:2023tjf,Chen:2025?} underlying 2-Chern-Simons theory, as well as that arising from Soergel bimodules \cite{liu2024braided} underlying knot homology, are both braided monoidal. We will say more about this in \S \ref{higherskein}.
\end{rmk}

For posterity, let us recall the following notion \cite{book-charclass}.
\begin{definition}\label{chernnumber}
    The \textbf{(total) Chern number} of a complex vector bundle $E\rightarrow X$ on $X$ is 
    \begin{equation*}
        \operatorname{ch}(E) = \int_{[X]}c(E),\qquad [X]\in H_{\operatorname{dim}X}(X,\mathbb{Z}),
    \end{equation*}
    where $c(E)$ is the total Chern class of $E$ and $[X]$ is the fundamental homology class.
\end{definition}

\subsection{Reflection-positivity of $\mathbb{G}$-decorated ribbons}\label{reflectionpositivity}
By considering PL 2-ribbons as PL 2-manifolds, the following is immediate.
\begin{proposition}\label{higherdaggerPL2ribbon}
    Orientation reversals and a $2\pi$-rotations of the framing on $D^4$ induces the following functors
    \begin{align*}
        -^{\dagger_1}&:\mathcal{T}'^{PL}_\text{mrk}\rightarrow (\mathcal{T}'^{PL}_\text{mrk})^\text{1-op,2-op},\\
        -^{\dagger_2}&:\mathcal{T}'^{PL}_\text{mrk}\rightarrow (\mathcal{T}'^{PL}_\text{mrk})^\text{2-op}.
    \end{align*}
    which identify a 2-$\dagger$ structure on $\mathcal{T}'^{PL}_\text{mrk}$ \cite{ferrer2024daggerncategories,stehouwer2023dagger}.
\end{proposition}
\noindent This notion, as well as the framing and orientation pairings that we have defined in \S \ref{higherdaggerpairings}, will be crucial for the reflection-positivity of the $\mathbb{G}$-decorated ribbons.

\subsubsection{Codimension-1}
The geometry we will consider is the following. Let $B\in \operatorname{Hom}_{\mathcal{T}'^{PL}_\text{mrk}}(n,0)$ denote a connected directed graph with an unframed outgoing anchor $v$, and take $\,_{B}P_\emptyset\in \operatorname{PLRib}_{2+\epsilon}'(D^4)$ to be a marked PL 2-ribbon with the trivial target boundary graph (recall \textit{Remark \ref{emptyboundary}}). Let $L^+$ denote the marking set of $P$, which are all incoming.

Pick any combinatorial triangulation $\Gamma_P$ of $P$. By rotating the framing $(e,\nu)\mapsto e^T=(e,-\nu)$ of the source edges in $B$ (see \S \ref{2dagger}), we obtain a marked PL 2-ribbon $\,_{\emptyset}\tilde P_{\tilde B}$ whose target graph is the oppositely-framed graph $\tilde B$, and the set $\tilde L$ of orientation-reversed markings $\tilde \ell$, which are all incoming as well. We equip it with the triangulation $\Gamma_{\tilde{P}}=\Gamma_P^{\dagger_2} = \tilde\Gamma_P$. This allows us to stack these PL 2-ribbons together to obtain $\,_\emptyset(\tilde P\cup P)_\emptyset$.

By functoriality, $\mathbb{G}$-decorated PL 2-ribbons on this configuration live in the pullback
\begin{equation*}
    \widehat{\C}_q(\mathbb{G}^{\tilde P})^{\text{op}}\times_{B}\widehat{\C}_q(\mathbb{G}^P)\subset \operatorname{Fun}_{\mathsf{Meas}}^{\bullet,\ast}({\C}_q(\mathbb{G}^{\tilde P})^{\text{op}}\times_{B}{\C}_q(\mathbb{G}^P),\mathsf{Hilb}).
\end{equation*}
Note the framing pairing of \textbf{Definition \ref{framingpairing}} is precisely a $\bullet$-module cone *-functor. It in fact defines a Wilson surface state, living in the left-hand side of the above.

    Denote by $(\tilde\omega_P,\omega_P)\in \widehat{\C}_q(\mathbb{G}^{\tilde P})^{\text{c-op}_v}\times_{B}\widehat{\C}_q(\mathbb{G}^P)$ the framing pairing state given in \eqref{geometrypair2}. The composition law $\circ$ in \textbf{Definition \ref{compositionskein}} sends it to a Wilson surface state on $\bar P\cup_B P$:
\begin{equation}
    \Omega_P= \Omega_{\tilde P\cup_BP } = \tilde\omega_P\circ\omega_P\in \widehat{\C}_q\big(\mathbb{G}^{\tilde P\cup_B P}\big).\nonumber
\end{equation}   

\subsubsection{Codimension-2}
Next, we start with the composite PL 2-ribbon $\,_\emptyset(\tilde P\cup_B P)_\emptyset$, which contains $n$ markings equipped with the marking set $\tilde L\ast L=(\tilde L\ast L)^+$. Each marking in $\tilde L\ast L$ are incoming, and takes the form $\tilde \ell^+ \ast \ell^+$ concatenated along the middle anchors in $B$, with endpoints given by the trivial graph $\emptyset$ with unframed base point $v$.

Consider the PL 2-ribbon $\big(\,_\emptyset (\tilde P\cup_B P)_\emptyset\big)^{\dagger_1}= \,_\emptyset (\tilde P\cup_B P)^{\dagger_1}_\emptyset$. It has equipped a marking set $\overline{(\tilde L\ast L)}= \bar L \ast \bar{\tilde{L}}$ containing the concatenation of framing-reversed \textit{outgoing} paths $\overline{\ell^+}=\bar\ell^-,~\overline{\tilde\ell^+} = \bar{\tilde \ell}^-$ along the orientation-reversed boundary graph $\bar B$ in the middle. Hence up to ambient PL homeomorphism we have
\begin{equation*}
    (\tilde P\cup_B P)^{\dagger_1} \cong \bar P\cup_{\bar B}\bar {\tilde P}.
\end{equation*}
Importantly, each marking in $\tilde L\ast L$ is framing-reversing PL homotopous to some marking in $\overline{(\tilde L\ast L)}$, which allows us to form the connected summation 
\begin{equation}
     \mathscr{P}_B = \big(\bar P\cup_{\bar B}\bar {\tilde P}\big)\#_H (\tilde P\cup_B P) \cong \big(\bar P \#_{H_1}\tilde P\big)\cup_{\bar B\vee B}\big(\bar{\tilde P}\#_{H_2}P\big),\label{2-ribboninterchange}
\end{equation}
where we have used the interchanger diffeomorphism mentioned in \textit{Remark \ref{skeinint}}, and $H = H_1\ast H_2: \bar L \ast \bar{\tilde{L}}\Rightarrow \tilde L \ast L$ are the given summation collars. See fig. \ref{fig:4-ball}.

\begin{figure}[h]
    \centering
    \includegraphics[width=1\linewidth]{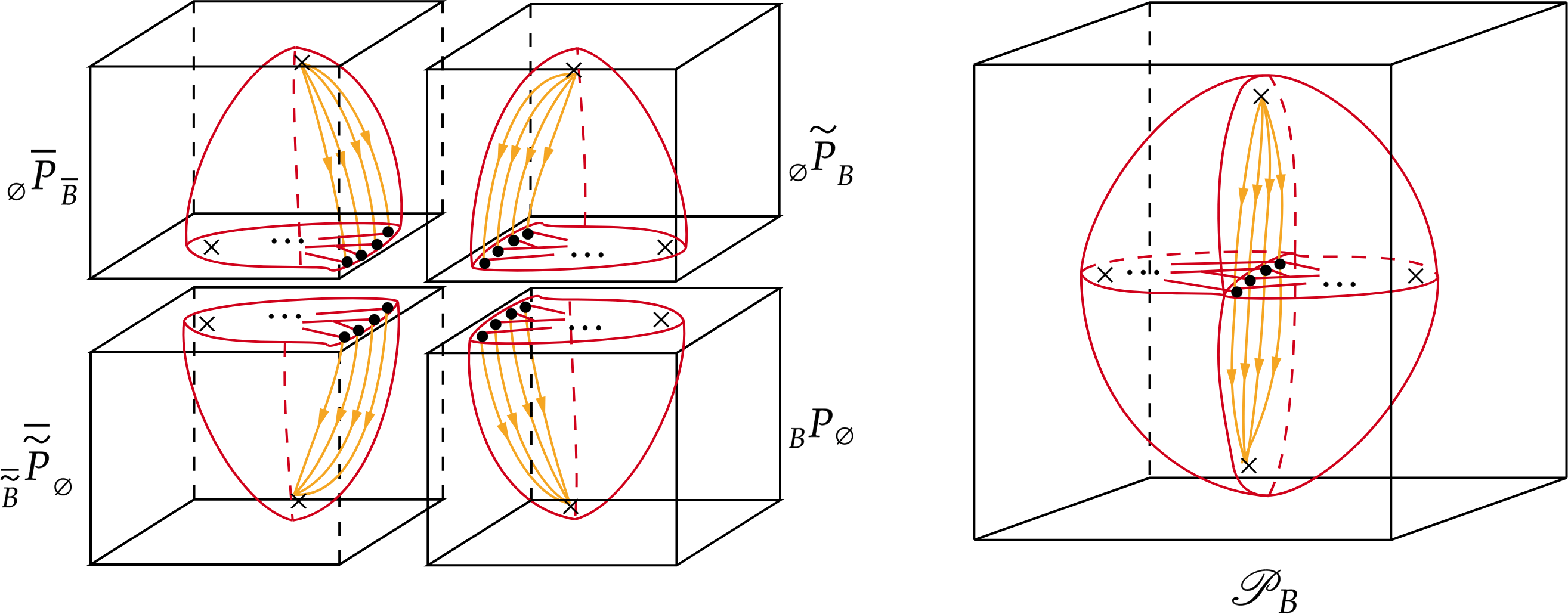}
    \caption{The "embellished" closed 2-ribbon $\mathscr{P}_B$ obtained from the construction. The trivial unframed anchors are marked with the symbol "$\times$".}
    \label{fig:4-ball}
\end{figure}

Let us then denote by $(\bar{\tilde\omega}_P,\bar\omega_P)\in \widehat{\C}_q(\mathbb{G}^{\bar{\tilde P}})^{\text{c-op}_v,\text{c-op}_h}\times_{\bar B}\widehat{\C}_q(\mathbb{G}^{\bar P})^{\text{c-op}_h}$ the framing pairing state \eqref{framingpairing} under the action of $-^{*_1}$. An argument analogous to the above then gives a Wilson surface state 
\begin{equation*}
    \bar\Omega_P=\Omega_{\bar{\tilde P}\cup_{\bar B } \bar P } = \bar{\tilde\omega}_P\circ\bar\omega_P\in \widehat{\C}_q\big(\mathbb{G}^{\bar{\tilde P}\cup_{\bar B} \bar P}\big).
\end{equation*}
By monoidality of $\mathbb{G}$-decorated PL 2-ribbons, we are then able to form the {monoidal product}
\begin{equation}
    \mathscr{O}_{P;B}=\bar\Omega_P~\hat\otimes_H~ \Omega_P \in \widehat{\C}_q(\mathbb{G}^{\mathscr{P}_B})\label{alterfoldstate}
\end{equation}
between these two Wilson surfaces. This distinguished state \eqref{alterfoldstate} has some interesting properties, which we will briefly mention in \S \ref{conclusion}.

We finally come to the main definition of this section. 
\begin{definition}
    We say the $\mathbb{G}$-decorated PL 2-ribbons $\operatorname{PLRib}'^{\mathbb{G}}_{(1+1)+\epsilon}(D^4)$ satisfy \textbf{reflection-positivity} iff for each marked PL 2-ribbon $\,_BP_\emptyset\in \operatorname{PLRib}'_{(1+1)+\epsilon}(D^4)$, the bigraded total Chern $q$-polynomial $c_\mathscr{O}=[\mathscr{O}_{P;B}]\in  H_{\mathbb{G}^B}(B(\mathsf{H}\rtimes G)^{\mathscr{P}_B},\bbZ)[t][q,q^{-1}]$ defined in \textbf{Proposition \ref{bigradedQ}} whose Chern number 
    \begin{equation*}
        \operatorname{ch}_\mathscr{O} = \int_{\big[(\mathsf{H}\rtimes G)^{\mathscr{P}_B}\big]} c_\mathscr{O}\in \bbZ[q,q^{-1}]
    \end{equation*}
    is a positive $q$-polynomial; namely $\operatorname{ch}_\mathscr{O}$ only has positive coefficients.
\end{definition}
\noindent Note the 1-holonomy degrees-of-freedom on $G$ is kept, since the boundary graph $B$ is kept fixed. 

\begin{rmk}
    Neglecting the $q$-grading in $\operatorname{ch}_\mathscr{O}$ for the moment, the positivity means that the Chern classes $c_{\mathscr{O},r}$ can be represented by positive real $(r,r)$-forms on $\mathbb{G}^{\mathscr{P}_B}$ for all $r\leq\operatorname{rk}\mathscr{O}_{P;B}$. Such conditions can in fact determine the geometry of $\mathbb{G}^{\mathscr{P}_B}$: for instance, the positivity of the first Chern class of a $\bbC$-line bundle $L\rightarrow X$ means that $c_1(L)$ can be represented by a K{\"a}hler form, making $X$ into a K{\"a}hler manifold; see \cite{book-symplectic}.
\end{rmk}

If we glue a 3-disc onto $\mathscr{P}_B$, then the embedded graph $B$ (or rather $\bar B\vee B$) keeps track of a \textit{separating surface} $M$ in a 3-manifold $\Sigma$ for whom $\mathscr{P}_B$ is its type-0 partition. Incidentally, these separating surfaces are crucial ingredients for the construction of the so-called \textbf{alterfold TQFTs} \cite{Liu:2023dhj}; we will say a bit more in regards to this connection in \S \ref{conclusion}.

\section{Stably equivalent $\mathbb{G}$-decorated 2-ribbons: $\operatorname{PLRib}^{\mathbb{G};q}_{(1+1)+\epsilon}(D^4)$}\label{stableequivG2ribbons}
Recall if a 3-manifold $\Sigma$ admits $P$ as a simple type-(0) partition, then $M\setminus P\cong D^3$ is a PL 3-disc. By performing a PL homeomorphism which "shrinks" this 3-disc to be small enough, the 3-manifold $\Sigma$ can be submersed into the slab $D^3\times[0,1]$, provided the original 2d polyhedron $P$ is already embedded into the slab. 

Conversely, given a 2d polyhedron $P$, we can obtain a 3-manifold $\Sigma$ by "filling in" $P$ by gluing a genus-0 3-handle $D^3$ along $\partial D^3 \xrightarrow{\sim}P$. As for the boundary of the simple polyhedron $P$, we first perform a PL homeomorphism that makes $P$ intersect the boundary slabs $D^3\times\{0,1\}$ transversally (see Thm. 2.32 in \cite{Liu:2024qth}) at the graphs $B_0,B_1$. This transversal intersection grants us an $\epsilon$-small collar $B_0\times[0,\epsilon]$ above $B_0$, say. Gluing in a PL 3-disc $D^3\simeq D^2\times[0,1]$ onto $P$ then looks, around this $\epsilon$-collar, like filling $B_0\times[0,\epsilon]$ with a PL 2-cylinder $D^2\times[0,\epsilon]$ along a PL homeomorphism $\partial D^2 \times[0,\epsilon] \cong B_0\times[0,\epsilon]$. 

If $B_0$ itself is closed, then filling in a 2-handle like this nets us a compact oriented Riemann surface $M_0$; see Def. 11 of \cite{Alekseev:1994au}. For instance, if $B_0 \simeq S^1\vee S^1$, then filling in a 2-disc gives the 2-torus $M_0\simeq \mathbb{T}^2$ (see \S \ref{boundaryCS}).  Similar argument applies to the "target" graph $B_1$. 

Thus this describes a way in which we can assign a 3-dimensional bordism $\Sigma:M_0\rightarrow M_1$ to a PL 2-ribbon configuration $\_{B_0}P_{B_1}$ by filling in 3-handles. Moreover, this 3-dimensional bordism can be smoothly embedded into the 4-disc $D^4$.

\subsection{Stable equivalence of partitions}
A central result in 2-dimensional topology is that compact oriented Riemann surfaces $M$ are determined up to homeomorphism by filling its standard graph $B$ with a 2-handle \cite{Alekseev:1994au,matveev2007algorithmic}. As such, the boundary configurations $M_0,M_1$ can be determined completely by the boundary graphs $B_0,B_1$.

But what about the bulk? Given a compact oriented 3-manifold $\Sigma$ whose boundary components $\partial \Sigma = M_0\coprod\bar M_1$ determine the standard graphs $B_0,B_1$ uniquely up to PL homeomorphism, we can find \textit{a} type-(0) simple partition $P$ of $\Sigma$ such that $\,_{B_0}P_{B_1}\in\operatorname{PLRib}_{2+\epsilon}'(D^4)$ is a PL 2-ribbon configuration. 

However, the problem is that $P$ may not be unique.
\begin{definition}
    We say two 2d partitions $P\sim P'$ associated to type-0 handlebody decompositions of a 3-manifold $\Sigma$ are equivalent iff they differ by an ambient isotopy in $\Sigma$. 
\end{definition}
\noindent Two equivalent simple partitions of course determine the same 3-manifold up to homotopy, but the problem is that a 3-manifold $\Sigma$ may admit various \textit{inequivalent} simple polyhedron partitions.\footnote{Recall \textit{Remark \ref{lengthhandle}} tells us that longer-length handlebody decompositions determine the underlying 3-manifold more accurately. Type-(0) decompositions have length one, so one does not expect 3-manifolds to have unique such partitions.}

How much distinct inequivalent type-(0) simple partitions of a given 3-manifold can differ is characterized by the following stable equivalence result of Thm. 3.5 in \cite{Sakata2022-il}.
\begin{theorem}
    Two handlebody decompositions of type-(0) of a closed connected oriented 3-manifold $\Sigma$ are equivalent $P\sim P'$ up to a finite number of 0-2/2-3 $\mathrm{handlebody ~moves}$ (fig. \ref{fig:handlemoves}).
\end{theorem}
\noindent Therefore, given a 3-manifold, its type-0 partitions are \textit{not} determined uniquely up to ambient isotopy, but instead up to stable equivalence. 

\begin{figure}
    \centering
    \includegraphics[width=1\linewidth]{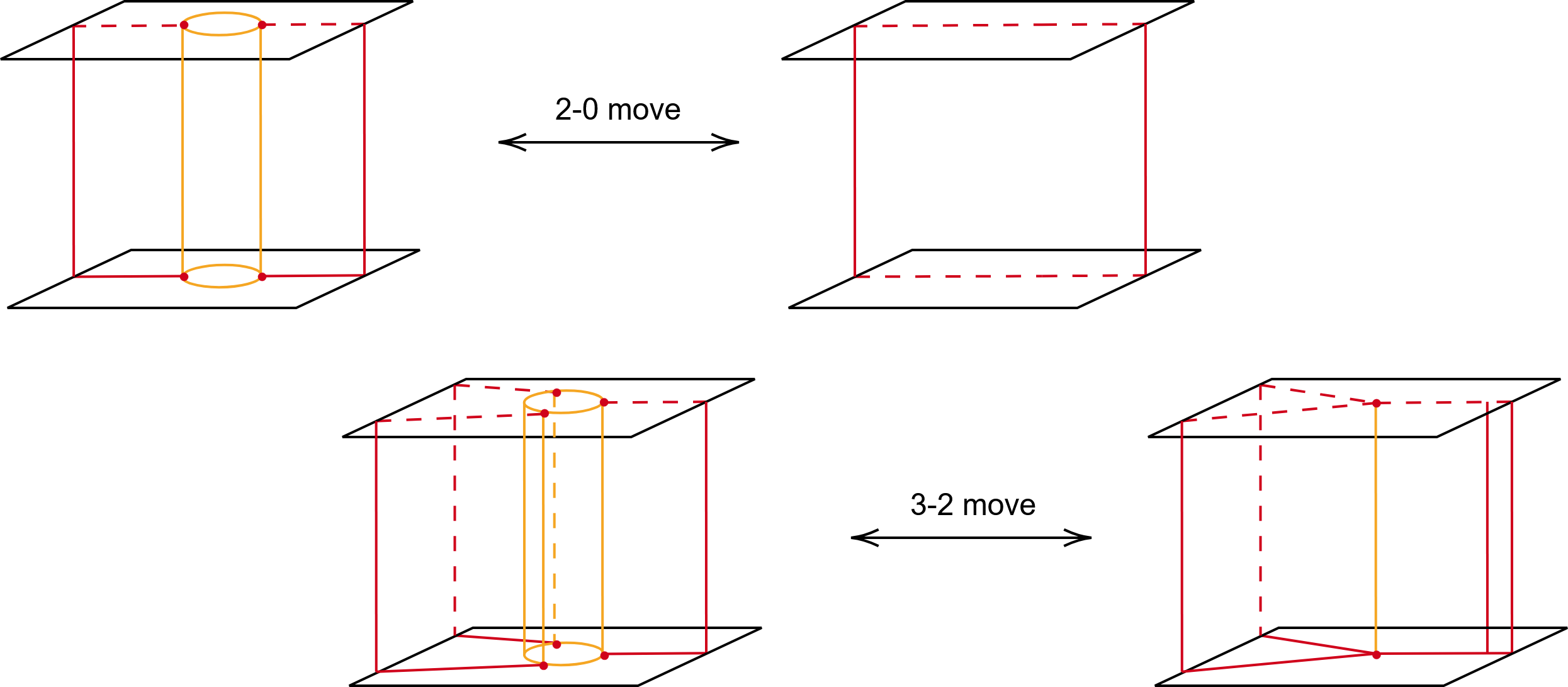}
    \caption{The 0-2/2-0 and 3-2/2-3 handlebody moves.}
    \label{fig:handlemoves}
\end{figure}

\begin{rmk}
    The full statement of stable equivalence in \cite{Sakata2022-il} is that two handlebody decompositions of types-$(g_1,\dots,g_n;P)$ and -$(g_1',\dots,g_n';P')$ of $\Sigma$ are equivalent upon a finite number of applications of handlebody moves of fig. \ref{fig:handlemoves}, \textit{as well as} stabilizations. This "stabilization" operation in essence adds handles to the partition, and hence increases the genera $g_i$. Of course, type-(0) partitions are by definition \textit{unstabilized}  (ie. one that does not come from performing stabilizations), and there has been work previously which classifies whether a given partition of general type is {unstabilized}. The result of Waldhausen \cite{WALDHAUSEN1968195}, for instance, states that any Heegaard splitting of $S^3$ with genus $g$ is stabilized for $g\geq 1.$
\end{rmk}

We must now quotient out the handlebody moves.
\begin{definition}
    The \textbf{stably-equivalent PL 2-ribbons}, $\operatorname{PLRib}_{(1+1)+\epsilon}(D^4)$, is the homotopy quotient $\operatorname{PLRib}_{(1+1)+\epsilon}(D^4)/\sim$, where $\,_{B_0}P_{B_1}\sim \,_{B_0}P'_{B_1}$ iff $P,P'$ are equivalent up to (a finite number of) handlebody moves away from (small $\epsilon$-collars of) the boundaries $B_0,B_1$. Define $$\operatorname{PLRib}_{(1+1)+\epsilon}(D^4)\equiv \bigoplus_n\operatorname{End}_{\mathcal{T}^{PL}_\text{mrk}}(n).$$
\end{definition}
\noindent Note we only perform handlebody moves in the bulk of the 4-disc.

\begin{proposition}\label{3epsilonbord}
    $\operatorname{PLRib}_{(1+1)+\epsilon}(D^4)$ is a monoidal double category equivalent to the category $\operatorname{Bord}^{SO}_{\langle 3,2\rangle+\epsilon}(D^4)$ of $(3+\epsilon)$-dimensional framed oriented bordisms equipped with a submersion into the 4-disc $D^4$, given by filling in a 3-disc.
\end{proposition}

\begin{rmk}\label{corners}
    The statement "filling in a 3-disc" needs more elaboration. In general, there are two ways to paste a handle to a smooth manifold smoothly: (i) a pair of small collars/tubular half-neighbourhoods \textit{with trivial normal bundles} around the attaching sites are chosen, then they are smoothly identified, or (ii) the handle boundary is attached directly, then the resulting manifold with corners are smoothed out. Details of the first construction can be found in \cite{kosinski2013differential}. In the second case, subtleties can arise since the smoothing of the corners is \textit{data}, which makes keeping track of $\operatorname{Bord}^{SO}_{\langle 3,2\rangle+\epsilon}(D^4)$ slightly tedious. As such, we shall take the first approach implicitly in the following. 
\end{rmk}

\subsection{Invariance under stable equivalence}\label{stableinvariance}
In this penultimate section of this paper, we shall prove the following central result. Recall the $\mathbb{G}$-decorated marked PL ribbons in \textbf{Definition \ref{markedhigherGskein}}.
\begin{theorem}\label{handlebodyinvariance}
    Each additive monoidal internal functor $\Omega: \operatorname{PLRib}_{(1+1)+\epsilon}'(D^4)\rightarrow \widehat{\C}_q(\mathbb{G})$ descends to $\operatorname{PLRib}_{(1+1)+\epsilon}(D^4)$. The \textbf{quantum 2-Chern-Simons 2-ribbon invariant} on the 4-disc $D^4$ is therefore defined as $$2\mathcal{CS}^\mathbb{G}_{q}(D^4)\equiv \operatorname{Fun}\big(\operatorname{PLRib}_{1+1}(D^4),[\widehat{\C}_q(\mathbb{G})]\big).$$
\end{theorem}
\begin{proof}
Since we have an equivalence $\widehat{\C}_q(\mathbb{G})\simeq\C_q(\mathbb{G})$ of measureable categories thanks to the Yoneda embedding, we will work directly with the 2-graph states in the following.

\begin{lemma}
    All PL 2-ribbons involved in the following need not have boundary components.
    \begin{itemize}
        \item Let $P,P'$ be connected summable PL 2-ribbons with two summation collars given by framing-reversing homotopies $H,H': \ell_j^-\Rightarrow\ell_k'^+$, then a 0-2 handlebody move is equivalent to the PL isomorphism $H'\ast H^{-1}=\id_{\ell_j^-}$.
        \item Let $P_1,P_2,P_3$ be pairwise connected summable PL 2-ribbons, and let $H_{12},H_{23},H_{13}$ be the associated summation collars. Then a 2-3 handlebody move is equivalent to the PL isomorphism $H^{-1}_{13}\ast H_{23}\ast H_{12} = \id_{\ell_j^-}.$
    \end{itemize}
\end{lemma}
\begin{proof}
    By $H_1^{-1}\ast H_2$, we mean the gluing $H_1^{\dagger_1}\cup_L H_2$ of the orientation-reversal of $H_1$ with $H_2$ along a  PL homeomorphism of their boundaries $L=\ell^-\coprod \ell'^+$ .
    
    The statement follows directly from the geometry; see fig. \ref{fig:handlecollar}. Away from (collars of) the boundary slices, the restriction of $H'\ast H^{-1}=\id$ to a neighbourhood in the interior is exactly a 2-0 handlebody move. Similarly, the equation $H^{-1}_{13}\ast H_{23}\ast H_{12} =\id$ gives rise to a 3-2 handlebody move.
\end{proof}

\begin{figure}[h]
    \centering
    \includegraphics[width=1\linewidth]{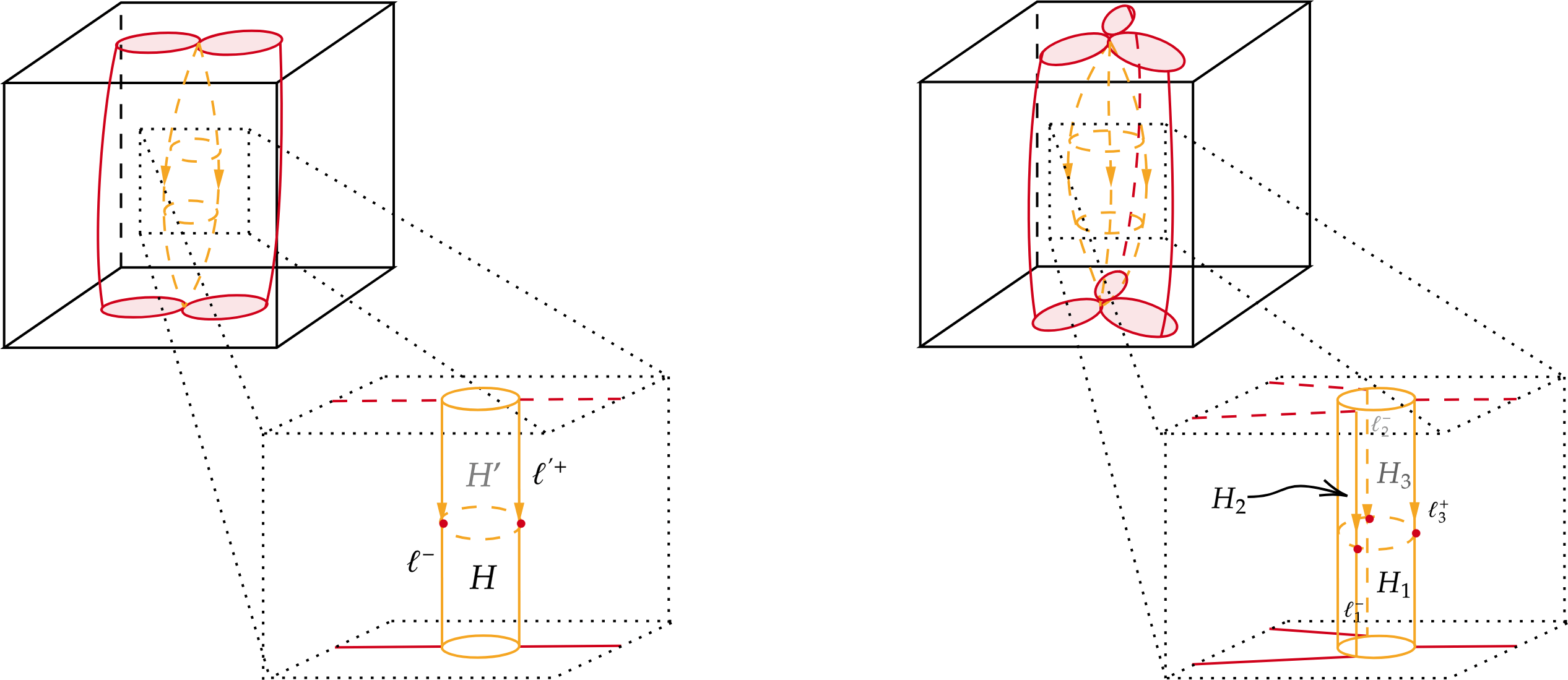}
    \caption{Configurations which relate the handlebody moves to homotopies between the summation collars.}
    \label{fig:handlecollar}
\end{figure}

The invariance under stable equivalence then follows provided the connected summation on $\mathbb{G}$-decorated 2-ribbon graphs do not depend on the summation collar $H$ up to homotopy. In other words, we have the diagram
\[\begin{tikzcd}
	{\widehat{\C}_q(\mathbb{G}^P)\times_H\widehat{\C}_q(\mathbb{G}^{P'})} && {\widehat{\C}_q(\mathbb{G}^P)\times_{H'}\widehat{\C}_q(\mathbb{G}^{P'})} \\
	{\widehat{\C}_q(\mathbb{G}^{P\#_HP'})} && {\widehat{\C}_q(\mathbb{G}^{P\#_{H'}P'})}
	\arrow["\simeq", from=1-1, to=1-3]
	\arrow["{\hat\otimes_H}"', from=1-1, to=2-1]
	\arrow["\cong", shorten <=14pt, shorten >=14pt, Rightarrow, from=1-1, to=2-3]
	\arrow["{\hat\otimes_{H'}}", from=1-3, to=2-3]
	\arrow["\simeq"', from=2-1, to=2-3]
\end{tikzcd}\]
From the formula \eqref{summationcollar} for the monoidal structure $\hat\otimes_H$, it is clear that it suffices to exhibit the homotopy commutative diagram
\begin{equation}\begin{tikzcd}
	{\C_q((\mathsf{H}\rtimes G)^H)} & {\C_q((\mathsf{H}\rtimes G)^{H'})} \\
	{\mathsf{Hilb}} & {\mathsf{Hilb}}
	\arrow["\simeq", from=1-1, to=1-2]
	\arrow["{\int_{\mathbb{G}^H}^\oplus d\mu_H(-)}"', from=1-1, to=2-1]
	\arrow["\cong",shorten <=8pt, shorten >=8pt, Rightarrow, from=1-1, to=2-2]
	\arrow["{\int_{\mathbb{G}^{H'}}^\oplus d\mu_{H'}(-)}", from=1-2, to=2-2]
	\arrow["="', from=2-1, to=2-2]
\end{tikzcd}\label{commute}    
\end{equation}
with respect to the direct Haar integral functors. 

\begin{tcolorbox}[breakable]
    \subsubsection*{Disclaimer.} Strictly speaking, we will need to pick a combinatorial triangulation $\Gamma_H,\Gamma_{H'}$ of the collars $H,H'$ for following argument. But due to \textbf{Theorem \ref{invariance}}, this choice does not matter up to equivalence, so for the sake of clarity we will work directly with $H,H'$.
\end{tcolorbox}

\begin{lemma}\label{2flats}
    If $H,H'$ are two homotopic summation collars, ie. they bound a contractible 3-cell in $D^2\times[0,1]^2\subset D^4$, then \eqref{commute} commutes.
\end{lemma}
\begin{proof}
We leverage the underlying geometry to extract the following two ingredients. 
\begin{enumerate}
    \item Recall from \textbf{Definition \ref{summable}} that $H,H'$ must be oriented and framed in the same way. Let $L=\ell^-\coprod \ell'^+$ and denote by $H'^{\dagger_1}\cup_L H\Rightarrow \id_{\ell^-}$ the given PL homotopy. 2-flatness \textbf{Definition \ref{2flat}} then guarantees a 2-gauge transformation $f: \mathbb{G}^{H'}\rightarrow \mathbb{G}^{H}$ on the 2-holonomies, which is a Lie 2-group diffeomorphism. 
    \item 
Let $\hat F:\C_q(\mathbb{G}^{H'^{\dagger_1}\cup_{L} H})\simeq\mathsf{Hilb}$ be the equivalence given to us by \textbf{Proposition \ref{reversestack}}. Holonomy-density $\ostar: \C_q(\mathbb{G}^{\bar H'})^{\text{m-op}}\times_L \C_q(\mathbb{G}^{H})\xrightarrow{\sim}\C_q(\mathbb{G}^{\bar H'\cup_{L} H})$ allows us to view $\hat F: \C_q(\mathbb{G}^{\bar H'})^{\text{m-op}}\times_L \C_q(\mathbb{G}^{H})\rightarrow \mathsf{Hilb}.$ From this, we can then use \textbf{Proposition \ref{cylinderembeddings}} to deduce that $\hat F$ in fact lives in the essential image of the embedding\footnote{$\hat F$ actually comes from the functor \eqref{disjointpair}, in fact, since it just performs a $\ostar$-tensor product on the two given 2-graph states. This is true for any equivalence provided by \textbf{Proposition \ref{reversestack}}.}
\begin{align*}
    \operatorname{Fun}_\mathsf{Meas}^{*,\bullet}(\C_q(\mathbb{G}^{H}),\C_q(\mathbb{G}^{H'})) \rightarrow \operatorname{Fun}^{*,\bullet}_{\mathsf{Meas}}(\C_q(\mathbb{G}^{\bar H'})^{\text{m-op}}\times\C_q(\mathbb{G}^{H}),\mathsf{Hilb}).
\end{align*}
Its preimage gives the equivalence $F: \C_q(\mathbb{G}^{H})\simeq \C_q(\mathbb{G}^{H'})$ which fits on the top row of \eqref{commute}.
\end{enumerate}

We now use $f$ and $F$ to construct a Lie 2-group diffeomorphism $G:\mathbb{G}^{H'}\rightarrow\mathbb{G}^{H}$ such that $\mu_{H}$ is equivalent to the pushforward $\mu_{H'}\circ G^{-1}$. First, using $f$ we induce the direct image functor $f_*: \C_q(\mathbb{G}^{H'})\xrightarrow{\sim} \C_q(\mathbb{G}^{H})$. The composite $F\circ f_*$ is then a measureable automorphism on $ \C_q(\mathbb{G}^{H'})\subset \mathsf{Meas}_{\mathbb{G}^{H'}}$, which by \textbf{Proposition \ref{pullbackmeas}} is measureably naturally isomorphic $G'^* \cong F\circ  f_*$ to the pull-back measureable functor along a Lie 2-group diffeomorphism $G':\mathbb{G}^{H'}\rightarrow\mathbb{G}^{H'}$. 

We put $G=f\circ G': \mathbb{G}^{H'}\rightarrow\mathbb{G}^{H}$ as the requisite Lie 2-group diffeomorphism. The push-forward measure $\mu_{H}' = \mu_{H'}\circ G^{-1}$ is an invariant Haar measure on $\mathbb{G}^{H}$, which by uniqueness \textbf{Proposition \ref{haarunique}} we have an equivalence $\mu_{H}\sim\mu_H'=\mu_{H'}\circ G^{-1}$. \textbf{Theorem \ref{isointegral}} then finally gives us the desired measureable natural isomorphism (in the first line)
\begin{align*}
    \int_{\mathbb{G}^{H}}^\oplus d\mu_{H}(-)&\cong \int_{\mathbb{G}^H}^\oplus d(\mu_{H'}\circ G^{-1})(-) \cong \int^\oplus_{G(\mathbb{G}^{H'})}d\mu_{H'}(-)\\
    & \cong \int^\oplus_{\mathbb{G}^{H'}}d\mu_{H'}(-)\circ \big(f\circ G'\big)^* \cong \int^\oplus_{\mathbb{G}^{H'}}d\mu_{H'}(-)\circ (F\circ f_*\circ f^*) \\ 
    &\Rightarrow \int^\oplus_{\mathbb{G}^{H'}}d\mu_{H'}(-)\circ F
\end{align*}
where we have used the composition associativity in $\mathsf{Meas}$ in the second line, and the adjunction $f^*\dashv f_*$ for coherent sheaves of $C_q(\mathbb{G})$-modules \cite{Fausk:2003,Kashiwara1990SheavesOM} in the last line.
\end{proof}
To treat the case with three summation collars $H_{12},H_{23},H_{13}$, we can simply pick $H'=H_{13}, H=H_{12}\cup_{L_2} H_{23}$ and apply the above result.
\end{proof}

    For \textit{weak 2-Chern-Simons} 2-ribbon invariants $2\mathcal{CS}^{\mathbb{G};\tau}_q(D^4)$, it can be seen from the above proof that the non-trivial associator $\tau$ contributes directly to an anomaly in the 3-2 handlebody move. On the other hand, the 1-2 handlebody move instead receives anomaly contribution from weak unitors of $\mathbb{G}$, which we do not typically enter into the data of the 2-holonomies.


\subsection{Connected summation with corners}\label{alterfold}
By combining the above main theorem and \textbf{Proposition \ref{2ribboninvariant}}, the 2-Chern-Simons 2-ribbon invariants are parameterized as a set by the invariant subset of the Chern $q$-polynomials $$H_{\mathbb{G}^B}(B\mathbb{G}^{P},\bbZ)[t][q,q^{-1}], \qquad  \,_BP_\emptyset\in\operatorname{PLRib}_{2+1}(D^4)$$ living on PL homeomorphism classes of PL 2-ribbons.

Now in accordance with \textbf{Proposition \ref{3epsilonbord}}, these 2-ribbon invariants should extend to invariants of framed oriented $(2+1)+\epsilon$-dimensional bordisms $\operatorname{Bod}^{SO}_{\langle 3,2\rangle+\epsilon}(D^4)$ via the handlebody decomposition. This then begs the question: what is the monoidal structure on $3+\epsilon$ bordisms induced from PL connected summation $\#$? 

For PL 2-ribbons without boundary graphs, this is simple: the idea is to interpret a summation collars $H$ as the \textit{core} of an attaching handle $\mathring{H} = S^2\times [0,1]$ associated to the usual \textit{interior} connected summation 
\begin{equation*}
    \Sigma_1 \# \Sigma_2 = (\Sigma_1\setminus D^3)\cup_{S^2}(\Sigma_2\setminus D^3),\qquad \partial \mathring{H} =S^2 \times S^0,
\end{equation*}
where $S^2$ is the sphere boundary $\partial D^3\simeq S^2$ of open 3-discs $D^3$ in the interior of the 3-manifolds $\Sigma_1,\Sigma_2$. Note that all notion of "attaching" is in the sense mentioned in \textit{Remark \ref{corners}}.

In the presence of boundary, we turn to the following notion from \cite{kosinski2013differential}.
\begin{definition}
    Let $\Sigma_1,\Sigma_2$ be smooth $n$-manifolds with connected boundary. The \textbf{boundary connected sum} $\Sigma_1\#_{\partial} \Sigma_2$ is the gluing $\Sigma_1\cup_f \Sigma_2$ along a diffeomorphism $f:D^{n-1}\rightarrow D'^{n-1}$ of (tame) $(n-1)$-discs $D^{n-1}\subset \partial \Sigma_1,~D'^{n-1}\subset \partial \Sigma_2$.
\end{definition}
\noindent Notice that, in contrast to ordinary interior connected summation, the \textit{entire} tame 2-discs are identified, not just its boundary. The idea is then that the anchors on a PL 2-ribbon are interpreted as the core of this 2-disc. 

The PL connected summation operation $\#_H$ can therefore be interpreted as a "combination" of both an interior connected sum and a boundary connected sum. Indeed, since the attaching handle $H^2$ whose core is given by the summation collar $H$ \textit{must} meet the boundary of the 3-manifold by construction, this meeting generates \textit{corners} upon connected summation. The prototypical form of a connected attaching handle in the interior is the cylinder $\mathring H= S_+^2\times [0,1]$ on a hemisphere $S_+^2\cong D^2\subset S^2$, whose corner is given by two (oppositely-framed) circles $S^1\times S^0$. See the top left corner of fig. \ref{fig:cornersummation}.

The more precise definition is the following, as inspired by "connected summations with corners" described in \S 2.1 of \cite{sarkar2018cohomologyringsclasstorus} and the "end summation" operation of Gompf \cite{Gompf:1983,Bennett_2016}.
\begin{definition}
    Let $\Sigma$ denote a 3-manifold with boundary $M$. An immersed 3-disc $D^3$ is called \textbf{partially embedded} iff
    \begin{itemize}
        \item it intersects the boundary $M$ at a 2-disc $D^3\cap M\cong S^2_-\cong D^2$, and
        \item there exists an $\epsilon$-collar $k_\epsilon$ of the boundary away from which the remaining portion $\tilde D^3$ of $D^3$ embeds into the interior $\operatorname{int}\Sigma$ of $\Sigma$.
    \end{itemize}
    The \textbf{corner connected summation} $\Sigma \#_{\mathring{H}}\Sigma'$ between two such 3-manifolds $\Sigma,\Sigma'$ with partially embedded 3-discs $D^3,D'^3$ is the result of gluing an attaching half-cylinder $\mathring{H}\cong S_+^2\times[0,1]$ (the summation collar), subject to the following conditions:
    \begin{enumerate}
        \item away from the $\epsilon$-collars $k_\epsilon,k'_\epsilon$, we have a diffeomorphism $f: \partial\mathring{H}\xrightarrow{\sim}\partial (\operatorname{int}\Sigma\setminus \tilde D^3) \coprod \partial (\operatorname{int}\Sigma \setminus \tilde D'^3)$, 
        \item on the boundary, we have a diffeomorphism $f_\partial: D^3\cap M \xrightarrow{\sim} D'^3\cap M'$, and finally,
        \item on the $\epsilon$-collars, we have a smooth interpolation from $f_\epsilon$ to $f_\partial$ around the corners of $\mathring{H}$.
    \end{enumerate}
\end{definition}
\noindent An illustration of this procedure is given in fig. \ref{fig:cornersummation}.

\begin{figure}[h]
    \centering
    \includegraphics[width=0.85\linewidth]{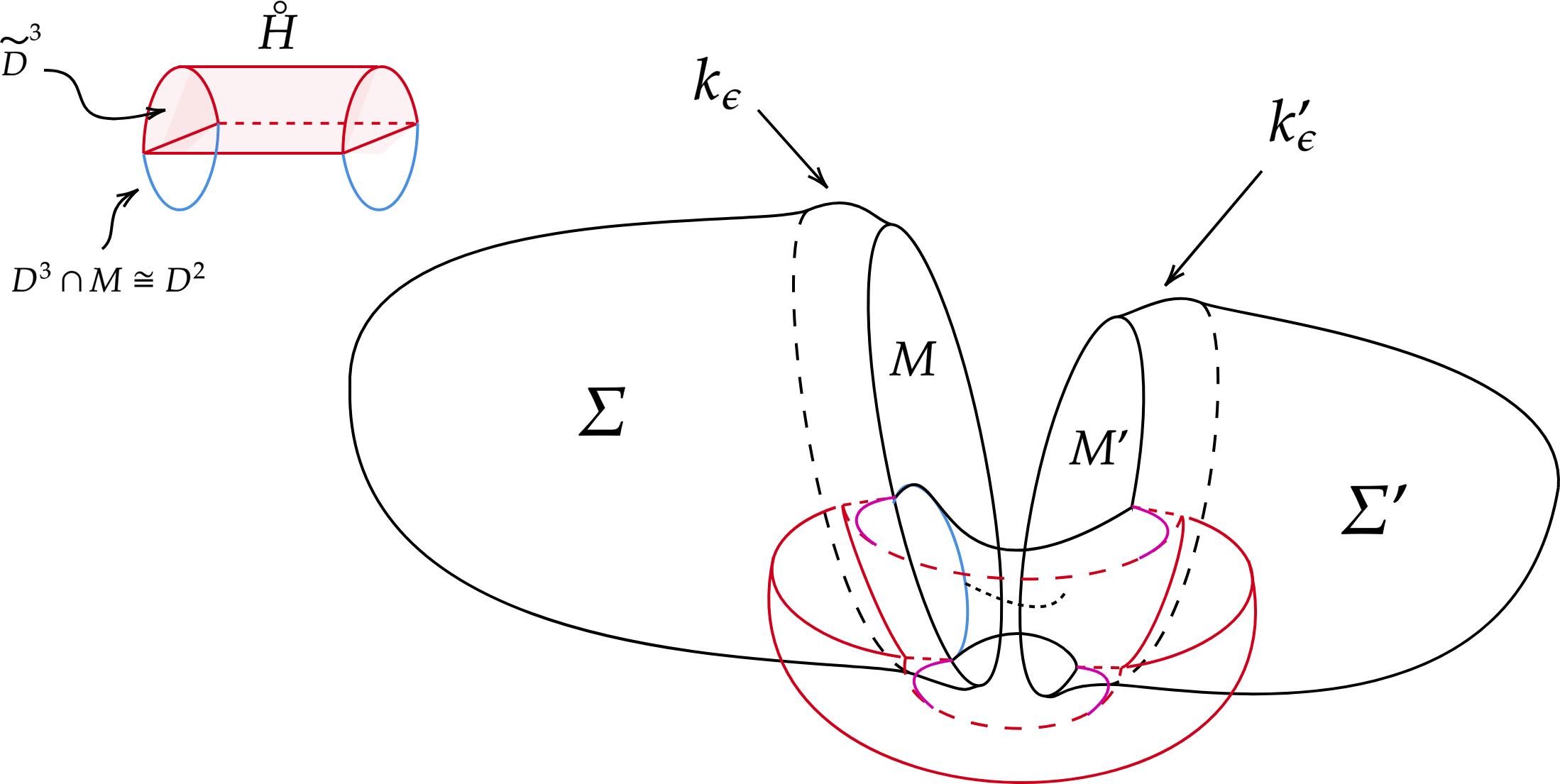}
    \caption{A demonstration of the corner connected summation operation on $\Sigma,\Sigma'$. The summation collar $\mathring{H}$ is colour-coded as red, while the boundary portions of the 3-disc $D^2\cong D^3\cap M$ are blue. Within the $\epsilon$-collars $k_\epsilon,k_{\epsilon}'$, the handle attachment map $f_\epsilon$ is smoothly interpolated into the boundary gluing map $f_\partial$ of the 2-disc; this is colour-coded in purple.}
    \label{fig:cornersummation}
\end{figure}

The composition of $\Sigma$ as bordisms in $\operatorname{Bord}_{\langle3,2\rangle+\epsilon}^{SO}(D^4)$ are once again given by stacking, but with the additional condition that there must be a diffeomorphism 
\begin{equation*}
    k_\epsilon \cup_M k_{\epsilon}' \cong M\times[0,2\epsilon]
\end{equation*}
between the $\epsilon$-collars of $\Sigma,\Sigma'$ around the middle 2-manifold $M$ and the cylinder on $M$. Moreover, the partially embedded 3-discs should become a genuinely embedded 3-disc in the bulk $\Sigma\cup\Sigma'$. This reflects the "stackability condition" for PL 2-ribbons described in \textit{Remark \ref{skeinint}}.

\begin{rmk}\label{freedman}
    It is interesting to observe the close relationship between the special handles with corners $\mathring{H}$ described in \S \ref{alterfold} and the Casson handles in $M^4$ \cite{Casson:1986}. This may allow one to perform Freedman's exotic 4-manifold surgery \cite{Freedman:1982} with 2-Chern-Simons 2-ribbon invariants $2\mathcal{CS}_q^\mathbb{G}(M^4)$. We will study this more explicitly in a future work down the line.
\end{rmk}

\section{Conclusion}\label{conclusion}
In this paper, we have constructed the 2-ribbon invariant $2\mathcal{CS}_q^\mathbb{G}(D^4)$ in a 4-disc of 2-Chern-Simons theory. This is a crucial towards the definition of the \textit{2-Chern-Simons TQFT}, with the ultimate goal of performing 4-manifold handlebody surgery on $M^4$ with them. For this, the 2-ribbon invariants $2\mathcal{CS}^\mathbb{G}_q(D^4)$ must of course first be extended to arbitrary 4-manifolds $M^4$.

In analogy with the Witten-Turaev-Reshetikhin TQFT in 3-dimensions \cite{Reshetikhin:1991tc,Turaev:1992,WITTEN1990285}, this presents a series of challenges that one must overcome.  Aside from extracting the higher-skein relations --- which we shall mention in \S \ref{higherskein} --- these include:
\begin{itemize}
    \item What is the notion of "2-sphericality" for the 2-ribbon invariants on $M^4=S^4$?
    \item What is the quantization condition for 2-Chern-Simons theory?
    \item How do we \textit{actually} compute $2\mathcal{CS}^\mathbb{G}_q(D^4)$?
\end{itemize}
These are actually the same question: 
\begin{center}
    {\large{\emph{What is the 2-representation theory for (vertical isomorphism classes of) $\mathbb{U}_q\G$?}}}
\end{center}
Indeed, in the usual skein theory {\`a} la Witten-Reshetikhin-Turaev, sphericality requires a notion of \textit{quantum dimension}, which is what allows us to compute knot polynomials/Kauffman bracket from irreducible representations of, for instance, $U_q\frak{sl}_2$. Moreover, positivity of the quantum dimension immediately implies the Chern-Simons level-quantization $q\in\mu_\infty$ \cite{Alekseev:1994au}.

\medskip

Toward this, there has been some discussions in the literature about what "higher-dimensional sphericality" and "2-categorical dimension" means one level up \cite{Mackaay:hc,Douglas:2018,Liu:2024qth}. Further, a definition of the 2-categorical \textit{quantum} dimension was given in \cite{Chen:2025?},
\begin{equation*}
    \mathfrak{Dim}_q(\cD): 1_\cD\Rightarrow 1_\cD,\qquad \cD\in \operatorname{2Rep}(\mathbb{U}_q\G),
\end{equation*}
which was shown to bypass the difficulty (\textit{Warning 2.5} of \cite{Douglas:2018}) suffered by the strict-pivotal setting.

\medskip

In a companion work, we will dive deeper into the categorical representation/character theory of $\mathbb{U}_q\G$ and make \textit{Remark \ref{modelchange}} precise. Based on its structures as a {$\mathsf{Meas}$-internal Hopf category}, we will tackle the aforementioned issues of 2-categorical "quantum dimensions/quantum 2-traces". This servers, together with smooth 4-manifold theory (cf. \textit{Remark \ref{freedman}}), as the foundation for the \textbf{4d 2-Chern-Simons TQFT}.

\medskip

We mention some more interesting aspects of the 2-Chern-Simons TQFT in the following.

\subsubsection*{Gapped and gapless boundaries of the 2-Chern-Simons TQFT.} We will show in \S \ref{boundaryCS} that the 3d Chern-Simons degrees-of-freedom can be extracted as the "degree-0 part" of its 4d derived counterpart. However, we note here that this is \textit{not} a form of "{transgression}" --- the latter is well-known to govern the Chern-Simons/Wess-Zumino-Witten holography \cite{Willerton:2008gyk,Carey_1997,Waldorf:2012,Waldorf2015TransgressiveLG}.

The works \cite{Chen:2024axr,Chen:2023integrable} suggest that transgressing the 2-Chern-Simons theory leads to a gapless 3d topological-holomorphic field theory that hosts \textit{derived} current algebras (cf. \cite{FAONTE2019389,Kapranov2021InfinitedimensionalL,Alfonsi:2024qdr}). This means that, at the level of TQFTs, there are two different types of boundaries for 2-Chern-Simons theory: the Chern-Simons/Witten-Reshetikhin-Turaev TQFT (which is gapped) and a topological-holomorphic field theory of "affine raviolo" type \cite{Garner:2023zqn,alfonsi2024raviolo} (which is gapless). 

An upcoming work by the author will describe this "affine raviolo Kac-Moody VOA" in more detail. 
\[\begin{tikzcd}
	{\text{2-Chern-Simons TQFT}} && \begin{array}{c} \substack{\text{Chern-Simons}\\ \text{Witten-Reshetikhin-Turaev}} \text{ TQFT} \end{array} \\
	\begin{array}{c} 3d ~\substack{\text{derived Kac-Moody} \\ \text{affine raviolo}}\text{ VOA} \end{array} && {2d~\text{affine Kac-Moody VOA} }
	\arrow["{\text{deg-0}}", from=1-1, to=1-3]
	\arrow["{\text{“2-transgression”}}"', dashed, from=1-1, to=2-1]
	\arrow["{\text{transgression}}", from=1-3, to=2-3]
	\arrow["{\text{deg-0?}}", dashed, from=2-1, to=2-3]
\end{tikzcd}\]
This presents a very interesting 4d/3d example of the topological bulk-boundary correspondence as described in, for instance, \cite{Kong:2022hjj,Kong2024-vr,Wen:2019}. 

\subsubsection*{Alterfolds with corners.} Recall the closed PL 2-ribbon $\mathscr{P}_B$ constructed in \S \ref{reflectionpositivity}. By pasting a 3d genus-0 3-handle onto $\mathscr{P}_B$, we obtain a stratified 3-manifold $M^3=M^3_\mathscr{P}$ for whom the associated distinguished Wilson surface state $\mathscr{O}_{P;B}\in\widehat{\C}_q(\mathbb{G}^{\mathscr{P}_B})$ \eqref{alterfoldstate} can be thought of as the decorations on $M^3$ \cite{Liu:2024qth}. 

However, the 3-manifold constructed in this way not only has a separating surface, but also \textit{corners} given by the marked anchors of the PL 2-ribbon $\mathscr{P}_B$. If we view $\mathscr{P}_B: \emptyset\Rightarrow \bar B\vee B\Rightarrow \emptyset$ is a split higher-idempotent (or better yet, a \textit{condensation higher-monad} \cite{Gaiotto:2019xmp,Johnson-Freyd:2020usu,decoppet2022morita,decoppet2022morita,Douglas:2018}), then it can be shown (more details will appear in a future work) that $\mathscr{O}_{P;B}$ determines a \textit{von Neumann $D^3$-algebra} $A_{P}\subset \cB(H_B)$ on some (separable, possibly infinite-dimensional) Hilbert space $H_B\in \mathsf{Hilb}\simeq \widehat{\C}_q(G^\emptyset)$. 

The functional integral construction \cite{Liu:2024qth} then gives us a \textit{3d alterfold TQFT} $Z_A$, whose value on $M^3=M^3_\mathscr{P}$ is given by a non-degenerate positive tracial state $\operatorname{tr}_{H_B}:A_{P} \rightarrow \R_{\geq 0}$. Such tracial states present an interesting challenge: its existence \textit{must}, in general, combine techniques from operator algebras \cite{book-operators,Takesaki1979} and the theory of modified traces \cite{geer:hal-00603999,Geer2011-cr,Geer_Patureau-Mirand_Turaev_2009}. 


\subsubsection*{Relation to Soergel bimodules.} In view of the results of \S \ref{boundaryCS}, 2-Chern-Simons theory contains a categorification $\C_q(G)$ of the Chern-Simons degrees-of-freedom decorating 1-tangles. In accordance with \textbf{Proposition \ref{bigradedQ}}, it determines a bigraded ring $H_G(BG,\bbZ)[t][q,q^{-1}]$ localized at the graph $B$. Due to \textit{Remark \ref{knothomology}}, one may wonder how this invariant is related to Khovanov-Rozansky homology.

Following \cite{2024arXiv240704891L}, we take $G=U_N$ with its maximal torus $T=U_1^N$, and consider the standard parabolics $G_i = U_1^{i-1}\times U_2\times U_1^{N-i-1}\subset G$ associated to each permutation $s_{i,i+1}$ in the Weyl group. One can extract from the integral cohomology $H^\bullet(BU_N,\bbZ) = H^\bullet(BU_N)$ (or any generalized cohomology $E$ over any $E_\infty$-ring spectrum with a complex orientation) the data of the so-called \textit{Bott-Samelson $H^\bullet(BT)$-$H^\bullet(BT)$ bimodules} $(H\bbZ) B_{i_1,\dots,i_m}^\bullet$, which are closely related to the $U_N$ Soergel bimodules that govern Khovanov-Rozansky homology \cite{webster2013knot,liu2024braided,Elias2010ADT,elias2020introduction}.

Together with the observations made in \S \ref{higherskein}, it may therefore be possible to relate the 2-Chern-Simons TQFT with the lasagna higher-skein modules of \cite{Morrison2019InvariantsO4,Manolescu2022SkeinLM}. 



\newpage

\appendix

\section{Relation to previous works}\label{prevworks}
In this appendix, we organize the relationship between the combinatorial quantization framework developed here with many of the (mostly) recent existing literature.

\subsection{Recovering the Chern-Simons observables}\label{boundaryCS}
The fact that 2-Chern-Simons action can recover Chern-Simons action at the boundary is known semiclassically \cite{Jurco:2018sby,Soncini:2014,Chen:2024axr}. Here, we provide a quantum version of this fact, by recovering the combinatorial framework of \cite{Alekseev:1994pa,Alekseev:1994au}. 


Let  ${\C}_q(G)$ denote the objects part\footnote{Given a (co)category object $C$ internal to a bicategory $\cV$, the functor $\operatorname{Cat}_\cV\rightarrow\cV: C\mapsto C_0$ which extracts the objects $C_0$ of $C$ is the right-adjoint of the \textit{discretization} functor $\cV\rightarrow\operatorname{Cat}_\cV: C_0\mapsto \big(C_0\rightrightarrows C_0\big)$ \cite{Miranda:2025}.} of the quantum categorical coordinate ring $\C_q(\mathbb{G})$. By construction, $\C_q(G)$ serves as the categorification of a quasitriangular Hopf *-algebra  isomorphic to the quantum coordinate ring $C_q(G)$ on $G$.

If the boundary $\partial P=B$ has a single component, then its objects part determines a Hopf cocategory $\C_q(G^B)$ localized on $B$. This object $\C_q(G^B)$ serves as the categorification of the degrees-of-freedom in Chern-Simons theory, in the sense that $\C_q(G^B)$ are given by measureable sheaves of modules over the quasitrigular Hopf algebra $\C_q(G^B)$, which is isomorphic to the one defined in Def. 12 of \cite{Alekseev:1994au}.  It is also not hard to see that the *-operation $-^{*_1}$ descends to the orientation reversal *-operation on $C_q(G)$ as defined in \cite{Alekseev:1994pa}.

Indeed, if $\phi^I_e\in C_q(G^B)$ denotes a basis of localized 1-graph states $e\in\Gamma^1$ such that $\phi_e^{IJ}(\{h_{e'}\}_{e'}) = h^{IJ}_e$ is the $(I,J)$-th entry of $h_e$, then we can see from \S \ref{locality} that the coproduct restricted on $C_q(G^B)$ satisfies
    \begin{equation*}
        (-\cdot-)\big(\Delta_0(\phi^{IJ}_e)\big) =\sum_K\Big( \sum_{e_1\ast e_2=e} \phi_{e_1}^{IK} \phi^{KJ}_{e_2} - \sum_{e_2\ast e_1 =e} \phi^{IJ}_{e_2} \phi^{JK}_{e_1} \Big),
    \end{equation*}
which is precisely the coproduct on the Chern-Simons holonomies \cite{Alekseev:1994pa}. The $R$-matrices $(\bar R_0)_e$ on each edge $e\in B$ can also be checked to be of the same form as eqs. (2.45)-(2.48) in \cite{Alekseev:1994pa}; they govern the cocommutativity of the Wilson lines localized on adjacent edges in $B$.

\subsubsection*{Example: the standard Chern-Simons algebra on the 2-torus}
Let us make the above more precise, with the example of the unpunctured 2-torus $\mathbb{T}^2 = \Sigma_{1,0}$. The \textbf{standard graph} $B_{1,0}$ (see Def. 11 in \cite{Alekseev:1994au}) is a(n oriented) graph with a single 4-valent crossing, homotopically equivalent to the bouquet $S^1\vee S^1$ of two circles based at the crossing vertex $v$.

The first step is to recover $B_{1,0}$ from the marked PL 2-ribbons $\mathcal{T}'^{PL}_\text{mrk}$ in \textbf{Definition \ref{markedPL2skeins}}.
\begin{lemma}\label{edgecontraction}
    The standard graph $B_1$ of the 2-torus $\Sigma_1=  \mathbb{T}^2$ can be recovered from objects in the ribbon 2-algebra $\operatorname{End}_{\mathcal{T}'^{PL}_\text{mrk}}(2)$.
\end{lemma}
\begin{proof}
     We call a connected graph $B\in \operatorname{End}_{\mathcal{T}'^{PL}_\text{mrk}}(n)$ \textit{minimal} when it is indecomposable as a wedge sum of graphs in $\operatorname{End}_{\mathcal{T}'^{PL}_\text{mrk}}(n)$. Setting $n=2$, there are three connected minimal graphs up to ambient PL homeomorphism; they are the identity $1_2$ (two parallel lines) and the two diagrams $B_+,B_\times$ illustrated in fig. \ref{fig:standardgraph}.

\begin{figure}[h]
    \centering
    \includegraphics[width=1\linewidth]{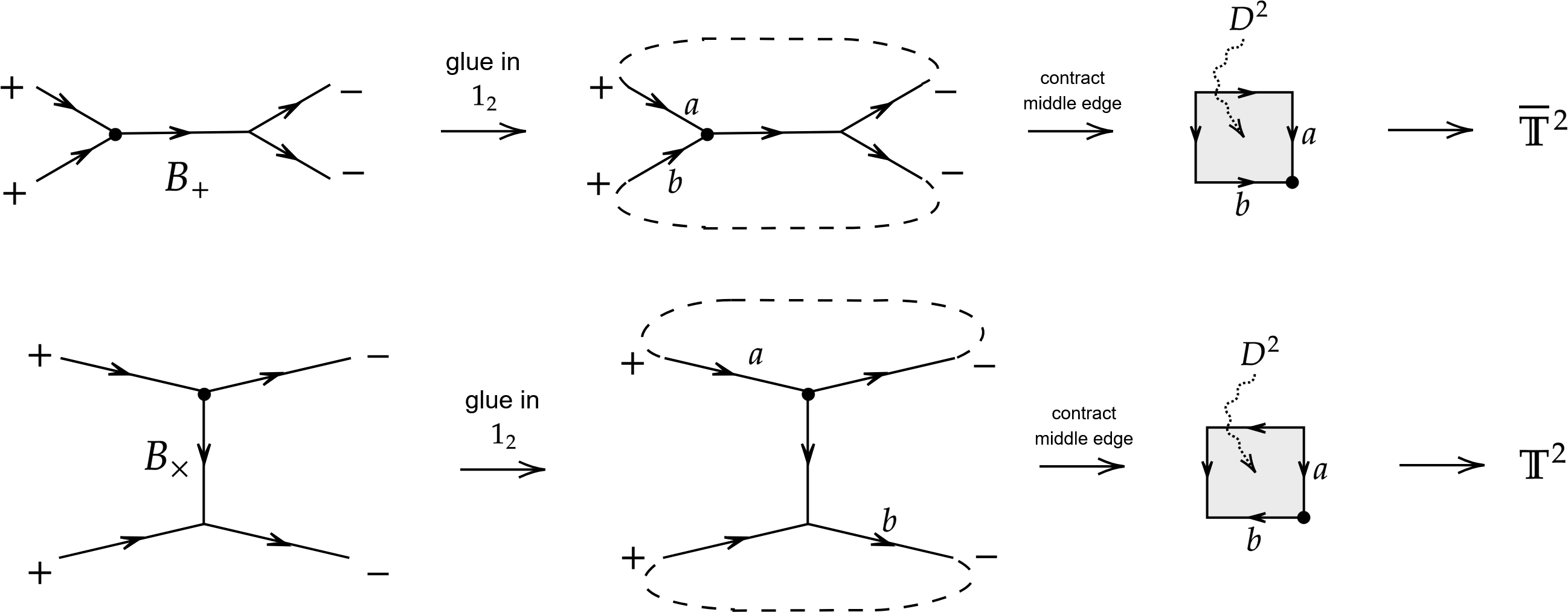}
    \caption{The minimal $s$- and $t$-channel graphs $B_+,B_\times\in \operatorname{End}_{\mathcal{T}'^{PL}_\text{mrk}}(2)$, from which we can obtain the 2-torus $\mathbb{T}^2$ and its orientation reversal $\bar{\mathbb{T}}^2$.}
    \label{fig:standardgraph}
\end{figure}

    We can close off $B_\times$, say, by gluing the identity graph $1_2$ into its incoming and outgoing vertices. The standard graph $B_1$ on $\mathbb{T}^2$, which is a closed 4-valent crossing graph as oriented in fig. 1 of \cite{Alekseev:1994au}, can then be obtained from it by contracting the middle internal edge via a PL homotopy. See the right side of fig. \ref{fig:standardgraph}.
\end{proof}

Now by closing $B_\times$ off as described in \textbf{Lemma \ref{edgecontraction}}, additional $R$-matrix relations governing the locality between the holonomies on the incoming and outgoing edges (see eg. line 4 of Def. 12 in \cite{Alekseev:1994au}) are introduced. The edge contraction result (Prop. 9) in \textit{loc. sit.} then provides the desired isomorphism of $C_q(G^{B})$ with the Chern-Simons standard graph algebra on $B_{1,0}$.

    \begin{rmk}\label{orientationgraph}   
        The standard graph of the 2-torus $\bar{\mathbb{T}}^2$ with the opposite orientation can be obtained by contracting the middle internal edge of $B_+$. This is illustrated in the top row of fig. \ref{fig:standardgraph}. This introduces different locality/braiding relations in $C_q(G^B)$ which produces the Chern-Simons standard graph algebra for the oppositely-oriented 2-torus.
    \end{rmk}


\subsection{Geometry of 2-tangles in 4-dimensions}\label{baezlangford}
The above result, as well as the definition of the PL 2-ribbons in \S \ref{PL2ribbons}, suggests a close relationship between the double bicategory $\mathcal{T}'^{PL}_\text{mrk}$ and the 2-category encoding the geometric/homotopic properties of the 2-tangles in 4-dimensions. 

Let us therefore begin by recalling the following notion \cite{BAEZ2003705}.
\begin{definition}\label{2tangledefinition}
    Consider the following data.
    \begin{enumerate}
        \item \textit{Objects:} these are finite subsets of $D^2$, and are in one-to-one correspondence with the natural numbers $\bbZ_{\geq 0}$,
        \item \textit{1-morphisms:} these are tangles --- namely embedded 1-manifolds $T\subset D^2\times[0,1]$ such that
        \begin{enumerate}
            \item its boundary points $\partial T$ lie in $\operatorname{int}D^2\times\{0,1\}$, and
            \item it has a "product structure": there exists $\epsilon>0$ such that, if $|z-z_0|<\epsilon$ for $z_0=0,1$ and $(x,y,z_0)\in T$, then  $(x,y,z)\in T$.
        \end{enumerate}
        \item \textit{2-morphisms:} these are surfaces with corners --- namely embedded 2-manifolds $S\subset D^2\times[0,1]\times[0,1]$ such that
        \begin{enumerate}
            \item its boundary is embedded in $D^2\times\partial([0,1]^2),$ such that $S\cap \big(D^3\times\{0,1\}\big)$ are a pair of tangles and $S\cap \big(D^2\times\{0,1\}\times[0,1]\big)$ consist of finitely many straight lines.
            \item $S$ has a "product structure near the boundary": there exist $\epsilon>0$ such that (i) if $|z-z'|<\epsilon$ then $(x,y,z,t)\in S \iff (x,y,z',t)\in S$, and (ii) if  $|t-t_0|<\epsilon$ for $t_0=0,1$ and $(x,y,z,t_0)\in S$, then $(x,y,z,t)\in S$.
        \end{enumerate}
        See eg. fig. \ref{fig:2tangleeg}.
    \end{enumerate}
    The \textbf{Baez-Langford 2-category $\mathcal{T}$ of (unframed unoriented) 2-tangles} is the 2-category obtained from the above geometric data up to level-preserving smooth isotopies in $D^4$, with the obvious composition laws for 1- and 2-morphisms (see Lemma 5 of \cite{BAEZ2003705}).
\end{definition}

\begin{figure}[h]
    \centering
    \includegraphics[width=0.55\linewidth]{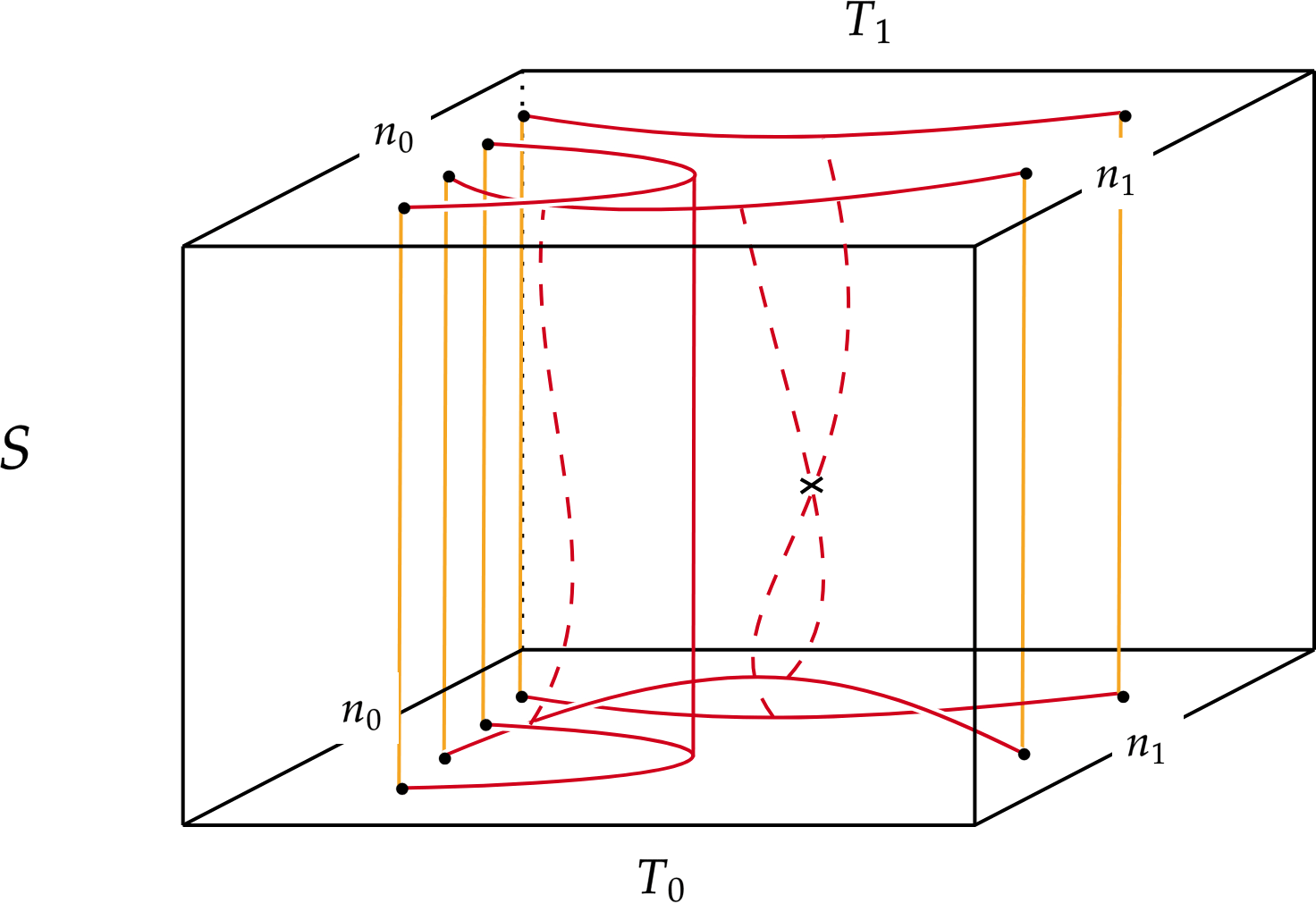}
    \caption{An example of a 2-tangle $S: T_0\Rightarrow T_1$ in $\mathcal{T}$, where $T_0,T_1: n_0\rightarrow n_1$ are embedded tangles.}
    \label{fig:2tangleeg}
\end{figure}

Each ambient isotopy class of the above data (1-/2-morphisms) have a "generic" representative. We define what this means here.
\begin{definition}
    Let $T$ be a tangle and $S$ an embedded surface as above.
    \begin{itemize}
        \item $T\subset D^3$ is called \textbf{generic} iff (i) its projection to the last two coordinates $[0,1]\times[0,1]$ is an embedding except at finitely many separated crossings, (ii) critical points of the Morse height function on $T$ are non-degenerate local extrema and (iii) all crossings and critical points are at different heights.
        \item $S\subset D^4$ is called \textbf{generic} iff its interseciton with the constant $t$-leaves is a generic tangle except at finitely many values of $t\in[0,1]$, at which one of the following "full set of elementary string interactions" \cite{CARTER19971} occur
        \begin{enumerate}
            \item the Reidemester I, II, III moves,
            \item birth/death of an unknotted circle,
            \item a saddle point of the Morse height function $S\rightarrow \R: (x,y,z,t)\mapsto t$,
            \item a "cusp on a fold line",
            \item a "double point crossing on a fold line", and
            \item moves that change the heights of the tangle crossings/extrema.
        \end{enumerate}
    \end{itemize}
    An example of a 2-tangle $S$ exhibiting the Reidemeister II move and a "double point corssing on a fold line", simultaneously, is displayed in fig. \ref{fig:2tangleeg}.
\end{definition}
\noindent The following is then proved in \cite{BAEZ2003705} by arguing with generic representatives in $\mathcal{T}.$
\begin{theorem}
    $\mathcal{T}$ is a "braided monoidal 2-category with duals\footnote{This means that the objects have duals such that the duality-mates of the 1-morphisms coincide with their adjoints. This notion was noted in \cite{Chen:2025?} to be a weak form of the so-called "$SO(3)$-volutive property" for ribbon tensor 2-categories, but it suffices for unframed unoriented 2-tangles.}" equipped with a self-dual generator $Z\in\mathcal{T}$, which is given by a single unframed point $Z\in D^2$ in the cube. 
\end{theorem}
 
Moreover, there is an equivalence $\mathcal{T}\simeq\mathcal{C}$ which describes unframed unoriented 2-tangles in 4-dimensions using a combinatorial description $\mathcal{C}$ studied in \cite{CARTER19971}. It was also conjectured in \cite{BAEZ2003705} that $\mathcal{T}$ should coincide with the "2-category of higher tangles" studied earlier by \cite{Kharlamov:1993}.

\medskip

From the above description, it is clear that PL 2-ribbons $\mathcal{T}^{PL}_\text{mrk}$ up to diffeomorphisms differ from $\mathcal{T}$ by its end-categories; $\mathcal{T}^{PL}_\text{mrk}$ seems to be much more related to $\mathfrak{gl}_N$-webs and foams \cite{Morrison2019InvariantsO4} at first glance. Thus, the goal for us here is to describe a formal procedure that relates the marked PL 2-ribbons to triangulations \cite{Cairns:1961} of the 2-tangles.\footnote{Notice that the "straight lines" in \textbf{Definition \ref{2tangledefinition}} of a 2-tangle $S$ are \textit{precisely} the markings on a PL 2-ribbon $P$.}

\begin{figure}[h]
    \centering
    \includegraphics[width=1\linewidth]{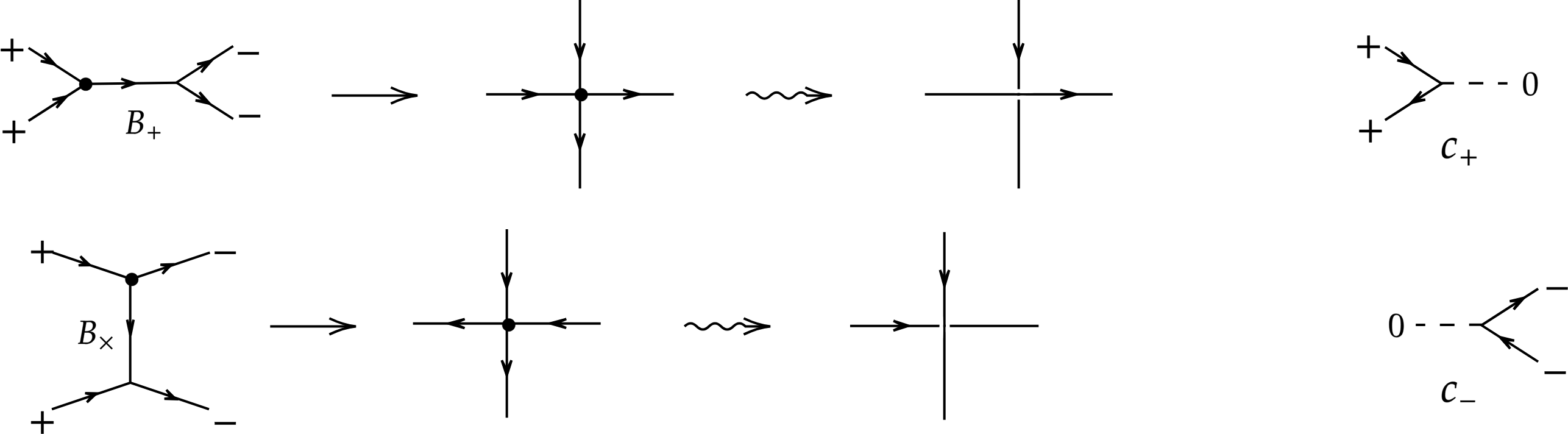}
    \caption{Conventions for interpreting the directed oriented graphs as certain embedded 1-tangles. The dashed edges are to indicate the trivial unframed "invisible" graph $1_0: 0\rightarrow 0$. These graphs $B_\times,B_+$ were also used as resolutions of tangle crossings in (2.3) of \cite{Morrison2019InvariantsO4}.}
    \label{fig:crossingsfolds}
\end{figure}

To setup the demonstration, we shall adopt the following conventions. All tangles will be assumed to be given a consistent blackboard framing.
\begin{itemize}
    \item Crossings (see the left side of fig. \ref{fig:crossingsfolds}): recall the  4-valent diagrams obtained from the graphs $B_+,B_\times$ in fig. \ref{fig:standardgraph}. The convention is that, if one stands on the oriented edge facing towards the crossing, then the crossing edge is associated with an under-crossing tangle. Otherwise it is an over-crossing. 
    \item Folds (see the right side of fig. \ref{fig:crossingsfolds}): we shall interpret the folds of 1-tangles as directed graphs $c_+: 2\rightarrow 0,c_-:0\rightarrow 2$ with the trivially marked point $0$ as source/target. One of the edges ending at two framed points are oriented "incorrectly", such that both of these points can be viewed as having the same framing.
\end{itemize}
We now construct the PL 2-ribbons on the graphs $B_+,B_\times$ which correspond to elementary string interactions involving the crossings, while those the graphs $c_\pm$ for the ones involving folds.

The isotopies which change the height of the string interactions are obvious, so we shall neglect them in the following.
\begin{enumerate}
    \item \textbf{Birth/death of an unknotted circle.} Consider the wedge sum $c_+\vee_{2} c_-$ along \textit{both} of its endpoints, then there is a PL 2-ribbon $c_+\vee_2 c_-\Rightarrow 1_0$ as shown on the left of fig. \ref{fig:cupsaddles}. We call this PL 2-ribbon "\textit{building a house}".
    \item \textbf{Saddle points.} Consider the wedge sum $c_-\vee_{0} c_+$, then there is a PL 2-ribbon $c_-\vee_0 c_+ \Rightarrow 1_2$ as shown on the right of fig. \ref{fig:cupsaddles}.
    \item \textbf{Cusp on a fold line.} Consider the wedge sum $c_-\vee_1 c_+$ along only one of its endpoints, then there is a PL 2-ribbon $c_-\vee_1 c_+ \Rightarrow1_0$ as shown in the middle of fig. \ref{fig:cupsaddles}.
    \begin{figure}[h]
        \centering
        \includegraphics[width=0.85\linewidth]{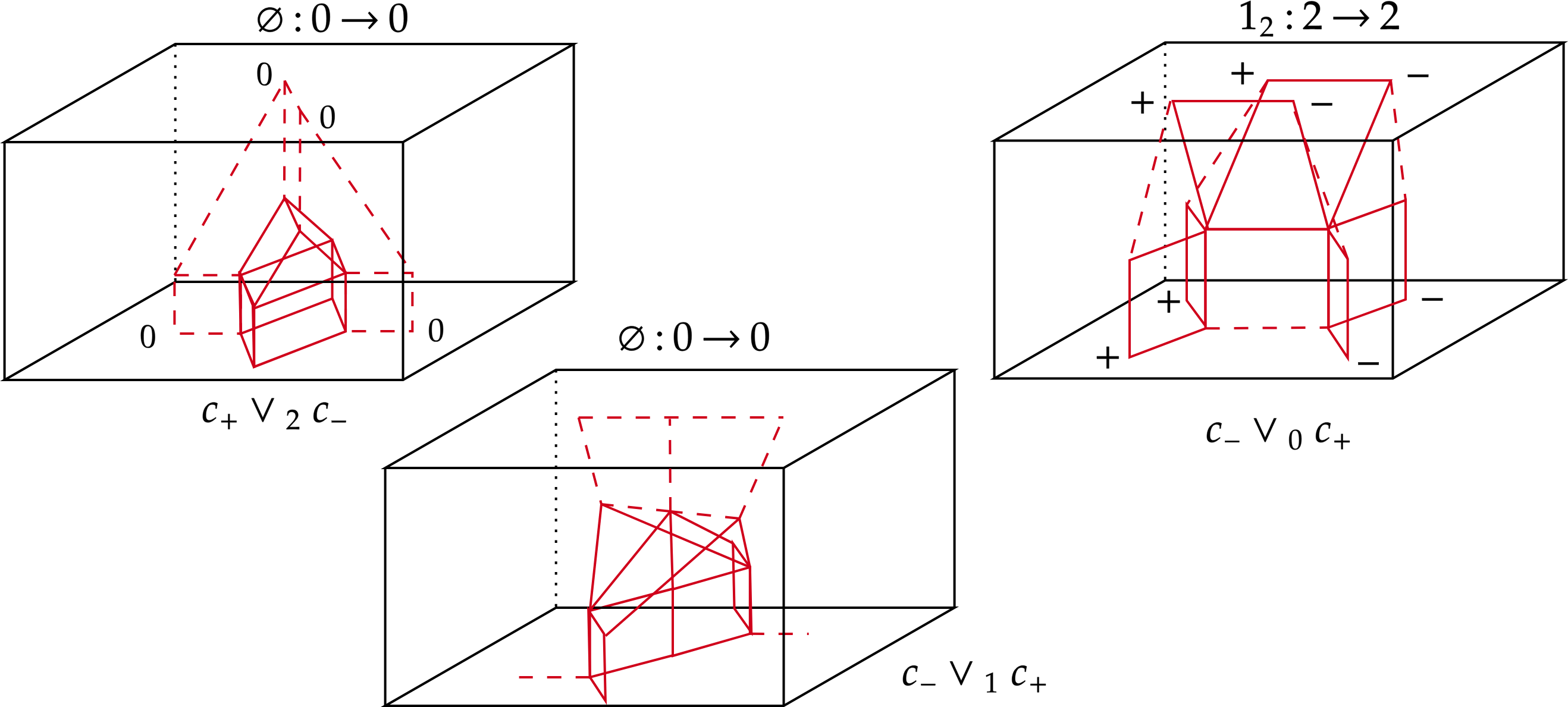}
        \caption{The PL 2-ribbon configurations which, upon smoothing, produces the birth/death of a circle, a saddle point and a cusp on a fold line. We have neglected the orientation and framing data of the graph for clarity.}
        \label{fig:cupsaddles}
    \end{figure}
    \item \textbf{Double point crossing on a fold line.} Consider the wedge sum $B_+\vee c_+$, then there is a PL 2-ribbon $B_+\vee c_+\Rightarrow c_-\vee B_+$ as in the left side of fig. \ref{fig:foldline}. Rotating the slab by $\pi/2$, we obtain $B_\times\vee c_+\Rightarrow c_-\vee B_\times$.
    \item \textbf{Reidemeister moves.} Consider the configurations $c_-\vee (B_+\coprod B_\times)\vee c_+,~c_-\vee_1 B_\times\vee_1 c_+$ as displayed on the right side of fig. \ref{fig:foldline}. The PL 2-ribbons witnessing Reidemeister I \& II moves can be obtained from "building a house", contracting the closed cycle present in these graphs.
    \begin{figure}[h]
        \centering
        \includegraphics[width=1\linewidth]{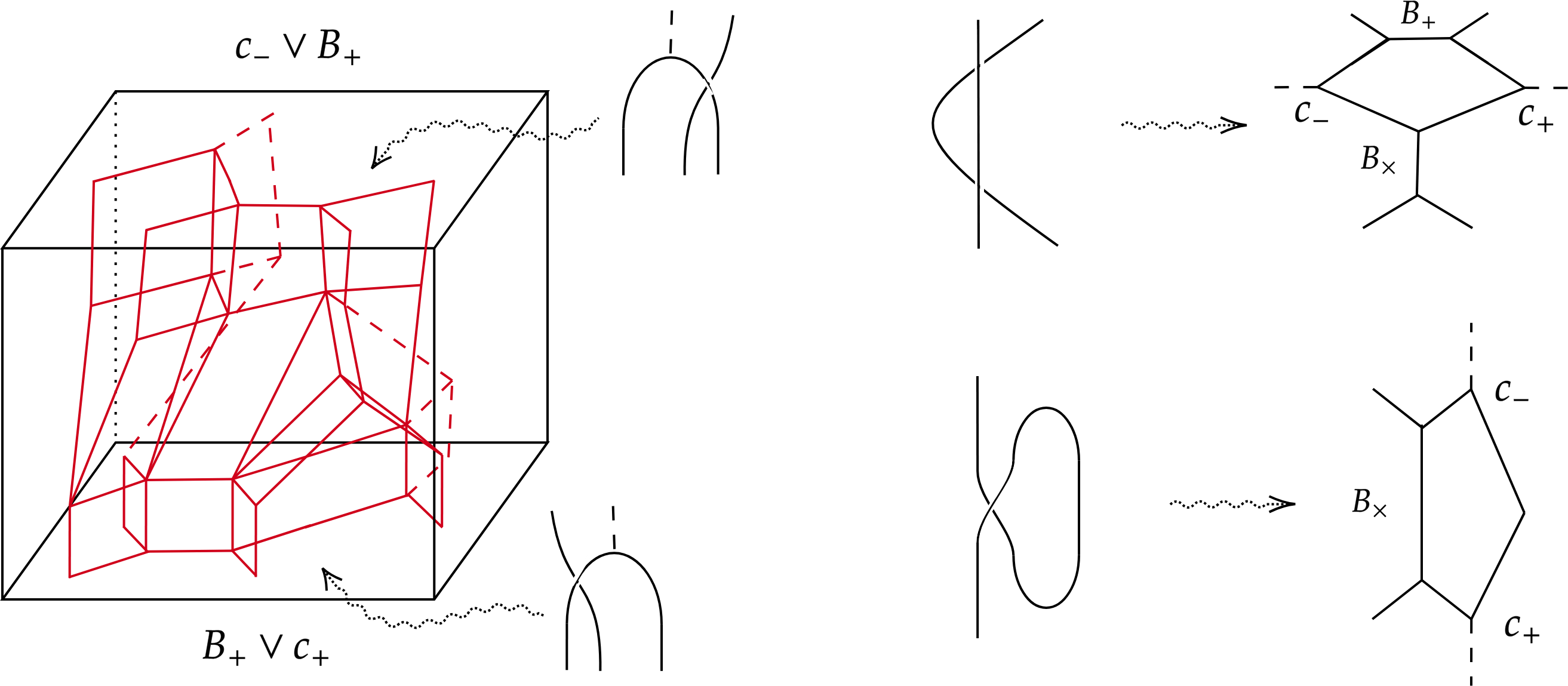}
        \caption{The PL 2-ribbons which, upon smoothing, produce a double point crossing on a fold line and the tangles involved in the Reidemeister I \& II moves.}
        \label{fig:foldline}
    \end{figure}
    The Reidemeister III move can also be constructed in the same way.
\end{enumerate}

There are, however, key differences between $\mathcal{T}$ and $\mathcal{T}'^{PL}_\text{mrk}$.
\begin{itemize}
    \item none of the (PL linearized) string interactions involve a trisection vertex (fig. \ref{fig:interchanger}), and
    \item $\mathcal{T}$ is not 2-$\dagger$; indeed, 1-/2-tangles in $\mathcal{T}$ are unframed and unoriented.
\end{itemize}
These mean that $\mathcal{T}'^{PL}_\text{mrk}$ could potentially capture more geometric data than $\mathcal{T}$; evidence for this was emphasized also in \cite{Chen:2025?}.

\subsection{Higher-dimensional skein relations}\label{higherskein}
As mentioned in \textit{Remark \ref{knothomology}}, both the $\mathfrak{gl}_N$ Khovanov homology and the 2-Chern-Simons Wilson surface states give rise to bigraded\footnote{In fact $\operatorname{KhR}^N$ is tri-graded, with the additional grading coming from \textit{blob homology} \cite{Morrison_2012}. However, this grading does not appear on the 4-disc $D^4$.} Abelian $\bbZ$-modules. These 2-ribbon invariants that arise form them --- though closely related geometrically --- have an important distinction.

In the former case, the usual skein relations from the quantum $\frak{gl}_N$, say, were first inserted into the skein polynomials $\mathcal{R}=\bbZ[q,q^{-1}]$,
\begin{equation*}
    s_{GL_N;q}(M^3) = \frac{\operatorname{Span}_R\left\{\text{framed links in $M^3$}\right\}}{\left\{\text{isotopies} ~\cup ~\substack{\text{skein relations} \\ \text{in $D^3\hookrightarrow M^3$}} \right\}}, 
\end{equation*}
which were then categorified to a homology theory $\mathcal{S}_{GL_N;q}^\ast(M^4)$. In the latter case, on the other hand, the underlying structure gauge group is first categorified, then from which an intrinsically higher-dimensional skein relation for decorated 2-ribbons can be extracted from the cobraiding $(\mathcal{R},\mathsf{T})$ on the 2-graph states. 

\medskip

Now geometrically, given the well-known "string-surface crossing" diagrams in braided monoidal 2-categories \cite{Douglas:2018,Johnson_Freyd_2023,Barrett_2024}, these higher-skein relations should encode the four ways in which string-surface crossings can be resolved; see fig. \ref{fig:stringsurface}. 

\begin{figure}[h]
    \centering
    \includegraphics[width=0.8\linewidth]{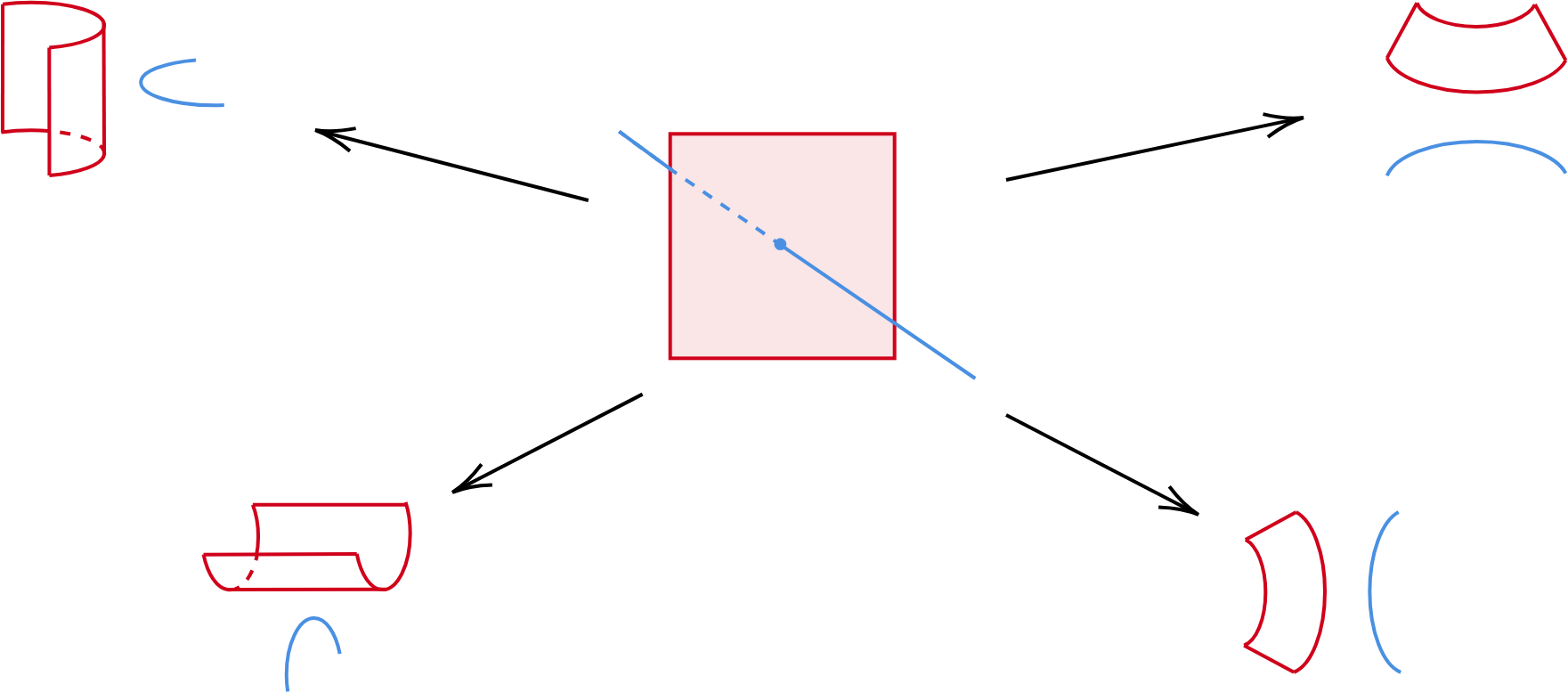}
    \caption{The four ways of resolving a string-surface crossing.}
    \label{fig:stringsurface}
\end{figure}

Provided the $R$-matrix cobraiding on a Hopf category determines a braiding structure on the 2-category of its 2-representations \cite{Chen:2025?} (see also \cite{neuchl1997representation,Chen:2023tjf}), these higher-skein relations should be captured by the $R$-matrix cobraiding on $\mathbb{U}_q\G$ --- which turn out to have \textit{precisely} four components (recall $\mathbb{G}$ as a category has $\mathsf{H}\rtimes G$ as morphisms),
\begin{align*}
    & \tilde R = \tilde R^l\times \tilde R^r \in \mathbb{U}_q\G\times \mathbb{U}_q\G,\\
    & \tilde R^l = \sum \tilde R^l_{(1)}\times \tilde R^l_{(2)} \in \mathbb{U}_q(\h\rtimes \g) \sim \mathbb{U}_q\h\times\mathbb{U}_q\g \\
    &\tilde R^r = \sum \tilde R^r_{(1)}\times \tilde R^r_{(2)}\in  \mathbb{U}_q(\h\rtimes \g) \sim \mathbb{U}_q\h\times\mathbb{U}_q\g,
\end{align*}
each governing the coarrow-part vs. the object-part components of the cobraiding $\mathfrak{R}$. The infinitesimal version of this idea is studied recently in \cite{Kemp_2025}.

Such \textbf{higher-skein relations} on $\mathscr{R}^\ast=H^\bullet_\mathbb{G}(B\mathbb{G},\bbZ)[q,q^{-1}]$ inherited upon $2\mathcal{CS}^\mathbb{G}_{q}(D^4)$ are what enters the skein-theoretic definition of the (tentative) 4-dimensional multiply-graded 2-Chern-Simons invariant 
\begin{equation*}
    \mathscr{S}_{\mathbb{G};q}^\ast(M^4) = \frac{\operatorname{Span}_{\mathscr{R}^\ast}\left\{\text{framed oriented 2-ribbons in $M^4$} \right\}} {\left\{\text{isotopies } \cup \substack{\text{ 2-skein relations} \\ \text{in $D^4\hookrightarrow M^4$}} \right\}},
\end{equation*}
in complete analogy with the Reshetikhin-Turaev construction \cite{Reshetikhin:1991tc,Turaev+2016}. 

The situation can be summarized in the following way,
\[\begin{tikzcd}
	&& \begin{array}{c} s_{G;q}(M^3)~\substack{\text{ribbon} \\ \text{invariants}} \end{array} && \begin{array}{c} \mathcal{S}^\ast_{GL_N;q}(M^4)~ \substack{\text{Khovanov-Rozansky} \\ \text{lasagna modules}} \end{array} \\
	\begin{array}{c} C_q(G)~\substack{\text{quantum} \\ \text{groups}}\end{array} \\
	&& \begin{array}{c} \C_q(\mathbb{G})~ \substack{\text{categorical} \\ \text{quantum groups}} \end{array} && \begin{array}{c} \mathscr{S}^\ast_{\mathbb{G};q}(M^4) ~\substack{\text{2-Chern-Simons}\\ \text{2-ribbon invariants}} \end{array}
	\arrow["{\text{categorify}}", from=1-3, to=1-5]
	\arrow["\begin{array}{c} \substack{\text{skein} \\ \text{relations}} \end{array}", from=2-1, to=1-3]
	\arrow["{\text{categorify}}"', from=2-1, to=3-3]
	\arrow["\begin{array}{c} \substack{\text{"2-skein"}\\ \text{relations}} \end{array}"', from=3-3, to=3-5]
	\arrow["{?}"{description}, squiggly, tail reversed, from=3-5, to=1-5]
\end{tikzcd}\]
It would be interesting to pin these 2-skein relations down and explicit compute the 2-ribbon invariants on, eg., $M^4=\mathbb{C}P^2,\overline{\bbC P}^2$ or $S^2\times S^2$. We shall leave this for a future work.

\newpage

\printbibliography

@article{TRENTINAGLIA2010750,
title = {Tannaka duality for proper Lie groupoids},
journal = {Journal of Pure and Applied Algebra},
volume = {214},
number = {6},
pages = {750-768},
year = {2010},
issn = {0022-4049},
doi = {https://doi.org/10.1016/j.jpaa.2009.08.004},
url = {https://www.sciencedirect.com/science/article/pii/S0022404909001765},
author = {Giorgio Trentinaglia}
}

@article{Miranda:2025,
    author = {C. Hughes, A. Miranda},
    title = {The Elementary Theory of the 2-Category of Small Categories},
    journal = {Theory and Applications of Categories},
    year = 2025,
volume=43,
number=8,
pages={196–242}
}

@article{Kemp_2025,
   title={Infinitesimal 2-braidings from 2-shifted Poisson structures},
   volume={212},
   ISSN={0393-0440},
   url={http://dx.doi.org/10.1016/j.geomphys.2025.105456},
   DOI={10.1016/j.geomphys.2025.105456},
   journal={Journal of Geometry and Physics},
   publisher={Elsevier BV},
   author={Kemp, Cameron and Laugwitz, Robert and Schenkel, Alexander},
   year={2025},
   month=jun, pages={105456} }

@InProceedings{Day:1970,
author="Day, Brian",
editor="MacLane, S.
and Applegate, H.
and Barr, M.
and Day, B.
and Dubuc, E.
and Phreilambud
and Pultr, A.
and Street, R.
and Tierney, M.
and Swierczkowski, S.",
title="On closed categories of functors",
booktitle="Reports of the Midwest Category Seminar IV",
year="1970",
publisher="Springer Berlin Heidelberg",
address="Berlin, Heidelberg",
pages="1--38",
isbn="978-3-540-36292-0"
}

@ARTICLE{2024arXiv240704891L,
       author = {{Liu}, Yu Leon},
        title = "{Braiding on complex oriented Soergel bimodules}",
      journal = {arXiv e-prints},
         year = 2024,
        month = jul,
          eid = {arXiv:2407.04891},
        pages = {arXiv:2407.04891},
          doi = {10.48550/arXiv.2407.04891},
archivePrefix = {arXiv},
       eprint = {2407.04891},
 primaryClass = {math.AT},
       adsurl = {https://ui.adsabs.harvard.edu/abs/2024arXiv240704891L}
}

@book{elias2020introduction,
  title={Introduction to Soergel Bimodules},
  author={Elias, B. and Makisumi, S. and Thiel, U. and Williamson, G.},
  isbn={9783030488260},
  series={RSME Springer Series},
  url={https://link.springer.com/book/10.1007/978-3-030-48826-0},
  year={2020},
  publisher={Springer International Publishing}
}

@article{Witten:2011zz,
    author = "Witten, Edward",
    title = "Fivebranes and Knots",
journal={Quantum Topol.},
vol=3,
number=1,
pages= {1–137},
    eprint = "1101.3216",
    archivePrefix = "arXiv",
    primaryClass = "hep-th",
    month = "1",
    year = "2012"
}

@article{SOZER2023109155,
title = {Monoidal categories graded by crossed modules and 3-dimensional HQFTs},
journal = {Advances in Mathematics},
volume = {428},
pages = {109155},
year = {2023},
issn = {0001-8708},
doi = {https://doi.org/10.1016/j.aim.2023.109155},
url = {https://www.sciencedirect.com/science/article/pii/S0001870823002980},
author = {K{\"u}rsat S{\"o}zer and Alexis Virelizier},
}

@article{Geer_Patureau-Mirand_Turaev_2009, 
title={Modified quantum dimensions and re-normalized link invariants}, 
volume={145}, 
DOI={10.1112/S0010437X08003795}, 
number={1}, 
journal={Compositio Mathematica}, 
author={Geer, Nathan and Patureau-Mirand, Bertrand and Turaev, Vladimir}, 
year={2009}, 
pages={196–212}
}

@article{geer:hal-00603999,
  TITLE = {{Ambidextrous objects and trace functions for nonsemisimple categories}},
  AUTHOR = {Geer, Nathan and Kujawa, Jonathan and Patureau-Mirand, Bertrand},
  URL = {https://hal.science/hal-00603999},
  NOTE = {15 pages},
  JOURNAL = {{Proceedings of the American Mathematical Society}},
  PUBLISHER = {{American Mathematical Society}},
  PAGES = {141 (2013), 2963--2978},
  YEAR = {2013},
  DOI = {10.1090/S0002-9939-2013-11563-7},
  HAL_ID = {hal-00603999},
  HAL_VERSION = {v1},
}

@ARTICLE{Geer2011-cr,
  title    = "Generalized trace and modified dimension functions on ribbon
              categories",
  author   = "Geer, Nathan and Kujawa, Jonathan and Patureau-Mirand, Bertrand",
  journal  = "Selecta Mathematica",
  volume   =  17,
  number   =  2,
  pages    = "453--504",
  month    =  "jun",
  year     =  "2011"
}

@article{Calaque:2021sgp,
    author = "Calaque, Damien and Haugseng, Rune and Scheimbauer, Claudia",
    title = "{The AKSZ Construction in Derived Algebraic Geometry as an Extended Topological Field Theory}",
    eprint = "2108.02473",
    archivePrefix = "arXiv",
journal={Memoirs of the AMS},
    primaryClass = "math.CT",
    reportNumber = "CPH-SYM-DNRF92",
    month = "8",
    year = "2021"
}

@article{Kong:2011jf,
    author = "Kong, Liang",
    editor = "Sati, Hisham and Schreiber, Urs",
    title = "{Conformal field theory and a new geometry}",
    eprint = "1107.3649",
    archivePrefix = "arXiv",
    primaryClass = "math.QA",
    journal = "Proc. Symp. Pure Math.",
    volume = "83",
    pages = "199--244",
    year = "2011"
}

@article{Gaiotto:2011,
author = {Davide Gaiotto and Edward Witten},
title = {{Knot invariants from four-dimensional gauge theory}},
volume = {16},
journal = {Advances in Theoretical and Mathematical Physics},
number = {3},
publisher = {International Press of Boston},
pages = {935 -- 1086},
year = {2012},
}

@book{Casson:1986,
    author = {A. Casson},
    title = {A la Recherche de la Topologie Perdue},
series={Progress in Mathematics},
volume={62},
    publisher = {Birkhäuser},
    year = {1986},
    chapter = {II. Three lectures on new infinite constructions in 4-dimensional manifolds},
pages={201–244}
}

@article{Freedman:1982,
author = {Michael Hartley Freedman},
title = {{The topology of four-dimensional manifolds}},
volume = {17},
journal = {Journal of Differential Geometry},
number = {3},
publisher = {Lehigh University},
pages = {357 -- 453},
year = {1982},
doi = {10.4310/jdg/1214437136},
URL = {https://doi.org/10.4310/jdg/1214437136}
}

@article{Bennett_2016,
   title={Exotic smoothings via large $\mathbb{R}^4$’s in Stein surfaces},
   volume={16},
   ISSN={1472-2747},
   url={http://dx.doi.org/10.2140/agt.2016.16.1637},
   DOI={10.2140/agt.2016.16.1637},
   number={3},
   journal={Algebraic \& Geometric Topology},
   publisher={Mathematical Sciences Publishers},
   author={Bennett, Julia},
   year={2016},
   month=jul, pages={1637–1681} }

@article{Gompf:1983,
author = {Robert E. Gompf},
title = {{Three exotic $\mathbb{R}^4$'s and other anomalies}},
volume = {18},
journal = {Journal of Differential Geometry},
number = {2},
publisher = {Lehigh University},
pages = {317 -- 328},
year = {1983},
doi = {10.4310/jdg/1214437666},
URL = {https://doi.org/10.4310/jdg/1214437666}
}

@misc{sarkar2018cohomologyringsclasstorus,
      title={Cohomology Rings of a Class of Torus Manifolds}, 
      author={Soumen Sarkar and Donald Stanley},
      year={2018},
      eprint={1807.03830},
      archivePrefix={arXiv},
      primaryClass={math.AT},
      url={https://arxiv.org/abs/1807.03830}, 
}

@article{Fausk:2003,
    author = {H. Fausk, P. Hu, and J.P. May },
    title = {Isomorphisms between left and right adjoints},
    journal = {Theory and Applications of Categories},
volume={11},
number=4,
pages={107-131},
    year = 2003
}

@misc{froehlich2024yonedalemmarepresentationtheorem,
      title={Yoneda lemma and representation theorem for double categories}, 
      author={Benedikt Fr{\"o}hlich and Lyne Moser},
      year={2024},
      eprint={2402.10640},
      archivePrefix={arXiv},
      primaryClass={math.CT},
      url={https://arxiv.org/abs/2402.10640}, 
}

@article{Cairns:1961,
author = {Stewart S. Cairns},
title = {{A simple triangulation method for smooth manifolds}},
volume = {67},
journal = {Bulletin of the American Mathematical Society},
number = {4},
publisher = {American Mathematical Society},
pages = {389 -- 390},
year = {1961},
}

@Inbook{Takesaki1979,
author="Takesaki, Masamichi",
editor="Takesaki, Masamichi",
title="Types of von Neumann Algebras and Traces",
bookTitle="Theory of Operator Algebras I",
year="1979",
publisher="Springer New York",
address="New York, NY",
pages="289--374",
isbn="978-1-4612-6188-9",
doi="10.1007/978-1-4612-6188-9_5",
url="https://doi.org/10.1007/978-1-4612-6188-9_5"
}

@inproceedings{Fiorenza:2013jz,
    author = "Fiorenza, Domenico and Sati, Hisham and Schreiber, Urs",
    title = "{A higher stacky perspective on Chern-Simons theory}",
    booktitle = "{Winter School in Mathematical Physics: Mathematical Aspects of Quantum Field Theory}",
    eprint = "1301.2580",
    archivePrefix = "arXiv",
    primaryClass = "hep-th",
    doi = "10.1007/978-3-319-09949-1_6",
    publisher = "Springer",
    pages = "153--211",
    month = "1",
    year = "2013"
}

@inproceedings{Davydov:2011kb,
    author = "Davydov, Alexei and Kong, Liang and Runkel, Ingo",
    editor = "Sati, Hisham and Schreiber, Urs",
    title = "{Field theories with defects and the centre functor}",
    eprint = "1107.0495",
    archivePrefix = "arXiv",
    primaryClass = "math.QA",
    reportNumber = "ZMP-HH-11-8",
series="Mathematical foundations of quantum field theory and perturbative string theory",
booktitle="Proc. Sympos. Pure Math.,",
 volume="83",
publisher="American Mathematical Society, Providence",
    pages = "71--128",
    month = "7",
    year = "2011"
}

@article{Sati:2008eg,
    author = "Sati, Hisham and Schreiber, Urs and Stasheff, Jim",
    title = "{$L_{\infty}$ algebra connections and applications to String- and Chern-Simons n-transport}",
    eprint = "0801.3480",
    archivePrefix = "arXiv",
    primaryClass = "math.DG",
    doi = "10.1007/978-3-7643-8736-5_17",
journal="Quantum Field Theory",
year="(2009)",
pages={303-424},
}

@article{Lam:2023xng,
    author = "Lam, Ho Tat",
    title = "{Classification of dipolar symmetry-protected topological phases: Matrix product states, stabilizer Hamiltonians, and finite tensor gauge theories}",
    eprint = "2311.04962",
    archivePrefix = "arXiv",
    primaryClass = "cond-mat.str-el",
    reportNumber = "MIT-CTP/5617",
    doi = "10.1103/PhysRevB.109.115142",
    journal = "Phys. Rev. B",
    volume = "109",
    number = "11",
    pages = "115142",
    year = "2024"
}

@book{book-symplectic,
    author = {Ana Cannas Silva},
    title = {Lectures on Symplectic Geometry},
    publisher = {Springer Berlin, Heidelberg},
doi={https://doi.org/10.1007/978-3-540-45330-7},
edition=1,
pages={XII, 220},
    year = 2001
}

@article{Chevyrev:2024eng,
    author = "Chevyrev, Ilya and Garban, Christophe",
    title = "{Villain Action in Lattice Gauge Theory}",
    eprint = "2404.09928",
    archivePrefix = "arXiv",
    primaryClass = "math.PR",
    doi = "10.1007/s10955-025-03420-1",
    journal = "J. Statist. Phys.",
    volume = "192",
    number = "3",
    pages = "38",
    year = "2025"
}

@article{Chen:2024ddr,
    author = "Chen, Jing-Yuan",
    title = "{Instanton Density Operator in Lattice QCD from Higher Category Theory}",
    eprint = "2406.06673",
    archivePrefix = "arXiv",
    primaryClass = "hep-lat",
    month = "6",
    year = "2024"
}

@article{Jacobson:2023cmr,
    author = "Jacobson, Theodore and Sulejmanpasic, Tin",
    title = "{Modified Villain formulation of Abelian Chern-Simons theory}",
    eprint = "2303.06160",
    archivePrefix = "arXiv",
    primaryClass = "hep-th",
    doi = "10.1103/PhysRevD.107.125017",
    journal = "Phys. Rev. D",
    volume = "107",
    number = "12",
    pages = "125017",
    year = "2023"
}

@book{book-operators,
    author = {Bruce Blackadar},
    title = {Operator Algebras},
subtitle={Theory of C*-Algebras and von Neumann Algebras},
series={Encyclopaedia of Mathematical Sciences},
doi={https://doi.org/10.1007/3-540-28517-2},
    publisher = {Springer Berlin, Heidelberg},
edition=1,
pages={XX, 517},
    year =2005 
}

@book{kosinski2013differential,
  title={Differential Manifolds},
  author={Kosinski, A.A.},
  isbn={9780486318158},
  series={Dover Books on Mathematics},
  url={https://books.google.com.hk/books?id=PcmjAQAAQBAJ},
  year={2013},
  publisher={Dover Publications}
}

@article{Pare:2011,
      title={Yoneda theory for double categories}, 
      author={Robert Pare},
      year={2011},
journal={Theory Appl. Categ.},
volume={25},
number=17,
pages={436-489}
}

@article{Shulman2011,
author = {Shulman, Michael},
journal = {The New York Journal of Mathematics [electronic only]},
language = {eng},
pages = {75-125},
publisher = {University at Albany, Deptartment of Mathematics and Statistics},
title = {Comparing composites of left and right derived functors.},
url = {http://eudml.org/doc/229181},
volume = {17},
year = {2011},
}

@Inbook{Kerler2001,
title="Double Categories and Double Functors",
bookTitle="Non-Semisimple Topological Quantum Field Theories for 3-Manifolds with Corners",
year="2001",
publisher="Springer Berlin Heidelberg",
address="Berlin, Heidelberg",
pages="343--352",
author={Thomas Kerler and Volodymyr V. Lyubashenko},
isbn="978-3-540-44625-5",
doi="10.1007/3-540-44625-7_10",
url="https://doi.org/10.1007/3-540-44625-7_10"
}

@article{schreiber2013connectionsnonabeliangerbesholonomy,
      title={Connections on non-abelian Gerbes and their Holonomy}, 
      author={Urs Schreiber and Konrad Waldorf},
      year={2013},
      eprint={0808.1923},
      archivePrefix={arXiv},
      primaryClass={math.DG},
journal={Theory Appl. Categ.},
volume={28},
number=17,
pages={476-540}
}

@article{Morrison_2012,
   title={Blob homology},
   volume={16},
   ISSN={1465-3060},
   url={http://dx.doi.org/10.2140/gt.2012.16.1481},
   DOI={10.2140/gt.2012.16.1481},
   number={3},
   journal={Geometry \& Topology},
   publisher={Mathematical Sciences Publishers},
   author={Morrison, Scott and Walker, Kevin},
   year={2012},
   month=jul, pages={1481–1607} }

@article{Kharlamov:1993,
     author = {Kharlamov, V. M. and Turaev, V. G.},
     title = {On the {Definition} of $2${-Category} of $2${-Knots}},
     journal = {Les rencontres physiciens-math\'ematiciens de Strasbourg -RCP25},
     note = {talk:7},
     pages = {151--166},
     publisher = {Institut de Recherche Math\'ematique Avanc\'ee - Universit\'e Louis Pasteur},
     volume = {45},
     year = {1993},
     language = {en},
     url = {https://www.numdam.org/item/RCP25_1993__45__151_0/}
}

@article{Liu:2023dhj,
    author = "Liu, Zhengwei and Ming, Shuang and Wang, Yilong and Wu, Jinsong",
    title = "{Alterfold Topological Quantum Field Theory}",
    eprint = "2312.06477",
    archivePrefix = "arXiv",
    primaryClass = "math-ph",
    month = "12",
    year = "2023"
}

@ARTICLE{Beck2024-fs,
  title    = "Combinatorial 2d higher topological quantum field theory from a
              local cyclic $A_\infty$ algebra",
  author   = "Beck, Justin and Losev, Andrey and Mnev, Pavel",
  journal  = "Letters in Mathematical Physics",
  volume   =  114,
  number   =  6,
  pages    = "125",
  month    =  oct,
  year     =  2024
}

@article{shulman2009framedbicategoriesmonoidalfibrations,
      title={Framed bicategories and monoidal fibrations}, 
      author={Michael A. Shulman},
      year={2008},
journal={Theory Appl. Categ.},
volume={20},
 number={18}, 
pages={650--738},
      eprint={0706.1286},
      archivePrefix={arXiv},
      primaryClass={math.CT},
      url={https://arxiv.org/abs/0706.1286}, 
}

@book{book-2cats,
    author = {Johnson, Niles and Yau, Donald},
    title = {2-Dimensional Categories},
    publisher = {Oxford University Press},
    year = {2021},
    month = {01},
    isbn = {9780198871378},
    doi = {10.1093/oso/9780198871378.001.0001},
    url = {https://doi.org/10.1093/oso/9780198871378.001.0001},
}

@article{Ehresmann1963,
author = {Ehresmann, Charles},
journal = {Annales scientifiques de l'École Normale Supérieure},
language = {fre},
number = {4},
pages = {349-426},
publisher = {Elsevier},
title = {Catégories structurées},
url = {http://eudml.org/doc/81794},
volume = {80},
year = {1963},
}

@misc{haioun2025nonsemisimplewrtboundarycraneyetter,
      title={Non-semisimple WRT at the boundary of Crane-Yetter}, 
      author={Benjamin Ha{\" i}oun},
      year={2025},
      eprint={2503.20905},
      archivePrefix={arXiv},
      primaryClass={math.QA},
      url={https://arxiv.org/abs/2503.20905}, 
}

@book{brylinski2007loop,
  title={Loop Spaces, Characteristic Classes and Geometric Quantization},
  author={Brylinski, J.L.},
  isbn={9780817647308},
  lccn={2007936854},
  series={Modern Birkh{\"a}user Classics},
  url={https://books.google.com.hk/books?id=ta5UB1D64_gC},
  year={2007},
  publisher={Birkh{\"a}user Boston}
}

@article{WALDHAUSEN1968195,
title = {Heegaard-Zerlegungen der 3-sphäre},
journal = {Topology},
volume = {7},
number = {2},
pages = {195-203},
year = {1968},
issn = {0040-9383},
doi = {https://doi.org/10.1016/0040-9383(68)90027-X},
url = {https://www.sciencedirect.com/science/article/pii/004093836890027X},
author = {Friedhelm Waldhausen}
}

@ARTICLE{Gomez_Larranaga1987-dz,
  title    = "3-Manifolds which are unions of three solid tori",
  author   = "G{\'o}mez Larra{\~n}aga, Jos{\'e} Carlos",
  journal  = "manuscripta mathematica",
  volume   =  59,
  number   =  3,
  pages    = "325--330",
  month    =  sep,
  year     =  1987
}

@book{Loregian_2021, place={Cambridge}, series={London Mathematical Society Lecture Note Series}, title={(Co)end Calculus}, publisher={Cambridge University Press}, author={Loregian, Fosco}, year={2021}, collection={London Mathematical Society Lecture Note Series}}

@article{Castler:1965,
 ISSN = {00029939, 10886826},
 URL = {http://www.jstor.org/stable/2033878},
 author = {B. G. Casler},
 journal = {Proceedings of the American Mathematical Society},
 number = {4},
 pages = {559--566},
 publisher = {American Mathematical Society},
 title = {An Imbedding Theorem for Connected 3-Manifolds with Boundary},
 urldate = {2025-03-24},
 volume = {16},
 year = {1965}
}

@ARTICLE{Sakata2022-il,
  title    = "Handlebody decompositions of three-manifolds and polycontinuous
              patterns",
  author   = "Sakata, N and Mishina, R and Ogawa, M and Ishihara, K and Koda, Y
              and Ozawa, M and Shimokawa, K",
  journal  = "Proc Math Phys Eng Sci",
  volume   =  478,
  number   =  2260,
  pages    = "20220073",
  month    =  apr,
  year     =  2022,
  address  = "England",
  language = "en"
}

@book{matveev2007algorithmic,
  title={Algorithmic Topology and Classification of 3-Manifolds},
  author={Matveev, S.},
  isbn={9783540458982},
  lccn={2007927936},
  series={Algorithms and Computation in Mathematics},
  url={https://link.springer.com/book/10.1007/978-3-662-05102-3},
  year={2007},
  publisher={Springer Berlin Heidelberg}
}

@article{Bunk_2025,
   title={Lorentzian bordisms in algebraic quantum field theory},
   volume={115},
   ISSN={1573-0530},
   url={http://dx.doi.org/10.1007/s11005-025-01906-3},
   DOI={10.1007/s11005-025-01906-3},
   number={1},
   journal={Letters in Mathematical Physics},
   publisher={Springer Science and Business Media LLC},
   author={Bunk, Severin and MacManus, James and Schenkel, Alexander},
   year={2025},
   month=feb }

@article{Frank:1993,
    author = {Michael Frank},
    title = {Direct integrals and Hilbert W*-Modules},
    journal = {Problems in algebra, geometry and discrete mathematics (in russian)},
    year = 1992,
    publisher={Moscow State University}, 
    pages={162-177}
}

@inproceedings{Segal1951DecompositionsOO,
booktitle={Memoirs of the American Mathematical Scociety},
  title={Decompositions of Operator Algebras I and II},
  author={Irving Ezra Segal},
  year={1951},
  url={https://api.semanticscholar.org/CorpusID:117110498}
}

@article{ACKERMAN_2016,
   title={On computability and disintegration},
   volume={27},
   ISSN={1469-8072},
   url={http://dx.doi.org/10.1017/S0960129516000098},
   DOI={10.1017/s0960129516000098},
   number={8},
   journal={Mathematical Structures in Computer Science},
   publisher={Cambridge University Press (CUP)},
   author={ACKERMAN, NATHANAEL L. and FREER, CAMERON E. and ROY, DANIEL M.},
   year={2016},
   month=jul, pages={1287–1314} }

@article{Bridgeland:1999,
author = {Bridgeland, Tom},
title = {Equivalences of Triangulated Categories and Fourier–Mukai Transforms},
journal = {Bulletin of the London Mathematical Society},
volume = {31},
number = {1},
pages = {25-34},
doi = {https://doi.org/10.1112/S0024609398004998},
url = {https://londmathsoc.onlinelibrary.wiley.com/doi/abs/10.1112/S0024609398004998},
eprint = {https://londmathsoc.onlinelibrary.wiley.com/doi/pdf/10.1112/S0024609398004998},
year = {1999}
}

@book{hudson1969piecewise,
  title={Piecewise Linear Topology},
  author={Hudson, J.F.P.},
  isbn={9780805345513},
  lccn={lc72075219},
  series={Mathematics lecture note series},
  url={https://books.google.com.hk/books?id=UAY_AAAAIAAJ},
  year={1969},
  publisher={W. A. Benjamin}
}

@misc{douglas2016internalbicategories,
      title={Internal bicategories}, 
      author={Christopher L. Douglas and André G. Henriques},
      year={2016},
      eprint={1206.4284},
      archivePrefix={arXiv},
      primaryClass={math.CT},
      url={https://arxiv.org/abs/1206.4284}, 
}

@article{Ferreira:2015,
author = {Nelson Martins-Ferreira},
title = {{On the notion of pseudocategory internal to a category with a 2-cell structure}},
volume = {8},
journal = {Tbilisi Mathematical Journal},
number = {1},
publisher = {Tbilisi Centre for Mathematical Sciences},
pages = {107 -- 141},
year = {2015},
doi = {10.1515/tmj-2015-0007},
URL = {https://doi.org/10.1515/tmj-2015-0007}
}

@article{Swan1962VectorBA,
  title={Vector bundles and projective modules},
  author={Richard G. Swan},
  journal={Transactions of the American Mathematical Society},
  year={1962},
  volume={105},
  pages={264-277},
  url={https://api.semanticscholar.org/CorpusID:51735278}
}

@article{Serre:1955,
 ISSN = {0003486X, 19398980},
 URL = {http://www.jstor.org/stable/1969915},
 author = {Jean-Pierre Serre},
 journal = {Annals of Mathematics},
 number = {2},
 pages = {197--278},
 publisher = {Princeton University Press},
 title = {Faisceaux Algébriques Cohérents},
 urldate = {2025-01-18},
 volume = {61},
 year = {1955}
}

@book{Turaev+2016,
url = {https://doi.org/10.1515/9783110435221},
title = {Quantum Invariants of Knots and 3-Manifolds},
author = {Vladimir G. Turaev},
publisher = {De Gruyter},
address = {Berlin, Boston},
doi = {doi:10.1515/9783110435221},
isbn = {9783110435221},
year = {2016},
lastchecked = {2025-01-17}
}

@article{Liu:2024qth,
    author = "Liu, Zhengwei",
    title = "{Functional Integral Construction of Topological Quantum Field Theory}",
    eprint = "2409.17103",
    archivePrefix = "arXiv",
    primaryClass = "math-ph",
    month = "9",
    year = "2024"
}

@article{MACKENZIE200046,
title = {Double Lie Algebroids and Second-Order Geometry, II},
journal = {Advances in Mathematics},
volume = {154},
number = {1},
pages = {46-75},
year = {2000},
issn = {0001-8708},
doi = {https://doi.org/10.1006/aima.1999.1892},
url = {https://www.sciencedirect.com/science/article/pii/S0001870899918923},
author = {K.C.H. Mackenzie}
}

@book{Kashiwara1990SheavesOM,
  title={Sheaves on Manifolds},
  author={Masaki Kashiwara and Pierre Schapira},
series={Grundlehren der mathematischen Wissenschaften},
doi={https://doi.org/10.1007/978-3-662-02661-8},
publisher={Springer-Verlag Berlin Heidelberg},
day={ 14},
month={March},
year={2013},
edition={1},
pages={X,512},
url={https://link.springer.com/book/10.1007/978-3-662-02661-8}
}

@ARTICLE{Kong2024-vr,
  title    = "Categories of Quantum Liquids {II}",
  author   = "Kong, Liang and Zheng, Hao",
  journal  = "Communications in Mathematical Physics",
  volume   =  405,
  number   =  9,
  pages    = "203",
  month    =  aug,
  year     =  2024
}

@article{Radford:1976,
 ISSN = {00029327, 10806377},
 URL = {http://www.jstor.org/stable/2373888},
 author = {David E. Radford},
 journal = {American Journal of Mathematics},
 number = {2},
 pages = {333--355},
 publisher = {Johns Hopkins University Press},
 title = {The Order of the Antipode of a Finite Dimensional Hopf Algebra is Finite},
 urldate = {2024-12-13},
 volume = {98},
 year = {1976}
}

@ARTICLE{Etingof:2004,
  author={Etingof, Pavel and Nikshych, Dmitri and Ostrik, Viktor},
  journal={International Mathematics Research Notices}, 
  title={An analogue of Radford's S4 formula for finite tensor categories}, 
  year={2004},
  volume={2004},
  number={54},
  pages={2915-2933},
  doi={10.1155/S1073792804141445}}

@article{Barrett_2024,
   title={Gray categories with duals and their diagrams},
   volume={450},
   ISSN={0001-8708},
   url={http://dx.doi.org/10.1016/j.aim.2024.109740},
   DOI={10.1016/j.aim.2024.109740},
   journal={Advances in Mathematics},
   publisher={Elsevier BV},
   author={Barrett, John W. and Meusburger, Catherine and Schaumann, Gregor},
   year={2024},
   month=jul, pages={109740} }

@article{Grandis2004,
author = {{Grandis, Marco} and {Pare, Robert}},
journal = {Cahiers de Topologie et G{\`e}om{\`e}trie Diff{\`e}rentielle Cat{\`e}goriques},
language = {eng},
number = {3},
pages = {193-240},
publisher = {Dunod {\`e}diteur, publi{\`e} avec le concours du CNRS},
title = {Adjoint for double categories},
url = {http://eudml.org/doc/91684},
volume = {45},
year = {2004},
}

@article{PUTROV2017254,
title = {Braiding statistics and link invariants of bosonic/fermionic topological quantum matter in 2+1 and 3+1 dimensions},
journal = {Annals of Physics},
volume = {384},
pages = {254-287},
year = {2017},
issn = {0003-4916},
doi = {https://doi.org/10.1016/j.aop.2017.06.019},
url = {https://www.sciencedirect.com/science/article/pii/S0003491617301859},
author = {Pavel Putrov and Juven Wang and Shing-Tung Yau}
}

@misc{Chen1:2025?,
      title={Combinatorial quantization of 4d 2-Chern-Simons theory I: the Hopf category of higher-graph states}, 
      author={Hank Chen},
      year={2025},
      eprint={2501.06486},
      archivePrefix={arXiv},
      primaryClass={math-ph},
      url={https://arxiv.org/abs/2501.06486}, 
}

@misc{Chen:2025?,
    author = "Hank Chen",
    title = "Categorical quantum symmetries and ribbon tensor 2-categories",
    journal = "To appear",
    year = "2025"
}

@article{Reshetikhin:1990pr,
    author = "Reshetikhin, N. Yu. and Turaev, V. G.",
    title = "{Ribbon graphs and their invariants derived from quantum groups}",
    doi = "10.1007/BF02096491",
    journal = "Commun. Math. Phys.",
    volume = "127",
    pages = "1--26",
    year = "1990"
}

@article{CARTER19971,
title = {A Combinatorial Description of Knotted Surfaces and Their Isotopies},
journal = {Advances in Mathematics},
volume = {127},
number = {1},
pages = {1-51},
year = {1997},
issn = {0001-8708},
doi = {https://doi.org/10.1006/aima.1997.1618},
url = {https://www.sciencedirect.com/science/article/pii/S0001870897916182},
author = {J.Scott Carter and Joachim H. Rieger and Masahico Saito}
}

@article{Delvaux2006ANO,
  title={A note on Radford's $S^4$ formula},
  author={Lydia Delvaux and Alfons Van Daele and Shuanhong Wang},
  journal={arXiv: Rings and Algebras},
  year={2006},
  url={https://api.semanticscholar.org/CorpusID:1770960}
}

@article{RADFORD19921,
title = {On the antipode of a quasitriangular Hopf algebra},
journal = {Journal of Algebra},
volume = {151},
number = {1},
pages = {1-11},
year = {1992},
issn = {0021-8693},
doi = {https://doi.org/10.1016/0021-8693(92)90128-9},
url = {https://www.sciencedirect.com/science/article/pii/0021869392901289},
author = {David E Radford}
}

@article{SHUM199457,
title = {Tortile tensor categories},
journal = {Journal of Pure and Applied Algebra},
volume = {93},
number = {1},
pages = {57-110},
year = {1994},
issn = {0022-4049},
doi = {https://doi.org/10.1016/0022-4049(92)00039-T},
url = {https://www.sciencedirect.com/science/article/pii/002240499200039T},
author = {Mei Chee Shum}
}

@ARTICLE{Jones:2017,
       author = {{Jones}, Corey and {Penneys}, David},
        title = "{Operator Algebras in Rigid C*-Tensor Categories}",
      journal = {Communications in Mathematical Physics},
         year = 2017,
        month = nov,
       volume = {355},
       number = {3},
        pages = {1121-1188},
          doi = {10.1007/s00220-017-2964-0},
archivePrefix = {arXiv},
       eprint = {1611.04620},
 primaryClass = {math.OA},
       adsurl = {https://ui.adsabs.harvard.edu/abs/2017CMaPh.355.1121J}
}

@article{Pachl_1978, title={Disintegration and compact measures.}, volume={43}, url={https://www.mscand.dk/article/view/11771}, DOI={10.7146/math.scand.a-11771}, abstractNote={&amp;lt;p&amp;gt;A probability space is compact in the sense of Marczewski if and only if it admits countably additive disintegrations. It follows that the restriction of a compact measure to a sub-$\sigma$-algebra is compact.&amp;lt;/p&amp;gt;}, journal={MATHEMATICA SCANDINAVICA}, author={Pachl, Jan K.}, year={1978}, month={Jun.}, pages={157–168} }

@article{Yetter2003MeasurableC,
  title={Measurable Categories},
  author={David N. Yetter},
  journal={Applied Categorical Structures},
  year={2003},
  volume={13},
  pages={469-500},
  url={https://api.semanticscholar.org/CorpusID:15178814}
}

@article{stehouwer2023dagger,
  title={Dagger categories via anti-involutions and positivity},
  author={Stehouwer, Luuk and Steinebrunner, Jan},
  journal={arXiv preprint arXiv:2304.02928},
  year={2023}
}

@misc{ferrer2024daggerncategories,
      title={Dagger $n$-categories}, 
      author={Giovanni Ferrer and Brett Hungar and Theo Johnson-Freyd and Cameron Krulewski and Lukas Müller and Nivedita and David Penneys and David Reutter and Claudia Scheimbauer and Luuk Stehouwer and Chetan Vuppulury},
      year={2024},
      eprint={2403.01651},
      archivePrefix={arXiv},
      primaryClass={math.CT},
      url={https://arxiv.org/abs/2403.01651}, 
}

@book{book-mathphys1,
  address = {New York},
  author = {Reed, M. and Simon, B.},
  edition = {Second},
  publisher = {Academic Press},
  title = {Methods of Modern Mathematical Physics. I {F}unctional
 Analysis},
  year = 1980
}

@ARTICLE{Fuller:2015,
       author = {{Fuller}, Ben},
        title = "{Semidirect Products of Monoidal Categories}",
      journal = {arXiv e-prints},
         year = 2015,
        month = oct,
          eid = {arXiv:1510.08717},
        pages = {arXiv:1510.08717},
          doi = {10.48550/arXiv.1510.08717},
archivePrefix = {arXiv},
       eprint = {1510.08717},
 primaryClass = {math.CT},
       adsurl = {https://ui.adsabs.harvard.edu/abs/2015arXiv151008717F}
}

@article{DAY199799,
title = {Monoidal Bicategories and Hopf Algebroids},
journal = {Advances in Mathematics},
volume = {129},
number = {1},
pages = {99-157},
year = {1997},
issn = {0001-8708},
doi = {https://doi.org/10.1006/aima.1997.1649},
url = {https://www.sciencedirect.com/science/article/pii/S0001870897916492},
author = {Brian Day and Ross Street}
}

@article{Soncini:2014,
  title={4-D semistrict higher Chern-Simons theory I},
  author={ Soncini, Emanuele  and  Zucchini, Roberto },
  journal={Journal of High Energy Physics},
  volume={2014},
  number={10},
  year={2014},
}

@inproceedings{Baez2023HoangXS,
  title={Ho{\` a}ng Xu{\^ a}n S{\'i}nh's Thesis: Categorifying Group Theory},
  author={John C. Baez},
  year={2023},
  url={https://api.semanticscholar.org/CorpusID:260775694}
}

@article{Bursztyn2000DeformationQO,
  title={Deformation Quantization of Hermitian Vector Bundles},
  author={Henrique Bursztyn and Stefan Waldmann},
  journal={Letters in Mathematical Physics},
  year={2000},
  volume={53},
  pages={349-365},
  url={https://api.semanticscholar.org/CorpusID:117082030}
}

@ARTICLE{Crane:2003gk,
  title    = "Measurable Categories and 2-Groups",
  author   = "Crane, Louis and Yetter, David N",
  journal  = "Applied Categorical Structures",
  volume   =  13,
  number   =  5,
  pages    = "501--516",
  month    =  dec,
  year     =  2005
}

@article{Williams2015HaarSO,
  title={Haar Systems on Equivalent Groupoids},
  author={Dana P. Williams},
  journal={arXiv: Operator Algebras},
  year={2015},
  url={https://api.semanticscholar.org/CorpusID:33769918}
}

@article{Waldorf2015TransgressiveLG,
  title={Transgressive loop group extensions},
  author={Konrad Waldorf},
  journal={Mathematische Zeitschrift},
  year={2015},
  volume={286},
  pages={325-360},
  url={https://api.semanticscholar.org/CorpusID:119641403}
}

@article{Nikolaus2011FOUREV,
  title={Four Equivalent Versions of Non-Abelian Gerbes},
  author={Thomas Nickelsen Nikolaus and Konrad Waldorf},
  journal={Pacific Journal of Mathematics},
  year={2011},
  volume={264},
  pages={355-420},
  url={https://api.semanticscholar.org/CorpusID:55265704}
}

@article{Schreiber:2013pra,
    author = "Schreiber, Urs",
    title = "{Differential cohomology in a cohesive infinity-topos}",
    eprint = "1310.7930",
    archivePrefix = "arXiv",
    primaryClass = "math-ph",
    month = "10",
    year = "2013"
}

@article{Waldorf2012ACO,
  title={A Construction of String 2-Group Models using a Transgression-Regression Technique},
  author={Konrad Waldorf},
  journal={Contemp. Math.},
volume={584},
pages={99-115},
  year={2012},
  url={https://api.semanticscholar.org/CorpusID:119267307}
}

@article{Schommer_Pries_2011,
   title={Central extensions of smooth 2–groups and a finite-dimensional string 2–group},
   volume={15},
   ISSN={1465-3060},
   url={http://dx.doi.org/10.2140/gt.2011.15.609},
   DOI={10.2140/gt.2011.15.609},
   number={2},
   journal={Geometry \& Topology},
   publisher={Mathematical Sciences Publishers},
   author={Schommer-Pries, Christopher-J},
   year={2011},
   month=may, pages={609–676} }

@article{Waldorf:2012,
      title={Transgression to Loop Spaces and its Inverse, I: Diffeological Bundles and Fusion Maps}, 
      author={Konrad Waldorf},
journal = {	Cah. Topol. Geom. Differ. Categ.}, 
year={2012},
volume={LIII}, 
pages={162-210},
      primaryClass={math.DG},
      url={https://arxiv.org/abs/0911.3212}
}

@article{Henriques2006Integrating,
  title={Integrating $L_\infty$-algebras},
  author={Andr{\'e} Henriques},
  journal={Compositio Mathematica},
  year={2006},
  volume={144},
  pages={1017 - 1045},
  url={https://api.semanticscholar.org/CorpusID:17931112}
}

@article{Khovanov_2010,
   title={A categorification of quantum $\mathrm{sl}(n)$},
   volume={1},
   ISSN={1664-073X},
   url={http://dx.doi.org/10.4171/QT/1},
   DOI={10.4171/qt/1},
   number={1},
   journal={Quantum Topology},
   publisher={European Mathematical Society - EMS - Publishing House GmbH},
   author={Khovanov, Mikhail and Lauda, Aaron D.},
   year={2010},
   month=feb, pages={1–92} }

@article{Fiorenza2020TwistorialCI,
  title={Twistorial cohomotopy implies Green–Schwarz anomaly cancellation},
  author={Domenico Fiorenza and Hisham Sati and Urs Schreiber},
  journal={Reviews in Mathematical Physics},
  year={2020},
  url={https://api.semanticscholar.org/CorpusID:221173118}
}

@article{FAONTE2019389,
title = {Higher Kac–Moody algebras and moduli spaces of G-bundles},
journal = {Advances in Mathematics},
volume = {346},
pages = {389-466},
year = {2019},
issn = {0001-8708},
doi = {https://doi.org/10.1016/j.aim.2019.01.040},
url = {https://www.sciencedirect.com/science/article/pii/S0001870819300763},
author = {Giovanni Faonte and Benjamin Hennion and Mikhail Kapranov},
}

@article{Kapranov2021InfinitedimensionalL,
  title={Infinite-dimensional (dg) Lie algebras and factorization algebras in algebraic geometry},
  author={Mikhail Kapranov},
  journal={Japanese Journal of Mathematics},
  year={2021},
  pages={1-32},
  url={https://api.semanticscholar.org/CorpusID:231392391}
}

@article{Schenkel:2024dcd,
    author = "Schenkel, Alexander and Vicedo, Beno\textasciicircum{}it",
    title = "{5d 2-Chern-Simons Theory and 3d Integrable Field Theories}",
    eprint = "2405.08083",
    archivePrefix = "arXiv",
    primaryClass = "hep-th",
    doi = "10.1007/s00220-024-05170-9",
    journal = "Commun. Math. Phys.",
    volume = "405",
    number = "12",
    pages = "293",
    year = "2024"
}

@article{Morrison2019InvariantsO4,
  title={Invariants of 4–manifolds from Khovanov–Rozansky
link homology},
  author={Scott Morrison and Kevin Walker and Paul Wedrich},
  journal={Geometry \& Topology},
  year={2019},
  url={https://api.semanticscholar.org/CorpusID:198967556}
}

@article{Alfonsi:2024qdr,
    author = "Alfonsi, Luigi and Kim, Hyungrok and Young, Charles A. S.",
    title = "{Raviolo vertex algebras, cochains and conformal blocks}",
    eprint = "2401.11917",
    archivePrefix = "arXiv",
    primaryClass = "math.QA",
    month = "1",
    year = "2024"
}

@article{Manolescu2022SkeinLM,
  title={Skein lasagna modules and handle decompositions},
  author={Ciprian Manolescu and Kevin Walker and Paul Wedrich},
  journal={Advances in Mathematics},
  year={2022},
  url={https://api.semanticscholar.org/CorpusID:249538143}
}

@article{Chen:2024axr,
    author = "Chen, Hank and Liniado, Joaquin",
    title = "{Higher gauge theory and integrability}",
    eprint = "2405.18625",
    archivePrefix = "arXiv",
    primaryClass = "hep-th",
    doi = "10.1103/PhysRevD.110.086017",
    journal = "Phys. Rev. D",
    volume = "110",
    number = "8",
    pages = "086017",
    year = "2024"
}

@article{alfonsi2024raviolo,
  title={Raviolo vertex algebras, cochains and conformal blocks},
  author={Alfonsi, Luigi and Kim, Hyungrok and Young, Charles AS},
  journal={arXiv preprint arXiv:2401.11917},
  year={2024}
}

@article{huang2023tannaka,
  title={Tannaka-Krein duality for finite 2-groups},
  author={Huang, Mo and Zhang, Zhi-Hao},
  journal={arXiv preprint arXiv:2305.18151},
  year={2023}
}

@article{BAEZ2003705,
title = {Higher-dimensional algebra IV: 2-tangles},
journal = {Advances in Mathematics},
volume = {180},
number = {2},
pages = {705-764},
year = {2003},
issn = {0001-8708},
doi = {https://doi.org/10.1016/S0001-8708(03)00018-5},
url = {https://www.sciencedirect.com/science/article/pii/S0001870803000185},
author = {John C. Baez and Laurel Langford}
}

@article{lurie2008classification,
  title={On the classification of topological field theories},
  author={Lurie, Jacob},
  journal={Current developments in mathematics},
  volume={2008},
  number={1},
  pages={129--280},
  year={2008},
  publisher={International Press of Boston}
}

@article{Atiyah:1988,
     author = {Atiyah, Michael F.},
     title = {Topological quantum field theory},
     journal = {Publications Math\'ematiques de l'IH{\'E}S},
     pages = {175--186},
     publisher = {Institut des Hautes \'Etudes Scientifiques},
     volume = {68},
     year = {1988},
     mrnumber = {1001453},
     zbl = {0692.53053},
     language = {en},
     url = {http://www.numdam.org/item/PMIHES_1988__68__175_0/}
}

@article{Sean:private,
    author = {Sanford, Sean},
    title = {Universal traces and 2-characters for 2-representations of groups},
    journal = {private communiations},
    year = 2024
}

@article{liu2024braided,
  title={A braided $(\infty,2)$-category of Soergel bimodules},
  author={Liu, Yu Leon and Mazel-Gee, Aaron and Reutter, David and Stroppel, Catharina and Wedrich, Paul},
    eprint = "2401.02956",
    archivePrefix = "arXiv",
    year={2024}
}

@article{Fiorenza:2020iax,
    author = "Fiorenza, Domenico and Sati, Hisham and Schreiber, Urs",
    title = "{Twistorial cohomotopy implies Green\textendash{}Schwarz anomaly cancellation}",
    eprint = "2008.08544",
    archivePrefix = "arXiv",
    primaryClass = "hep-th",
    doi = "10.1142/S0129055X22500131",
    journal = "Rev. Math. Phys.",
    volume = "34",
    number = "05",
    pages = "2250013",
    year = "2022"
}

@article{Jurco:2018sby,
    author = {Jur\v{c}o, Branislav and Raspollini, Lorenzo and S\"amann, Christian and Wolf, Martin},
    title = "{$L_\infty$-Algebras of Classical Field Theories and the Batalin-Vilkovisky Formalism}",
    eprint = "1809.09899",
    archivePrefix = "arXiv",
    primaryClass = "hep-th",
    reportNumber = "EMPG-18-19, DMUS-MP-18/05",
    doi = "10.1002/prop.201900025",
    journal = "Fortsch. Phys.",
    volume = "67",
    number = "7",
    pages = "1900025",
    year = "2019"
}

@article{Alvarez:1997ma,
    author = "Alvarez, Orlando and Ferreira, Luiz A. and Sanchez Guillen, J.",
    title = "{A New approach to integrable theories in any dimension}",
    eprint = "hep-th/9710147",
    archivePrefix = "arXiv",
    reportNumber = "US-FT-31-97, IFT-P-066-97, UMTG-202",
    doi = "10.1016/S0550-3213(98)00400-3",
    journal = "Nucl. Phys. B",
    volume = "529",
    pages = "689--736",
    year = "1998"
}

@article{Chen:2021xks,
    author = "Chen, Yu-An and Hsin, Po-Shen",
    title = "{Exactly solvable lattice Hamiltonians and gravitational anomalies}",
    eprint = "2110.14644",
    archivePrefix = "arXiv",
    primaryClass = "cond-mat.str-el",
    doi = "10.21468/SciPostPhys.14.5.089",
    journal = "SciPost Phys.",
    volume = "14",
    number = "5",
    pages = "089",
    year = "2023"
}

@article{baez19982,
  title={2-Tangles},
  author={Baez, John C and Langford, Laurel},
  journal={Letters in Mathematical Physics},
  volume={43},
  pages={187--197},
  year={1998},
  publisher={Springer}
}

@article{Johnson-Freyd:2013oea,
    author = "Johnson-Freyd, Theo",
    editor = "Donagi, Ron and Douglas, Michael R. and Kamenova, Ljudmila and Rocek, Martin",
    title = "{Poisson AKSZ theories and their quantizations}",
    eprint = "1307.5812",
    archivePrefix = "arXiv",
    primaryClass = "math-ph",
    doi = "10.1090/pspum/088/01465",
    journal = "Proc. Symp. Pure Math.",
    volume = "88",
    pages = "291--306",
    year = "2014"
}

@article{Song_2023,
	doi = {10.1007/jhep07(2023)207},
  
	url = {https://doi.org/10.1007%2Fjhep07%282023%29207},
  
	year = 2023,
	month = {jul},
  
	publisher = {Springer Science and Business Media {LLC}
},
  
	volume = {2023},
  
	number = {7},
  
	author = {Danhua Song and Mengyao Wu and Ke Wu and Jie Yang},
  
	title = {Higher Chern-Simons based on (2-)crossed modules},
  
	journal = {Journal of High Energy Physics}
}

@article{Kong:2022hjj,
    author = "Kong, Liang and Zheng, Hao",
    title = "{Categories of quantum liquids III}",
    eprint = "2201.05726",
    archivePrefix = "arXiv",
    primaryClass = "hep-th",
    month = "1",
    year = "2022"
}

@article{Baez1996HigherDimensionalAI,
  title={Higher-Dimensional Algebra II. 2-Hilbert Spaces},
  author={John C. Baez},
  journal={Advances in Mathematics},
  year={1996},
  volume={127},
  pages={125-189},
  url={https://api.semanticscholar.org/CorpusID:2792589}
}

@ARTICLE{Ganter:2014,
       author = {{Ganter}, Nora and {Usher}, Robert},
        title = "{Representation and character theory of finite categorical groups}",
        volume = {31}, 
        year = {2016}, 
        number = {21}, 
    pages = {542-570},
      journal = {Theory and Applications of Categories},
        month = jul,
          eid = {arXiv:1407.6849},
          doi = {10.48550/arXiv.1407.6849},
archivePrefix = {arXiv},
       eprint = {1407.6849},
 primaryClass = {math.CT},
       adsurl = {http://www.tac.mta.ca/tac/volumes/31/21/31-21abs.html}
}

@ARTICLE{Ganter:2006,
       author = {{Ganter}, Nora and {Kapranov}, Mikhail},
        title = "{Representation and character theory in 2-categories}",
      journal = {arXiv Mathematics e-prints},
         year = 2006,
        month = feb,
          eid = {math/0602510},
        pages = {math/0602510},
          doi = {10.48550/arXiv.math/0602510},
archivePrefix = {arXiv},
       eprint = {math/0602510},
 primaryClass = {math.KT},
       adsurl = {https://ui.adsabs.harvard.edu/abs/2006math......2510G}
}

@article{Carqueville:2023aak,
    author = {Carqueville, Nils and M\"uller, Lukas},
    title = "{Orbifold completion of 3-categories}",
    eprint = "2307.06485",
    archivePrefix = "arXiv",
    primaryClass = "math.QA",
    month = "7",
    year = "2023"
}

@article{Carqueville:2017aoe,
    author = "Carqueville, Nils and Runkel, Ingo and Schaumann, Gregor",
    title = "{Orbifolds of n-dimensional defect TQFTs}",
    eprint = "1705.06085",
    archivePrefix = "arXiv",
    primaryClass = "math.QA",
    doi = "10.2140/gt.2019.23.781",
    journal = "Geom. Topol.",
    volume = "23",
    pages = "781--864",
    year = "2019"
}

@article{Carqueville:2016kdq,
    author = "Carqueville, Nils and Meusburger, Catherine and Schaumann, Gregor",
    title = "{3-dimensional defect TQFTs and their tricategories}",
    eprint = "1603.01171",
    archivePrefix = "arXiv",
    primaryClass = "math.QA",
    doi = "10.1016/j.aim.2020.107024",
    journal = "Adv. Math.",
    volume = "364",
    pages = "107024",
    year = "2020"
}

@article{Bartsch:2022mpm,
    author = "Bartsch, Thomas and Bullimore, Mathew and Ferrari, Andrea E. V. and Pearson, Jamie",
    title = "{Non-invertible symmetries and higher representation theory I}",
    eprint = "2208.05993",
    archivePrefix = "arXiv",
    primaryClass = "hep-th",
    doi = "10.21468/SciPostPhys.17.1.015",
    journal = "SciPost Phys.",
    volume = "17",
    number = "1",
    pages = "015",
    year = "2024"
}

@book{costello_gwilliam_2016, place={Cambridge}, series={New Mathematical Monographs}, title={Factorization Algebras in Quantum Field Theory}, volume={1}, DOI={10.1017/9781316678626}, publisher={Cambridge University Press}, author={Costello, Kevin and Gwilliam, Owen}, year={2016}, collection={New Mathematical Monographs}}

@article{Porst2008Strict2A,
  title={Strict 2-Groups are Crossed Modules},
  author={Sven-S. Porst},
  journal={arXiv: Category Theory},
  year={2008}
}

@article{Wockel2008Principal2A,
url = {https://doi.org/10.1515/form.2011.020},
title = {Principal 2-bundles and their gauge 2-groups},
title = {},
author = {Christoph Wockel},
pages = {565--610},
volume = {23},
number = {3},
journal = {Forum Mathematicum},
doi = {doi:10.1515/form.2011.020},
year = {2011},
lastchecked = {2023-08-29}
}

@article{Garner:2023zqn,
    author = "Garner, Niklas and Williams, Brian R.",
    title = "{Raviolo vertex algebras}",
    eprint = "2308.04414",
    archivePrefix = "arXiv",
    primaryClass = "math.QA",
    month = "8",
    year = "2023"
}

@article{Gaiotto:2014kfa,
    author = "Gaiotto, Davide and Kapustin, Anton and Seiberg, Nathan and Willett, Brian",
    title = "{Generalized Global Symmetries}",
    eprint = "1412.5148",
    archivePrefix = "arXiv",
    primaryClass = "hep-th",
    doi = "10.1007/JHEP02(2015)172",
    journal = "JHEP",
    volume = "02",
    pages = "172",
    year = "2015"
}

@article{Chen:2023integrable,
    author = "Chen, Hank and Girelli, Florian",
    title = "{Integrability from categorification and the 2-Kac-Moody Algebra}",
    eprint = "2307.03831",
    archivePrefix = "arXiv",
    primaryClass = "math-ph",
    month = "7",
    year = "2023"
}

@article{Turaev:1992,
author = {TURAEV, VLADIMIR G.},
title = {MODULAR CATEGORIES AND 3-MANIFOLD INVARIANTS},
journal = {International Journal of Modern Physics B},
volume = {06},
number = {11n12},
pages = {1807-1824},
year = {1992},
doi = {10.1142/S0217979292000876},

URL = { 
    
        https://doi.org/10.1142/S0217979292000876
    
    

},
eprint = { 
    
        https://doi.org/10.1142/S0217979292000876
    
    

}
}

@article{Elias2010ADT,
  title={A Diagrammatic Temperley-Lieb Categorification},
  author={Ben Elias},
  journal={Int. J. Math. Math. Sci.},
  year={2010},
  volume={2010},
  pages={530808:1-530808:47}
}

@article{decoppet2022morita,
   title={The Morita Theory of Fusion 2-Categories},
   volume={7},
   url={http://dx.doi.org/10.21136/HS.2023.07},
   DOI={10.21136/hs.2023.07},
   number={1},
   journal={Higher Structures},
   publisher={Institute of Mathematics, Czech Academy of Sciences},
   author={Décoppet, Thibault D.},
   year={2023},
   month={May},
   pages={234–292} }

@book{neuchl1997representation,
  title={Representation Theory of Hopf Categories},
  author={Neuchl, M.},
  url={https://books.google.ca/books?id=gpgLHAAACAAJ},
  year={1997},
  publisher={Verlag nicht ermittelbar}
}

@inbook{Kapranov:1994,
	title = {2-categories and Zamolodchikov tetrahedra equations},
	booktitle = {Algebraic groups and their generalizations: quantum and infinite-dimensional methods (University Park, PA, 1991)},
	series = {Proc. Sympos. Pure Math.},
	volume = {56},
	year = {1994},
	pages = {177{\textendash}259},
	publisher = {Amer. Math. Soc., Providence, RI},
	organization = {Amer. Math. Soc., Providence, RI},
	author = {Kapranov, M. M. and Voevodsky, V. A.}
}

@article{Khovanov:2000,
author = {Mikhail Khovanov},
title = {{A categorification of the Jones polynomial}},
volume = {101},
journal = {Duke Mathematical Journal},
number = {3},
publisher = {Duke University Press},
pages = {359 -- 426},
year = {2000},
doi = {10.1215/S0012-7094-00-10131-7},
URL = {https://doi.org/10.1215/S0012-7094-00-10131-7}
}

@article{Khovanov:2006,
 ISSN = {00029947},
 URL = {http://www.jstor.org/stable/3845459},
 author = {Mikhail Khovanov},
 journal = {Transactions of the American Mathematical Society},
 number = {1},
 pages = {315--327},
 publisher = {American Mathematical Society},
 title = {An Invariant of Tangle Cobordisms},
 urldate = {2025-04-18},
 volume = {358},
 year = {2006}
}

@article{Wan:2014woa,
    author = "Wan, Yidun and Wang, Juven C. and He, Huan",
    title = "{Twisted Gauge Theory Model of Topological Phases in Three Dimensions}",
    eprint = "1409.3216",
    archivePrefix = "arXiv",
    primaryClass = "cond-mat.str-el",
    doi = "10.1103/PhysRevB.92.045101",
    journal = "Phys. Rev. B",
    volume = "92",
    number = "4",
    pages = "045101",
    year = "2015"
}

@article{Else:2017yqj,
    author = "Else, Dominic V. and Nayak, Chetan",
    title = "{Cheshire charge in (3+1)-dimensional topological phases}",
    eprint = "1702.02148",
    archivePrefix = "arXiv",
    primaryClass = "cond-mat.str-el",
    doi = "10.1103/PhysRevB.96.045136",
    journal = "Phys. Rev. B",
    volume = "96",
    number = "4",
    pages = "045136",
    year = "2017"
}

@article{Kong:2020wmn,
    author = "Kong, Liang and Tian, Yin and Zhang, Zhi-Hao",
    title = "{Defects in the 3-dimensional toric code model form a braided fusion 2-category}",
    eprint = "2009.06564",
    archivePrefix = "arXiv",
    primaryClass = "cond-mat.str-el",
    doi = "10.1007/JHEP12(2020)078",
    journal = "JHEP",
    volume = "12",
    pages = "078",
    year = "2020"
}

@article{Johnson-Freyd:2020usu,
    author = "Johnson-Freyd, Theo",
    title = "{On the Classification of Topological Orders}",
    eprint = "2003.06663",
    archivePrefix = "arXiv",
    primaryClass = "math.CT",
    doi = "10.1007/s00220-022-04380-3",
    journal = "Commun. Math. Phys.",
    volume = "393",
    number = "2",
    pages = "989--1033",
    year = "2022"
}

@book{atiyah2018k,
  title={K-theory},
  author={Atiyah, M.},
  isbn={9780429973178},
  url={https://books.google.ca/books?id=f1NPDwAAQBAJ},
  year={2018},
  publisher={CRC Press}
}

@article{Chen2z:2023,
    author = "Chen, Hank",
    title = "{Drinfeld double symmetry of the 4d Kitaev model}",
    eprint = "2305.04729",
    archivePrefix = "arXiv",
    primaryClass = "cond-mat.str-el",
    doi = "10.1007/JHEP09(2023)141",
    journal = "JHEP",
    volume = "09",
    pages = "141",
    year = "2023"
}

@article{Chen:2023tjf,
    author = "Chen, Hank and Girelli, Florian",
    title = "{Categorified Quantum Groups and Braided Monoidal 2-Categories}",
    eprint = "2304.07398",
    archivePrefix = "arXiv",
    primaryClass = "math.QA",
    month = "4",
    year = "2023"
}

@article{Chen:2022hct,
    author = "Chen, Hank and Girelli, Florian",
    title = "{Gauging the Gauge and Anomaly Resolution}",
    eprint = "2211.08549",
    archivePrefix = "arXiv",
    primaryClass = "hep-th",
    month = "11",
    year = "2022"
}

@article{Johnson_Freyd_2023,
   title={Minimal nondegenerate extensions},
   ISSN={1088-6834},
   url={http://dx.doi.org/10.1090/jams/1023},
   DOI={10.1090/jams/1023},
   journal={Journal of the American Mathematical Society},
   publisher={American Mathematical Society (AMS)},
   author={Johnson-Freyd, Theo and Reutter, David},
   year={2023},
   month={Jul} }

@book{etingof2016tensor,
  title={Tensor Categories},
  author={Etingof, P. and Gelaki, S. and Nikshych, D. and Ostrik, V.},
  isbn={9781470434410},
  lccn={2015006773},
  series={Mathematical Surveys and Monographs},
  url={https://books.google.ca/books?id=Z6XLDAAAQBAJ},
  year={2016},
  publisher={American Mathematical Society}
}

@article{Woronowicz1988,
author = {Woronowicz, S.L.},
journal = {Inventiones mathematicae},
number = {1},
pages = {35-76},
title = {Tannaka-Krein duality for compact matrix pseudogroups. Twisted SU (N) groups.},
url = {http://eudml.org/doc/143589},
volume = {93},
year = {1988},
}

@article{WITTEN1990285,
title = {Gauge theories, vertex models, and quantum groups},
journal = {Nuclear Physics B},
volume = {330},
number = {2},
pages = {285-346},
year = {1990},
issn = {0550-3213},
doi = {https://doi.org/10.1016/0550-3213(90)90115-T},
url = {https://www.sciencedirect.com/science/article/pii/055032139090115T},
author = {Edward Witten}
}

@inproceedings{Rouquier2005CategorificationOS,
  title={Categorification of sl 2 and braid groups},
  author={Raphael Rouquier},
  year={2005}
}

@unpublished{rouquier:hal-00002981,
  TITLE = {{Categorification of the braid groups}},
  AUTHOR = {Rouquier, Raphael},
  URL = {https://hal.science/hal-00002981},
  NOTE = {working paper or preprint},
  YEAR = {2004},
  MONTH = Sep,
  HAL_ID = {hal-00002981},
  HAL_VERSION = {v1},
}

@article{Reshetikhin:1991tc,
    author = "Reshetikhin, N. and Turaev, V. G.",
    title = "{Invariants of three manifolds via link polynomials and quantum groups}",
    doi = "10.1007/BF01239527",
    journal = "Invent. Math.",
    volume = "103",
    pages = "547--597",
    year = "1991"
}

@article{Delcamp:2023kew,
    author = "Delcamp, Clement and Tiwari, Apoorv",
    title = "{Higher categorical symmetries and gauging in two-dimensional spin systems}",
    eprint = "2301.01259",
    archivePrefix = "arXiv",
    primaryClass = "hep-th",
    doi = "10.21468/SciPostPhys.16.4.110",
    journal = "SciPost Phys.",
    volume = "16",
    number = "4",
    pages = "110",
    year = "2024"
}

@book{book-quasihopf,
title = "Quasi-Hopf Algebras: A Categorical Approach",
author = "Stefaan Caenepeel and Daniel Bulacu and Florin Panaite and {Van Oystaeyen}, Freddy",
year = "2019",
language = "English",
isbn = "978-1-108-42701-2",
series = "Encyclopedia of Mathematics and its Applications",
publisher = "Cambridge University Press",
}

@article{KongTianZhou:2020,
title = {The center of monoidal 2-categories in 3+1D Dijkgraaf-Witten theory},
journal = {Advances in Mathematics},
volume = {360},
pages = {106928},
year = {2020},
issn = {0001-8708},
doi = {https://doi.org/10.1016/j.aim.2019.106928},
url = {https://www.sciencedirect.com/science/article/pii/S0001870819305432},
author = {Liang Kong and Yin Tian and Shan Zhou}
}

@article{Wen:2019,
  title = {Classification of $3+1\mathrm{D}$ Bosonic Topological Orders (II): The Case When Some Pointlike Excitations Are Fermions},
  author = {Lan, Tian and Wen, Xiao-Gang},
  journal = {Phys. Rev. X},
  volume = {9},
  issue = {2},
  pages = {021005},
  numpages = {37},
  year = {2019},
  month = {Apr},
  publisher = {American Physical Society},
  doi = {10.1103/PhysRevX.9.021005},
  url = {https://link.aps.org/doi/10.1103/PhysRevX.9.021005}
}

@article{Sati:2009ic,
    author = "Sati, Hisham and Schreiber, Urs and Stasheff, Jim",
    title = "{Differential twisted String and Fivebrane structures}",
    eprint = "0910.4001",
    archivePrefix = "arXiv",
    primaryClass = "math.AT",
    doi = "10.1007/s00220-012-1510-3",
    journal = "Commun. Math. Phys.",
    volume = "315",
    pages = "169--213",
    year = "2012"
}

@article{Gaiotto:2019xmp,
    author = "Gaiotto, Davide and Johnson-Freyd, Theo",
    title = "{Condensations in higher categories}",
    eprint = "1905.09566",
    archivePrefix = "arXiv",
    primaryClass = "math.CT",
    month = "5",
    year = "2019"
}

@article{Ciambelli:2021nmv,
    author = "Ciambelli, Luca and Leigh, Robert G. and Pai, Pin-Chun",
    title = "{Embeddings and Integrable Charges for Extended Corner Symmetry}",
    eprint = "2111.13181",
    archivePrefix = "arXiv",
    primaryClass = "hep-th",
    doi = "10.1103/PhysRevLett.128.171302",
    month = "11",
    year = "2021"
}

@article{Freidel:2020xyx,
    author = "Freidel, Laurent and Geiller, Marc and Pranzetti, Daniele",
    title = "{Edge modes of gravity. Part I. Corner potentials and charges}",
    eprint = "2006.12527",
    archivePrefix = "arXiv",
    primaryClass = "hep-th",
    doi = "10.1007/JHEP11(2020)026",
    journal = "JHEP",
    volume = "11",
    pages = "026",
    year = "2020"
}

@article{Freidel:2020svx,
    author = "Freidel, Laurent and Geiller, Marc and Pranzetti, Daniele",
    title = "{Edge modes of gravity. Part II. Corner metric and Lorentz charges}",
    eprint = "2007.03563",
    archivePrefix = "arXiv",
    primaryClass = "hep-th",
    doi = "10.1007/JHEP11(2020)027",
    journal = "JHEP",
    volume = "11",
    pages = "027",
    year = "2020"
}

@article{Bochniak:2020vil,
    author = "Bochniak, Arkadiusz and Hadasz, Leszek and La\.zej Ruba, B.",
    title = "{Dynamical generalization of Yetter's model based on a crossed module of discrete groups}",
    eprint = "2010.00888",
    archivePrefix = "arXiv",
    primaryClass = "math-ph",
    doi = "10.1007/JHEP03(2021)282",
    journal = "JHEP",
    volume = "03",
    pages = "282",
    year = "2021"
}

@article{Willerton:2008gyk,
    author = "Willerton, Simon",
    title = "{The twisted Drinfeld double of a finite group via gerbes and finite groupoids}",
    doi = "10.2140/agt.2008.8.1419",
    journal = "Algebr. Geom. Topol.",
    volume = "8",
    number = "3",
    pages = "1419--1457",
    year = "2008"
}

@article{Ritter:2016,
	doi = {10.1142/s0129055x16500215},
  
	url = {https://doi.org/10.1142},
  
	year = 2016,
	month = {oct},
  
	publisher = {World Scientific Pub Co Pte Lt},
  
	volume = {28},
  
	number = {09},
  
	pages = {1650021},
  
	author = {Patricia Ritter and Christian Sämann},
  
	title = {L$\infty$-algebra models and higher Chern{\textendash}Simons theories},
  
	journal = {Reviews in Mathematical Physics}
}

@article{Douglas:2018,
  title={Fusion 2-categories and a state-sum invariant for 4-manifolds},
  author={Christopher L. Douglas and David J. Reutter},
  journal={arXiv: Quantum Algebra},
  year={2018},
  url={https://api.semanticscholar.org/CorpusID:119305837}
}

@Inbook{Kapustin2017,
author="Kapustin, Anton
and Thorngren, Ryan",
editor="Auroux, Denis
and Katzarkov, Ludmil
and Pantev, Tony
and Soibelman, Yan
and Tschinkel, Yuri",
title="Higher Symmetry and Gapped Phases of Gauge Theories",
bookTitle="Algebra, Geometry, and Physics in the 21st Century: Kontsevich Festschrift",
year="2017",
publisher="Springer International Publishing",
address="Cham",
pages="177--202",
isbn="978-3-319-59939-7",
doi="10.1007/978-3-319-59939-7{\_}5",
url="https://doi.org/10.1007/978-3-319-59939-7{\_}5"
}

@article{Baez:2012,
	doi = {10.1090/s0065-9266-2012-00652-6},
  
	url = {https://doi.org/10.1090%2Fs0065-9266-2012-00652-6},
  
	year = 2012,
	publisher = {American Mathematical Society ({AMS})},
  
	volume = {219},
  
	number = {1032},
  
	author = {John Baez and Aristide Baratin and Laurent Freidel and Derek Wise},
  
	title = {Infinite-Dimensional Representations of 2-Groups},
  
	journal = {Memoirs of the American Mathematical Society}
}

@article{Baez:2004,
author = {Baez, John C. and Lauda, Aaron D.},
journal = {Theory and Applications of Categories [electronic only]},
language = {eng},
pages = {423-491},
publisher = {Mount Allison University, Department of Mathematics and Computer Science, Sackville},
title = {Higher-dimensional algebra. V: 2-Groups.},
url = {http://eudml.org/doc/124217},
volume = {12},
year = {2004},
}

@inproceedings{Angulo2024TheVE,
  title={The van Est homomorphism for strict Lie 2-groups},
  author={Camilo Angulo and Miquel Cueca},
  year={2024},
  url={https://api.semanticscholar.org/CorpusID:269790865}
}

@article{walker2012,
  title={(3+ 1)-TQFTs and topological insulators},
  author={Walker, Kevin and Wang, Zhenghan},
  journal={Frontiers of Physics},
  volume={7},
  number={2},
  pages={150--159},
  year={2012},
  publisher={Springer},
    url = {https://arxiv.org/abs/1104.2632},
    doi = {10.48550/ARXIV.1104.2632}
}

@article{Carey_1997,
	doi = {10.1016/s0393-0440(96)00014-9},
  
	url = {https://doi.org/10.1016%2Fs0393-0440%2896%2900014-9},
  
	year = 1997,
	month = {jan},
  
	publisher = {Elsevier {BV}
},
  
	volume = {21},
  
	number = {2},
  
	pages = {183--197},
  
	author = {A.L. Carey and M.K. Murray and B.L. Wang},
  
	title = {Higher bundle gerbes and cohomology classes in gauge theories},
  
	journal = {Journal of Geometry and Physics}
}

@article{Johnson-Freyd:2020,
    author = "Johnson-Freyd, Theo",
    title = "{(3+1)D topological orders with only a $\mathbb{Z}_2$-charged particle}",
    eprint = "2011.11165",
    archivePrefix = "arXiv",
    primaryClass = "math.QA",
    month = "11",
    year = "2020"
}

@article{KitaevKong_2012,
	doi = {10.1007/s00220-012-1500-5},
  
	url = {https://doi.org/10.1007%2Fs00220-012-1500-5},
  
	year = 2012,
	month = {jun},
  
	publisher = {Springer Science and Business Media {LLC}
},
  
	volume = {313},
  
	number = {2},
  
	pages = {351--373},
  
	author = {Alexei Kitaev and Liang Kong},
  
	title = {Models for Gapped Boundaries and Domain Walls},
  
	journal = {Communications in Mathematical Physics}
}

@misc{Baez:2002highergauge,
  doi = {10.48550/ARXIV.HEP-TH/0206130},
  
  url = {https://arxiv.org/abs/hep-th/0206130},
  
  author = {Baez, John C.},
  
  title = {Higher Yang-Mills Theory},
  
  publisher = {arXiv},
  
  year = {2002},
  
  copyright = {Assumed arXiv.org perpetual, non-exclusive license to distribute this article for submissions made before January 2004}
}

@article{Bochniak_2021,
	doi = {10.1007/jhep09(2021)068},
  
	url = {https://doi.org/10.1007%2Fjhep09%282021%29068},
  
	year = 2021,
	month = {sep},
  
	publisher = {Springer Science and Business Media {LLC}
},
  
	volume = {2021},
  
	number = {9},
  
	author = {A. Bochniak and L. Hadasz and P. Korcyl and B. Ruba},
  
	title = {Dynamics of a lattice 2-group gauge theory model},
  
	journal = {Journal of High Energy Physics}
}

@article{Kong:2020,
  title = {Algebraic higher symmetry and categorical symmetry: A holographic and entanglement view of symmetry},
  author = {Kong, Liang and Lan, Tian and Wen, Xiao-Gang and Zhang, Zhi-Hao and Zheng, Hao},
  journal = {Phys. Rev. Research},
  volume = {2},
  issue = {4},
  pages = {043086},
  numpages = {53},
  year = {2020},
  month = {Oct},
  publisher = {American Physical Society},
  doi = {10.1103/PhysRevResearch.2.043086},
  url = {https://link.aps.org/doi/10.1103/PhysRevResearch.2.043086}
}

@article{Song:2021,
    author = "Song, Danhua and Lou, Kai and Wu, Ke and Yang, Jie",
    title = "{Generalized higher connections and Yang-Mills}",
    eprint = "2112.13370",
    archivePrefix = "arXiv",
    primaryClass = "hep-th",
    month = "12",
    year = "2021"
}

@book{maclane:71,
  address = {New York},
  author = {MacLane, Saunders},
  mrclass = {18-02},
  mrnumber = {MR0354798 (50 \#7275)},
  mrreviewer = {H.-B. Brinkmann},
  note = {Graduate Texts in Mathematics, Vol. 5},
  pages = {ix+262},
  publisher = {Springer-Verlag},
  title = {Categories for the Working Mathematician},
  year = 1971
}

@misc{webster2013knot,
      title={Knot invariants and higher representation theory II: the categorification of quantum knot invariants}, 
      author={Ben Webster},
      year={2013},
      eprint={1005.4559},
      archivePrefix={arXiv},
      primaryClass={math.GT}
}

@article{Pretko:2020,
author = {Pretko, Michael and Chen, Xie and You, Yizhi},
title = {Fracton phases of matter},
journal = {International Journal of Modern Physics A},
volume = {35},
number = {06},
pages = {2030003},
year = {2020},
doi = {10.1142/S0217751X20300033},

URL = { 
        https://doi.org/10.1142/S0217751X20300033
    
},
eprint = { 
        https://doi.org/10.1142/S0217751X20300033
    
}
}

@article{Bullivant:2021,
   title={Gapped boundaries and string-like excitations in (3+1)d gauge models of topological phases},
   volume={2021},
   ISSN={1029-8479},
   url={http://dx.doi.org/10.1007/JHEP07(2021)025},
   DOI={10.1007/jhep07(2021)025},
   number={7},
   journal={Journal of High Energy Physics},
   publisher={Springer Science and Business Media LLC},
   author={Bullivant, Alex and Delcamp, Clement},
   year={2021},
   month={Jul} }

@book{book-charclass,
	author = {John Milnor and James D. Stacheff},
	publisher = {Princeton University Press},
	series = {Annals of Mathematics Studies},
	title = {{Characteristic Classes}},
	year = {1974},
	}

@article{Pfeiffer2007,
title = {2-Groups, trialgebras and their Hopf categories of representations},
journal = {Advances in Mathematics},
volume = {212},
number = {1},
pages = {62-108},
year = {2007},
issn = {0001-8708},
doi = {https://doi.org/10.1016/j.aim.2006.09.014},
url = {https://www.sciencedirect.com/science/article/pii/S0001870806003343},
author = {Hendryk Pfeiffer},
eprint = "0411468",
    archivePrefix = "arXiv",
    primaryClass = "math-ph",
}

@article{chen:2022,
  title = {Categorified Drinfel'd double and $BF$ theory: 2-groups in 4D},
  author = {Chen, Hank and Girelli, Florian},
  journal = {Phys. Rev. D},
  volume = {106},
  issue = {10},
  pages = {105017},
  numpages = {35},
  year = {2022},
  month = {Nov},
  publisher = {American Physical Society},
  doi = {10.1103/PhysRevD.106.105017},
  url = {https://link.aps.org/doi/10.1103/PhysRevD.106.105017}
}

@article{Mackaay:hc,
	author = {Marco Mackaay and Roger Picken},
	eprint = {math/0104285},
	title = {2-Categories, 4d state-sum models and gerbes},
	url = {https://arxiv.org/pdf/math/0104285.pdf}}

@article{Pachner1991Pachner,
    title = "{\textit{P.L. Homeomorphic Manifolds are Equivalent by Elementary Shellings}}",
    journal = {European Journal of Combinatorics},
    volume = {12},
    number = {2},
    pages = {129-145},
    year = {1991},
    issn = {0195-6698},
    doi = {https://doi.org/10.1016/S0195-6698(13)80080-7},
    url = {https://www.sciencedirect.com/science/article/pii/S0195669813800807},
    author = {Udo Pachner},
}

@article{Thorngren2015,
	author = {Thorngren, Ryan},
	doi = {10.1007/JHEP02(2015)152},
	journal = {Journal of High Energy Physics},
	number = {152},
	publisher = {Springer},
	title = {Framed Wilson operators, fermionic strings, and gravitational anomaly in 4d},
	url = {https://doi.org/10.1007/JHEP02(2015)152},
	volume = {2015},
	year = {2010}}

@article{Brown,
	author = {R. Brown},
	collection = {London Mathematical Society Lecture Note Series},
	doi = {10.1017/CBO9780511526305},
	journal = {London Mathematical Society Lecture Note Series},
	pages = {187--210},
	place = {Cambridge},
	publisher = {Cambridge University Press},
	title = {Computing Homotopy Types Using Crossed $N$-Cubes of Groups},
	volume = {1},
	year = {1992}}

@article{baez2004,
author = {Baez, John C. and Lauda, Aaron D.},
journal = {Theory and Applications of Categories [electronic only]},
language = {eng},
pages = {423-491},
publisher = {Mount Allison University, Department of Mathematics and Computer Science, Sackville},
title = {Higher-dimensional algebra. V: 2-Groups.},
url = {http://eudml.org/doc/124217},
volume = {12},
year = {2004},
}

@article{Baez:1995xq,
	archiveprefix = {arXiv},
	author = {Baez, J. C. and Dolan, J.},
	doi = {10.1063/1.531236},
	eprint = {q-alg/9503002},
	journal = {J. Math. Phys.},
	pages = {6073--6105},
	title = {{Higher dimensional algebra and topological quantum field theory}},
	volume = {36},
	year = {1995}}

@article{Baez:2004in,
	archiveprefix = {arXiv},
	author = {Baez, John and Schreiber, Urs},
	eprint = {hep-th/0412325},
	month = {12},
	title = {{Higher gauge theory: 2-connections on 2-bundles}},
	year = {2004}}

@article{Baez:2003fs,
	archiveprefix = {arXiv},
	author = {Baez, John C. and Crans, Alissa S.},
	eprint = {math/0307263},
	journal = {Theor. Appl. Categor.},
	pages = {492--528},
	title = {{Higher-Dimensional Algebra VI: Lie 2-Algebras}},
	volume = {12},
	year = {2004}}

@article{Mikovic:2015hza,
	archiveprefix = {arXiv},
	author = {Mikovi\'c, A. and Oliveira, M. A. and Vojinovi\'c, M.},
	doi = {10.1088/0264-9381/33/6/065007},
	eprint = {1508.05635},
	journal = {Class. Quant. Grav.},
	number = {6},
	pages = {065007},
	primaryclass = {gr-qc},
	title = {{Hamiltonian analysis of the BFCG theory for the Poincar\'e 2-group}},
	volume = {33},
	year = {2016}}

@article{Mikovic:2016xmo,
	archiveprefix = {arXiv},
	author = {Mikovic, Aleksandar and Oliveira, Miguel Angelo and Vojinovic, Marko},
	eprint = {1610.09621},
	month = {10},
	primaryclass = {math-ph},
	title = {{Hamiltonian analysis of the BFCG theory for a strict Lie 2-group}},
	year = {2016}}

@article{Mikovic:2011si,
	archiveprefix = {arXiv},
	author = {Mikovic, A. and Vojinovic, M.},
	doi = {10.1088/0264-9381/29/16/165003},
	eprint = {1110.4694},
	journal = {Class. Quant. Grav.},
	pages = {165003},
	primaryclass = {gr-qc},
	title = {{poincar{\'e} 2-group and quantum gravity}},
	volume = {29},
	year = {2012}}

@article{Drinfeld:1986in,
	author = {Drinfeld, V. G.},
	doi = {10.1007/BF01247086},
	journal = {Zap. Nauchn. Semin.},
	pages = {18--49},
	title = {{Quantum groups}},
	volume = {155},
	year = {1986}}

@article{Jimbo:1985zk,
    author = "Jimbo, Michio",
    title = "{A q difference analog of U(g) and the Yang-Baxter equation}",
    doi = "10.1007/BF00704588",
    journal = "Lett. Math. Phys.",
    volume = "10",
    pages = "63--69",
    year = "1985"
}

@inproceedings{Crane:1993if,
    author = "Crane, Louis and Yetter, David",
    title = "{A Categorical construction of 4-D topological quantum field theories}",
    eprint = "hep-th/9301062",
    archivePrefix = "arXiv",
    reportNumber = "PRINT-93-0299 (KANSAS-STATE)",
    month = "3",
    year = "1993"
}

@article{Crane:1994ty,
	archiveprefix = {arXiv},
	author = {Crane, Louis and Frenkel, Igor},
	doi = {10.1063/1.530746},
	eprint = {hep-th/9405183},
	journal = {J. Math. Phys.},
	pages = {5136--5154},
	title = {{Four-dimensional topological field theory, Hopf categories, and the canonical bases}},
	volume = {35},
	year = {1994}}

@article{Kim:2019owc,
	archiveprefix = {arXiv},
	author = {Kim, Hyungrok and Saemann, Christian},
	doi = {10.1088/1751-8121/ab8ef2},
	eprint = {1911.06390},
	journal = {J. Phys. A},
	number = {44},
	pages = {445206},
	primaryclass = {hep-th},
	reportnumber = {EMPG-19-24},
	title = {{Adjusted parallel transport for higher gauge theories}},
	volume = {53},
	year = {2020}}

@article{Benini:2018reh,
	archiveprefix = {arXiv},
	author = {Benini, Francesco and C\'ordova, Clay and Hsin, Po-Shen},
	doi = {10.1007/JHEP03(2019)118},
	eprint = {1803.09336},
	journal = {JHEP},
	pages = {118},
	primaryclass = {hep-th},
	reportnumber = {SISSA 10/2018/FISI, SISSA-10-2018-FISI},
	title = {{On 2-Group Global Symmetries and their Anomalies}},
	volume = {03},
	year = {2019}}

@article{Baez:1995ph,
	archiveprefix = {arXiv},
	author = {Baez, John C.},
	doi = {10.1007/BF00398315},
	eprint = {q-alg/9507006},
	journal = {Lett. Math. Phys.},
	pages = {129--143},
	title = {{Four-Dimensional BF theory with cosmological term as a topological quantum field theory}},
	volume = {38},
	year = {1996}}

@article{Martins:2010ry,
	archiveprefix = {arXiv},
	author = {Martins, Joao Faria and Mikovic, Aleksandar},
	doi = {10.4310/ATMP.2011.v15.n4.a4},
	eprint = {1006.0903},
	journal = {Adv. Theor. Math. Phys.},
	number = {4},
	pages = {1059--1084},
	primaryclass = {hep-th},
	title = {{Lie crossed modules and gauge-invariant actions for 2-BF theories}},
	volume = {15},
	year = {2011}}

@article{Chen:2012gz,
	archiveprefix = {arXiv},
	author = {Chen, Zhuo and Sti{\'e}non, Mathieu and Xu, Ping},
	doi = {10.4310/jdg/1367438648},
	eprint = {1202.0079},
	journal = {J. Diff. Geom.},
	number = {2},
	pages = {209--240},
	primaryclass = {math.DG},
	title = {{Poisson 2-groups}},
	volume = {94},
	year = {2013}}

@article{Cordova:2018cvg,
	archiveprefix = {arXiv},
	author = {C\'ordova, Clay and Dumitrescu, Thomas T. and Intriligator, Kenneth},
	doi = {10.1007/JHEP02(2019)184},
	eprint = {1802.04790},
	journal = {JHEP},
	pages = {184},
	primaryclass = {hep-th},
	title = {{Exploring 2-Group Global Symmetries}},
	volume = {02},
	year = {2019}}

@article{Pretko:2017fbf,
	archiveprefix = {arXiv},
	author = {Pretko, Michael},
	doi = {10.1103/PhysRevD.96.024051},
	eprint = {1702.07613},
	journal = {Phys. Rev. D},
	number = {2},
	pages = {024051},
	primaryclass = {cond-mat.str-el},
	title = {{Emergent gravity of fractons: Mach\textquoteright{}s principle revisited}},
	volume = {96},
	year = {2017}}

@article{Dubinkin:2020kxo,
	archiveprefix = {arXiv},
	author = {Dubinkin, Oleg and Rasmussen, Alex and Hughes, Taylor L.},
	doi = {10.1016/j.aop.2020.168297},
	eprint = {2007.05539},
	journal = {Annals Phys.},
	pages = {168297},
	primaryclass = {cond-mat.str-el},
	title = {{Higher-form Gauge Symmetries in Multipole Topological Phases}},
	volume = {422},
	year = {2020}}

@Inbook{Kapustin:2013uxa,
author="Kapustin, Anton
and Thorngren, Ryan",
editor="Auroux, Denis
and Katzarkov, Ludmil
and Pantev, Tony
and Soibelman, Yan
and Tschinkel, Yuri",
title="Higher Symmetry and Gapped Phases of Gauge Theories",
bookTitle="Algebra, Geometry, and Physics in the 21st Century: Kontsevich Festschrift",
year="2017",
publisher="Springer International Publishing",
address="Cham",
pages="177--202",
isbn="978-3-319-59939-7",
doi="10.1007/978-3-319-59939-7_5",
url="https://doi.org/10.1007/978-3-319-59939-7_5"
}

@article{Baez:2005sn,
author = {John C. Baez and Danny Stevenson and Alissa S. Crans and Urs Schreiber},
title = {{From loop groups to 2-groups}},
volume = {9},
journal = {Homology, Homotopy and Applications},
number = {2},
publisher = {International Press of Boston},
pages = {101 -- 135},
year = {2007},
}

@phdthesis{Delcamp:2018kqc,
	author = {Delcamp, Clement},
	school = {Waterloo U.},
	title = {{Gauge Models of Topological Phases and Applications to Quantum Gravity}},
	year = {2018}}

@article{Martins:2006hx,
	archiveprefix = {arXiv},
	author = {Martins, Joao Faria and Porter, Timothy},
	eprint = {math/0608484},
	journal = {Theor. Appl. Categor.},
	pages = {118--150},
	title = {{On Yetter's invariant and an extension of the Dijkgraaf-Witten invariant to categorical groups}},
	volume = {18},
	year = {2007}}

@article{Yetter:1993dh,
	author = {Yetter, D. N.},
	doi = {10.1142/S0218216593000076},
	journal = {J. Knot Theor. Ramifications},
	pages = {113--123},
	title = {{TQFT's from homotopy 2 types}},
	volume = {2},
	year = {1993}}

@article{Bai_2013,
	author = {Chengming Bai and Yunhe Sheng and Chenchang Zhu},
	doi = {10.1007/s00220-013-1712-3},
	journal = {Communications in Mathematical Physics},
	month = {apr},
	number = {1},
	pages = {149--172},
	publisher = {Springer Science and Business Media {LLC}},
	title = {Lie 2-Bialgebras},
	url = {https://doi.org/10.1007%2Fs00220-013-1712-3},
	volume = {320},
	year = 2013}

@article{Girelli:2021khh,
	archiveprefix = {arXiv},
	author = {Girelli, Florian and Laudonio, Matteo and Tsimiklis, Panagiotis},
	eprint = {2105.10616},
	month = {5},
	primaryclass = {hep-th},
	title = {{Polyhedron phase space using 2-groups: $\kappa$-Poincar\'e as a Poisson 2-group}},
	year = {2021}}

@article{Alekseev:1994pa,
	archiveprefix = {arXiv},
	author = {Alekseev, Anton {\relax Yu}. and Grosse, Harald and Schomerus, Volker},
	doi = {10.1007/BF02099431},
	eprint = {hep-th/9403066},
	journal = {Commun. Math. Phys.},
	pages = {317-358},
	primaryclass = {hep-th},
	reportnumber = {HUTMP-94-B336, ESI-79-1994, UUITP-5-94, UWTHPH-1994-8},
	slaccitation = {%%CITATION = HEP-TH/9403066;%%},
	title = {{Combinatorial quantization of the Hamiltonian Chern-Simons theory}},
	volume = {172},
	year = {1995}}

@article{Alekseev:1994au,
	archiveprefix = {arXiv},
	author = {Alekseev, Anton {\relax Yu}. and Grosse, Harald and Schomerus, Volker},
	doi = {10.1007/BF02101528},
	eprint = {hep-th/9408097},
	journal = {Commun. Math. Phys.},
	pages = {561-604},
	primaryclass = {hep-th},
	reportnumber = {HUTMP-94-B-337, ESI-113-1994, UUITP-11-94A, UWTHPH-1994-26},
	slaccitation = {%%CITATION = HEP-TH/9408097;%%},
	title = {{Combinatorial quantization of the Hamiltonian Chern-Simons theory. 2.}},
	volume = {174},
	year = {1995}}

@article{Fock:1998nu,
	archiveprefix = {arXiv},
	author = {Fock, V. V. and Rosly, A. A.},
	eprint = {math/9802054},
	journal = {Am. Math. Soc. Transl.},
	pages = {67-86},
	primaryclass = {math-qa},
	slaccitation = {%%CITATION = MATH/9802054;%%},
	title = {{Poisson structure on moduli of flat connections on Riemann surfaces and r matrix}},
	volume = {191},
	year = {1999}}

@article{Bullivant:2016clk,
	archiveprefix = {arXiv},
	author = {Bullivant, A. and Calcada, M. and Kadar, Z. and Martin, P. and Martins, J.},
	doi = {10.1103/PhysRevB.95.155118},
	eprint = {1606.06639},
	journal = {Phys. Rev.},
	number = {15},
	pages = {155118},
	primaryclass = {cond-mat.str-el},
	slaccitation = {%%CITATION = ARXIV:1606.06639;%%},
	title = {Topological phases from higher gauge symmetry in 3+1 dimensions},
	volume = {B95},
	year = {2017}}

@article{Wang:2016rzy,
	archiveprefix = {arXiv},
	author = {Wang, Zitao and Chen, Xie},
	doi = {10.1103/PhysRevB.95.115142},
	eprint = {1611.09334},
	journal = {Phys. Rev.},
	number = {11},
	pages = {115142},
	primaryclass = {cond-mat.str-el},
	slaccitation = {%%CITATION = ARXIV:1611.09334;%%},
	title = {{Twisted gauge theories in three-dimensional Walker-Wang models}},
	volume = {B95},
	year = {2017}}

@article{Delcamp:2017pcw,
	archiveprefix = {arXiv},
	author = {Delcamp, Clement},
	doi = {10.1007/JHEP12(2017)128},
	eprint = {1709.04924},
	journal = {JHEP},
	pages = {128},
	primaryclass = {hep-th},
	slaccitation = {%%CITATION = ARXIV:1709.04924;%%},
	title = {{Excitation basis for (3+1)d topological phases}},
	volume = {12},
	year = {2017}}

@article{Semenov1992,
	author = {Semenov-Tyan-Shanskii, M. A.},
	day = {01},
	doi = {10.1007/BF01083527},
	issn = {1573-9333},
	journal = {Theoretical and Mathematical Physics},
	month = {Nov},
	note = {\url{https://doi.org/10.1007/BF01083527}},
	number = {2},
	pages = {1292--1307},
	title = {{Poisson-Lie groups. The quantum duality principle and the twisted quantum double}},
	volume = {93},
	year = {1992}}

\end{document}